\documentclass{amsart} 

\usepackage[english]{babel} 
\usepackage{accents} 

\makeatletter
\ifcsname c@insertedfiles\endcsname
\else 
  \newcounter{insertedfiles}
\fi
\makeatother

\providecommand{\noinsert}[1]{ 
	\ifnum\value{insertedfiles} = 0
	    #1 
	  \else
 	 \fi
}

\ifdefined\macroalreadyinput 
\else 
	\newcommand{\macroalreadyinput} 

\usepackage{adjustbox} 
\usepackage{amsfonts}
\usepackage{amsmath}
\usepackage{amssymb}
\usepackage{amsthm}
\usepackage{array}
\usepackage[utf8]{inputenc}
\usepackage[T1]{fontenc}
\usepackage{color}
\usepackage{comment} 
\usepackage{csquotes}
\usepackage[inline]{enumitem} 

\usepackage{fancybox}
\usepackage{fancyhdr}
\usepackage{geometry}
\usepackage{hhline}
\usepackage{glossaries} 
\usepackage{graphicx}
\usepackage{hyperref}
\hypersetup{
	colorlinks=true,
	allcolors=black,
}
\usepackage{iflang}
\usepackage{ifthen}
\usepackage{lastpage}
\usepackage{mathrsfs}
\usepackage{mathtools}
\usepackage{multicol}
\usepackage{multirow} 
\usepackage[normalem]{ulem} 
\usepackage{tikz-cd}
\usepackage{xspace}
\usepackage{xstring}

\newsavebox{\boxcontainer}

\newtheorem{thm}{\IfLanguageName{english}{Theorem}{Théorème}}[subsection]
\newtheorem{defi}[thm]{\IfLanguageName{english}{Definition}{Définition}}
\newtheorem{prop}[thm]{Proposition}
\newtheorem{lem}[thm]{\IfLanguageName{english}{Lemma}{Lemme}}
\newtheorem{cor}[thm]{\IfLanguageName{english}{Corollary}{Corollaire}}

\newtheorem{exe}{\IfLanguageName{english}{Example}{Exemple}}[subsection]

\theoremstyle{definition} 
\newtheorem{nota}[thm]{Notation}
\newtheorem{remark}{Remark}[subsection]

\newenvironment{groupe}{{}{}} 

\DeclareMathOperator{\Sp}{Sp}

\DeclareMathOperator{\id}{id}

\DeclareMathOperator{\GL}{GL}

\newcounter{subsubsubsection} 

\newcommand{\bs}{\backslash}

\newcommand{\rr}{\mathbb{R}}
\newcommand{\qq}{\mathbb{Q}}

\newcommand{\nn}{\mathbb{N}}

\newcommand{\didots}{\setbox1=\hbox{\,.} \hbox{.\raise4pt\copy1\raise8pt\copy1}}

\newcommand{\quotient}[2]{{\raisebox{.2em}{$#1$}\left/\raisebox{-.2em}{$#2$}\right.}}
\newcommand{\blank}{\hbox{}}

\ProvideDocumentCommand{\command}{m o m}{
	\IfValueTF{#2}{
		\providecommand{#1}[#2]{#3} 
		\renewcommand{#1}[#2]{#3} 
	}{
		\providecommand{#1}{#3} 
		\renewcommand{#1}{#3} 
	}
}

\ProvideDocumentCommand{\DocumentCommand}{m m m}{
		\ProvideDocumentCommand{#1}{#2}{#3} 
		\RenewDocumentCommand{#1}{#2}{#3} 
}

\NewDocumentCommand{\dc}{o}{\IfValueTF{#1}{\left[\!\!\left[ #1 \right]\!\!\right]}{\left[\!\!\left[ \right]\!\!\right]}}
\newcommand{\ul}[1]{\uline{#1}}
\newcommand{\ol}[1]{\overline{#1}}
\newcommand{\ust}[2]{\underset{#2}{#1}}
\newcommand{\ovst}[2]{\overset{#2}{#1}}
\newcommand{\round}[1]{\mathring{\overbrace{{#1}}}}
\NewDocumentCommand{\ps}{O{\cdot} O{\cdot}}{\left\langle #1; #2 \right\rangle}
\NewDocumentCommand{\norm}{O{\cdot} o}{\IfValueTF{#2}{\left\Vert #1 \right\Vert_{#2}}{\left\Vert #1 \right\Vert}}
\NewDocumentCommand{\abs}{O{\cdot} o}{\IfValueTF{#2}{\left| #1 \right|_{#2}}{\left| #1 \right|}}

\definecolor{grey}{gray}{0.5} 
\newcommand{\blue}[1]{\colorlet{currentColor}{.}\color{blue}{#1}\color{currentColor}\xspace}

\newcommand{\red}[1]{\colorlet{currentColor}{.}\color{red}{#1}\color{currentColor}}
\newcommand{\white}[1]{\colorlet{currentColor}{.}\color{white}{#1}\color{currentColor}}
\let\bf\textbf
\let\it\textit

\newcommand{\sub}[1]{\texorpdfstring{\textsubscript{{#1}}}{}}

\let\intemp\in
\renewcommand{\in}{~\intemp~} 
\let\modelstemp\models
\renewcommand{\models}{~\modelstemp~} 
\let\captemp\cap
\renewcommand{\cap}{~\captemp~} 
\let\cuptemp\cup
\renewcommand{\cup}{~\cuptemp~} 

\NewDocumentCommand{\set}{o o o}{} 

\let\Vee\bigvee 
\let\Wedge\bigwedge 
\RenewDocumentCommand{\bigvee}{o o}{ 
	\IfValueTF{#1}{ 
		\IfValueTF{#2}{ 
			\ovst{\ust{\Vee}{#1}}{#2}
		}{
			\ust{\Vee}{#1}
		}
	}{
	\Vee
	}
}

\RenewDocumentCommand{\bigwedge}{o o}{
	\IfValueTF{#1}{ 
		\IfValueTF{#2}{ 
			\ovst{\ust{\Wedge}{#1}}{#2}
		}{
			\ust{\Wedge}{#1}
		}
	}{
	\Wedge
	}
}

\let\ssup\sup 
\let\iinf\inf 
\RenewDocumentCommand{\sup}{o o}{ 
	\IfValueTF{#1}{ 
		\IfValueTF{#2}{ 
			\ovst{\ust{\ssup}{#1}}{#2}
		}{
			\ust{\ssup}{#1}
		}
	}{
	\ssup
	}
	\, 
}

\RenewDocumentCommand{\inf}{o o}{
	\IfValueTF{#1}{ 
		\IfValueTF{#2}{ 
			\ovst{\ust{\iinf}{#1}}{#2}
		}{
			\ust{\iinf}{#1}
		}
	}{
	\iinf
	}
	\, 
}

\let\llim\lim 
\RenewDocumentCommand{\lim}{o o}{ 
	\IfValueTF{#1}{ 
		\IfValueTF{#2}{ 
			\ovst{\ust{\llim}{#1}}{#2}
		}{
			\ust{\llim}{#1}
		}
	}{
	\llim
	}
	\, 
}

\let\Prod\prod
\let\Sum\sum 
\RenewDocumentCommand{\prod}{o o}{ 
	\IfValueTF{#1}{ 
		\IfValueTF{#2}{ 
			\ovst{\ust{\Prod}{#1}}{#2}
		}{
			\ust{\Prod}{#1}
		}
	}{
	\Prod
	}
}

\RenewDocumentCommand{\sum}{o o}{ 
	\IfValueTF{#1}{ 
		\IfValueTF{#2}{ 
			\ovst{\ust{\Sum}{#1}}{#2}
		}{
			\ust{\Sum}{#1}
		}
	}{
	\Sum
	}
}

\let\Ccup\bigcup 
\let\Ccap\bigcap 
\RenewDocumentCommand{\bigcup}{o o}{ 
	\IfValueTF{#1}{ 
		\IfValueTF{#2}{ 
			\ovst{\ust{\Ccup}{#1}}{#2}
		}{
			\ust{\Ccup}{#1}
		}
	}{
	\Ccup
	}
}

\RenewDocumentCommand{\bigcap}{o o}{
	\IfValueTF{#1}{ 
		\IfValueTF{#2}{ 
			\ovst{\ust{\Ccap}{#1}}{#2}
		}{
			\ust{\Ccap}{#1}
		}
	}{
	\Ccap
	}
}

\let\Odot\bigodot
\let\Oplus\bigoplus
\let\Otimes\bigotimes
\RenewDocumentCommand{\bigodot}{o o}{ 
	\IfValueTF{#1}{ 
		\IfValueTF{#2}{ 
			\ovst{\ust{\Odot}{#1}}{#2}
		}{
			\ust{\Odot}{#1}
		}
	}{
	\Odot
	}
}

\RenewDocumentCommand{\bigoplus}{o o}{ 
	\IfValueTF{#1}{ 
		\IfValueTF{#2}{ 
			\ovst{\ust{\Oplus}{#1}}{#2}
		}{
			\ust{\Oplus}{#1}
		}
	}{
	\Oplus
	}
}

\RenewDocumentCommand{\bigotimes}{o o}{ 
	\IfValueTF{#1}{ 
		\IfValueTF{#2}{ 
			\ovst{\ust{\Otimes}{#1}}{#2}
		}{
			\ust{\Otimes}{#1}
		}
	}{
	\Otimes
	}
}

\let\Sqcup\bigsqcup
\RenewDocumentCommand{\bigsqcup}{o o}{ 
	\IfValueTF{#1}{ 
		\IfValueTF{#2}{ 
			\ovst{\ust{\Sqcup}{#1}}{#2}
		}{
			\ust{\Sqcup}{#1}
		}
	}{
	\Sqcup
	}
}

\let\Uplus\biguplus
\RenewDocumentCommand{\biguplus}{o o}{ 
	\IfValueTF{#1}{ 
		\IfValueTF{#2}{ 
			\ovst{\ust{\Uplus}{#1}}{#2}
		}{
			\ust{\Uplus}{#1}
		}
	}{
	\Uplus
	}
}

\makeatletter
\newcommand{\TestUpperCase}[1]{
    \def\temp{#1}%
    \def\upcase{0}
    \@for\i:=A,B,C,D,E,F,G,H,I,J,K,L,M,N,O,P,Q,R,S,T,U,V,W,X,Y,Z\do{%
        \ifx\temp\i%
            \def\upcase{1}
        \fi%
    }%
    \let\temp\relax
}
\newcommand{\TestLowerCase}[1]{
    \def\temp{#1}%
    \def\lowcase{0}
    \@for\i:=,a,b,c,d,e,f,g,h,i,j,k,l,m,n,o,p,q,r,s,t,u,v,w,x,y,z\do{%
        \ifx\temp\i%
            \def\lowcase{1}
        \fi%
    }%
    \let\temp\relax
}
\makeatother

\ExplSyntaxOn 

\let\Max\max 
\RenewDocumentCommand{\max}{o}{ 
	\clist_clear:N \l_tmpa_clist 
	\tl_clear:N \l_tmpa_tl 
	\tl_clear:N \l_tmpb_tl 
	\IfValueTF{#1}{ 
		\clist_set:Nn \l_tmpa_clist {#1} 
		\int_set:Nn \l_tmpa_int {\clist_count:N \l_tmpa_clist} 
	}{ 
	}
	
	\int_compare:nNnTF {\l_tmpa_int} = {0}{ 
		\Max 
	}{ 
		\int_compare:nNnTF {\l_tmpa_int} = {1}{ 
			\clist_pop:NN \l_tmpa_clist \l_tmpa_tl 
			\Max({#1}) 
		}{ 
			\int_compare:nNnTF {\l_tmpa_int} = {2}{ 
				\clist_pop:NN \l_tmpa_clist \l_tmpa_tl 
				\clist_pop:NN \l_tmpa_clist \l_tmpb_tl 
				\Max({\l_tmpa_tl} , {\l_tmpb_tl})
			}{ 
				ERRRRRRRREEEEEEEEEEUUUUUUUUUUUUR 
			}
		}
	}
}

\let\Min\min 
\RenewDocumentCommand{\min}{o}{ 
	\clist_clear:N \l_tmpa_clist
	\tl_clear:N \l_tmpa_tl
	\tl_clear:N \l_tmpb_tl
	\IfValueTF{#1}{ 
		\clist_set:Nn \l_tmpa_clist {#1}
		\int_set:Nn \l_tmpa_int {\clist_count:N \l_tmpa_clist}
	}{ 
	}
	
	\int_compare:nNnTF {\l_tmpa_int} = {0}{ 
		\Min
	}{
		\int_compare:nNnTF {\l_tmpa_int} = {1}{ 
			\clist_pop:NN \l_tmpa_clist \l_tmpa_tl
			\Min({#1})
		}{
			\int_compare:nNnTF {\l_tmpa_int} = {2}{ 
				\clist_pop:NN \l_tmpa_clist \l_tmpa_tl
				\clist_pop:NN \l_tmpa_clist \l_tmpb_tl
				\Min({\l_tmpa_tl} , {\l_tmpb_tl})
			}{ 
				ERRRRRRRREEEEEEEEEEUUUUUUUUUUUUR
			}
		}
	}
}

\ExplSyntaxOff 

\ExplSyntaxOn

\newlength{\acs} 
\clist_new:N \g_arr_clist 
\NewDocumentCommand{\arr}{o m o}{ 
	\tl_clear:N \l_tmpa_tl 
	\tl_clear:N \l_tmpb_tl 
	\clist_clear:N \l_tmpa_clist
	\IfValueTF{#3}{ 
		\clist_set:Nn \l_tmpa_clist {#3} 
	}{
		\clist_set_eq:NN \l_tmpa_clist \g_arr_clist 
	}
	
    \int_set:Nn \l_tmpa_int {\clist_count:N \l_tmpa_clist}
	\int_compare:nNnTF \l_tmpa_int < #2{
		\textbf{ERREUR \, \textbackslash arr : \, longueur \, de \, la \, liste \, < \, nombre \, fourni} 
	}{ 
    	\int_do_until:nNnn {\clist_count:N \l_tmpa_clist} < {#2+1}{
    		\int_step_inline:nnnn 1 1 {#2-1} {
    			\clist_pop:NN \l_tmpa_clist \l_tmpb_tl
        		\tl_put_right:NV \l_tmpa_tl \l_tmpb_tl 
        		\tl_put_right:Nn \l_tmpa_tl {&}
    		}
        	\clist_pop:NN \l_tmpa_clist \l_tmpb_tl
        	\tl_put_right:NV \l_tmpa_tl \l_tmpb_tl 
        	\tl_put_right:Nn \l_tmpa_tl {\\ }
    	}
    	\int_set:Nn \l_tmpa_int {\clist_count:N \l_tmpa_clist}
    	\int_compare:nNnT \l_tmpa_int > 0{
    		\int_step_inline:nnnn 1 1 {\l_tmpa_int - 1}{
    			\clist_pop:NN \l_tmpa_clist \l_tmpb_tl
          		\tl_put_right:NV \l_tmpa_tl \l_tmpb_tl
          		\tl_put_right:Nn \l_tmpa_tl {&}
        	}
        	\clist_pop:NN \l_tmpa_clist \l_tmpb_tl
        	\tl_put_right:Nn \l_tmpa_tl {\multicolumn{\int_eval:n {#2 - \l_tmpa_int + 1}}{c}{\l_tmpb_tl}}
    	}
		\begin{array}[#1]{*{#2}{c}} 
		\tl_use:N \l_tmpa_tl 
		\end{array}
	}
}

\clist_new:N \g_addToOdd_clist
\NewDocumentCommand{\addToOddPositions}{m}{
    \clist_set:NV \l_temp_clist \g_addToOdd_clist 
    \clist_clear:N \g_addToOdd_clist
    \int_step_inline:nn { \clist_count:N \l_temp_clist } {
        \clist_gput_right:Nx \g_addToOdd_clist { \clist_item:Nn \l_temp_clist {##1} } 
        \int_compare:nNnF { \int_mod:nn {##1} {2} } = { 0 } { 
            \clist_gput_right:Nn \g_addToOdd_clist {#1} 
        }
    }
}

\newcommand{\acc}[1]{
	\clist_clear:N \l_tmpa_clist 
	\clist_set:Nn \l_tmpa_clist {#1}
	
	\clist_set:NV \g_addToOdd_clist \l_tmpa_clist
	\addToOddPositions{\si} 
	\clist_set_eq:NN \g_arr_clist \g_addToOdd_clist 
	
	\int_set:Nn \l_tmpa_int {\clist_count:N \l_tmpa_clist}
	\int_set:Nn \l_tmpb_int {\int_div_truncate:nn {\l_tmpa_int}{2}}
	
	\int_compare:nNnTF {\l_tmpa_int} = {2 * \l_tmpb_int}{ 
	}{ 
		\clist_reverse:N \g_arr_clist
		\clist_pop:NN \g_arr_clist \l_tmpa 
		\clist_reverse:N \g_arr_clist
		\clist_put_right:Nn \g_arr_clist {\sinon}
	}
	\left\{\arr{3}\right. 
}

\newcommand{\barr}[1]{
	\clist_clear:N \l_tmpa_clist 
	\clist_set:Nn \l_tmpa_clist {#1}
	
	\clist_set:NV \g_addToOdd_clist \l_tmpa_clist
	\addToOddPositions{\si}
	\clist_set_eq:NN \g_arr_clist \g_addToOdd_clist 
	
	\int_set:Nn \l_tmpa_int {\clist_count:N \l_tmpa_clist}
	\int_set:Nn \l_tmpb_int {\int_div_truncate:nn {\l_tmpa_int}{2}}
	
	\int_compare:nNnTF {\l_tmpa_int} = {2 * \l_tmpb_int}{
	}{
		\clist_reverse:N \g_arr_clist
		\clist_pop:NN \g_arr_clist \l_tmpa 
		\clist_reverse:N \g_arr_clist
		\clist_put_right:Nn \g_arr_clist {\sinon}
	}
	\left|\arr{3}\right.
}

\NewDocumentCommand{\fct}{o o m}{ 
	\IfValueTF{#2}{
		\clist_set:Nn \l_fcto_clist {#2}
		\tl_set:Nn \l_f_tl {#1}
	}{
		\IfValueTF{#1}{
			\clist_set:Nn \l_fcto_clist {#1}
		}{
			\clist_set:Nn \g_addToOdd_clist {#1} 
		}
	}
	\clist_set:Nn \g_addToOdd_clist {#3} 
	\let\temp\arraycolsep 
	\setlength{\arraycolsep}{\acs} 
	\addToOddPositions{\mapsto} 
	\clist_set:NV \l_fct_clist \g_addToOdd_clist 
	\IfValueTF{#1}{
		\tl_clear:N \l_ta_tl
		\tl_clear:N \l_tb_tl
		\clist_pop:NN \l_fcto_clist \l_ta_tl
		\clist_pop:NN \l_fcto_clist \l_tb_tl
		\clist_put_left:Nn \l_fct_clist \l_tb_tl
		\clist_put_left:Nn \l_fct_clist {\rightarrow}
		\clist_put_left:Nn \l_fct_clist \l_ta_tl
		\clist_set:NV \g_arr_clist \l_fct_clist
	}{
	}
	\clist_set:NV \g_arr_clist \l_fct_clist
	\IfValueTF{#2}{
		#1 \colon \arr[t]{3}
	}
	{
		\arr[t]{3}
	}
	\setlength{\arraycolsep}{\temp} 
	\let\temp\relax 
}

\newcommand{\bseqf}[1]{ 
	
	\clist_set:Ne \l_bseqf_clist #1
	\clist_pop:NN \l_bseqf_clist \l_bseqf_sense_tl
	\clist_pop:NN \l_bseqf_clist \l_bseqf_align_tl
	\clist_pop:NN \l_bseqf_clist \l_bseqf_left_tl
	\clist_pop:NN \l_bseqf_clist \l_bseqf_topl_tl
	\clist_pop:NN \l_bseqf_clist \l_bseqf_topr_tl
	\clist_pop:NN \l_bseqf_clist \l_bseqf_bot_tl
	\clist_pop:NN \l_bseqf_clist \l_bseqf_right_tl
	\clist_pop:NN \l_bseqf_clist \l_bseqf_nbline_tl
	
	\tl_if_eq:VnT  {\l_bseqf_sense_tl} {d} { \tl_set:Nn \l_bseqf_sense_tl \rootAtTop }
	\tl_if_eq:VnT  {\l_bseqf_sense_tl} {u} { \tl_set:Nn \l_bseqf_sense_tl \rootAtBottom }
	
	\tl_if_eq:VnT { \l_bseqf_align_tl } {c}  { \tl_set:Nn \l_bseqf_align_tl {\centerAlignProof} } 
	\tl_if_eq:VnT { \l_bseqf_align_tl } {n}  { \tl_set:Nn \l_bseqf_align_tl {\normalAlignProof} } 
	\tl_if_eq:VnT { \l_bseqf_align_tl } {b}  { \tl_set:Nn \l_bseqf_align_tl {\bottomAlignProof} } 
	
	\tl_if_eq:VnT {\l_bseqf_nbline_tl} {o} { \tl_set:Nn \l_bseqf_nbline_tl \noLine }
	\tl_if_eq:VnT {\l_bseqf_nbline_tl} {0} { \tl_set:Nn \l_bseqf_nbline_tl \noLine }
	\tl_if_eq:VnT {\l_bseqf_nbline_tl} {-} { \tl_set:Nn \l_bseqf_nbline_tl \singleLine } 
	\tl_if_eq:VnT {\l_bseqf_nbline_tl} {1} { \tl_set:Nn \l_bseqf_nbline_tl \singleLine }
	\tl_if_eq:VnT {\l_bseqf_nbline_tl} {=} { \tl_set:Nn \l_bseqf_nbline_tl \doubleLine }
	\tl_if_eq:VnT {\l_bseqf_nbline_tl} {2} { \tl_set:Nn \l_bseqf_nbline_tl \doubleLine }
	
	\clist_if_in:NnT {\l_bseqf_clist} {/} { \tl_set:Nn \l_bseqf_tline_tl \solidLine }
	\clist_if_in:NnT {\l_bseqf_clist} {.} { \tl_set:Nn \l_bseqf_tline_tl \dottedLine }
	\clist_if_in:NnT {\l_bseqf_clist} {--} { \tl_gset:Nn \l_bseqf_tline_tl \dashedLine }
	
	\l_bseqf_sense_tl
	\l_bseqf_align_tl
	\AxiomC{\l_bseqf_topl_tl}
	\AxiomC{\l_bseqf_topr_tl}
	\LeftLabel{\l_bseqf_left_tl}\RightLabel{\l_bseqf_right_tl}
	\l_bseqf_nbline_tl
	\l_bseqf_tline_tl
	\BinaryInfC{\l_bseqf_bot_tl}
}

\NewDocumentCommand{\bseq}{o m}{
	
	\int_zero_new:N \l_bseq_env_int 
	
	\tl_set:Nn \l_bseq_sense_tl {u}
	\tl_set:Nn \l_bseq_align_tl {n} 
	\tl_clear_new:N \l_bseq_left_tl 
	
	\tl_clear_new:N \l_bseq_topl_tl 
	\tl_clear_new:N \l_bseq_topr_tl 
	
	\tl_clear_new:N \l_bseq_right_tl 
	\tl_set:Nn \l_bseq_nbline_tl {1}
	\tl_set:Nn \l_bseq_tline_tl {/}

	\clist_clear_new:N \l_bseq_param_clist
	
	\IfValueT{#1}{\clist_set:Nn \l_bseq_param_clist {#1}} 
	\clist_set:Nn \l_bseq_seq_clist {#2}
	
	\clist_if_in:NnT \l_bseq_param_clist {env} {\int_set:Nn \l_bseq_env_int {1}}
	\clist_if_in:NnT \l_bseq_param_clist {c} {\tl_set:Nn \l_bseq_align_tl {c}}
	\clist_if_in:NnT \l_bseq_param_clist {b} {\tl_set:Nn \l_bseq_align_tl {b}}
			
	\clist_if_in:NnT \l_bseq_param_clist {0} {\tl_set:Nn \l_bseq_nbline_tl {0}}
	\clist_if_in:NnT \l_bseq_param_clist {o} {\tl_set:Nn \l_bseq_nbline_tl {o}}
	\clist_if_in:NnT \l_bseq_param_clist {2} {\tl_set:Nn \l_bseq_nbline_tl {2}}
	\clist_if_in:NnT \l_bseq_param_clist {=} {\tl_set:Nn \l_bseq_nbline_tl {=}}
	\clist_if_in:NnT \l_bseq_param_clist {--} {\tl_set:Nn \l_bseq_tline_tl {--}}
	\clist_if_in:NnT \l_bseq_param_clist {.} {\tl_set:Nn \l_bseq_tline_tl {.}}

	\clist_clear_new:N \l_bseq_top_clist
	\tl_clear_new:N \l_bseq_bot_tl
	\tl_clear_new:N \l_bseq_left_tl
	\tl_clear_new:N \l_bseq_right_tl
	
	\keys_define:nn { sequent } {
		Ax     .clist_set:N = \l_bseq_top_clist,
		Ccl    .tl_set:N = \l_bseq_bot_tl,
		left   .tl_set:N = \l_bseq_left_tl,
		right  .tl_set:N = \l_bseq_right_tl
	}
	
	\exp_args:NnV \keys_set:nn { sequent } \l_bseq_seq_clist
	
	\clist_if_empty:NTF \l_bseq_top_clist {
		\tl_set:Nn \l_bseq_topl_tl {{ }} 
		\tl_set:Nn \l_bseq_topr_tl {{ }} 
	}{
		\clist_pop:NN \l_bseq_top_clist \l_bseq_topl_tl
		\clist_pop:NN \l_bseq_top_clist \l_bseq_topr_tl
	}
	
	\tl_if_empty:NT \l_bseq_right_tl {
		\tl_set:Nn \l_bseq_right_tl {{ }}
	}
	\tl_if_empty:NT \l_bseq_left_tl {
		\tl_set:Nn \l_bseq_left_tl {{ }}
	}
	\tl_if_empty:NT \l_bseq_bot_tl {
		\tl_set:Nn \l_bseq_bot_tl {{ }}
	}
	
	\clist_set:NV \l_bseq_clist {\l_bseq_sense_tl, \l_bseq_align_tl, \l_bseq_left_tl, \l_bseq_topl_tl, \l_bseq_topr_tl, \l_bseq_bot_tl, \l_bseq_right_tl, \l_bseq_nbline_tl, \l_bseq_tline_tl}
	
	\int_compare:nNnTF {\l_bseq_env_int} = {0} {
		\bseqf{\l_bseq_clist}
		\DisplayProof{}
	}{
		\begin{prooftree}
		\bseqf{\l_bseq_clist}
		\end{prooftree}
	}
}

\NewDocumentCommand{\seq}{o m m}
{
	\IfValueTF{#1}{
		\tl_set:Nn \l_tmpa_tl {\hhline{=}} 
	}{
		\tl_set:Nn \l_tmpa_tl {\hhline{-}} 
	}
	\int_zero_new:N  \l_max_int
	\clist_set:Nn \l_tmpa_clist {#2} 
	\clist_set:Nn \l_tmpb_clist {#3} 
	
	\int_set:Nn \l_max_int {\int_max:nn {\clist_count:N \l_tmpa_clist} {\clist_count:N \l_tmpb_clist}}
	
	\ifmmode 
		\begin{array}{c} 
			
			\begin{array}{*{\l_max_int}{c}}
				
				\clist_use:Nn \l_tmpa_clist {& \qquad} 
				
			\end{array} \\ \l_tmpa_tl 
			
			\begin{array}{*{\l_max_int}{c}}
				
				\clist_use:Nn \l_tmpb_clist {& \qquad} 
				
			\end{array} 
			
		\end{array} 
	
	\else 
		\begin{tabular}{c} 
			
			\begin{tabular}{*{\l_max_int}{c}}
				
				\clist_use:Nn \l_tmpa_clist {& \qquad} 
				
			\end{tabular} \\ \l_\tmpa_tl 
			
			\begin{tabular}{*{\l_max_int}{c}}
				
				\clist_use:Nn \l_tmpb_clist {& \qquad} 
				
			\end{tabular} 
			
		\end{tabular} 
	
	\fi
}

\ExplSyntaxOff

\newlength{\indlength}
\newcommand{\ind}[1][1]{\makeatletter\ifx\@classname\zml-article\ifnum#1=2\hspace*{#1\indlength}\fi\else\hspace*{#1\indlength}\fi\makeatother} 

\ExplSyntaxOn

\NewDocumentCommand{\target}{m o}{
	\tl_set:Nn \l_target_tl {#1}
	
	\IfValueTF{#2}{
		\hypertarget{\tl_to_str:n {#1}}{#2}
	}{
		\hypertarget{\tl_to_str:n {#1}}{#1}
	}
}

\NewDocumentCommand{\link}{m o}{
	\tl_set:Nn \l_target_tl {#1}
	
	\IfValueTF{#2}{
		\hyperlink{\tl_to_str:n {#1}}{#2}
	}{
		\hyperlink{\tl_to_str:n {#1}}{#1}
	}
}

\ExplSyntaxOff

\NewDocumentCommand{\deftarget}{m o}{
	\IfValueTF{#2}{
		\target{#1}[\emph{#2}]
	}{
		\target{#1}[\emph{#1}]
	}
}

\NewDocumentCommand{\deflink}{m o}{
	\IfValueTF{#2}{
		\link{#1}[\emph{#2}]
	}{
		\link{#1}[\emph{#1}]
	}
}

\newcommand{\adjusttopage}[1]{\adjustbox{width=\textwidth}{#1}}

\def\inv{^{-1}}

\def\acal{\mathscr{A}} 
\def\bcal{\mathscr{B}} 
\def\ccal{\mathscr{C}} 
\def\dcal{\mathscr{D}} 
\def\ecal{\mathscr{E}} 
\def\fcal{\mathscr{F}}

\def\kcal{\mathscr{K}} 
\def\lcal{\mathscr{L}}

\def\rcal{\mathscr{R}} 
\def\scal{\mathscr{S}}

\def\vcal{\mathscr{V}}

\def\Ical{\mathcal{I}}

\def\Lcal{\mathcal{L}}

\def\Tcal{\mathcal{T}}

\renewcommand{\phi}{\varphi} 
\let\oldepsilon\epsilon
\let\oldvarepsilon\varepsilon
\renewcommand{\epsilon}{\oldvarepsilon}
\renewcommand{\varepsilon}{\oldepsilon}

\NewDocumentCommand{\satisfies}{o}{\IfValueTF{#1}{\models_{#1}}{\models}}

\def\inv{^{-1}}

\def\Ens{\IfLanguageName{english}{\scal\!et}{\ecal\!ns}}

\def\Sup{\scal\!up}

\def\equ{\IfLanguageName{english}{the following statements are equivalent\xspace}{les assertions suivantes sont équivalentes\xspace}}

\def\0{\textbf{0}} 
\def\1{\textbf{1}} 
\def\2{\textbf{2}} 
\def\A{\textbf{A}}

\def\F{\textbf{F}}

\def\I{\textbf{I}}

\def\L{\textbf{L}}

\def\T{\textbf{T}}

\def\W{\textbf{W}}

\newcommand{\ifmmodeTF}[2]{
  \ifmmode
    #1 
  \else
    #2 
  \fi
}

\newcommand{\ifmmodeFT}[2]{
  \ifmmode
    #2 
  \else
    #1 
  \fi
}

\newcommand{\ifmmodeT}[1]{
  \ifmmode
    #1 
  \fi
}

\newcommand{\ifmmodeF}[1]{
  \ifmmode
  \else
    #1 
  \fi
}

\newcommand{\rescale}[2]{\ifmmodeTF{\scalebox{#1}{${#2}$}}{\scalebox{#1}{{#2}}}} 

\NewDocumentCommand{\texte}{m o o o}{\ifmmode\IfValueTF{#2}{\IfLanguageName{english}{\text{ #2 }}{\text{ #1 }}}{\text{ #1 }}\else\IfValueTF{#2}{\IfValueTF{#3}{\IfLanguageName{english}{\IfValueTF{#4}{#4\xspace}{#3\xspace}}{#3\xspace}}{\IfLanguageName{english}{#2\xspace}{#1\xspace}}}{#1\xspace}\fi} 

\def\tq{\texte{tq}[s.t.][tel que][such that]} 
\def\et{\texte{et}[and]} 
\def\ou{\texte{ou}[or]} 
\def\si{\texte{si}[if]} 
\def\sinon{\texte{sinon}[otherwise]} 
\def\ssi{\texte{si, et seulement si,}[if and only if]}

\ExplSyntaxOn 

\seq_new:N \l_local_enum_seq 

\newcommand{\storethestuff}[1]{ 
  \seq_set_from_clist:Nn \l_local_enum_seq {#1} 
}

\newcommand{\dotheenumstuff}{ 
\int_zero:N \l_tmpa_int 
\seq_map_inline:Nn \l_local_enum_seq { 
    \int_incr:N \l_tmpa_int
    \item ##1
  } 
} 

\ExplSyntaxOff

\NewDocumentCommand{\restr}{m o o}{
	\IfValueTF{#3}{
		#1\left|\ovst{\ust{\textcolor{white}{A}}{{#2}}}{{#3}}\right.
		}{
		\IfValueTF{#2}{
			#1\left|\ust{\textcolor{white}{A}}{{#2}}\right.
		}{
		#1\left|\right.
		}
	}
}

\newcommand{\citup}[1]{\textup{\cite{#1}}}

\RequirePackage{tikz-cd}
\RequirePackage{amssymb}
\usetikzlibrary{calc}
\usetikzlibrary{decorations.pathmorphing}

\tikzset{curve/.style={settings={#1},to path={(\tikztostart)
    .. controls ($(\tikztostart)!\pv{pos}!(\tikztotarget)!\pv{height}!270:(\tikztotarget)$)
    and ($(\tikztostart)!1-\pv{pos}!(\tikztotarget)!\pv{height}!270:(\tikztotarget)$)
    .. (\tikztotarget)\tikztonodes}},
    settings/.code={\tikzset{quiver/.cd,#1}
        \def\pv##1{\pgfkeysvalueof{/tikz/quiver/##1}}},
    quiver/.cd,pos/.initial=0.35,height/.initial=0}

\tikzset{tail reversed/.code={\pgfsetarrowsstart{tikzcd to}}}
\tikzset{2tail/.code={\pgfsetarrowsstart{Implies[reversed]}}}
\tikzset{2tail reversed/.code={\pgfsetarrowsstart{Implies}}}
\tikzset{no body/.style={/tikz/dash pattern=on 0 off 1mm}}

\fi 

\newtheorem{letterthm}{Theorem}

\NewDocumentCommand{\con}{o}{\IfValueTF{#1}{C^0({[0,1]}^{#1} , {[0,1]})}{C^0({[0,1]} , {[0,1]})}}
\NewDocumentCommand{\conc}{o o}{\IfValueTF{#2}{\IfValueTF{#1}{C^0_{{#2},\nearrow}({[0,1]}^{#1} , {[0,1]})}{C^0_{{#2},\nearrow}({[0,1]} , {[0,1]})}}{\IfValueTF{#1}{C^0_\nearrow({[0,1]}^{#1} , {[0,1]})}{C^0_\nearrow({[0,1]} , {[0,1]})}}}
\providecommand{\+}{\dot +}
\let\.\-
\renewcommand{\-}{
	\mathrel{
		\tikz[baseline=(x.base)]{
    			\node[inner sep=0pt] (x) {$-$};
    			\node at ([xshift=0pt,yshift=-1pt]x.north) {\small .};
  		}
  	}
}
\newcommand{\dotp}{\dot +} 
\newcommand{\dotm}{\-} 

\newcommand{\e}[2]{#2} 
\newcommand{\alg}{\acal_\mathcal{L}} 
\providecommand{\USC}{{USC(\lcal)}} 
 
\def\usc{\texte{s.c}[u.s.c][semi-continue][upper semi-continuous]} 
\def\sp{\texte{s.c.}[s.p.][sup-continue][sup-preserving]} 
\NewDocumentCommand{\wb}{o o}{\IfValueTF{#1}{\IfValueTF{#2}{{#1} < {#2}}{{#1}^{<}}}{<}} 
\NewDocumentCommand{\wa}{o o}{\IfValueTF{#1}{\IfValueTF{#2}{{#1} > {#2}}{{#1}^{>}}}{>}} 
\NewDocumentCommand{\nwb}{o o}{\IfValueTF{#1}{\IfValueTF{#2}{{#1} \not < {#2}}{{#1} \not <}}{\not <}} 
\NewDocumentCommand{\nwa}{o o}{\IfValueTF{#1}{\IfValueTF{#2}{{#1} \not > {#2}}{{#1} \not >}}{\not >}}

\NewDocumentCommand{\USCB}{o}{\IfValueTF{#1}{USC(#1)}{USC(\lcal,\bcal)}}

\newcommand{\excl}[1]{\IfValueTF{#1}{!_{\!\mathclap{\color{white}\rule{0.6ex}{0.6ex}}_{\mathclap{#1}}}\,}{!}} 
\newcommand{\que}[1]{\IfValueTF{#1}{?_{\!\!{{{\mathclap{\color{white}\rule{0.5ex}{0.6ex}}}}}_{\mathclap{#1}}}\,}{?}}

\newcommand{\ouv}[1]{{#1}^u} 
\newcommand{\cop}{\text{cop}}

\NewDocumentCommand{\Con}{o}{\IfValueTF{#1}{C^0(#1)}{C^0(\lcal,{[0,1]})}} 
\DocumentCommand{\Conc}{o o}{\IfValueTF{#2}{\IfValueTF{#1}{C^0_{{#2},\nearrow}([0,1]^{#1})}{C^0_{{#2},\nearrow}([0,1])}}{\IfValueTF{#1}{C^0_\nearrow([0,1]^{#1})}{C^0_\nearrow([0,1])}}}
\newcommand{\crcl}{\texte{commutative residuated complete lattice}} 
\newcommand{\crcls}{\texte{commutative residuated complete lattices}} 
\newcommand{\Crcl}{\ccal\!rcl}

\NewDocumentCommand{\Lusc}{o}{\IfValueTF{#1}{USC([0,1]_u^{#1})}{(USC([0,1]_u^n))_{n \in \nn}}} 
\NewDocumentCommand{\Lcinc}{o}{\IfValueTF{#1}{C^0_\nearrow([0,1]^{#1})}{(C^0_\nearrow([0,1]^n))_{n \in \nn}}}

\def\InFm{\I n\F_m} 
\def\InL{\I n\L} 
\RenewDocumentCommand{\bs}{o}{\IfValueTF{#1}{{#1}^	\sim}{\sim}} 
\RenewDocumentCommand{\ng}{o}{\IfValueTF{#1}{{#1}^	\neg}{\neg}} 
\def\InSqt{\Ical n \Sqt} 
\def\GL{\textit{\textbf{GL}}\xspace} 
\def\MGL{\textit{\textbf{MGL}}\xspace} 
\def\InGL{\textit{\textbf{InGL}}\xspace} 
\def\InMGL{\textit{\textbf{InMGL}}\xspace} 
\def\CFLew{\textit{\textbf{CFL}}\sub{\bf{\it{ew}}}\xspace} 
\def\InLJK{\textit{\textbf{InLJK}}\xspace}

\DocumentCommand{\con}{o}{\IfValueTF{#1}{C^0([0,1]^{#1} , [0,1])}{C^0([0,1] , [0,1])}}
\DocumentCommand{\conc}{o o}{\IfValueTF{#2}{\IfValueTF{#1}{C^0_{{#2},\nearrow}([0,1]^{#1})}{C^0_{{#2},\nearrow}([0,1])}}{\IfValueTF{#1}{C^0_\nearrow([0,1]^{#1})}{C^0_\nearrow([0,1])}}}

\DocumentCommand{\Lusc}{o}{\IfValueTF{#1}{USC([0,1]_u^{#1})}{(USC([0,1]_u^n))_{n \in \nn}}} 
\DocumentCommand{\Lcinc}{o}{\IfValueTF{#1}{C^0_\nearrow([0,1]^{#1})}{(C^0_\nearrow([0,1]^n))_{n \in \nn}}}  

\command{\crcl}{\texte{commutative residuated complete lattice}} 
\command{\crcls}{\texte{commutative residuated complete lattices}} 

\command{\<}[1]{#1} 
\command{\e}[2]{#2_#1} 
\command{\alg}{\acal_\mathcal{L}} 
\def\CFLew{\textbf{\textit{CFL}}\sub{\bf{\it{ew}}}\xspace} 

\def\MGL{\textbf{\textit{MGL}}\xspace} 
\def\GL{\textbf{\textit{GL}}\xspace} 
\def\LJK{\textbf{\textit{LJK}}\xspace} 
\def\int{\text{int}}

\def\Marco{{\T_M}} 
\def\LMarco{{\mathcal{L}_M}} 
\def\usc{\texte{s.c}[u.s.c][semi-continue][upper semi-continuous]} 
\def\sp{\texte{s.c.}[s.p.][sup-continue][sup-preserving]} 
\def\lsc{\texte{s.c}[l.s.c][][lower semi-continuous]} 
\newcommand{\Cn}[1]{C^0_\nearrow({#1})}

\usepackage{bussproofs}

\def\fCenter{\vdash}

\EnableBpAbbreviations

\def\InL{\I n\L} 
\RenewDocumentCommand{\bs}{o}{\IfValueTF{#1}{{#1}^\sim}{\sim}} 
\RenewDocumentCommand{\ng}{o}{\IfValueTF{#1}{{#1}^\neg}{\neg}} 
\def\GL{\textit{\textbf{GL}}\xspace} 
\def\MGL{\textit{\textbf{MGL}}\xspace} 
\def\InGL{\textit{\textbf{InGL}}\xspace} 
\def\InMGL{\textit{\textbf{InMGL}}\xspace} 
\def\LJK{\textit{\textbf{LJK}}\xspace} 
\def\InCFLew{\textit{\textbf{InCFL}}\sub{\textit{\textbf{ew}}}\xspace} 
\def\inv{\text{inv}} 

\def\Sqt{\scal\! qt}

\def\InSqt{\Ical n \Sqt} 
\def\InFm{\I n\F_m} 
\command{\USC}{USC(\lcal)}
\DocumentCommand{\these}{o}{\IfValueTF{#1}{\texte{N$_#1$}}{\texte{N}}}

\def\InFm{\I n\F_m} 
\def\InL{\I n\L} 
\RenewDocumentCommand{\bs}{o}{\IfValueTF{#1}{{#1}^	\sim}{\sim}} 
\RenewDocumentCommand{\ng}{o}{\IfValueTF{#1}{{#1}^	\neg}{\neg}} 
\DocumentCommand{\Con}{o}{\IfValueTF{#1}{C^0([0,1]^{#1})}{(C^0[0,1]^n))_{n \in \nn}}}  
\def\InSqt{\Ical n \Sqt} 
\def\GL{\textit{\textbf{GL}}\xspace} 
\def\MGL{\textit{\textbf{MGL}}\xspace} 
\def\InGL{\textit{\textbf{InGL}}\xspace} 
\def\InMGL{\textit{\textbf{InMGL}}\xspace} 
\def\LJK{\textit{\textbf{LJK}}\xspace} 
\def\InLJK{\textit{\textbf{InLJK}}\xspace} 
\def\inv{\text{inv}} 
\def\int{\text{int}} 
\def\class{\text{class}} 

\providecommand{\Cn}[1]{C^0_\nearrow({#1})} 
\RenewDocumentCommand{\bs}{o}{\IfValueTF{#1}{{#1}^\sim}{\sim}} 
\RenewDocumentCommand{\ng}{o}{\IfValueTF{#1}{{#1}^\neg}{\neg}}

\def\s+{{\Sigma_+}}

\def\Sqt{\scal\! qt}

\DocumentCommand{\these}{o}{\IfValueTF{#1}{\texte{N$_#1$}}{\texte{N}}}

\DocumentCommand{\Ralg}{o}{\IfValueTF{#1}{\W_#1}{\rcal alg}}

 \def\sequences{structures\xspace}
\def\InFm{\I n\F_m} 
\def\InL{\I n\L} 
\DocumentCommand{\bs}{o}{\IfValueTF{#1}{{#1}^	\sim}{\sim}} 
\DocumentCommand{\ng}{o}{\IfValueTF{#1}{{#1}^	\neg}{\neg}} 
\def\InSqt{\Ical n \Sqt} 
\def\GL{\textit{\textbf{GL}}\xspace} 
\def\InGL{\textit{\textbf{InGL}}\xspace} 
\def\MGL{\textit{\textbf{MGL}}\xspace} 
\def\InMGL{\textit{\textbf{InMGL}}\xspace} 
\def\cut{\text{cut}} 

\usepackage{stackengine}

\command{\dotlozenge}{\stackinset{c}{}{c}{}{\scalebox{.5}{$\bullet$}}{$\lozenge$}} 
\command{\dotblacklozenge}{\stackinset{c}{}{c}{}{\textcolor{white}{\scalebox{.5}{$\bullet$}}}{$\blacklozenge$}} 
\command{\dotsquare}{\stackinset{c}{}{c}{}{\scalebox{.5}{$\bullet$}}{$\square$}} 
\newcommand{\dotblacksquare}{\stackinset{c}{}{c}{}{\textcolor{white}{\scalebox{.5}{$\bullet$}}}{$\blacksquare$}}

\title{Cut-free Deductive System for Continuous Intuitionistic Logic} 

\author[1]{Guillaume Raymond Geoffroy} 

\thanks{
I warmly thank Yoann Dabrowski and Itaï Ben Yaacov for their wise advice and the avenues of reflexion they provided me. I particularly thank Yoann Dabrowski for his proofreading. I also thank the Istituto Grothendieck for hosting me at their expense during the Toposes in Mondovì conference, during which I was working on this article.	
}

\begin{document} 

\begin{titlepage} 

\vspace*{\fill} 

\begin{abstract} 

We introduce and develop propositional \emph{continuous intuitionistic logic} and propositional \emph{continuous affine logic} through the study of two classes of algebras, and provide sequent-style deductive systems with cut-admissibility for these logics. Our approach centres on AC-algebras, which are algebras $\USC$ of sup-preserving functions from $[0,1]$ to an integral commutative residuated complete lattice $\lcal$ (in the intuitionistic case, $\lcal$ is a locale). We give an algebraic axiomatisation of AC-algebras in the language of continuous logic and prove, using the Macneille completion, that every Archimedean model embeds into some AC-algebra. We also show that (i) $\USC$ satisfies $v \+ v = 2v$ exactly when $\lcal$ is a locale, (ii) involutiveness of negation in $\USC$ corresponds to that in $\lcal$, and that (iii) adding those conditions recovers classical continuous logic. For each variant—--affine, intuitionistic, involutive, classical—--we provide a sequent style deductive system and prove completeness and cut admissibility. This yields the first sequent style formulation of classical continuous logic enjoying cut admissibility. 

\end{abstract} 

\maketitle 

\vspace*{\fill} 

\end{titlepage} 

\tableofcontents

\section{Introduction} 

\begin{groupe} 

\ind On the one hand, Continuous Logic is a very prolific area of mathematics first introduced in \cite{benyaacovContinuousFirstOrder2010} and its model theoretic framework was then developed in \cite{benyaacovModelTheoryMetric2008}, \cite{farahModelTheoryOperator2013}, \cite{farahModelTheoryOperator2014} and \cite{farahModelTheoryOperator2014a} (see also \cite{hartIntroductionContinuousModel2023} for an introduction). By reinterpreting equality as distance and quantifiers as suprema and infima, it extends classical model-theoretic methods to encompass classes of complete metric structures, domains out of reach of classical first-order logic due to their lack of finitary axiomatisation. While model theory has historically focused on algebraic structures and their first-order theories, continuous first-order logic offers the necessary expressive tools to handle infinitary properties. For instance, it enables to study Hilbert spaces and probability algebras with classical tools of logic, interpreting independence as orthogonality or probabilistic independence depending on context. Contrary to previous attempts \cite{changContinuousModelTheory1966} and \cite{hajekMetamathematicsFuzzyLogic1998}, the current framework of \cite{benyaacovModelTheoryMetric2008}, that we follow, is more closely aligned with syntax and reasoning of classical logic, permitting broader access to foundational results such as compactness, Löwenheim-Skolem theorems, and omitting types theorems. In \cite{benyaacovProofCompletenessContinuous2010} is proven a completeness theorem for a Hilbert-style deductive system in Continuous Logic, which is, to the best of our knowledge, the sole try for a proof theory for Continuous Logic. 

On the other hand, Intuitionistic Logic was first developed by Brouwer as a logical basis for constructivism \cite{brouwerIntuitionismFormalism1913} as opposed to the formalism of Hilbert. It is the logic obtained from classical logic by removing the principle of excluded middle, or equivalently the rule of \textit{reductio ad absurdum}. However, it has, by now, found very important applications in computer science and proof assistants through the Curry-Howard correspondence \cite{howardFormulaeastypesNotionConstruction1980}. The main mathematical use of Intuitionistic Logic may be its application to the study of internal objects of toposes (whether elementary or Grothendieck) through Kripke-Joyal semantics \cite{maclaneSheavesGeometryLogic1994}. A sequent calculus for it is well-known \cite{gentzenUntersuchungenUeberLogische1935} and \cite{gentzenUntersuchungenUeberLogische1935a} (cf. \cite{gentzenInvestigationsLogicalDeduction1964} and \cite{gentzenInvestigationsLogicalDeduction1965} for english translations) and the completeness of the class of Heyting algebras and the class of Kripke propositional models are well established (for Heyting algebras, the original papers are \cite{heytingFormalenRegelnIntuitionistischen1930}, and a proof can be found in english in \cite{esakiaHeytingAlgebrasDuality2019} and \cite{fittingIntuitionisticLogicModel1969}; for completeness of Kripke semantics, see \cite{fittingIntuitionisticLogicModel1969}). This article aims at defining the propositional theory of continuous intuitionistic logic. 

A previous development of intuitionistic continuous logic had been set up by Jérémie Marquès in \cite{marquesCategoricalLogicPerspective2023} under the name \emph{Fuzzy Intuitionistic Logic} relying on a previous work of Marco Abbadini on a positive version of continuous logic \cite{abbadiniDualCompactOrdered2019}. In this paper, Abbadini started from compact ordered topological spaces, which were introduced by L. Nachbin in \cite{nachbinTopologyOrder1965}. They are to topology and partial order what compact Hausdorff spaces are to topology. In \cite{abbadiniDualCompactOrdered2019}, the author proved that the category of compact ordered spaces is dual to a category of algebras he called MC-algebras. To prove that MC-algebras form a variety, he gave a sophisticated axiomatisation of them. In \cite{marquesCategoricalLogicPerspective2023}, Jérémie Marquès introduced the notion of intuitionistic compact ordered spaces and showed that the duality in \cite{abbadiniDualCompactOrdered2019} restricts to a duality between MC-algebras with a residuation and intuitionistic compact ordered spaces. We will show that these algebras are the metrically complete algebras for our alternative approach to continuous intuitionistic logic. 

In this paper, we aim at providing sequent style deductive systems for various kind of continuous logics and prove cut admissibility for these systems. In order to prove a cut admissibility theorem for the logics presented in this paper, we will rely on Algebraic Proof Theory.  Algebraic Proof Theory was first theorised in \cite{ciabattoniAlgebraicProofTheory2012} and finds its roots in \cite{galatosCutEliminationStrong2010} and \cite{okadaPhaseSemanticCutelimination1999}. Algebraic Proof Theory is a research program aimed at systematically interrelating proof-theoretic and algebraic methods, particularly in the study of substructural logics—understood as extensions of the full Lambek calculus, typically characterized by the absence of structural rules like exchange, weakening, and contraction. It builds on the discovery that the admissibility of the cut rule and subformula property correspond closely to algebraic properties of their semantic counterparts, which are subvarieties of FL-algebras also known as residuated lattices with an additional constant 0 \cite{galatos2007residuated} \footnote{On contrary to \cite{galatos2007residuated}, in this paper, the neutral element of a residuated lattice will always be the top one.}. A key focus is on the transformation of some axioms into analytic structural rules for sequent calculi. It shows a strong link between cut admissibility and stability under Macneille completion. \cite{galatosResiduatedFramesApplications2012} also deals with involutive logics, thus enabling us to obtain a sequent calculus style system for classical continuous logic having the cut admissibility property. 

\end{groupe}

\begin{groupe} 

\ind Let us give the motivations of our work. The first objective of this article is to lay the ground work for an analysis of metric structures internal to Grothendieck toposes well handled by continuous logic when the topos is $\Ens$. Contrary to topological spaces, locales are internalisable into toposes. There are several objects of real numbers internal to a topos (\cite{johnstoneSketchesElephantTopos2002}, section D 4.7). However, the good notion of norm for internal $C^*$-algebras, Banach spaces and metric spaces in general (\cite{hofmannRepresentationsAlgebrasContinuous1972}, \cite{hofmannSheafTheoreticalConcepts1979}, \cite{reichmanSemicontinuousRealNumbers1983}, \cite{johnstoneSketchesElephantTopos2002}) is valued into the so-called semi-continuous real numbers. For a topological space $X$, this object of real numbers is the sheaf of upper semi-continuous functions into $\rr$, which externalisation is $USC(X,\rr)$. It motivates the study of $USC(X,\rr)$ and thus of $USC(X,[0,1])$ that we’ll simply denote by $USC(X)$ and generalisations of this notion. For a locale $\lcal$, the object that we’ll call $\USC$ is the natural expansion of $\lcal$ by $[0,1]$. Indeed, as element of the category $\Sup$ of complete lattices and sup-preserving functions, $\lcal$ is isomorphic to $\Sup(\{0,1\}, \lcal)$ which means it can be seen as expanded by $\{0,1\}$. $\USC = \Sup([0,1], \lcal)$ is then the natural expansion of $\lcal$ above $[0,1]$. The class of $\USC$ for $\lcal$ a locale generalises the one of $USC(X)$ for $X$ a topological space as $USC(X)$ is isomorphic to $USC(\Tcal)$, where $\Tcal$ is the topology of $X$. 

	In order to study these structures from a logical point of view, we will give an axiomatisation of the class $\USC$. The language we retained is the one of continuous logic \cite{benyaacovProofCompletenessContinuous2010}. For cut admissibility purposes only, we add some unary symbols that are definable in the language of continuous logic. This approach differs from the one of \cite{abbadiniDualCompactOrdered2019} since, in order to study a positive version of continuous logic, he uses a symbol for the truncated addition one has in $[0,1]$ and a symbol for fusion. We only have a symbol for addition, fusion being definable using subtraction by a constant, thus giving a more natural way to think about our logic from an intuitionist point of view. In \cite{abbadiniDualCompactOrdered2019}, the author studies MC-algebras and we show that every MC-algebra can be embedded into a $\USC$, for $\lcal$ a locale. Hence, having interpreted its language in ours, our theory is then a conservative extension of the one of \cite{abbadiniDualCompactOrdered2019}. Together with the cut admissibility property of the sequent calculus style system we give here, we can thus claim having found a sequent calculus style system for MC-algebras. 
	
	Let us provide further insight into the main results. In order to embrace logics for constructive mathematics besides intuitionistic and classic continuous logic and tackle logics of which the negation is involutive, we start the work in a generalised framework. In this setting, locales are replaced by residuated commutative complete lattices, which are supposed integral, also known as normal commutative quantales. The unary operations that were only introduced for cut admissibility purposes become here unavoidable, because multiplication by $2$ may not be obtained as a sum. It turns out that, for $\lcal$ a residuated commutative complete lattice, the algebra $\USC$ satisfies $v \+ v = 2v$ if and only if $\lcal$ is a locale (\bf{Corollary} \ref{2+}) and the negation of $\USC$ is involutive if and only if so is the one of $\lcal$ (\bf{Theorem} \ref{involutive correspondence}). From here on, we derive the equivalence between our theory to which are added the assumptions $v \+ v = 2v$ and of involutiveness of the negation and the theory of classical continuous logic. The involutive case, the intuitionistic case and the classical case are dealt with on their own. 
	
In the general framework, as well as for each of the aforementioned particular cases, we provide a sequent calculus-style system, and we prove a cut admissibility theorem. We emphasize here that we obtained the first sequent calculus-style system for classical continuous logic that enjoys a cut admissibility theorem. In the intuitionistic case, our logic, from a proof theoretic point of view, lacks weakening but has distributivity. Contrary to the logics for which Bunched (hyper)sequent calculus is suited, where distributivity is required (\cite{dunnGentzenSystemPositive1973}, \cite{mintsCuteliminationTheoremRelevant1976}, \cite{paoliSubstructuralLogicsPrimer2002}, \cite{ciabattoniBunchedHypersequentCalculi}), in our case, distributivity is a consequence of the other rules. We emphasize that our approach does not rely on bunched calculus, and the system presented here follows a sequent style in the sense that the only binary structure symbol is " , ". To achieve a proof of cut admissibility for all the logics discussed in this article, we introduce two systems in the section \ref{Annexes}: $\MGL$ for the non involutive case and $\InMGL$ for the involutive one. We then introduce a new system for each logic studied here and prove a completeness theorem and a cut admissibility theorem. For Intuitionistic Continuous Logic, we call the system \textbf{\textit{LJK}} and we prove the following two theorems 

\begin{letterthm}[Completeness theorem] \label{Completeness theorem int} 

The classes $IC$ and $MC$ are both sound and complete for $\textbf{\textit{LJK}}$. 

\end{letterthm} 

\begin{letterthm}[Cut Admissibility theorem] \label{Cut Admissibility theorem int} 

In the system $\textbf{\textit{LJK}}$, for all formulas $a_1, \ldots, a_n \et b$ and $\{, \, , \circ_2, \bullet_2, \circ_\alpha, \epsilon\}$-term $G$ such that there exists a deduction of $G(a_1, \ldots, a_n) \vdash b$ using the cut rule, there exists a deduction of $G(a_1, \ldots, a_n) \vdash b$ not using the cut rule. 

\end{letterthm} 

\end{groupe} 

\begin{groupe} 

\vspace{\baselineskip} 

\ind We now sketch a plan for our article. In section \ref{section definition USC}, we will define the set of upper semi-continuous functions ($\USC$) from a commutative residuated lattice $\lcal$ to $[0,1]$ and give it inherited structures from $[0,1]$ and $\lcal$. We will call them AC-algebras (\bf{Definition} \ref{def AC-algèbres}). They constitute an algebraic semantics for an affine continuous logic which is a continuous version of affine logic (\cite{shulmanAffineLogicConstructive2022}), also known as $\text{FL}_{\text{ew}}$, or Intuitionistic Multiplicative Additive Linear Logic (IMALL) with weakening. We will study how properties are transferred from $\lcal$ to  $\USC$ in subsection \ref{subsection crcl} and how properties are transferred from $[0,1]$ to $\USC$ in subsection \ref{subsection structure de [0,1]}. Then in section \ref{section algebraic axiomatisation}, we will give an algebraic axiomatisation $\T$ of these algebras and prove that \begin{letterthm} \label{big theorem} 

For all model $A$ of $\T$, there exists a commutative residuated complete lattice $\lcal$ such that the quotient of the Macneille completion of $A$ by the equivalence relation induced by the preorder $\preceq$ is isomorphic to $\USC$. 

For all model $A$ of $\T$, there exists a commutative residuated complete lattice $\lcal$ such that the quotient of $A$ by $\simeq$ embeds into $\USC$. 

\end{letterthm} 

To that end, in subsection \ref{subsection Complete Archimedean models}, we will introduce two auxiliary theories and show that every complete Archimedean model of these theories is isomorphic to some AC-algebra. In subsection \ref{subsection Models of T}, we will prove \bf{Theorem} \ref{big theorem} by showing that all complete Archimedean model of $\T$ is a model of the auxilliary theories. In section \ref{section general Cut Admissibility}, we will then give a sequent-style cut-free deductive system for AC-algebras, prove the class of all AC-algebras is complete for this system and that this system has the cut admissibility property. 

In section \ref{section Intuitionistic continuous logic}, we study the class of IC-algebras, that-is-to-say AC-algebras for which the underlying commutative residuated complete lattice a locale. The first subsection (subsection \ref{subsection Topological Preliminaries}), is independent from the rest of the article and can be read on its own. It deals with (compact) ordered topological spaces. In the intuitionistic case, we actually axiomatize algebras whose Archimedean quotient (quotient by the $\simeq$ relation, \bf{Notation} \ref{preceq 0}) embeds into some $USC(X)$ for some topological space $X$. We first give an axiomatisation of IC-algebras in subsection \ref{subsection algebraic axiomatisation of IC-algebras}, and then study their relationship with MC-algebras in subsection \ref{subsection reduction}. We are finally able to prove that the class of all $USC(X)$ for $X$ a topological space and the class of IC-algebras are equivalent in a wide language, namely $(USC([0,1]_u^n))_{n \in \nn}$ (\bf{Theorem} \ref{Reduction}). Finally, in subsection \ref{subsection intuitionistic Cut Admissibility}, give a sequent-style cut-free deductive system for IC-algebras and prove the class of all IC-algebras is complete for this system and this system has the cut admissibility property. 

In section \ref{Involutive case}, we study the property of involutiveness of the negation. We first prove that the negation of an AC-algebra is involutive if and only if this is the case for the negation of the underlying commutative residuated complete lattice (\bf{Theorem} \ref{involutive correspondence}), which leads to an axiomatisation of these involutive AC-algebras. Finally, we give a sequent-style cut-free deductive system admitting the cut rule that describes involutive AC-algebras. 

In section \ref{section Boolean case}, we study involutive IC-algebras. They are the analogue of complete Boolean algebras in the continuous setting. We first show that the theory obtained to describe this Boolean Continuous Logic is equivalent to the theory of classical continuous logic (\bf{Theorem} \ref{equivalence to classical logic}). Second, we prove that the ordered topological space associated to any IC-algebra (\bf{Corollary} \ref{corollary order is equality}) is actually just a topological space, thus proving they are analogous to complete Boolean algebras. Finally, we exhibit a sequent-style cut-free deductive system admitting the cut rule that describes this logic. 

Finally, in the section \bf{Annexes} \ref{Annexes}, we prove a cut-admissibility theorem (\bf{Theorem} \ref{Cut Admissibility annexes}) which we rely on to prove all other cut-admissibility theorems of this paper. 

\end{groupe} 

\section{Definition of the algebra $\USC$ of the sup-preserving functions} \label{section definition USC} 

In this section, we want to define the main object of study of this article ($\USC$), which is built upon $[0,1]$ and commutative residuated lattices ($\lcal$), and give it inherited structures from $[0,1]$ and the commutative residuated lattice $\lcal$. The definition of $\USC$ is inspired by \cite{banaschewskiRealNumbersPointfree1997} and \cite{banaschewskiExtendedRealFunctions2012}, which deal with the real numbers in pointfree topology. The language for this study is \linebreak $\mathcal{L} = \{\vee,\, \wedge,\, \+,\, \-,\, 2,\, \frac{\cdot}{2},\, j_\ast,\, j,\, \alpha,\, \underline{0},\, \underline{1}\}$, whose interpretation in $[0,1]$ is defined in \bf{Definition} \ref{Définition du langage}. The language of commutative residuated lattices is $\mathcal{L}_{crl} = \{\wedge,\, \vee,\, \otimes,\, \nrightarrow,\, \bot,\, \top\}$, whereas the one of $[0,1]$ is $\mathcal{L}_{[0,1]} = \{\max,\, \min,\, \+,\, \-,\, 2,\, \frac{\cdot}{2},\, \ul 0,\, \ul 1\}$. As we can see, the structure of $[0,1]$ contains the one of a commutative residuated lattice. However, $\mathcal{L}_{crl}$ will act on $\USC$ using the pointwise structure of $\lcal$ while $\mathcal{L}_{[0,1]}$ will act by convolution using $\otimes$ from $\mathcal{L}$. Hence, on contrary to expectations when comparing the interpretations of $\mathcal{L}_{crl}$ and $\mathcal{L}_{[0,1]}$, $\max$ is naturally interpreted as $\otimes$, because of the use of $\otimes$ in convolution. However, $\+$ and $\otimes$ are intertwinned in such a manner that we will be able to forget $\otimes$ and work with the language $\mathcal{L}$ (\bf{Theorems} \ref{de lcal à USC(lcal)}, \ref{théorème de comparaison} and \ref{from [0,1] to L}). The aim of this section is thus to determine how properties of $[0,1]$ and $\lcal$ are transfered and transformed into properties of $\USC$. 

\subsection{Preliminaries} 

\begin{nota} 

We denote by $[0,1]_u$ the set $[0,1]$ endowed with the topology whose open sets are the $[0,q)$, $q \in [0,1]$ and $[0,1]$ itself. For all topological space $X$, we denote by $\Tcal(X)$ the topology of $X$. We denote $\Tcal([0,1]_u)$ the topology of $[0,1]_u$, and by $\Tcal([0,1]_u^n)$ the topology of $[0,1]_u^n$ \label{locale}. Note that $[0,1]_u$ is sober. 

\end{nota} 

\begin{defi} 

A subset $D$ of an ordered set $X$ is \emph{sup-dense} in $X$ if every $x \in X$ is the supremum of a part of $D$. We denote by $\Sup$ the category of complete orders and sup-preserving functions. 

\end{defi} 

Here comes the interpretation of $\mathcal{L}$ in $[0,1]$. 

\begin{defi} \label{Définition du langage} 

Let $x \et y \in [0,1]$. 

\begin{tabular}{ccc} 

$\max[x,y]$ is the maximum of $x$ and $y$ & \white{a a a a a a a a} &$\min[x,y]$ is the minimum of $x$ and $y$ 
\\ 
$x \+ y = x \dotp y = \min[(x + y),1]$ & &$x \- y = x \dotm y = \max[(x - y),0]$ 
\\ 
$2x = \min[x + x,1] = x \dotp x$ & &$\frac{\cdot}{2}(x) = \frac{x}{2}$ 
\\ 
$j_\ast(x) = \frac{x}{2} + \frac{1}{2}$ & &$j(x) = 2\left(x \- \frac{1}{2}\right) = \max[x + x - 1,0]$ 
\\ 
$\alpha(x) = \max[\frac{x}{2},j(x)]$ & & 
\\ 
$\ul 0 = 0$ & & $\ul 1 = 1$ 

\end{tabular} 

\end{defi} 

\subsection{Definition of the \crcl $\USC$} \label{subsection crcl} 

\subsubsection{Definition of the set underlying $USC(\lcal)$} 

An upper semi-continuous function from a topological space $X$ to $[0,1]$ is a continuous function from $X$ to $[0,1]_u$, which, when $X$ is sober, equivalently is a morphism of locales from $\Tcal([0,1]_u)$ to the topology of $X$ \cite[Proposition IX.3.2]{maclaneSheavesGeometryLogic1994}. According to \bf{Corollary} \ref{corscale}, it is also the same data as the one of a sup-preserving function from $[0,1]$ to the topology of $X$. This idea, also presented in \cite[Definition 3.1]{garciagutierrezLowerUpperRegularizations2009} and \cite[Definition 4.4]{gutierrezgarciaAlgebraicRepresentationSemicontinuity2007}, together with \bf{Corollary} \ref{corscale}, lead us to \bf{Definition} \ref{def USC}.  

\begin{nota} 

We denote by $USC(X)$ the set of all upper semi-continuous functions from $X$ to $[0,1]$ and by $f^*$ the sup-preserving function from $[0,1]$ to the topology of $X$, that takes $q \in [0,1]$ to $f^{-1}([0,q))$ for each $f \in USC(X)$. 

\end{nota} 

\begin{defi} 

A \emph{commutative residuated lattice} is a lattice $\lcal$ endowed with a commutative monoid operation $\otimes$ whose neutral element is the top one, denoted by $\top$, and a binary operation $\nrightarrow$ such that, for all $u$, $v \et w \in \lcal$, $u \otimes v \leq w \Leftrightarrow u \leq v \nrightarrow w$. $\nrightarrow$ is called a residual. A commutative residuated complete lattice is a \emph{commutative residuated lattice} whose order is complete. 

A \emph{lax morphism of commutative residuated lattices} $f \colon \lcal \rightarrow \kcal$ is an order-preserving function such that, for all $u \et v \in \lcal$, $f(u \otimes v) \geq f(u) \otimes f(v)$ and $f(\top) = \top$. 

A \emph{lax morphism of commutative residuated complete lattices} $f \colon \lcal \rightarrow \kcal$ is a sup-preserving function such that, for all $u \et v \in \lcal$, $f(u \otimes v) \geq f(u) \otimes f(v)$. 

We denote the category of commutative residuated complete lattices and lax morphisms by $\Crcl$. 

\end{defi} 

{\remark Residuated complete lattices are also called \emph{integral}, or \emph{normal quantales}.} 

For all this section, let $\lcal$ be a \crcl, with internal implication $\nrightarrow$, maximum $\top$ and minimum $\bot$. 

\begin{lem} \label{scale} 

Let $D$ be sup-dense in $[0,1]$, $n \in \nn$ and $f \colon \e< D^n \rightarrow \lcal$ and $g\colon D^n \rightarrow D$ be functions. 

The function $\fct[F_{f,g}][\Tcal([0{,}1]_u), \lcal]{U, \acc{\bigvee[\ust{p \in D^n}{g(p) < q}] f(p), U =  {[0,q)}, \top}}$ is a lax morphism of \crcls. 

\end{lem} 

\begin{proof} 

For every $q \leq p \in [0,1]$, $$F_{f,g}([0, q \wedge  p)) = F_{f,g}([0,q)) = F_{f,g}([0,q)) \wedge F_{f,g}([0,p)) \geq F_{f,g}([0,q)) \otimes F_{f,g}([0,p)).$$ 

For every $(q_i)_{i \in I} \in [0,1]^I$, $$F_{f,g}\left(\left[0,\bigvee[i \in I] q_i\right)\right) = \bigvee[\ust{p \in D^n}{g(p) < \bigvee[i \in I] q_i}] f(p) = \bigvee[\ust{p \in D^n}{\exists i \in I \tq g(p) < q_i}] f(p) = \bigvee[i \in I] \bigvee[\ust{p \in D^n}{g(p) < q_i}] f(p) = \bigvee[i \in I] F_{f,g}([0,q_i)).$$ 

Since $[0,1]$ is sent to $\top$, $F_{f,g}$ is a lax morphism of \crcls from $\Tcal([0,1]_u)$ to $\lcal$.
\end{proof} 

\begin{cor} \label{corscale} Let $D$ be an sup-dense subset of $[0,1]$. $\fct[G][\Sup(\e< D{,}\lcal), \Crcl(\Tcal([0{,}1]_u){,}\lcal)]{f, F_{f,\id}}$ is an isomorphism. 

\end{cor} 

{\remark $G$ is called a Raney's transform in \cite{santocanaleDualizingSuppreservingEndomaps2021}, and is denoted $(\cdot)\textasciicircum$.} 

\begin{proof} 

Since $D$ is sup-dense in $[0,1]$, every morphism of locales arises from a sup-preserving function from $\e< D$ to $\lcal$, which is its restriction, thus being unique. 
\end{proof} 

\begin{cor} \label{building} 

Let $D$ be sup-dense in $[0,1]$, $n \in \nn$ and $f \colon \e< D^n \rightarrow \lcal$ and $g\colon D^n \rightarrow D$ be functions. $$\fct[{[0,1]} , \lcal]{q, \bigvee[\ust{p \in D^n}{g(p) < q}] f(p)} \in \Sup([0,1] , \lcal).$$ 

\end{cor} 

\begin{defi} \label{def USC} 

We define $\USC$ as the set of sup-preserving functions from $[0,1]$ to $\lcal$. 

\end{defi} 

\begin{lem} \label{adj} 

Let us define, for all $f\colon [0,1] \rightarrow \lcal$, $\fct[\ouv{f}][{[0,1]}, \lcal]{q, \bigvee[p \wb q] \; \bigwedge[r \geq p] f(r)}$. \blue{According to \bf{Corollary} \ref{building}, $f^u \in \USC$}. 

For all $f\colon [0,1] \rightarrow \lcal$ and $g \in \USC$, $\ouv{f} \leq g \Leftrightarrow \forall q \in [0,1] \; g(q) \leq f(q)$, so, for all $q \in [0,1]$, $\ouv{f}(q) \leq f(q)$. 

Moreover, for all $f\colon [0,1] \rightarrow \lcal$ if $f$ is non-decreasing, for all $q \in [0,1]$, $\ouv{f}(q) = \bigvee[p \wb q] f(p)$, and, if $f \in \USC$, $\ouv{f} = f$. 

\end{lem} 

\begin{proof} 

Let $f\colon [0,1] \rightarrow \lcal$, $g \in \USC$. \begin{align*} 
\forall q \in [0,1] \; f(q) \geq g(q) &\Leftrightarrow \forall r \geq q \in [0,1] f(r) \geq g(q)\\ 
&\Leftrightarrow \forall q \in [0,1] \; \bigwedge[r \geq q] f(r) \geq g(q)\\ 
&\Leftrightarrow \forall q \in [0,1] \; \bigvee[p \wb q] \bigwedge[r \geq p] f(r) \geq g(q)\\ 
&\Leftrightarrow \ouv{f} \leq g 
\end{align*}

If $f$ is non-decreasing, then for all $q \in [0,1]$, $\ouv{f}(q) = \bigvee[p \wb q] \bigwedge[r \geq p] f(r) = \bigvee[p \wb q] f(p)$. 

Finally, if $f \in \USC$, for all $q \in [0,1]$, $\ouv {f}(q) = \bigvee[p \wb q] f(p) = f(q)$. 
\end{proof} 

Let $\dcal$ denote the set of dyadic numbers in $[0,1]$. According to \textbf{Corollary} \ref{corscale}, an upper semi-continous function from $\lcal$ to $[0,1]$ is entirely characterised by its values on $q$, for $q \in \qq \cap [0,1]$, or $q \in \dcal$. 

\subsubsection{The pointwise induced structure from $\lcal$ on $\USC$} 

The aim of this construction is to obtain a commutative residuated complete lattice (\bf{Theorem} \ref{USC is a crcl}) that satisfies some formulas that are true in $\lcal$ (\bf{Theorem} \ref{de lcal à USC(lcal)}). 

To define the order, let us first recall that a function $f \in USC(X)$, for $X$ a sober topological space, is lower than a function $g \in USC(X)$ if and only if $g^*\colon [0,1] \rightarrow \Tcal(X)$ is lower than $f^*\colon [0,1] \rightarrow \Tcal(X)$ in the sense that for all $q \in {[0,1]}$, $g^*(q) \subset f^*(q)$. Thus, we define the order on $\USC$ by $f \leq g \Leftrightarrow \forall q \in {[0,1]}\;  g(q) \leq f(q)$, for all $f$ and $g$ being in $\USC$. Notice that $\vee$ and $\wedge$ respectively correspond to the lower and upper bounds of two functions for the pointwise order induced by $\lcal$. 

\begin{lem}[{\cite[Proposition 1.]{banaschewskiExtendedRealFunctions2012}}] \label{def wedge et vee} 

Let $f$ and $g \in \USC$, and $(f_i)_{i \in I} \in \USC^I$. 

\ind The maximum of $\USC$ is $\fct[\ul {1}][{[0,1]}, \lcal]{q, \bot}$ and its minimum is $\fct[\ul {0}][{[0,1]}, \lcal]{q, \acc{\bot, q = 0, \top}}$. 
\\ 
\ind The lower bound of $(f_i)_{i \in I}$ exists and assigns to each $ q \in \e<{{[0,1]}}$ $\bigvee[i \in I] f_i(q)$. 
\\ 
\ind The upper bound of $f$ and $g$ exists and assigns to each $ q \in \e<{{[0,1]}}$ $\bigvee[\wb[p][q]] f(p) \wedge g(p)$. 

Hence $\USC$ is a complete lattice. 

\end{lem} 

\begin{proof} 

If $\ul {1}$ and $\ul 0$ are \sp, then, it is clear that they are respectively the maximum and minimum of $\USC$. 

For all $(q_j)_{j \in J} \in [0,1]^J$, $\bigvee[j \in J] \ul {1}(q_j) = \bot = \ul {1}\left(\bigvee[j \in J] q_j\right)$. For all $(q_j)_{j \in J} \in [0,1]^J$, if $\bigvee[j \in J] q_j > 0$ then there exists $j_0 \in J$ such that $q_{j_0} > 0$, so $\bigvee[j \in J] \ul 0 (q_j) = \top = \ul 0 \left(\bigvee[j \in J] q_j\right)$, and if $\bigvee[j \in J] q_j = 0$ then $\bigvee[j \in J] \ul 0 (q_j) = \bot = \ul 0 \left(\bigvee[j \in J] q_j\right)$. 

Thus, $\ul {1} \et \ul 0 \in \USC$. 

Clearly, the function that assigns $\bigvee[i \in I] f_i(q)$ to each $q \in {[0,1]}$ preserves suprema. Thus, \linebreak $q \mapsto \bigvee[i \in I] f_i(q)$ is the lower bound of $(f_i)_{i \in I}$. 

Since $\otimes$ is non-decreasing in each coordinate, for all $U,\, V,\, U' \et V' \in \lcal$, $$(U \otimes V) \wedge (U' \otimes V') \geq (U \wedge U') \otimes (V \wedge V').$$ 

Let us denote by $f \wedge g$ the pointwise lower bound of $f$ and $g$. $\ouv{(f \wedge g)} \in \USC$. What we need to prove is that $\ouv{(f \wedge g)}$ is the actual upper bound of $f$ and $g$. 

However, by \bf{Lemma} \ref{adj}, for all $h \in \USC$, $$h \geq f \et h \geq g \Leftrightarrow \forall q \in [0,1] \; h(q) \leq f(q) \wedge g(q) \Leftrightarrow h \geq \ouv{f \wedge g}.$$ 

\end{proof} 

\begin{defi} \label{defi otimes} 

$\otimes$ and $\nrightarrow$ are defined on $\USC$ by, for all $f \et g \in \USC$, \linebreak $\fct[f \otimes g][{[0,1]}, \lcal]{q, f(q) \otimes g(q)}$ and $\fct[f \nrightarrow g][{[0,1]}, \lcal]{q, \bigvee[p \wb q] \bigwedge[r \geq p] f(r) \nrightarrow g(r)} = \ouv{(q \mapsto f(q) \nrightarrow g(q))}$. 

\end{defi} 

\begin{lem} \label{lem nrightarrow} 

For all $f \et g \in \USC$, $f \otimes g \et f \nrightarrow g \in \USC$. Moreover, $\otimes$ is associative and commutative and its neutral element is  $\ul 0$. Finally, $\nrightarrow$ is the residual of $\otimes$ and satisfies, for all $f \et g \in \USC$, $(f \nrightarrow g) (q) \leq f(q) \nrightarrow g(q)$. 

\end{lem} 

\begin{proof} 

Let $f$, $g$ et $h \in \USC$. For all $(q_i)_{i \in I} \in [0,1]^I$, 

\begin{align*} 
(f \otimes g)\left(\bigvee[i \in I] q_i\right) &= f\left(\bigvee[i \in I] q_i\right) \otimes g\left(\bigvee[i \in I] q_i\right)\\ 
&= \bigvee[i,j \in I] f(q_i) \otimes g(q_j)\\ 
&\leq \bigvee[i,j \in I] f(q_i \vee q_j) \otimes g(q_i \vee q_j)\\ 
&\leq \bigvee[i \in I] f(q_i) \otimes g(q_i),
\end{align*}

so $(f \otimes g)\left(\bigvee[i \in I] q_i\right) = \bigvee[i \in I] f(q_i) \otimes g(q_i) = \bigvee[i \in I] (f \otimes g)(q_i)$. Hence $f \otimes g \in \USC$. 

$\otimes$ is clearly associative and commutative. 

For all $q \in [0,1]$, $(\ul 0 \otimes f) (q) = \top \otimes f(q) = f(q)$. 

Since $f \nrightarrow g = \ouv{(q \mapsto f(q) \nrightarrow g(q))}$, according to \bf{Lemma} \ref{adj}: \begin{enumerate} 

\item[-\;] $f \nrightarrow g \in \USC$, 

\item[-\;] for all $q \in [0,1]$, $(f \nrightarrow g)(q) \leq f(q) \nrightarrow g(q)$ 

\item[-\;] and, for all $h \in \USC$, $$f \nrightarrow g \leq h \Leftrightarrow \forall q \in [0,1] \; h(q) \leq f(q) \nrightarrow g(q) \Leftrightarrow \forall q \in [0,1] \; f(q) \otimes h(q) \leq g(q) \Leftrightarrow g \leq f \otimes h.$$ 

\end{enumerate} 
\end{proof} 

We have thus proven the following theorem. 

\begin{thm} \label{USC is a crcl} 

$(\USC, \otimes, \nrightarrow)$ with the reverse order is a commutative residuated complete lattice. 

\end{thm} 

\begin{lem} 

For all $f \et g \in \USC$, for all $q \in [0,1]$, 

\begin{enumerate} 

\item $(f \vee g)  = \ouv{(q \mapsto f(q) \wedge g(q))}$, 

\item $(f \wedge g)  = \ouv{(q \mapsto f(q) \vee g(q))}$, 

\item $f \otimes g = \ouv{(q \mapsto f(q) \otimes g(q))}$, 

\item $f \nrightarrow g = \ouv{(q \mapsto f(q) \nrightarrow g(q))}$, 

\item $\ul 0 = \ouv{(q \mapsto \top)}$, 

\item $\ul 1 = \ouv{(q \mapsto \bot)}$. 

\end{enumerate} 

\end{lem} 

\begin{proof} 

Let $f \et g \in \USC$. The argument relies on \bf{Lemma} \ref{adj}. 

\begin{enumerate} 

\item $q \mapsto f(q) \wedge g(q)$ is non-decreasing, so, for all $p \in [0,1]$, $$\ouv{(q \mapsto f(q) \wedge g(q))}(p) = \bigvee[r \wb p] f(r) \wedge g(r) = (f \vee g)(p).$$ 

\item $f \wedge g \colon q \mapsto f(q) \vee g(q) \in \USC$, so $\ouv{(f \wedge g)} = f \wedge g$, i.e. $f \wedge g = \ouv{(q \mapsto f(q) \vee g(q))}$. 

\item $f \otimes g \colon q \mapsto f(q) \otimes g(q) \in \USC$, so $\ouv{(f \otimes g)} = f \otimes g$, i.e. $f \otimes g = \ouv{(q \mapsto f(q) \otimes g(q))}$. 

\item By definition of $f \nrightarrow g$. 

\item For all $q \in [0,1]$, $\ouv{(q \mapsto \top)}(q) = \bigvee[p \wb q] \top = \ul 0(q)$. 

\item \blue{For all $q \in [0,1]$, $\ouv{(q \mapsto \bot)}(q) = \bot = \ul 1(q)$}. 

\end{enumerate} 

\end{proof} 

\begin{defi} 

For the purpose of \bf{Lemma} \ref{link lcal} and \bf{Theorem} \ref{de lcal à USC(lcal)}, let \linebreak $\mathcal{L}_{crl} = \{\otimes, \; \nrightarrow, \; \vee, \; \wedge, \; \bot, \; \top\}$. We consider a countable set $\vcal$ of variables. Let $E$ denote the set of terms of $\mathcal{L}_{crl}$ and let's define $E_{lax}$ and $E_{colax}$ as follows: $$E_{lax} = \{\phi \in E \, | \, \forall f \colon \vcal \rightarrow \USC \et q \in [0,1] \; \phi[f](q) \leq \phi[f(q)]\}$$ $$E_{colax} = \{\phi \in E \, | \, \forall f \colon \vcal \rightarrow \USC \et q \in [0,1] \; \phi[f](q) \geq \phi[f(q)]\}.$$ 

{\remark Care must be taken when interpreting $\vee$ in $\USC$, that is $\wedge$, and vice-versa. Hence, for all terms $\phi \et \psi$ of $\mathcal{L}_{crl}$, $(\phi \vee \psi)[f](q) = (\phi[f] \wedge \psi[f])(q) = \phi[f](q) \vee \psi[f](q)$.} 

\end{defi} 

\begin{lem} $ $ \label{link lcal} \begin{enumerate} 

\item $E_{colax}$ contains the variables and constants and is stable by $\otimes$ and $\vee$. 

\item $E_{lax}$ contains the variables and constants and is stable by $\otimes$, $\vee$ and $\wedge$. 

\item $E_{colax} \nrightarrow E_{lax} \subset E_{lax}$ \blue{and $E_{lax} \nrightarrow \bot \subset E_{colax}$}. 

\end{enumerate} 

\end{lem} 

\begin{proof} To begin with, it is clear that both $E_{lax}$ and $E_{colax}$ contain the variables and constants. Let $f\colon \vcal \rightarrow \USC$. We here remind the reader that the interpretation of $\vee$ in $\USC$ is $\wedge$, and $\wedge$ is defined as the pointwise upper bound. 

\begin{enumerate} 

\item Let $\phi \et \psi \in E_{colax}$ and $q \in [0,1]$. 

$(\phi \otimes \psi)[f](q) = \phi[f](q) \otimes \psi[f](q) \geq \phi[f(q)] \otimes \psi[f(q)] = (\phi \otimes \psi) [f(q)]$ (\bf{Definition} \ref{defi otimes}). 

$(\phi \vee \psi)[f](q) = \phi[f](q) \vee \psi[f](q) \geq \phi[f(q)] \vee \psi[f(q)] = (\phi \vee \psi) [f(q)]$ (by \bf{Lemma} \ref{def wedge et vee}). 

\item Let $\phi \et \psi \in E_{colax}$ and $q \in [0,1]$. 

$(\phi \otimes \psi)[f](q) = \phi[f](q) \otimes \psi[f](q) \leq \phi[f(q)] \otimes \psi[f(q)] = (\phi \otimes \psi) [f(q)]$. 

$(\phi \vee \psi)[f](q) = \phi[f](q) \vee \psi[f](q) \leq \phi[f(q)] \vee \psi[f(q)] = (\phi \vee \psi) [f(q)]$. 

$(\phi \wedge \psi)[f](q) \leq \phi[f](q) \wedge \psi[f](q) \leq \phi[f(q)] \wedge \psi[f(q)] = (\phi \wedge \psi) [f(q)]$. 

\item Let $\phi \in E_{lax}$, $\psi \in E_{colax}$ and $q \in [0,1]$. \blue{$$(\psi \nrightarrow \phi)[f](q) \leq \psi[f](q) \nrightarrow \phi[f](q) \leq \psi[f(q)] \nrightarrow \phi[f(q)]\text{ (by \bf{Lemma} \ref{lem nrightarrow})}.$$} 

\end{enumerate} 
\end{proof} 

\begin{thm} \label{de lcal à USC(lcal)} 

Let $\phi_0, \ldots, \phi_k \in E_{lax}$ and $\psi_0, \ldots, \psi_k \in E_{colax}$. 

If $\lcal \models (\psi_1 \leq \phi_1 , \ldots , \psi_k \leq \phi_k) \Rightarrow \phi_0 \leq \psi_0$, then \\ $\USC \models (\psi_1 \geq \phi_1 , \ldots , \psi_k \geq \phi_k) \Rightarrow \phi_0 \geq \psi_0$. 

\end{thm} 

\begin{proof} 

Let $f\colon \vcal \rightarrow \USC$, $q \in [0,1]$ $\phi_0, \ldots, \phi_k \in E_{lax}$ and $\psi_0, \ldots, \psi_k \in E_{colax}$ such that $\lcal \models (\psi_1 \leq \phi_1 , \ldots , \psi_k \leq \phi_k) \Rightarrow \phi_0 \leq \psi_0$. Assume $\USC \models \psi_1[f] \geq \phi_1[f], \ldots, \psi_k[f] \geq \phi_k[f]$. 

For all $1 \leq i \leq k$, $\psi_i[f(q)] \leq \psi_i[f](q) \leq \phi_i[f](q) \leq \phi_i[f(q)]$, so $$\lcal \models (\psi_1[f(q)] \leq \phi_1[f(q)] , \ldots , \psi_k[f(q)] \leq \phi_k[f(q)].$$ Thus $\lcal \models \phi_0[f(q)] \leq \psi_0[f(q)].$ Hence, $\phi_0[f](q) \leq \phi_0[f(q)] \leq \psi_0[f](q) \leq \psi_0[f](q)$. 

Hence $\USC \models \phi_0[f] \geq \psi_0[f]$. 
\end{proof} 

We can also embed $\lcal$ into $\USC$. 

\begin{defi} 

To each $U \in \lcal$, we associate the \emph{0-indicator} $\fct[0_U][{[0,1]}, \lcal]{q, \acc{U, {1} > q \wa 0, \bot, q \nwa 0}}$. 

To each $f \in \USC$, we associate $U_f \colon = f(1)$. 

\end{defi} 

\begin{lem} 

For every $U \in \lcal$, $0_U \in \USC$, and, for all $f \in \USC$, $$U_f \leq U \Leftrightarrow f \geq 0_U.$$ 

Moreover, for all $U \in \lcal$, $U_{0_U} = U$, so $U \mapsto 0_U$ is an order embedding of $\lcal$ into $\USC^{op}$ and $f \mapsto U_f$ is onto. 

\end{lem} 

\begin{proof} 

For every $U \in \lcal \et f \in \USC$, $$U_f \leq U \Leftrightarrow \forall q \in {[0,1]}\, f(q) \leq U \Leftrightarrow \forall q \in {[0,1]}\, f(q) \leq 0_U(q) \Leftrightarrow f \geq 0_U.$$ 

Moreover, for all $U \in \lcal$, $U_{0_U} = 0_U(1) = U$. 
\end{proof} 

\begin{lem} \label{Crl embedding} 

$\fct[\lcal^{op}, \USC]{U, 0_U }$ is an $\mathcal{L}_{crl}$-embedding. 

\end{lem} 

\begin{proof} 

$0_\bot = \ul 1$. 

$0_\top = \ul 0$. 

Let $U \et V \in \lcal$ and $q > 0$. Since $q > 0$, $0_U(q) = U \et 0_V(q) = V$. 

$0_{U \otimes V} (q) = U \otimes V = 0_U(q) \otimes 0_V(q) = (0_U \otimes 0_V)(q)$. 

$0_{U \wedge V} (q) = U \wedge V = 0_U(q) \wedge 0_V(q) = (0_U \vee 0_V)(q)$. 

$0_{U \vee V} (q) = U \vee V = 0_U(q) \vee 0_V(q) = (0_U \wedge 0_V)(q)$. 

$0_U \nrightarrow 0_V (q) = \bigvee[p \wb q] \; \bigwedge[r \geq p] 0_U(r) \nrightarrow 0_V(r) = \bigvee[0 \wb p \wb q] \; \bigwedge[r \geq p] U \nrightarrow V = U \nrightarrow V = 0_{U \nrightarrow V} (q)$. 
\end{proof} 

Now, we shall state a strong converse of \bf{Theorem} \ref{de lcal à USC(lcal)}, which is an immediate consequence of \bf{Lemma} \ref{Crl embedding}. 

\begin{thm} \label{first half} 

For all terms $\phi_0, \ldots, \phi_k$ and $\psi_0, \ldots, \psi_k$ in the language $\mathcal{L}_{crl}$, if $\USC$ satisfies $(\psi_1 \geq \phi_1 , \ldots , \psi_k \geq \phi_k) \Rightarrow \phi_0 \geq \psi_0$, then $\lcal \models (\psi_1 \leq \phi_1 , \ldots , \psi_k \leq \phi_k) \Rightarrow \phi_0 \leq \psi_0$. 

\end{thm} 

\subsection{The structure inherited from $[0,1]$ by convolution} \label{subsection structure de [0,1]} 

\subsubsection{The action of $\Lusc$ on $USC(\lcal)$} 

Here, we want to give an interpretation of any upper semi-continuous function from $[0,1]_u^n$ to $[0,1]$, $n \in \nn$, in $\USC$ and give some kind of formulas in this language that are true in $\USC$ if they are true in $[0,1]$ (\bf{Theorem} \ref{formule formules}). To this purpose, we use a copairing-like notion. We will then aim at reducing the language to only a few symbols. Before getting to the heart of the matter, we need some notations. 

\begin{nota} \label{nota a_phi}

We recall here that the topology of $[0,1]_u$ is denoted $\Tcal([0,1]_u)$ and the topology of $[0,1]_u^n$ is denoted $\Tcal([0,1]_u^n)$, for all $n \in \nn$ (\textup{\bf{Notation} \ref{locale}}). 

Let $n \in \nn$. For all $a \in \Lusc[n]$, let's denote by $a^*$ the sup-preserving function from $[0,1]$ to $\Tcal([0,1]_u^n)$ induced by $a$, namely $\fct[a^*][{[0,1]}, \Tcal({[0,1]}_u^n)]{q, a^{-1}({[0,q)})}$. \\ For all term $\phi$ in the language $\Lusc$, let $a_\phi$ be its interpretation in $[0,1]$. For all $f \in \USC^n$, let $G(f) = (G(f_1), \ldots, G(f_n))$ (cf. \bf{Corollary} \ref{corscale}). 

Finally, let $\vcal$ denote a countable set of variables. For all $v \in \vcal^n$ and $p\colon \vcal \rightarrow [0,1]$ and $f\colon \vcal \rightarrow \USC$, let $p_v = (p(v_1), \ldots, p(v_n))$ and $f_v = (f(v_1), \ldots, f(v_n))$. 

\end{nota} 

\begin{defi} 

Let $n \in \nn$, $X_1, \ldots, X_n$ be topological spaces , $\kcal$ be a commutative residuated complete lattice and $f_1 \colon \Tcal(X_1) \rightarrow \kcal, \ldots, f_n \colon \Tcal(X_n) \rightarrow \kcal$ be morphisms of commutative residuated complete lattices. 

We define $\fct[\cop(f_1, \ldots, f_n)][\Tcal\left(\prod[i = 1][n] X_i\right), \kcal]{U, \bigvee[U_1 \times \ldots \times U_n \subset U] f_1(U_1) \otimes \ldots \blue\otimes f_n(U_n)}$. 

\end{defi} 

\begin{lem} 

$\cop(f_1, \ldots, f_n)$ is the smallest lax morphism of commutative residuated complete lattices $g\colon \Tcal\left(\prod[i = 1][n] X_i\right) \rightarrow \kcal$ such that, for all $1 \leq i \leq n$, $$g(X_1 \times \ldots \times X_{i-1} \times U_i \times X_{i+1} \times \ldots \times X_n) \geq f_i(U_i).$$ Moreover, for $g = \cop(f_1, \ldots, f_n)$, the preceding inequality is an equality. 

\end{lem} 

\begin{proof} 

We prove it for $n = 2$ only, the proof being the same in general. 

Let $X_1$ and $X_2$ be topological spaces and $f_1 \colon \Tcal(X_1) \rightarrow \kcal$ and $f_2 \colon \Tcal(X_2) \rightarrow \kcal$ be morphisms of commutative residuated complete lattices. For all $U_1 \in \Tcal(X_1) \et U_2 \in \Tcal(X_2)$, \linebreak $\cop(f_1, f_2)(U_1 \times U_2) = f_1(U_1) \otimes f_2(U_2)$, so $\cop(f_1, f_2)(U_1 \times X_2) = f_1(U_1)$ and \linebreak $\cop(f_1, f_2)(X_1 \times U_2) = f_2(U_2)$. 

$\cop(f_1, f_2)$ clearly preserves upper bounds. 

For all $U \et V \in \Tcal(X_1 \times X_2)$, \begin{align*} 
\cop(f_1,f_2) (U \cap V) &= \bigvee[U_1 \times U_2 \subset U \cap V] f_1(U_1) \otimes f_2(U_2)\\ 
&= \bigvee[U_1 \times U_2 \subset U] \bigvee[V_1 \times V_2 \subset V] f_1(U_1 \cap V_1) \otimes f_2(U_2 \cap V_2)\\ 
&\geq \bigvee[U_1 \times U_2 \subset U] \; \bigvee[V_1 \times V_2 \subset V] f_1(U_1) \otimes f_1(V_1) \otimes f_2(U_2) \otimes f_2(V_2)\\ 
&= \cop(f_1, f_2) (U) \otimes \cop(f_1, f_2) (V). 
\end{align*} 

For all lax morphism of commutative residuated complete lattice $g \colon \Tcal\left(X_1 \times X_2\right) \rightarrow \kcal$ such that, for all $U_1 \in \Tcal(X_1)$ and $U_2 \in \Tcal(X_2)$, $g(U_1 \times X_2) \geq f_1(U_1)$ and $g(X_1 \times U_2) \geq f_2(U_2)$, for all $U \in \Tcal(X_1 \times X_2)$, \begin{align*} g(U) &= \bigvee[U_1 \times U_2 \subset U] g(U_1 \times U_2) 
\\ 
&= \bigvee[U_1 \times U_2 \subset U] g((U_1 \times X_2) \cap (X_1 \times U_2)) 
\\ 
&\geq \bigvee[U_1 \times U_2 \subset U] g(U_1 \times X_2) \otimes g(X_1 \times U_2) \\ 
&\geq \bigvee[U_1 \times U_2 \subset U] f_1(U_1) \otimes f_2(U_2) 
\\ 
&= \cop(f_1, f_2)(U). \end{align*}

\end{proof} 

\begin{cor} 

For all $n \in \nn$ and $f \in \USC^n$, $\cop(G(f)) \in \Crcl(\Tcal([0,1]_u^n), \lcal)$. 

\end{cor} 

\begin{defi} \label{interpretation} 

Let $n \in \nn$ and $a \in \Lusc[n]$. 

We define the interpretation of $a$ in $\USC$ by, for all $f \in \USC^n$, $a(f) = \cop(G(f)) \circ a^*$, i.e. for all $q \in [0,1]$, $a(f)(q) = \bigvee[\prod[i = 1][n] U_i \subset a^{-1}({[0,q)})] G(f_1)(U_1) \otimes \ldots \otimes G(f_n)(U_n)$. 

\end{defi} 

\begin{remark} \label{topological case} 

Notice that, if $\lcal$ is the topology of a topological space $X$, then, for all \linebreak $f \in USC(X)^n$ $\cop(f_1^*, \ldots, f_n^*) \colon \Tcal([0,1]_u^n) \rightarrow \lcal$ is the map associated to $f\colon X \rightarrow [0,1]^n$ through the isomorphism $USC(X) \simeq USC(\lcal)$, and thus $a(f) = (a \circ f)^*$. 

\end{remark} 

\begin{lem} \label{monotonie de l'action} 

Let $(a_i)_{i \in I} \in \Lusc[n]^I$ and $f \in \USC^n$. 

$\left(\bigwedge[i \in I] a_i\right)(f) = \bigwedge[i \in I] a_i(f)$. Thus, $a \mapsto a(f)$ is non-decreasing. 

\end{lem} 

\begin{proof} 

For all $q \in [0,1]$, \begin{align*} \left(\bigwedge[i \in I] a_i\right)(f)(q) &= \cop(G(f)) \left(\left(\bigwedge[i \in I] a_i\right){}^*(q)\right) = \cop(G(f)) \left(\bigvee[i \in I] a_i^*(q)\right) 
\\ 
&= \bigvee[i \in I] \cop(G(f)) (a_i^*(q)) = \left(\bigwedge[i \in I] a_i(f)\right)(q), \end{align*} so $\left(\bigwedge[i \in I] a_i\right)(f) = \bigwedge[i \in I] a_i(f)$. 
\end{proof} 

\begin{lem} \label{operad action} 

For all $a \in \Lusc[n]$ and $b_1 \in USC([0,1]_u^{k_1}), \ldots, b_n \in \USC([0,1]_u^{k_n})$, for all $f \in \USC^{k_1} \times \ldots \times \USC^{k_n}$, $$(a \circ (b_1, \ldots, b_n))(f) = a(b_1(f_{1,1}, \ldots, f_{1,k_1}), \ldots, b_n(f_{n,1}, \ldots, f_{n,k_n})).$$ 

\end{lem} 

\begin{proof} 

Let $a \in \Lusc[n]$, $b_1 \in USC([0,1]^{k_1}), \ldots, b_n \in \USC([0,1]^{k_n})$ and \linebreak $f \in \USC^{k_1} \times \ldots \times \USC^{k_n}$. Let $f_i$ denote $(f_{i,1}, \ldots, f_{i,k_i})$. 

\bf{Claim:} $\cop(G(f)) \circ \cop(b_1^{-1}, \ldots, b_n^{-1}) = \cop(\cop(G(f_1)) \circ b_1^{-1}, \ldots, \cop(G(f_n)) \circ b_n^{-1})$. Indeed, for all $U \in \Tcal\left(\prod[i = 1][n] X_i\right)$, \begin{align*} \cop(G(f)) \circ \cop(b_1^{-1}, \ldots, b_n^{-1}) (U) &= \cop(G(f))\left(\bigcup[\prod[i = 1][n] U_i \subset U] (b_1^{-1}(U_1) \times \ldots \times b_n^{-1}(U_n))\right) 
\\ 
&= \bigvee[\prod[i = 1][n] U_i \subset U] \cop(G(f))(b_1^{-1}(U_1) \times \ldots \times b_n^{-1}(U_n)) 
\\ 
&= \bigvee[\prod[i = 1][n] U_i \subset U] \cop(G(f_1))(b_1^{-1}(U_1)) \otimes \ldots \otimes \cop(G(f_n))(b_n^{-1}(U_n)) 
\\ 
&= \cop(\cop(G(f_1)) \circ b_1^{-1}, \ldots, \cop(G(f_n)) \circ b_n^{-1})(U). \end{align*} 

By noticing that, \color{blue} $(a \circ (b_1, \ldots, b_n))^{-1} = \cop(b_1^{-1}, \ldots, b_n^{-1}) \circ a^{-1}$, we can conclude that, for all $q \in [0,1]$, \begin{align*}(a \circ (b_1, \ldots, b_n))(f)(q) &= \cop(G(f)) \circ (a \circ (b_1, \ldots, b_n))^{-1}([0,q)) 
\\ 
&= \cop(G(f)) \circ \cop(b_1^{-1}, \ldots, b_n^{-1}) \circ a^{-1}([0,q)) 
\\ 
&= \cop(\cop(G(f_1)) \circ b_1^{-1}, \ldots, \cop(G(f_n)) \circ b_n^{-1}) \circ a^*(q) 
\\ 
&= a(b_1(f_1), \ldots, b_n(f_n))(q). \end{align*} \color{black} 


\end{proof} 

\begin{defi} 

A term of $\Lusc$ is linear if any variable occurs at most once in it. 

\end{defi} 

\begin{lem} \label{equation formules} 

Let $\phi[v_1, \ldots, v_n]$ and $\psi[v_1, \ldots, v_n]$ be two terms of $\Lusc$ and \linebreak $f\colon \vcal \rightarrow \USC$. Let $v = (v_1, \ldots, v_n)$. 

$\phi[f] \geq a_\phi(f_v)$, with equality if $\phi$ is linear (\bf{Notation} \ref{nota a_phi}) and \blue{$a_{\phi \wedge \psi}(f_v) = a_\phi(f_v) \wedge a_\psi(f_v)$}. 

\end{lem} 

\begin{proof} 

By induction over the linear terms of $\Lusc$. 

The proposition is true for constants and variables. 

Let $a \in \Lusc[n]$, $\phi_1[v_{1,1}, \ldots, v_{1,k_1}]$, $\ldots$, $\phi_n[v_{n,1}, \ldots, v_{n,k_n}]$ be terms of $\Lusc$, \linebreak $f\colon \vcal \rightarrow \USC$. Assume that $\phi_1[f] = a_{\phi_1}(f_{v_1})$, $\ldots$, $\phi_n[f] = a_{\phi_n}(f_{v_n})$ and that $a(\phi_1, \ldots, \phi_n)$ is linear, i.e., for all $1 \leq i, i' \leq n$, $1 \leq j \leq k_i$ and $1 \leq j' \leq k_{i'}$, if $(i,j) \neq (i',j')$, then $v_{i,j} \neq v_{i',j'}$. 

By linearity of $a[u_1, \ldots, u_n]$, for all $q \in [0,1]^n$, $a_{a[u_1, \ldots, u_n]}(q) = a(q_1, \ldots, q_n)$, so, thanks to \bf{Lemma} \ref{operad action}, 

\begin{align*} 
a[\phi_1/u_1, \ldots, \phi_n/u_n][f] &= a[\phi_1[f_{v_1}]/u_1, \ldots, \phi_n[f_{v_n}]/u_n]\\ 
&= a[a_{\phi_1}(f_{v_1})/u_1, \ldots, a_{\phi_n}(f_{v_n})/u_n]\\ 
&= a_{a[u_1, \ldots, u_n]} (a_{\phi_1}(f_{v_1}), \ldots, a_{\phi_n}(f_{v_n}))\\ 
&= a(a_{\phi_1}(f_{v_1}), \ldots, a_{\phi_n}(f_{v_n}))\\ 
&= (a \circ (a_{\phi_1}, \ldots, a_{\phi_n})) (f_v)\\ 
&= a_{a(\phi_1, \ldots, \phi_n)} (f_v)\\ 
\end{align*}

Hence, for all linear term $\phi[v]$ of $\Lusc$ and $f\colon \vcal \rightarrow \USC$, $\phi[f] = a_\phi(f_v)$. 

Let $\phi[v_1, \ldots, v_n]$ be a term of $\Lusc$, $f\colon \vcal \rightarrow \USC$ and $q \in [0,1]$. Let us denote by $k_i$ the number of occurrences of $v_i$ in $\phi$, for all $1 \leq i \leq n$. Let also, for all $k \in \nn$ and all set $u$, $k u$ denote the $k$-uple $(u, \ldots, u)$. There exists a linear term $\phi_0\left[u_{1, 1}, \ldots, u_{n, k_n}\right]$ such that $\phi[v_1, \ldots, v_n] = \phi_0[k_1 v_1, \ldots, k_n v_n]$. \begin{align*} 
\phi[f](q) &= \phi_0[f_{v_1}/u_{1,1}, \ldots, f_{v_1}/u_{1,k_1}, \ldots,f_{v_n}/u_{n,1}, \ldots, f_{v_n}/u_{n,k_n}](q) = a_{\phi_0}(g)(q) 
\\ 
&= \bigvee[\rescale{0.5}{\prod[i = 1][n] \prod[j = 1][k_i] U_{i,j} \subset a_{\phi_0}^*(q)}] \bigotimes[i = 1][n] \bigotimes[j = 1][k_i] G(f_{v_i})(U_{i,j}) 
\\ 
&\leq \bigvee[\rescale{0.5}{\prod[i = 1][n] \prod[j = 1][k_i] U_{i,j} \subset a_{\phi_0}^*(q)}] \bigotimes[i = 1][n] G(f_{v_i})\left(\bigcap[j = 1][k_i] U_{i,j}\right) && \target{( 1 )}
\\ 
&\ust{=}{(\ast)} \bigvee[\rescale{0.5}{\prod[i = 1][n] U_i \subset a_\phi^*(q)}] \bigotimes[i = 1][n] G(f_{v_i})(U_i) 
\\ 
&= a_\phi(f_v)(q) 
\end{align*} 

$(\ast)$ is justified by the following : for all family of open subsets $(U_{i,j})_{\ust{1 \leq i \leq n}{1 \leq j \leq k_i}}$ of $[0,1]$, $$\prod[i = 1][n] \prod[j = 1][k_i] U_{i,j} \subset a_{\phi_0}^*(q) \Rightarrow \prod[i = 1][n] \bigcap[j = 1][k_i] U_{i,j} \subset a_\phi^*(q)$$ and, for all family of open subsets of $[0,1]$ $(U_i)_{1 \leq i \leq n}$, $$\prod[i = 1][n] U_i \subset a_\phi(q) \Rightarrow \prod[i = 1][n] \prod[j = 1][k_i] U_i \subset a_{\phi_0}^*(q).$$ 

Let $\psi[v]$ be another term of $\Lusc$. 

\color{blue} First of all, we notice that $a_{\phi \wedge \psi} = a_\phi \wedge a_\psi$. Thus, $$a_{\phi \wedge \psi}(f_v)(q) = \hat{f_v} \circ (a_{\phi} \wedge a_{\psi})^*(q) = \hat{f_v}(a_\phi^*(q) \cup a_\psi^*(q)) = \hat{f_v}(a_\phi^*(q)) \vee \hat{f_v}(a_\psi^*(q)) = a_\phi(f_v) \wedge a_\psi(f_v)(q).$$ \color{black} 
\end{proof} 

\begin{remark} \label{Rq tense = inf} 

If $\otimes = \wedge$, then inequality $\link{( 1 )}$ is actually an equality and thus, for every term $\phi[v]$ of $\Lusc$ and $f\colon \vcal \rightarrow \USC$, $\phi[f] = a_\phi(f_v)$. 

\end{remark} 

We will crucially use the next result to give a conceptual proof of soundness in section \ref{section algebraic axiomatisation}. We here emphasize that some equations true in $[0,1]$ become false in $\USC$, such as $\max(v,v) \leq v$, which becomes $v \otimes v \leq v$ as stated by \bf{Lemma} \ref{calculatoire}, and may not be true in $\USC$. 

\begin{thm} \label{formule formules} 

Let $\phi_1[v], \ldots, \phi_n[v]$ and $\psi[v]$ be terms in the language $\Lusc$ and assume that for all $1 \leq i \leq n$ $\phi_i$ is linear. 

Then, if $[0,1] \models \bigwedge[i = 1][n] \phi_i \leq \psi$, $\USC \models \bigwedge[i = 1][n] \phi_i \leq \psi$. 

\end{thm} 

\begin{proof} 

Assume $[0,1] \models \bigwedge[i = 1][n] \phi_i \leq \psi$ and let $f\colon \vcal \rightarrow \USC$. 

\it{Remark} \ref{Rq tense = inf} implies that for all $p \colon \vcal \rightarrow [0,1]$, $a_{\bigwedge[i = 1][n] \phi_i}(p_v) = \bigwedge[i = 1][n] \phi_i[p] \leq \psi[p] = a_\psi(p_v)$, so \linebreak $a_{\bigwedge[i = 1][n] \phi_i} \leq a_\psi$. According to \bf{Lemma} \ref{equation formules} and \bf{Lemma} \ref{monotonie de l'action}, $$\left(\bigwedge[i = 1][n] \phi_i\right) [f] = \bigwedge[i = 1][n] \phi_i[f] = \bigwedge[i = 1][n] a_{\phi_i}(f_v) = a_{\bigwedge[i = 1][n] \phi_i}(f_v) \leq a_\psi(f_v) \leq \psi[f].$$ 
\end{proof} 

Seeing elements of $[0,1]$ as functions from $[0,1]^0$ to $[0,1]$, we can then embed $[0,1]$ into $\USC$. 

\begin{defi} 

To each $p \in [0,1]$, we associate $\fct[\ul p][{{[0,1]}} , \lcal]{q, \acc{\top, q > p, \bot}}$. We also recall that, following \bf{Definition} \ref{interpretation}, for all $p \in [0,1]$, $p$ defines a function $\fct[\USC^0, \USC]{\ast, p(\ast)}$. 

\end{defi} 

\begin{lem} \label{def cst} 

Let $p \in [0,1]$. We denote by $\ast$ the element of $\USC^0$. 

$\ul p = p(\ast)$. 

\end{lem} 

\begin{proof} 

For all $q \in [0,1]$, $p^*(q) = \acc{\{\emptyset\}, q > p, \emptyset}$, so $p(\ast)(q) = \acc{\top, q > p, \bot} = \ul p (q)$. 
\end{proof} 

\begin{lem} \label{[0,1] embedding} 

$\fct[{{[0,1]}} , \USC]{p, \ul p}$ is a $\Lusc$-embedding. Moreover, for all \linebreak $p \et q \in [0,1]$, $\ul p \otimes \ul q = \ul p \vee \ul q$. 

\end{lem} 

\begin{proof} 

Let $a \in \Lusc[n]$ and $p \in [0,1]^n$. 

$[0,1] \models a[p] \leq a(p)$ and $a[u_1, \ldots, u_n]$ is linear, so $\USC \models a[p] \leq a(p)$, i.e. $a(\ul p) \leq \ul{a(p)}$. \linebreak $[0,1] \models a(p) \leq a[p]$ and $a(p)$ is a constant and thus a linear term, so $\USC \models a(p) \leq a[p]$, i.e. $\ul{a(p)} \leq a(\ul p)$. Hence $a(\ul p) = \ul{a(p)}$. 

For all $p$, $q \et r \in [0,1]$, 
\\ 
\adjusttopage{$\ul p \vee \ul q (r) = \bigvee[r' \wb r] \ul p(r') \wedge \ul q(r') = \acc{\top, p \vee q \wb r, \bot}$ and $\ul p \otimes \ul q (r) = \ul p(r) \otimes \ul q(r) = \acc{\top, p \vee q \wb r, \bot}$.} 
\end{proof} 

We shall now state and prove the strong converse of \bf{Theorem} \ref{formule formules}. 

\begin{thm} \label{formules [0,1]} 

For all terms $\phi \et \psi$ of $\Lusc$, if $\USC \models \phi \leq \psi$, then \linebreak $[0,1] \models \phi \leq \psi$. 

\end{thm} 

\begin{proof} 

Let $\phi \et \psi$ be terms of $\Lusc$ such that $\USC \models \phi \leq \psi$. Let \linebreak $p\colon \vcal \rightarrow [0,1]$. 

$\USC \models \phi[\ul p] \leq \psi[\ul p]$, so, according to \bf{Lemma} \ref{[0,1] embedding},  $\USC \models \ul{a_\phi(p)} \leq \ul{a_\psi(p)}$. Hence, for all $q \in [0,1]$, $q > a_\psi(p) \Rightarrow q > a_\phi(p)$ and thus $a_\phi(p) \leq a_\psi(p)$, i.e. $[0,1] \models \phi[p] \leq \psi[p]$. 
\end{proof} 

\subsubsection{Reduction to the language $\mathcal{L}_{[0,1]}$} 

\begin{nota} \label{topological space} 

For all ordered topological space $(X, \leq)$, let's denote by $C^0_\nearrow(X)$ the set of continuous non-decreasing functions from $X$ to $[0,1]$, by $C^0(X)$ the set of continuous functions from $X$ to $[0,1]$ and by $USC_\nearrow(X)$ the set of non-decreasing upper semi-continuous functions from $X$ to $[0,1]$. 

\end{nota} 

\begin{defi} \label{Hausdorff} 

An ordered topological space $(X, \leq)$ is said \emph{Hausdorff} if $\leq$ is closed. 

\end{defi} 

\begin{lem} \label{X_u} 

Let $X$ be an ordered topological space and let's denote by $X_u$ the topological space whose underlying set is $X$ and whose topology is the set of all downward closed open sets of $X$. $$USC_\nearrow(X) = USC(X_u).$$ 

\end{lem} 

\begin{proof} 

For all $f \in USC_\nearrow(X)$, for all $q \in [0,1]$, $f^{-1}([0,q))$ is a downward closed open of $X$, so $f \in USC(X_u)$. 

Let $f \in USC(X_u)$. For all $q \in [0,1]$, $f^{-1}([0,q))$ is an open of $X$, so $f$ is \usc. Let $x \leq y \in X$. For all $q > f(y)$, $y \in f^{-1}([0,q))$, which is downward closed, so $x \in f^{-1}([0,q))$, i.e. $f(x) < q$. Thus $f(x) \leq f(y)$. 

Hence $f \in USC_\nearrow(X)$. 
\end{proof} 

\begin{lem}[Ordered version of Urysohn’s Lemma, {\cite[Chapter I, Theorem 1]{nachbinTopologyOrder1965}}] \label{Urysohn croissant} 

Let $X$ be a compact Hausdorff ordered topological space. For all downward closed subset $F$ and upward closed subset $G$ of $X$ such that $F \cap G = \emptyset$, there exists $f \in C^0_\nearrow(X)$ such that $\restr{f}[F] = 0$ and $\restr{f}[G] = 1$. 

\end{lem} 

\begin{lem} \label{cor Urysohn croissant} 

Let $X$ be a compact Hausdorff ordered topological space. $$USC_\nearrow(X) = \left\{\bigwedge A, A \subset C^0_\nearrow(X)\right\}.$$ 

\end{lem} 

\begin{proof} 

$C^0_\nearrow(X) \subset USC_\nearrow(X)$ and $USC_\nearrow(X)$ is stable by lower bounds, so \linebreak $\{\bigwedge A, A \subset C^0_\nearrow(X)\} \subset USC_\nearrow(X)$. 

Let $f \in USC_\nearrow(X)$, $x \in X$ and $q \in [0,1]$ such that $q > f(x)$. $f^{-1}( [q,1] )$ and $\{y \in X \,|\, y \leq x\}$ are closed, $f^{-1}( [q,1] )$ is upward closed, $\{y \in X \,|\, y \leq x\}$ is downward closed and $f^{-1}( [q,1] ) \cap \{x\} = \emptyset$, so, thanks to \textbf{Lemma} \ref{Urysohn croissant}, there exists $g \in C^0_\nearrow(X)$ such that $g(x) = 0 \et \restr{g}[{(f^{-1}( [q,1] ))}] = 1$, which gives $\ul q \+ g \geq f$ and $(\ul q \+ g)(x) = q$. Hence $f = \bigwedge[\ust{h \in C^0_\nearrow(X)}{h \geq f}] h$. 
\end{proof} 

\begin{thm} \label{reduction to Lcinc} 

There exists a unique family of lower-bounds-preserving functions \linebreak $\left(\cdot\colon \Lusc[n] \rightarrow \USC^{(\USC^n)}\right)_{n \in \nn \cup\{0\}}$ that is associative in the sense that, for all \linebreak $a \in \Lusc[n]$, $(b_1, \ldots, b_n) \in \prod[i = 1][n] \Lusc[k_i]$ and $(f_{i,1}, \ldots, f_{i,k_i})_{i \in \dc[1,n]} \in USC(\lcal)^{\sum[i = 1][n] k_i}$, 

$$(a \circ (b_1, \ldots, b_n)) \cdot (f_{1,1}, \ldots, f_{n,k_n}) = a \cdot (b_1 \cdot (f_{1,1}, \ldots, f_{1,k_1}), \ldots, b_n \cdot (f_{n,1}, \ldots, f_{n,k_n}))$$ 

and, for all $a \in \Lcinc[n]$ and $f \in \USC^n$, $a \cdot f = a(f)$. 

\end{thm} 

\begin{proof} 

By defining, for all $a \in \Lusc[n]$ and $f \in \USC^n$, $a \cdot f = a(f)$, we obtain a family of functions $\left(\cdot\colon \Lusc[n] \rightarrow \USC^{(\USC^n)}\right)_{n \in \nn \cup\{0\}}$ satisfying the required properties. 

Let now $\left(\cdot\colon \Lusc[n] \rightarrow \USC^{(\USC^n)}\right)_{n \in \nn \cup\{0\}}$ be such a family of functions. 

According to \bf{Lemma} \ref{cor Urysohn croissant}, since $[0,1]^n$ is compact and Hausdorff, for all $a \in \Lusc[n]$, $a = \bigwedge[\ust{\rescale{0.5}{b \in \Lcinc[n]}}{b \geq a}] b$, so, for all $f \in \USC^n$, $a \cdot f = \bigwedge[\ust{\rescale{0.5}{b \in \Lcinc[n]}}{b \geq a}] (b \cdot f) = \bigwedge[\ust{\rescale{0.5}{b \in \Lcinc[n]}}{b \geq a}] b(f) = a(f)$. 
\end{proof} 

\begin{nota} 

We recall that, we have defined $\fct[j][{[0,1]}, {[0,1]}]{x, 2\left(x \dotm \frac{1}{2}\right)}$, that is a non-decreasing continuous function and the maximum and minimum of two elements $x \et y$ of $[0,1]$ are respectively denoted $\max[x,y]$ and $\min[x,y]$ (cf. \bf{Definition} \ref{Définition du langage}). 

\end{nota} 

The interpretation of $\Lcinc$ on $\USC$ gives us an interpretation of the language \linebreak $\mathcal{L}_{[0,1]} = \{\min,\, \max,\, \dotp,\, 2 \cdot,\, \frac{\cdot}{2},\, j,\, \ul 0, \ul 1\}$ in $\USC$. The aim of this subsubsection is to prove that the action of $\Lcinc$ on $\USC$ is entirely characterised by the interpretation of $\mathcal{L}_{[0,1]}$ in $\USC$ (\bf{Theorem} \ref{PositiveAction}). For this, we need a Stone-Weierstrass type theorem and metric on $\USC$. Before going any further, we need some calculatory results.

\begin{lem} \label{calculatoire} 

Let $f \et g \in \USC$, $q$ and $q' \in [0,1]$. \begin{enumerate} 

\item $\max[f , g](q) = f(q) \otimes g(q)$. 

\item $\min[f , g](q) = f(q) \vee g(q) = (f \wedge g)(q)$. 

\item $(f \+ g)(q) = \bigvee[p \dotp r \wb q] f(p) \otimes g(r) = \bigvee[p \wb q] f(p) \otimes g(q - p)$ and $\+$ admits a residual $\-$, defined by $$(f \- g)(q) = \bigvee[p < q] \; \bigwedge[r \geq p] g(r - p) \nrightarrow f(r).$$ 

\item $(f \dotp q')(q) = f(q \dotm q') = (f \+ \ul q')(q)$. 

\item $(f \dotm q')(q) = (f \- \ul q')(q) = j\left(\frac{f}{2} \+ \ul{\frac{1 - \blue{q'}}{2}}\right)(q)$. 

\item $(2f)(q) = f\left(\frac{q}{2}\right)$ and $j_\ast(f)(q) = f(j(q))$. 

\end{enumerate} 

\end{lem} 

\begin{proof} 

Let $f \et g \in \USC$, $q$ and $q' \in (0,1]$. 

\begin{enumerate} 

\item \begin{align*} \max[f , g](q) = \bigvee[U \times V \subset \max[]^*(q)] G(f)(U) \otimes G(g)(V) &= \bigvee[\max[p , r] \wb q] G(f)([0,p)) \otimes G(g)([0,r)) 
\\ 
&= \bigvee[p \wb q] \bigvee[r \wb q] f(p) \otimes g(r) = f(q) \otimes g(q).\end{align*}  

\item \begin{align*} \min[f , g](q) &= \bigvee[U \times V \subset \min[]^*(q)] G(f)(U) \otimes G(g)(V) 
\\ 
&= \bigvee[\ust{U \subset [0,q)}{\ou V \subset [0,q)}] G(f)(U) \otimes G(g)(V) \;\;\;\; {{\texte{(since the order on}}} {{[0,1]}} \texte{is total)} 
\\ 
&= \left(\bigvee[U \subset {[0,q)}] G(f)(U) \otimes G(g)([0,1])\right) \vee \left(\bigvee[V \subset {[0,q)}] G(f)([0,1]) \otimes G(g)(V)\right) 
\\ 
&= f(q) \vee g(q).\end{align*}  

\item $(f \+ g)(q) = \bigvee[U \dotp V \subset {[0,q)}] G(f)(U) \otimes G(g)(V) = \bigvee[p \wb q] \bigvee[r \wb q \dotm p] f(p) \otimes g(r) = \bigvee[p \wb q] f(p) \otimes g(q - p)$. 

For all $h \in \USC$, \begin{align*} f \leq g \+ h &\Leftrightarrow \forall q \in {[0,1]} \; f(q) \geq \bigvee[p \+ r \wb q] g(p) \otimes h(r) 
\\ 
&\Leftrightarrow \forall q,p,r \in {[0,1]} \; \tq p \+ r \wb q \; f(q) \geq g(p) \otimes h(r) 
\\ 
&\Leftrightarrow \forall q,p,r \in {[0,1]} \; \tq p \+ r \wb q \; g(p) \nrightarrow f(q) \geq h(r) 
\\ 
&\Leftrightarrow  \forall q,r \in [0,1] \; \bigwedge[p \wb q \dotm r] g(p) \nrightarrow f(q) \geq h(r) 
\\ 
&\Leftrightarrow  \forall q,r \in [0,1] \; g(q \dotm r) \nrightarrow f(q) \geq h(r) 
\\ 
&\Leftrightarrow  \forall r \in [0,1] \; \bigwedge[q \geq r] g(q \dotm r) \nrightarrow f(q) \geq h(r) 
\end{align*} 

Thus, according to \bf{Lemma} \ref{adj}, since $r \mapsto \bigwedge[q \geq r] g(q \dotm r) \nrightarrow f(q)$ is non-decreasing, for all $r \in [0,1]$, $\bigvee[p \wb r] \bigwedge[q \geq p] g(\blue{q \dotm p}) \nrightarrow f(q) \geq h(r)$. 

\item $$(f \dotp q')(q) = \bigvee[U \dotp q' \subset {[0,q)}] G(f)(U) = \bigvee[p \dotp q' \wb q] f(p) = f(q \dotm q')$$. 

$(\_ \dotp q')[v]$ and $\dotp [v, q']$ are two terms of arity one, therefore linear, so, \linebreak since $[0,1] \models (\_ \dotp q')[v] = \dotp [v, q']$, thanks to \bf{Theorem} \ref{formule formules}, \linebreak $\USC \models (\_ \dotp q')[v] = \dotp [v, q']$ and so $f \dotp q' = f \+ \ul q'$. 

\item $(\_ \dotm q')[v]$, $\dotm [v, q']$ and $j\left(\frac{v}{2} \+ \frac{1 - q'}{2}\right)$ are three linear terms whose interpretations are equal in $[0,1]$, so, thanks to \bf{Theorem} \ref{formule formules}, $(f \dotm q')(q) = (f \- \ul q')(q) = j\left(\frac{f}{2} \+ \ul{\frac{1 - q'}{2}}\right)\blue{(q)}$. 

\item $$(2f)(q) = \bigvee[2 U \subset {[0,q)}] G(f)(U) = \bigvee[2 p \wb q] f(p) = f\left(\frac{q}{2}\right)$$ and $$j_\ast(f)(q) = \bigvee[j_\ast(U) \subset {[0,q)}] G(f)(U) = \bigvee[j_\ast(p) \wb q] f(p) = f(j(q)).$$ 

\end{enumerate} 
\end{proof} 

\begin{thm}[Increasing version of Stone-Weierstrass theorem for lattices] \label{IncreasingSW} 

Let $X$ be a compact topological space with at least two points endowed with an order $\leq$ and let $L$ be a sublattice of the lattice of continuous non-decreasing functions from $X$ to $[0,1]$. 

If, for all $y \not \leq x \in X$, $p \leq q \in [0,1]$ and $\epsilon > 0$, there exists $f \in L$ such that $\abs[f(x) - p] < \epsilon$ and $\abs[f(y) - q] < \epsilon$, then $L$ is dense in $C^0_\nearrow(X)$. 

\end{thm} 

\begin{proof} 

Let $g \in C^0_\nearrow(X)$ and $\epsilon > 0$. Let $x \in X$. 

For all $y \in X$, there exists $f \in L$ such that $\abs[f(x) - g(x)] < \epsilon \et \abs[f(y) - g(y)] < \epsilon$. Indeed, for all $y \neq x$, the assumption of the theorem gives such a function, and, since there are at least two points in $X$, if we take $y = x$, there exists $x' \neq x$ and the assumpion again gives such a function. 

$X$ being compact, there exists $f_1, \ldots , f_n \in L$ such that for all $i \in \dc[1,n]$ $\abs[f_i(x) - g(x)] < \epsilon$ and $X = \bigcup[i = 1][n] (f_i - g)^{-1} ((-\epsilon,\epsilon))$. Thus, $\bigvee[i = 1][n] f_i (y) > g(y) - \epsilon$ for all $y \in X$, $\bigvee[i = 1][n] f_i(x) < g(x) + \epsilon$ and $\bigvee[i = 1][n] f_i \in L$. 

This being true for every $x \in X$ and $X$ being compact, there exists $h_1, \ldots, h_k \in L$ such that for all $j \in \dc[1,k]$ and $y \in X$ $h_j (y) > g(y) - \epsilon$ and $X = \bigcup[j = 1][k] (h_j - g)^{-1} ((-{1},\epsilon))$. Thus, for all $x \in X$ $\abs[\left(\bigwedge[j = 1][k] h_j(x)\right) - g(x)] < \epsilon$, and $\bigwedge[j = 1][k] h_j \in L$. 
\end{proof} 

We notice here that \bf{Theorem} 8.3 of \cite{abbadiniDualCompactOrdered2019} is a corollary of \bf{Theorem} \ref{IncreasingSW}. 

\begin{cor} \label{corISW} Let $n \in \nn$. \\ The set $L_n = \{a \in \Lcinc[n] \, | \, a \text{ is a composition of } \max,\, \min,\, \dotp,\, 2 \cdot,\, \frac{\cdot}{2},\, j,\, \ul 0, \ul 1 \text{ and the projections}\}$ is dense in $\Lcinc[n]$. 

\end{cor} 

\begin{proof} 

For $n = 0$, $\Lcinc[n] = [0,1]$, and, for all $d \in \dcal$, there exists $k \et m \in \nn$ such that $d = \sum[i = 1][k] \frac{1}{2^m}$, so $d \in L_0$. Thus $L_0$ is dense in $\Lcinc[0]$. 

Let $n \in \nn$. $[0,1]^n$ has at least two distinct points. Since $L_n$ is stable by $\max$ and $\min$, $L_n$ is a lattice. 

Let $x \not \leq y \in [0,1]^n$, $p \leq q \in [0,1]$ and $\epsilon > 0$. There exists $i \in \dc[1,n]$ such that $x_i < y_i$, $d_1 \in \dcal \cap [x_i , y_i)$, $k \in \nn$ such that $2^k \geq \frac{1}{y_i - d_1}$, and $d_2 \et d_3 \in \dcal$ such that $\abs[d_2 - (q - p)] < \frac{\epsilon}{2}$ and $\abs[d_3 - p] < \frac{\epsilon}{2}$. 

We define $\fct[a][{[0,1]}^n , {[0,1]}]{z, {{\left(\min[2^k(z_i \dotm d_1)] , d_2\right) \+ d_3}}}$. 

$\dcal \subset L_0$ and, for all $z \in [0,1]$ and $d \in \dcal$, $z \dotm d = j\left(\frac{z}{2} \dotp \frac{1 - d}{2}\right) \in L_1$, so $a \in L_n$. $$\abs[a(x) - p] = \abs[\left(\min[\left(2^k (x_i \dotm a_1)\right) , d_2] \+ d_3\right) - p] = \abs[d_3 - p] < \frac{\epsilon}{2} < \epsilon.$$ \begin{align*}\abs[a(y) - q] &= \abs[\left(\min[\left(2^k (y_i \dotm a_1)\right) , d_2] \+ d_3\right) - q] = \abs[(d_2 \+ d_3) - q] \leq \abs[d_2 + d_3 - q] 
\\ 
&\leq \abs[(d_2 - (q - p))+ (d_3 - p)] \leq \abs[d_2 - (q - p)] + \abs[d_3 - p] < \epsilon.\end{align*} 

Finally, thanks to \textbf{Theorem} \ref{IncreasingSW}, $[0,1]^n$ being compact, $L_n$ is dense in $\Lcinc[n]$. 
\end{proof} 

In order to deal with density, we need a metric on $\USC$. 

\begin{lem} 

$\fct[{[0,1]} , USC(\lcal)]{q , \ul q}$ admits a left adjoint $\norm$. Thus, for all $f \in USC(\lcal)$ and $q \in [0,1]$, $\norm[f] \leq q \Leftrightarrow f \leq \ul q$. 

$\fct[d][USC(\lcal)^2, {[0,1]}]{{{(f,g)}}, {{\max[\norm[f \- g] , \norm[g \- f]]}}}$ defines a metric on $USC(\lcal)$. 
\end{lem} 

\begin{proof}  

Indeed, according to \bf{Lemma} \ref{def cst}, for all $p \in [0,1]$, $\ul p = p(\ast)$ and according to \bf{Lemma} \ref{monotonie de l'action}, $\fct[{[0,1]}, \USC^{(\USC^0)}]{p, p}$ preserves lower bounds, so $p \mapsto \ul p$ preserves lower bounds. 

Since $[0,1]$ is a complete order, $p \mapsto \ul p$ admits a left adjoint. 

$d$ is clearly positive and symmetric. For all $f \et g \in USC(\lcal)$, if $d(f,g) = 0$, then \linebreak $\norm[f \- g] = 0 \et \norm[g \- f] = 0$, thus $f \leq g \et g \leq f$, that-is-to-say $f = g$. 

Let $f$, $g$ and $h \in USC(\lcal)$. $f \leq (f \- g) \+ (g \- h) \+ h$, so $(f \- h) \- (f \- g) \leq g \- h$, which implies that $(f \- h) \- (f \- g) \leq \ul{\norm[g \- h]}$. Thus $(f \- h) \- \ul{\norm[g \- h]} \leq f \- g$, and so $f \- h \leq \ul{\norm[f \- g]} \+ \ul{\norm[g \- h]} = \ul{\norm[f \- g] \+ \norm[g \- h]}$. Hence $\norm[f \- h] \leq \norm[f \- g] \+ \norm[g \- h]$, which is equivalent to $\norm[f \- h] \leq \norm[f \- g] + \norm[g \- h]$, since $\norm[f \- h] \leq 1$. In the same way, $\norm[h \- f] \leq \norm[h \- g] + \norm[g \- f]$. 

Finally, \begin{align*} d(f,h) &= \max[\norm[f \- h] , \norm[h \- f]] \\ 
	&\leq \max[\norm[f \- g] + \norm[g \- h] , \norm[h \- g] + \norm[g \- f]] \\ 
	&\leq d(f,g) + d(g,h).\end{align*} 
\end{proof} 

The purpose of the structure on $USC(\lcal)$ is to obtain the following theorem. 

\begin{thm} 

\label{PositiveAction} 

Let us endow $C^0(USC(\lcal)^n,USC(\lcal))$, $n \in \nn \cup\{0\}$, with the metric $d_\infty$ defined by $d_\infty (F,G) = \!\!\!\!\! \bigvee[f \in USC(\lcal)^n] \!\!\!\! d(F(f),G(f))$, for all $(F,G) \in C^0(USC(\lcal)^n,USC(\lcal))^2$. 

\adjusttopage{There exists a unique family of continuous functions $\left(\cdot\colon \Lcinc[n] \rightarrow C^0_\nearrow(USC(\lcal)^n,USC(\lcal))\right)_{n \in \nn}$} that is associative in the sense that, for all $a \in \Lcinc[n]$, $(b_1, \ldots, b_n) \in \prod[i = 1][n] \Lcinc[k_i]$ and \linebreak $(f_{i,1}, \ldots, f_{i,k_i})_{i \in \dc[1,n]} \in USC(\lcal)^{\sum[i = 1][n] k_i}$, 

$$(a \circ (b_1, \ldots, b_n)) \cdot (f_{1,1}, \ldots, f_{n,k_n}) = a \cdot (b_1 \cdot (f_{1,1}, \ldots, f_{1,k_1}), \ldots, b_n \cdot (f_{n,1}, \ldots, f_{n,k_n}))$$ 

such that, for all $f \et g \in \USC$: 

\begin{enumerate}[label=(\arabic*)] 

\item \label{max cdot} $\max \cdot (f , g) = f \otimes g$ 
\item \label{min cdot} $\min \cdot (f , g) = f \wedge g$ 
\item \label{+ cdot} $\dotp \cdot (f , g) = f \+ g$ 
\item \label{2 cdot} $2 \cdot f = 2f$ 
\item \label{frac()(2) cdot} $\frac{\cdot}{2} \cdot f = \frac{f}{2}$ 
\item \label{j cdot} $j \cdot f = j(f)$

\end{enumerate} 

Moreover, $\cdot$ are isometries, and, for every $n \in \nn$, $a \in \Lcinc[n]$, $f \in USC(\lcal)^n$, $a \cdot f = a(f)$. 

\end{thm} 

\begin{remark} 

For every sublanguage $L$ of $\Lcinc$ containing $\mathcal{L}_{[0,1]}$, such as $\mathcal{L}$, \bf{Theorem} \ref{PositiveAction} works if one extends the list of axioms \ref{max cdot} to \ref{j cdot} to the symbols of $L$. 

\end{remark} 

As a corollary of \bf{Theorems} \ref{reduction to Lcinc} and \ref{PositiveAction}, we give the following theorem. 

\begin{thm} 

\label{USCAction} 

Let us endow, for all $n \in \nn$, $\Lusc[n]$ with the supremum metric, and $USC(\lcal)^{(USC(\lcal)^n)}$ with the metric $d_\infty$ defined by $d_\infty (F,G) = \!\!\!\!\! \bigvee[f \in USC(\lcal)^n] \!\!\!\! d(F(f),G(f))$, for all $F \et G \colon USC(\lcal)^n \rightarrow USC(\lcal)$. 

$(a,f) \mapsto a(f)$, as defined in \ref{interpretation}, is the unique family of continuous and lower-bound-preserving functions $\left(\cdot\colon \Lusc[n] \rightarrow USC(\lcal)^{(USC(\lcal)^n)})\right)_{n \in \nn}$ such that, for all $f \et g \in \USC$: 

\begin{enumerate}[label=(\arabic*)] 

\item $\max \cdot (f , g) = f \otimes g$ 
\item $\min \cdot (f , g) = f \wedge g$ 
\item $\dotp \cdot (f , g) = f \+ g$ 
\item $2 \cdot f = 2f$ 
\item $\frac{\cdot}{2} \cdot f = \frac{f}{2}$ 
\item $j \cdot f = j(f)$

\end{enumerate} 

\end{thm} 

{\remark \bf{Theorems} \ref{PositiveAction} and \ref{USCAction} mean that there is a unique structure of module on $USC(\lcal)$ over the operads $\Lcinc$ and $\Lusc$ satisfying points \ref{max cdot} to \ref{j cdot}.} 

To prove \bf{Theorems} \ref{PositiveAction} and \ref{USCAction}, we use the following lemma. 

\begin{lem} \label{isometry} 

For all $n \in \nn$, $\fct[\Lusc[n], \USC^{(\USC^n)}]{a, a(\_)}$ is an isometry. 

\end{lem} 

\begin{proof} 

Let $n \in \nn$, $a \et b \in \Lusc[n]$. 

For all $q \in [0,1]^n$, $a(\ul q) = \ul{a(q)} \et b(\ul q) = \ul{b(q)}$, so $d(a(\ul q) , b(\ul q)) = d(\ul{a(q)},\ul{b(q)}) = \abs[a(q) - b(q)]$, and so \linebreak $\norm[a - b]_\infty \leq d_\infty(a \cdot \_ , b \cdot \_)$. 

Let $f \in USC(\lcal)^n$. $a \leq (a \- b) \+ b$, so $a(f) \leq (a \- b)(f) \+ b(f)$, which amounts to \linebreak $a(f) \- b(f) \leq (a \- b)(f)$. However, $f \leq \ul{\norm[f]}$, so $(a \- b)(f) \leq (a \- b)(\ul{\norm[f]})$. Hence $d_\infty(a \cdot \_ , b \cdot \_) \leq \norm[a - b]_\infty$. 

Hence, $a \mapsto a(\_)$ is an isometry. 
\end{proof} 

\begin{proof}[Proof of \textbf{Theorem} \ref{PositiveAction}] 

\it{\ul{Existence part :}} 

Let, for all $n \in \nn$, $a \in \Lcinc[n]$ and $f \in \USC^n$, $a \cdot f = a(f)$. The associativity of $\cdot$ is what \bf{Lemma} \ref{operad action} states. By virtue of \bf{Lemma} \ref{isometry}, for all $n \in \nn$, $\fct[\Lcinc[n], C^0_\nearrow(\USC^n, \USC)]{a, a \cdot \_}$ is an isometry. This family of isometries tautologically satisfies axioms \ref{max cdot} to \ref{j cdot}. 

\it{\ul{Uniqueness part :}} 

Let now $L$ be a language and $\cdot$ be a family of functions as in \textup{\textbf{Theorem} \ref{PositiveAction}}. Let $n \in \nn$ and $a \in \Lcinc[n]$. 

There exists $(a_k)_{k \in \nn} \in L_n^\nn$ such that $a_k \rightarrow a$. The $a_k$s being compositions of elements of $L$, and $\cdot$ preserving the composition, for all $f \in USC(\lcal)^n$ $a_k \cdot f = a_k(f)$. $\cdot$ being continuous, for all $f \in USC(\lcal)^n$, $(a \cdot f) = \lim a_k \cdot f = \lim a_k(f) = a(f)$. 
\end{proof} 

\begin{proof}[Proof of \bf{Theorem} \ref{USCAction}]  

\it{\ul{Existence :}} 

By virtue of \bf{Lemma} \ref{isometry}, for all $n \in \nn$, $\fct[\Lusc[n], \USC^{(\USC^n)}]{a, a(\_)}$ is continuous. This family of continuous functions tautologically satisfies axioms \ref{max cdot} to \ref{j cdot}. 

\it{\ul{Uniqueness :}} 

Let $\cdot$ be such a family of functions. By \bf{Theorem} \ref{PositiveAction}, for all $n \in \nn$, $a \in \Lcinc[n]$ and $f \in \USC^n$, $a \cdot f = a(f)$. So, by \bf{Theorem} \ref{reduction to Lcinc}, for all $n \in \nn$, $a \in \Lusc[n]$ and $f \in \USC^n$, $a \cdot f = a(f)$. 
\end{proof} 

\subsection{The Continuous Logic structure on $\USC$} 

We remind the reader that the language we are finally interested in is $\mathcal{L} = \{\vee,\, \wedge,\, \+,\, \-,\, \frac{\cdot}{2},\, 2,\, j_\ast,\, j,\, \alpha,\, \ul 0,\, \ul 1\}$. 
\\ 
In order to give a systematic way to translate any inequation in the language $\mathcal{L}_{crl}$ and $\mathcal{L}_{[0,1]}$ to an inequation in the language $\mathcal{L}$ such that if the former is universally satisfied by $\USC$, so is the later, we have to replace $\otimes$ and $\nrightarrow$ in formulas of $\mathcal{L}_{crl}$ (\bf{Theorem} \ref{théorème de comparaison}) and compare the interpretation of $\max$ from the language $\mathcal{L}_{[0,1]}$ to $\vee$ (\bf{Theorem} \ref{from [0,1] to L}). 

\begin{defi} \label{def AC-algèbres} 

The interpretation of $\+$, $2$, $\frac{\cdot}{2}$, $2$, $j_\ast$, $j$, $\alpha$, $\ul 0$ and $\ul 1$ is given by the interpretation of $\mathcal{L}$ in $[0,1]$ (\bf{Definition} \ref{interpretation}). We call every $(\USC,\, \vee,\, \wedge,\, \+,\, \-,\, 2,\, \frac{\cdot}{2},\, j_\ast,\, j,\, \alpha,\, \ul 0,\, \ul 1)$, for $\lcal$ a \crcl, an \emph{Affine Continuous algebra}, or AC-algebra. 

\end{defi} 

\begin{lem} \label{calculs} 

Let $f$, $g \in \USC$ and $U$, $V \in \lcal$. 

$2\frac{f}{2} = f$ and, for all $q \in [0,1]$, $2(g \nrightarrow f)(q) \leq 2g(q) \nrightarrow 2f(q)$. 

$0_U \+ 0_V = 0_{U \otimes V}$ and $0_U \- 0_V = 0_{V \nrightarrow U}$. 

\end{lem} 

\begin{proof} 

For all $q \in (0,1]$, 

$$2\frac{f}{2}(q) = \frac{f}{2}\left(\frac{q}{2}\right) = \bigvee[\frac{U}{2} \subset {\left[0,\frac{q}{2}\right)]}] G(f)(U) = G(f)([0,q)] = f(q)$$ $$2(g \nrightarrow f)(q) = (g \nrightarrow f)\left(\frac{q}{2}\right) = \bigvee[p \wb \frac{q}{2}] \bigwedge[r \geq p] g(r) \nrightarrow f(r) \leq g\left(\frac{q}{2}\right) \nrightarrow f\left(\frac{q}{2}\right) = 2g(q) \nrightarrow 2f(q)$$ $$0_U \+ 0_V (q) = \bigvee[p \wb q] 0_U(p) \otimes 0_V(q - p) = \bigvee[0 \wb p \wb q] 0_U(p) \otimes 0_V(q - p) = U \otimes V = 0_{U \otimes V}(q)$$ $$0_U \- 0_V (q) = \bigvee[p \wb q] \bigwedge[r \geq p] 0_V(r) \nrightarrow 0_U(r) = \bigvee[0 \wb p \wb q] \bigwedge[r \geq p] 0_V(r) \nrightarrow 0_U(r) = V \nrightarrow U = 0_{V \nrightarrow U }(q).$$ 
\end{proof} 

The following theorem is used in the proof of \bf{Theorem} \ref{from [0,1] to L}, but also independently of \bf{Theorem} \ref{from [0,1] to L}, in the proof of \bf{Corollary} \ref{2+}. 

\begin{thm} \label{théorème de comparaison} 

For all $f \et g \in \USC$, $$\frac{f \+ g}{2} \leq f \otimes g \leq f \+ g \et f \- g \leq g \nrightarrow f \leq 2(f \- \frac{g}{2}).$$ 

\end{thm} 

\begin{proof} 

Let $f \et g \in \USC$ and $q \in [0,1]$. $[0,1] \models \frac{u \+ v}{2} \leq \max[u,v]$ and \linebreak $[0,1] \models \max[u,v] \leq u \+ v$, so, according to \bf{Theorem} \ref{formule formules} and \bf{Lemma} \ref{calculatoire}, \linebreak $\frac{f \+ g}{2} \leq f \otimes g \leq f \+ g$. 

For all $h \in \USC$, $g \nrightarrow f \leq h \Leftrightarrow f \leq g \otimes h \Rightarrow f \leq g \+ h \Leftrightarrow (f \- g) \leq h$, so $f \- g \leq g \nrightarrow f$, and $f \- g \leq h \Leftrightarrow f \leq h \+ g \Rightarrow f \leq 2(h \otimes g) \Leftrightarrow g \nrightarrow \frac{f}{2} \leq h$, so $g \nrightarrow \frac{f}{2} \leq f \- g$. \linebreak Since it is true for all $f \et g \in \USC$, for all $f \et g \in \USC$, \linebreak $2(f \- g) \geq 2\left(g \nrightarrow \frac{f}{2}\right) \geq 2g \nrightarrow 2\frac{f}{2} = 2g \nrightarrow f$. 
\end{proof} 

\begin{thm} \label{from [0,1] to L} 

Let \blue{$(\phi_i[v])_{1 \leq i \leq n}$} and $\psi[v]$ be terms in the language \blue{$\left(\mathcal{L}_{[0,1]} \setminus \{\-\}\right) \cup \{j,\, j_\ast,\, \alpha\}$} and assume \blue{each $\phi_i$} is linear. Let $\blue{\tilde \phi_i}$ be the $\mathcal{L}$-terms obtained by replacing every occurrence of $\max$ (resp. $\min$) in $\blue{\phi_i}$ by $\vee$ (resp. $\wedge$) and $\bar \psi$ be the $\mathcal{L}$-term obtained by replacing every occurrence of $\max$ (resp. $\min$) in $\phi$ by $\+$ (resp. $\wedge$). 

\blue{Then, $\USC \models \bigwedge[i = 1][n] \tilde \phi_i \leq \bar \psi$ if and only if $[0,1] \models \bigwedge[i = 1][n] \phi_i \leq \psi$.} 

\end{thm} 

\begin{proof} 

\ul{\it{Direct sense :}} Thanks to \bf{Theorem} \ref{formule formules} and since, for all term $\phi$ in $\mathcal{L}_{[0,1]}$ the interpretation of $\phi$ in $\USC$ is non-decreasing in every coordinate, it suffices to show that, for all $f$ and $g \in \USC$, $f \vee g \leq \max[f,g]$ and $\max[f,g] \leq f \+ g$. 

Let $f \et g \in \USC$ and $q \in [0,1]$. $\max[f,g](q) = f(q) \otimes g(q) \leq f(q)$ and \linebreak $\max[f,g](q) = f(q) \otimes g(q) \leq g(q)$, so $\max[f,g](q) \leq f(q) \wedge g(q) = (f \vee g)(q)$. According to \bf{Lemma} \ref{calculatoire} and \bf{Theorem} \ref{théorème de comparaison}, $\max[f,g] = f \otimes g \leq f \+ g$. 

\ul{\it{Converse sense :}} According to \bf{Lemma} \ref{[0,1] embedding}, $q \mapsto \ul q$ is an $\mathcal{L} \cup \{j_\ast,\, \alpha\}$-embedding, so, if $\USC \models \blue{\bigwedge[i = 1][n] \tilde \phi_i} \leq \bar \psi$, then, for all $p \colon \vcal \rightarrow [0,1]$, $\USC \models \blue{\bigwedge[i = 1][n] \ul{\phi_i[p]}} \leq \ul{\psi[p]}$, so $\blue{\bigwedge[i = 1][n] \phi_i[p]} \leq \psi[p]$ and thus $[0,1] \models \blue{\bigwedge[i = 1][n] \phi_i} \leq \psi$. 
\end{proof} 

\begin{remark} 

We end this subsection by noticing that, in the case where $\lcal$ is the topology $\Tcal(X)$ of a topological space $X$, this pointwise structure is the one induced by the $(\vee,\, \wedge,\, \blue\+,\, \blue\-,\, 2,\, \frac{\cdot}{2},\, j\, \ul 0,\, \ul 1)$-structure of $[0,1]$ through the isomorphism $USC(\Tcal(X)) \simeq USC(X)$. 

\end{remark} 

\section{Algebraic axiomatisation of AC-algebras} \label{section algebraic axiomatisation} 

\subsection{The three theories} 

Let us give an algebraic caracterisation of the $\USC$ for $\lcal$ a complete residuated lattice. The language in which we express the axioms is $\mathcal{L}$, that we recall to be $\{\vee,\, \wedge,\, \blue{\+,\, \-},\, \frac{\cdot}{2},\, 2,\, j_\ast,\, j,\, \alpha,\, \ul 0,\, \ul 1\}$. 

\begin{nota} \label{notation dyadiques} 

For all dyadic $d \in [0,1)$, we define $\ul d$ as $\sum[k = 1][2^n d] \frac{j_\ast(\ul 0)}{2^{n-1}}$ where $n$ is the smallest non negative integer such that $2^n d \in \nn$. Note that $\ul{\frac{1}{2}} = j_\ast(\ul 0)$. 

\end{nota} 

Here comes the list of axioms, the truth of which can be deduced from \bf{Theorem} \ref{formule formules} in every algebra of the form $USC(\lcal)$. 

\setlist[enumerate]{itemsep=5pt, topsep=10pt} 

\begin{enumerate}[label=(\arabic*)] 

\item \label{bounded commutative residuated lattice} \textbf{$\blue{(\vee,\, \wedge, \+, \-, \ul 0, \ul 1)}$ is a bounded commutative residuated lattice structure:} 

\begin{enumerate}[label=(\arabic{enumi}.\alph*)] 

\item \label{bounded lattice} $(A, \vee, \wedge, \blue{\ul 0}, \ul 1)$ is a bounded lattice 
\item \label{com monoid} $(A, \+, \ul 0)$ is a commutative monoid 
\item \label{residuation} $\-$ is the residual of $\+$: $w \leq u \+ v \Leftrightarrow w \- v \leq u$ 

\end{enumerate} 

\item \target{croissance} \label{croissance} \textbf{$2$, $\frac{\cdot}{2}$, $j_\ast$, $j$ and $\alpha$ non-decreasing.} 

\item \label{aDjunctions} \textbf{The adjunctions:} 

\begin{enumerate}[label=] 

\item \target{2 adjunction}[] \begin{enumerate*}[label=(\arabic{enumi}.\alph{enumii}.\arabic*) , itemjoin={\quad}] \item \label{(2v)/2<=v} $\frac{2 v}{2} \leq v$ \hspace{79pt} and \item \label{v<=2(v/2)} $v \leq 2 \frac{v}{2}$ \end{enumerate*} 

\item \begin{enumerate*}[label=(\arabic{enumi}.\alph{enumii}.\arabic*) , itemjoin={\quad}] \item \label{jj*v<=v} $j \circ j_\ast(v) \leq v$ \hspace{53pt} and \item \label{v<=j*jv} $v \leq j_\ast \circ j(v)$ \end{enumerate*} 

\end{enumerate} 

\item \target{Defining axioms}[] \label{Defining axioms} \textbf{Defining axioms:} 

\begin{enumerate}[label=\phantom{\alph*}] 

\item \hypertarget{cont de 2}{} \begin{enumerate*}[label=(\arabic{enumi}.\alph{enumii}.\arabic*) , itemjoin={\quad}] \item \label{2(+) <= 2 + 2} \blue{$2(u \+ v) \leq 2 u \+ 2 v$} \hspace{35pt} and \item \label{+ 2 <= 2(+)} $2 u \+ 2 v \leq 2 (u \+ v)$ \end{enumerate*} 

\item \color{blue} \hypertarget{def j*}{} \begin{enumerate*}[label=(\arabic{enumi}.\alph{enumii}.\arabic*) , itemjoin={\quad}] \item \label{def j* <=} $j_\ast(2v \+ u) \leq v \+ j_\ast(u)$ \hspace{19pt} and \item \label{def j* >=} $j_\ast(2v \+ u) \geq v \+ j_\ast(u)$ \end{enumerate*} \color{black} 

\item \hypertarget{>=}{} \begin{enumerate*}[label=(\arabic{enumi}.\alph{enumii}.\arabic*) , itemjoin={\quad}] \item \label{2 >=} $2 v \geq v$ \hspace{80pt} and \item \label{j* >=} $j_\ast(v) \geq v$ \end{enumerate*} 

\item \hypertarget{def alpha, 1}{} \begin{enumerate*}[label=(\arabic{enumi}.\alph{enumii}.\arabic*) , itemjoin={\quad}] \item \label{alpha beta v >= v} $v \leq \alpha(2v \wedge j_\ast(v))$ \hspace{32pt} and \item \label{alpha beta v <=  v} $\alpha(2v \wedge j_\ast(v)) \leq v$ \end{enumerate*} 

\item \hypertarget{def alpha, 2}{} \begin{enumerate*}[label=(\arabic{enumi}.\alph{enumii}.\arabic*) , itemjoin={\quad}] \item \label{beta alpha v <=  v} $2\alpha(v) \wedge j_\ast \circ \alpha(v)\leq v$ \hspace{17pt} and \item \label{beta alpha v >= v} $v \leq 2\alpha(v) \wedge j_\ast \circ \alpha(v)$ \end{enumerate*} 

\item \hypertarget{fin def 2 et j*}{} \begin{enumerate*}[label=(\arabic{enumi}.\alph{enumii}.\arabic*) , itemjoin={\quad}] \item \label{2j*(0) >= 1} $2j_\ast(\ul 0) \geq \ul 1$ \hspace{65pt} and \item \label{2v <=} \blue{$2v \leq v \+ v$} \end{enumerate*} 

\item \begin{enumerate*}[label=(\arabic{enumi}.\alph{enumii})] \item \label{1/2^n + 1/2^n} for all $n \in \nn^*$, $\alpha^n \circ j_\ast (\ul 0) \+ \alpha^n \circ j_\ast (\ul 0) \leq \alpha^{n-1} \circ j_\ast (\ul 0)$ \end{enumerate*} 

\end{enumerate} 

\color{blue} \item \color{black} \target{infinitesimals} \label{infinitesimals} \textbf{And the algebra of values is $[0,1]$ up to infinitesimals:} 

\begin{enumerate}[label=(\arabic{enumi}.\alph*)] \label{Values} 

\item \label{v<=d} for all dyadic number $d \in [0,1]$ and $n \in \nn$, $v \leq \ul{d} \+ \alpha^n (v \+ \ul{1 - d})$ 

\item \label{v>=d} for all $n \in \nn$, $\bigwedge[k = 1][2^{n+1} - 1] \ul{1 - \frac{k}{2^{n+1}}} \+ \alpha^{n+1} \blue{\left(v \+ \ul{\frac{k}{2^{n+1}}}\right)} \leq v \+ \ul{\frac{1}{2^n}}$ 

\end{enumerate} 

\item \textbf{The two non-algebraic properties:} \label{non-alg. properties} 

\begin{enumerate}[label=(\arabic{enumi}.\alph*)] 

\item \label{complete} completeness: the partial order is complete 

\item \label{archimedean} Archimedean: the lower bound of the family of $\ul{\frac{1}{2^n}}$ is $\ul 0$ 

\end{enumerate} 

\end{enumerate} 

We denote by $\T$ the theory consisting of axioms \ref{bounded lattice} to \ref{v>=d}. 

\begin{prop} 

\label{approx} For all $f$ and $g \in \USC$, $f \+ g = \bigwedge[q,p \in {[0,1]}] \ul q \+ \ul p \+ 0_{f(q)} \+ 0_{g(p)}$. Thus, for all $f \in \USC$, $f = \bigwedge[q \in {[0,1]}] \ul q \+ 0_{f(q)}$. 

\end{prop} 

{\remark The equality $f = \bigwedge[q \in {[0,1]}] \ul q \+ 0_{f(q)}$ is a corollary of the equality $\bar{e}(f) = f\textasciicircum$ found in \cite[p. 4]{santocanaleDualizingSuppreservingEndomaps2021}.} 

\begin{proof} 

For all $r \in {[0,1]}$, thanks to \textbf{Lemmas} \ref{calculs} and \ref{[0,1] embedding}, \begin{align*} 
\left(\bigwedge[q,p \in {[0,1]}] \ul q \+ \ul p \+ 0_{f(q)} \+ 0_{g(p)}\right) (r) &= \bigvee[q,p \in {[0,1]}] (\ul q \+ \ul p \+ 0_{f(q) \otimes g(p)}) (r)\\ 
&= \bigvee[q,p \in {[0,1]}] \bigvee[r' \+ r'' \wb r] (\ul q \+ \ul p)(r') \otimes (0_{f(q)} \otimes 0_{g(p)}) ({r''})\\ 
&= \bigvee[q,p \in {[0,1]}] \bigvee[0 \wb r''] \bigvee[r' \+ r'' \wb r] \Big(\ul{q \+ p}({r'}) \otimes (f(q) \otimes g(p))\Big)\\ 
&= \bigvee[q,p \in {[0,1]}] \bigvee[0 \wb r''] \bigvee[\ust{q \+ p \wb r'}{r' \+ r'' \wb r}] f(q) \otimes g(p)\\ 
&= \bigvee[q,p \in {[0,1]}] \bigvee[\ust{q \+ p \wb r'}{r' \+ 0 \wb r}] f(q) \otimes g(p)\\ 
&= \bigvee[p \+ q \wb r] f(q) \otimes g(p)\\ 
&= (f \+ g)(r) 
\end{align*} 
\end{proof} 

\begin{lem} \label{soundness for T} 

For every commutative residuated complete lattice $\lcal$, $\USC$ satisfies $\T$ and \ref{non-alg. properties}. 

\end{lem} 

\begin{proof} 

Clearly, $\USC$ satisfies axioms \ref{residuation} and \ref{croissance} and is complete and Archimedean. According to \bf{Theorem} \ref{from [0,1] to L}, $\USC$ satisfies axioms \ref{bounded lattice}, \ref{com monoid}, and axioms from \ref{(2v)/2<=v} to \ref{v>=d}, since they are satisfied by $[0,1]$. 
\end{proof} 

We aim at proving the following theorem: 

\begin{thm} \label{complete archimedean} 

Any complete Archimedean model of $\T$ is isomorphic to an AC-algebra. 

\end{thm} 

For this, we will actually prove that all models of another axiomatisation, simpler but in a \linebreak non-continuous language, are isomorphic to an AC-algebra, and prove it is a consequence of \linebreak the aforementioned theory $\T$. This axiomatisation is given in the language \linebreak $\mathcal{L}_1 = \{\preceq, +, -, 2, \frac{\cdot}{2}, j_\ast, j, \alpha, \beta, l, l^\ast, (\ul d)_{d \in [0,1] \texte{dyadic}}\}$. To define the new symbols on $\USC$, we use the following lemma. 

\begin{lem} \label{def} 

Let $\phi \et \psi \colon [0,1] \rightarrow [0,1]$ such that $\phi$ is right adjoint to $\psi$. $\phi \in USC([0,1])$ and, for all $f \in \USC$, $\phi(f) = f \circ \psi$. Moreover, $f \mapsto \ouv{(f \circ \phi)}$ is left adjoint to $\phi(\_)$. Hence, if $\psi \in USC([0,1])$, for all $f \in \USC$, $\psi(f) = \ouv{(f \circ \phi)}$. 

Thus we can define $\psi(\_)$ for all $\psi\colon [0,1] \rightarrow [0,1]$ such that $\psi$ admits a right adjoint $\phi$ by, for all \linebreak $f \in \USC$, $\phi(f) = \ouv{(f \circ \phi)}$. 

\end{lem} 

\begin{proof} 

For all $x \et q \in [0,1]$, $\psi(x) \wb q \Leftrightarrow x \wb \phi(q)$, so $\psi^*(q) = [0,\phi(q))$. Thus, for all $f \in \USC$, $\phi(f) = f \circ \psi$. 

Let $f \et g \in \USC$. \begin{align*} \ouv{(f \circ \phi)} \leq g &\Leftrightarrow \forall q \in [0,1] \, f(\phi(q)) \geq g(q) 
\\ 
&\Leftrightarrow \forall p \et q \in [0,1] \tq \phi(q) \wb p \, f(p) \geq g(q) 
\\ 
&\Leftrightarrow \forall p \et q \in [0,1] \tq q \wb \psi(p) \, f(p) \geq g(q) 
\\ 
&\Leftrightarrow \forall p \in [0,1] \, f(p) \geq g(\psi(p)) 
\\ 
&\Leftarrow f \leq \phi(g). \end{align*} 

If $b \in USC([0,1])$, then, since the adjunction equations $a(b(v)) \geq v$ and $b(a(v)) \leq v$ are satisfied by $[0,1]$, by virtue of \bf{Theorem} \ref{formule formules}, they are also satisfied by $\USC$ and thus $b(\_)$ is left adjoint to $a(\_)$, so, for all $f \in \USC$, $b(f) = \ouv{(f \circ a)}$. 
\end{proof} 

{\remark $b(f) \neq f \circ a$ in general.} 

For all $q \in [0,1]$, let $l(q) = \acc{1, q = 1, 0}$, which admits a left adjoint $l^\ast$ and let $\beta(q) = \alpha^{-1}(q) $. \textbf{Lemma} \ref{def} enables to endowe each $\USC$ with an $\mathcal{L}_1$-structure, by defining $l^\ast$ and $\beta$ on it. 

\begin{lem} \label{0-indicator} 

For all \crcl $\lcal$, $f \in \USC$, $l(f) = 0_{f(1)}$. 

\end{lem} 

\begin{proof} 

For all \crcl $\lcal$, $f \in \USC$ and $q \in [0,1]$, $$l(f)(q) = f(l^\ast(q)) = \bigvee[p \wb l^\ast(q)] f(p) = \bigvee[l(p) \wb q] f(p) = 0_{f(1)}(q).$$ 
\end{proof} 

The intermediate list of axioms, satisfied by every AC-algebra, by virtue of \bf{Theorem} \ref{formule formules}, is the following: 

\begin{enumerate}[label=(\arabic*)] 

\item \label{boun com res latt 2} \textbf{$(\preceq, \+, \-, \ul 0, \ul 1)$ is a bounded \blue{preordered} monoid with residuation:} 

\item \label{croissance 2} \textbf{$2$, $\frac{\cdot}{2}$, $j_\ast$, $j$, $\alpha$, $l$ and $l^\ast$ are non-decreasing.} 

\item \label{aDjunctions 2} \textbf{The adjunctions: \ref{(2v)/2<=v} to \ref{v<=j*jv} and } 

\begin{enumerate}[label=] 

\setcounter{enumii}{2} 

\item \begin{enumerate*}[label=\arabic{enumi}.\alph{enumii}.\arabic*. , itemjoin={\quad}] \item \label{alpha beta v >= v 2} $v \preceq \alpha \circ \beta (v)$, \quad and \item \label{alpha beta v <=  v 2} $\beta \circ \alpha(v) \preceq v$ \end{enumerate*} 

\item \label{l*} \begin{enumerate*}[label=\arabic{enumi}.\alph{enumii}.\arabic*. , itemjoin={\quad}] \item \label{ll* >=} $v \preceq l \circ l^\ast (v)$ \quad and \item \label{l*l <=} $l^\ast \circ l (v) \preceq v$ \end{enumerate*} 

\end{enumerate} 

\item \textbf{Defining axioms:} 

\begin{enumerate} 

\item \label{redef 2} $2 (\ul d \+ l(v)) \simeq \ul{2d} \+ l(v)$ 

\item \label{redef j*} $j_\ast (\ul d \+ l(v)) \simeq \ul{j_\ast(d)} \+ l(v)$ 

\item \label{redef alpha} $\alpha (\ul d \+ l(v)) \simeq \ul{\alpha(d)} \+ l(v)$ 

\item \label{def l} $l (\ul d \+ l(v)) \simeq \ul{l(d)} \+ l(v)$ 

\item \label{redef +} $\ul d \+ \ul{d'} \simeq \ul{d \dotp d'}$ 

\end{enumerate} 

\item \textbf{$+$ stabilizes the fixed points of $l$:} \begin{enumerate}[label=\arabic{enumi}.\alph*.] 

\setcounter{enumii}{3}

\item \label{Stability 2} $l(l(u) \+ l(v)) \simeq l(u) \+ l(v)$ 

\end{enumerate} 

\item \label{[0,1] 2} \textbf{And the algebra of values is $[0,1]$ to infinitesimals:} 

\begin{enumerate}[label=\arabic{enumi}.\alph*.] 

\setcounter{enumii}{2}

\item for all dyadic number $d \in [0,1]$, $v \preceq \ul d \+ l (v \+ \ul{1 - d})$ 

\item for all $n \in \nn$, $\bigwedge[k = 1][2^{n+1} -1] \ul{1 - \frac{k-1}{2^{n+1}}} \+ l (v \+ \ul{\frac{k-1}{2^{n+1}}}) \preceq v \+ \ul{\frac{1}{2^n}}$ 

\end{enumerate} 

\item \textbf{The two non-algebraic properties:} 

\begin{enumerate}[label=\arabic{enumi}.\alph*] 

\setcounter{enumii}{2}

\item \label{complete 2} completeness: there exists $\bigwedge\colon \{E \subset A\} \rightarrow A$ such that, for all $E \subset A$ and $b \in A$, $\bigwedge E \succeq b \Leftrightarrow \forall a \in E \; a \succeq b$ 
\item \label{archimedean 2} Archimedean: for all $n \in \nn$, $a \preceq \ul{\frac{1}{2^n}} \Rightarrow a \preceq \ul 0$ 

\end{enumerate} 

\end{enumerate}

We denote by $\T_1$ the set of axioms \ref{boun com res latt 2} to \ref{[0,1] 2}, $\mathcal{L}_0$ the language $\{\preceq, \+, \-, l, l^\ast, (\ul d)_{d \in [0,1] \texte{dyadic}}\}$ and \linebreak $\T_0$ the set of axioms of $\T_1$ that are formulas of $\mathcal{L}_0$, i.e. in which $2$, $\frac{\cdot}{2}$, $j_\ast$, $j$, $\alpha$ and $\beta$ do not appear. 

\subsection{Complete Archimedean models of $\T_0$ and $\T_1$} \label{subsection Complete Archimedean models} 

\subsubsection{Complete Archimedean models of $\T_0$} 

We prove here that any complete Archimedean model of $\T_0$ is isomorphic to an AC-algebra (\bf{Theorem} \ref{celui-là}). In the following two subsubsections, we prove that any complete model of $\T$ can be edowed with a preorder that makes it an Archimedean model of $\T_0$. It is thus $\lcal_0$-isomorphic to an AC-algebra and we also prove that the extra structure is preserved by the isomorphism. Then, we will deal with non complete non Archimedean models of $\T$. 

\blue{Let $A$ be a complete Archimedean model of $\T_0$. Let us denote by $\dcal$ the set of all the dyadic numbers in $[0,1]$.} 

\begin{lem} \label{remark dyadic} 

For all $(d_i)_{i \in I} \in \dcal^I$ such that $\bigwedge[i \in I] d_i \in \dcal$ and $a \in A$, $$a \preceq \ul{\bigwedge[i \in I] d_i} \Leftrightarrow \forall i \in I \; a \preceq \ul{d_i}.$$ 

\end{lem} 

\begin{proof} 

Let $(d_i)_{i \in I} \in \dcal^I$ such that $\bigwedge[i \in I] d_i \in \dcal$ and $a \in A$. If for all $i \in I$ $d_i = 1$, then $\bigwedge[i \in I] \ul{d_i} = \ul 1$, so $\bigwedge[i \in I] \ul{d_i} = \ul{\bigwedge[i \in I] d_i}$. 

Otherwise, we can assume that, for all $i \in I$, $d_i < 1$. Since for all $i_0 \in I$, $\bigwedge[i \in I] d_i \leq d_{i_0}$, if $a \preceq \ul{\bigwedge[i \in I] d_i}$, then $a \preceq \ul{d_{i_0}}$. 

Conversely, assume that, for all $i \in I$, $a \preceq \ul{d_i}$ and let $n \in \nn$. 

$\left(\bigwedge[i \in I] d_i\right) \dotp \frac{1}{2^n} > \left(\bigwedge[i \in I] d_i\right)$, so there exists $i_0 \in I$ such that $\left(\bigwedge[i \in I] d_i\right) \dotp \frac{1}{2^n} > d_{i_0}$. Thus, since $\left(\bigwedge[i \in I] d_i\right) \dotp \frac{1}{2^n} \in \dcal$, $a \preceq \ul{d_{i_0}} \preceq \ul{\left(\bigwedge[i \in I] d_i\right) \dotp \frac{1}{2^n}} \simeq \ul{\left(\bigwedge[i \in I] d_i\right)} \+ \ul{\frac{1}{2^n}}$. Thus, since $A$ is Archimedean, $a \preceq \ul{\left(\bigwedge[i \in I] d_i\right)}$. 
\end{proof} 

Since $[0,1]$ is generated by lower bounds of elements of $\dcal$, there is a unique lower bounds-preserving extension of $\ul \_$ to $[0,1]$, that we still denote $\ul \_$. 

\begin{thm} \label{celui-là} 

The quotient of every complete Archimedean model of $\T_0$ by the equivalence relation induced by its preorder is isomorphic to a $\USC$ for some commutative residuated complete lattice $\lcal$. 

\end{thm} 

\begin{proof} 

First, define $\lcal(A) = \quotient{\{e \in A \;|\; l(e) \simeq e\}}{\simeq}$, endowe it with the quotient order and the canonical surjection $\pi\colon \{e \in A \,|\, l(e) \simeq e\} \rightarrow \lcal(A)$ and call $\lcal(A)^{op}$ the set $\lcal(A)$ with the reversed order. 

\begin{lem} \label{celui-ci} 

\begin{enumerate} 

\item $\lcal(A)$ is complete and $\pi$ preserves lower bounds and thus admits a left adjoint, which is $\fct[\pi^\ast][\lcal(A), \{e \in A \;|\; l(e) \simeq e\}]{U, \bigwedge \{e \in A \, | \, l(e) \simeq e \et U \leq \pi(e)\}}$. 

Moreover $\pi \circ \pi^\ast = \id_\lcal(A)$ and $\pi^\ast \circ \pi \simeq \id_{\{e \in A \,|\, l(e) \simeq e\}}$. 

\item $\+$ induces a commutative and associative operation $\otimes$ on $\lcal(A)$ that admits a residual. 

\end{enumerate} 

\end{lem} 

\begin{proof} \begin{enumerate} \vspace{-7pt} 

\item Let $U \in \lcal$ and $(e_i)_{i \in I} \in \{e \in A \, | \, l(e) \simeq e\}^I$. There exists $e \in A$ such that \linebreak $l(e) \simeq e$ and $\pi(e) = U$. 

$\forall i \in I \; U \leq \pi(e_i) \Leftrightarrow \forall i \in I \; e \preceq e_i \Leftrightarrow e \preceq \bigwedge[i \in I] e_i \Leftrightarrow U \leq \pi\left(\bigwedge[i \in I] e_i\right)$. Hence, the lower bound of $(\pi(e_i))_{i \in I}$ is $\pi\left(\bigwedge[i \in I] \pi(e_i)\right)$. Thus $\lcal(A)$ is complete and, for all $U \in \lcal$, $$\pi \circ \pi^\ast(U) = \bigwedge[\ust{e \in A}{l(e) \simeq e \et U \leq \pi(e)}] \pi(e) = \bigwedge[\ust{e \in A}{l(e) \simeq e \et U = \pi(e)}] \pi(e) = U,$$ and, for all $e \in A$, $\pi^\ast \circ \pi(e) \simeq e$, because $\pi(\pi^\ast \circ \pi(e)) = \pi(e)$. 

\item For all $e \et e' \in A$ such that $l(e) = e \et l(e') = e'$, by \ref{Stability 2}, \linebreak $l(e \+ e') = l(l(e) \+ l(e')) = l(e) \+ l(e') = e \+ e'$. 

Since $+$ is non-decreasing, it induces an operation $\otimes\colon \lcal(A) \times \lcal(A) \rightarrow \lcal(A)$. 

Clearly, $\otimes$ is commutative and associative and admits $\pi(\ul 0)$ as neutral element. 

For all $(e_i)_{i \in I} \in A^I$ such that for all $i \in I$ $l(e_i) = e_i$ and $e \in A$ such that $l(e) = e$, \linebreak \resizebox{394pt}{!}{$\left(\bigwedge[i \in I] \pi(e_i)\right) \otimes \pi(e) = \pi\left(\left(\bigwedge[i \in I] e_i\right) \+ e\right) = \pi\left(\bigwedge[i \in I] (e_i \+ e)\right) = \bigwedge[i \in I] \pi(e_i \+ e) = \bigwedge[i \in I] (\pi(e_i) \otimes \pi(e))$.} 

Thus $\otimes$ admits a residual. 

\end{enumerate} 
\end{proof} 

Then, we define $\fct[i][A, \lcal(A)^{\e<{[0,1]}}]{a, \big(q \mapsto \bigwedge[\ust{d \in \dcal}{d < q}] \pi \circ l(a \+ \ul{1 - d})\big)}$ and $\fct[k][USC(\lcal(A)^{op}), A]{f, \bigwedge[d \in \dcal] \ul d \+ \pi^\ast(f(d))}$. 

We prove, in order, the following statements. 

\begin{enumerate} 

\item The image of $i$ is a subset of $USC(\lcal(A)^{op})$, and we still call $i$ the induced map, and, for all $a \in A$ and $d \in \dcal$, $i(a)(d) = \pi \circ l(a \+ \ul{1 - d})$. 

\item $k \circ i \simeq \id_A$. 

\item $i \circ k = \id_{USC(\lcal(A)^{op})}$. 

\item $i$ and $k$ are non-decreasing. 

\item As soon as $i$ preserves $\+$, $i$ preserves $\-$. As soon as $i$ preserves $l$, $i$ preserves $l^\ast$. 

\item $i$ preserves $\+$, $l$ and the $\ul d$'s for $d \in \dcal$. 

\end{enumerate} 

\begin{enumerate} 

\item For all $(q_i)_{i \in I} \in [0,1]^I$ and $a \in A$, \begin{align*} i(a)\left(\bigvee[i \in I] q_i\right) &= \bigwedge[\ust{d \in \dcal}{d < \bigvee[i \in I] q_i}] \pi \circ l\left(a \+ \ul{1 - d}\right) 
\\ 
&\ust{=}{\hypertarget{star}{(\ast)}} \bigwedge[\ust{d \in \dcal}{\exists i \in I \tq d < q_i}] \pi \circ l\left(a \+ \ul{1 - d}\right) 
\\ 
&= \bigwedge[i \in I] \bigwedge[\ust{d \in \dcal}{d < q_i}] \pi \circ l\left(a \+ \ul{1 - d}\right) 
\\ 
&= \bigwedge[i \in I] i(a)({q_i}),\end{align*} \hyperlink{star}{$(\ast)$} being a consequence of \textbf{Lemma} \ref{remark dyadic}. 

Hence, according to \textbf{Theorem} \ref{scale}, $i(a) \in USC(\lcal(A)^{op})$ for all $a \in A$. 

Let $a \in A$ and $d \in \dcal$. Since $\pi$, $l$ and $\+$ preserve lower bounds, \begin{align*} i(a)(d) &= \bigwedge[\ust{d' \in \dcal}{d' < d}] \pi \circ l(a \+ \ul{1 - d'})\\ 
&= \pi \circ l\left(\bigwedge[\ust{d' \in \dcal}{d' < d}] (a \+ \ul{1 - d'})\right)\\ 
&= \pi \circ l\left(a \+ \ul{\bigwedge[\ust{d' \in \dcal}{d' < d}] (1 - d')}\right)\\ 
&= \pi \circ l(a \+ \ul{1 - d}).
\end{align*}  

\item By \ref{[0,1] 2}, for all $a \in A$ and $d \in \dcal$, $a \preceq \ul d \+ l(a \+ \ul{1 - d})$, i.e. $a \preceq \bigwedge[d \in \dcal] \ul d \+ l(a \+ \ul{1 - d}) \simeq k \circ i(a)$, and, for all $n \in \nn$, $a \+ \ul{\frac{1}{2^n}} \succeq \bigwedge[k = 1][2^{n+1} -1] \ul{1 - \frac{k-1}{2^{n+1}}} \+ l(v \+ \ul{\frac{k-1}{2^{n+1}}}) \succeq \bigwedge[d \in \dcal] \ul d \+ l(a \+ \ul{1 - d}) \simeq k \circ i(a)$, so, since $A$ is Archimedean, $a \succeq k \circ i (a)$. 

\item Let $f \in USC(\lcal(A)^{op})$. For all $d \in [0,1]$, 

\vspace{-15 pt} \begin{align*} l(k(f) \+ \ul{1 - d}) &= l\left(\bigg(\bigwedge[d' \in \dcal] \ul{d'} \+ \pi^\ast(f( d'))\bigg) \+ \ul{1 - d}\right)\\ 
&\ust{\simeq}{\ref{ll* >=} \et \ref{l*l <=} } \bigwedge[d' \in \dcal] l\big(\pi^\ast(f( d')) \+ \ul{(1 - d + d') \wedge 1}\big)\\ 
&\ust{\simeq}{\ref{def l} \et \ref{redef +} } \bigwedge[d' \in [0,1]] l(\pi^\ast(f( d'))) \+ \ul{l\big((1 - d + d') \wedge 1\big)}\\ 
&= \bigwedge[\ust{d' \in \dcal}{d' < d}] \pi^\ast(f( d')),\end{align*}

so, thanks to \textbf{Lemma} \ref{celui-ci}, $$i(k(f))( d) = \pi\left(\bigwedge[\ust{d' \in \dcal}{d' < d}] \pi^\ast(f( d'))\right) = \bigwedge[\ust{d' \in \dcal}{d' < d}] \pi\circ \pi^\ast(f( d')) = \bigwedge[\ust{d' \in \dcal}{d' < d}] f( d') = f( d).$$ 

\item For all $a \leq b \in A$ and all $d \in \dcal$, $i(a)( d) = \pi \circ l(a \+ \ul{1 - d}) \leq \pi \circ l(b \+ \ul{1 - d}) = i(b)( d)$, and so $i(a) \leq i(b)$. 

For all $f \leq g \in USC(\lcal(A)^{op})$, $k(f) = \bigwedge[d \in \dcal] \ul d \+ \pi^\ast(f(d)) \preceq \bigwedge[d \in \dcal] \ul d \+ \pi^\ast(g(d)) = k(g)$. 

\item Assume that $i$ preserves $\+$. 

Then, for all $a$ and $b \in A$, and $f \in USC(\lcal(A)^{op})$, 

\vspace{-15 pt} \begin{align*} i(a) \- i(b) \leq f &\Leftrightarrow i(a) \leq i(b) \+ i \circ k(f) \\ 
&\Leftrightarrow i(a) \leq i(b \+ k(f)) \\ 
&\Leftrightarrow a \preceq b \+ k(f) \\ 
&\Leftrightarrow a \- b \preceq k(f) \\ 
&\Leftrightarrow i(a \- b) \leq f.\end{align*} 

Assume that $i$ preserves $l$. 

Then, for all $a$ and $b \in A$, and $f \in USC(\lcal(A)^{op})$, 

\vspace{-15 pt} \begin{align*} l^\ast(i(a)) \leq f &\Leftrightarrow i(a) \leq l(i(f)) \\ 
&\Leftrightarrow i(a) \leq i(l(f)) \\ 
&\Leftrightarrow a \preceq l(f) \\ 
&\Leftrightarrow l^\ast(a) \preceq k(f) \\ 
&\Leftrightarrow i(l^\ast(a)) \leq f.\end{align*} 

\item For all $d \et d'\in \dcal$, \begin{align*} i(\ul d)(d') &= \pi \circ l(\ul d \+ \ul{1 - d'}) 
\\ 
&\ust{=}{\ref{redef +}} \pi \circ l(\ul{(1 + d - d') \wedge 1}) 
\\ 
&\ust{=}{\ref{def l}} \pi(\ul{l\big((1 + d - d') \wedge 1\big)}) 
\\ 
&= \acc{\pi(\ul 1)\!\!\!\!, d \geq d', \pi(\ul 0)\!\!\!\!} 
\\ 
&= \ul d (d').\end{align*} 

For all $a \in A$ and $d \in \dcal$, $$i(l(a))(d) = l(l(a) \+ \ul{1 - d}) = l^2(a) \+ \ul{l(1 - d)} \ust{=}{\ref{def l}} l(a) \+ \ul{1 \- l^\ast(d)} = i(a)(l^\ast(d)) = l(i(a))(d).$$ 

For all $a$ and $b \in A$, \begin{align*}a \+ b &\simeq \left(\bigwedge[q \in \dcal] \ul q \+ l(a \+ \ul{1 - q})\right) \+ \left(\bigwedge[p \in \dcal] \ul p \+ l(b \+ \ul{1 - p})\right) \texte{by \ref{celui-ci} (1)}\\ 
&\simeq \bigwedge[q,p \in \dcal] \ul{q \dotp p} \+ l(a \+ \ul{1 - q}) \+ l(b \+ \ul{1 - p}) \texte{by \ref{residuation} \et \ref{redef +}}\\ 
&\simeq \bigwedge[p,q \in \dcal] \ul{q \dotp p} \+ \pi^\ast(i(a)(q)) \+ \pi^\ast(i(b)(p)) \texte{by definition of} i\\  
&\simeq \bigwedge[p,q \in \dcal] \ul{q \dotp p} \+ \pi^\ast(i(a)(q) \otimes i(b)(p)) \texte{by \ref{celui-ci} (2)} \\ 
&\simeq \bigwedge[p,q \in \dcal] \bigwedge[r \wa p \dotp q] \big(\ul{r} \+ \pi^\ast(i(a)(q) \otimes i(b)(p))\big) \\ 
&\simeq \bigwedge[r \in \dcal] \ul r \+ \pi^\ast\left(\bigwedge[p \dotp q \wb r] \big(i(a)(q) \otimes i(b)(p)\big)\right) \texte{by \ref{residuation}} \\ 
&\simeq \bigwedge[r \in \dcal] \ul r \+ \pi^\ast\big((i(a) \+ i(b))(r)\big) \\ 
&\simeq k\big(i(a) \+ i(b)\big) \end{align*}

so $i(a \+ b) = i(a) \+ i(b)$. 

\end{enumerate} 
\end{proof} 

\subsubsection{From $\T_0$ to $\T_1$} 

Now that we have proven any complete Arcimedean model of $\T_0$ to be isomorphic to an AC-algebra, we prove that the extra structure of $\T_1$, which contains the extra structure of $\T$, is preserved by the isomorphisms between models of $\T_0$ and AC-algebras, and therefore that every complete Archimedean model of $\T_1$ is isomorphic to an AC-algebra (\bf{Corollary} \ref{The corollaire}). Then, we will deal with non complete non Archimedean models of $\T$. 
\vspace{\baselineskip} 
Let $A$ be a complete Archimedean model of $\T_0$. 

\begin{lem} \label{Utile} 

For all $\phi \et \psi \colon [0,1] \rightarrow [0,1]$ such that $\phi$ is right adjoint to $\psi$, for all $a \in A$, $$\bigwedge[d \in \dcal] \ul d \+ l(a \+ \ul{1 - \psi(d)}) \simeq \bigwedge[d \in \dcal] \ul{\phi(d)} \+ l(a \+ \ul{1 - d}).$$ 

\end{lem} 

\begin{proof} 

For all $a \in A$ and $d \in \dcal$, $d \geq \phi(\psi(d))$, so $\ul d \+ l(a \+ \ul{1 - \psi(d)}) \succeq \ul{\phi(\psi(d))} \+ l(a \+ \ul{1 - \psi(d)})$ and $d \leq \psi(\phi(d))$, so $\ul{\phi(d)} \+ l(a \+ \ul{ 1 - d}) \succeq l{\phi(d)} \+ l(a \+ \ul{ 1 - \psi(\phi(d))})$. 
\end{proof} 

\begin{prop} \label{la proposition} 

Let $\phi$ and $\psi$ be two new function symbols of arity $1$, $a_\phi$ and $a_\psi$ be their interpretations in $[0,1]$ and define the following axioms 

\begin{enumerate} 

\item \label{phi et psi croissent} $\phi$ and $\psi$ are non-decreasing 

\item \label{phi psi} $\psi \circ \phi (v) \preceq v$ 

\item \label{psi phi} $v \preceq \phi \circ \psi (v)$ 

\item \label{def phi} $\phi(\ul d \+ l(v)) \simeq \ul{a_\phi(d)} \+ l(a)$, $d$ a dyadic in $[0,1]$ 

\end{enumerate} 

Assume \ref{phi et psi croissent}, \ref{phi psi} and \ref{psi phi} are satisfied by $([0,1], a_\phi, a_\psi)$ (notice that \ref{def phi} is satisfied by $[0,1]$). 

For every $\phi_A\colon A \rightarrow A$, the following two statements are equivalent: \begin{enumerate} 

\item there exists $\psi_A\colon A \rightarrow A$ such that $(A, \phi_A, \psi_A)$ satisfies \ref{phi et psi croissent}, \ref{phi psi}, \ref{psi phi} and \ref{def phi} 

\item for all $a \in A$, $\phi_A(a) \simeq \bigwedge[d \in \dcal] \ul d \+ l(a \+ \ul{1 - a_\psi(d)})$. 

\end{enumerate} 

In this case, for all $a \in A$, $\psi_A(a) \simeq \bigwedge \left\{b \in A \, | \, a \preceq \phi_A(b)\right\}$. 

\end{prop} 

\begin{proof} 

Assume that, for all $a \in A$, $\phi_A(a) \simeq \bigwedge[d \in \dcal] \ul d \+ l(a \+ \ul{1 - a_\psi(d)})$. Then, for all $(a_i)_{i \in I} \in A^I$, \linebreak $\phi_A\left(\bigwedge[i \in I] a_i\right) \simeq \bigwedge[d \in \dcal] \ul d \+ l\left(\bigwedge[i \in I] a_i \+ \ul{1 - a_\psi(d)}\right) \simeq \bigwedge[i \in I] \bigwedge[d \in \dcal] \ul d \+ l(a_i \+ \ul{1 - a_\psi(d)}) \simeq \bigwedge[i \in I] \phi_A(a_i)$. Let $\psi_A$ be define as in \bf{Proposition} \ref{la proposition}. For all $a \et b \in A$, if $\phi_A(b) \succeq a$, then $b \succeq \psi_A(a)$. Let now $a \et b \in A$ such that $\psi_A(a) \preceq b$. $\phi_A(b) \succeq \phi_A(\psi_A(a)) = \bigwedge \left\{\phi_A(c), \; c \in A \tq a \preceq \phi_A(c)\right\} \succeq a$. 

Thus $\phi_A$ admits $\psi_A$ a left adjoint, i.e. $(A, \phi_A, \psi_A)$ satisfies \ref{phi et psi croissent}, \ref{phi psi} and \ref{psi phi}. Finally, by \ref{def l} and \ref{redef +}, for all $a \in A$ and $q \in [0, 1]$, \begin{align*} \phi_A(\ul q \+ l(a)) &\simeq \bigwedge[d \in \dcal] \ul d \+ l(\ul q \+ l(a) \+ \ul{1 - a_\psi(d)}) \simeq \bigwedge[d \in \dcal] \ul d \+ l(l(a) \+ \ul{1 - (a_\psi(d) \- q)}) 
\\ 
&\simeq \bigwedge[d \in \dcal] \ul d \+ l(a) \+ \ul{l(1 - (a_\psi(d) \- q))} \simeq \bigwedge[d \geq a_\phi(q) \texte{dyadic}] \ul d \+ l(a) \simeq \ul{a_\phi(q)} \+ l(a). \end{align*} 

Hence $(A, \phi_A, \psi_A)$ satisfies \ref{def phi}. 

Assume now that there exists $\psi_A\colon A \rightarrow A$ such that $(A, \phi_A, \psi_A)$ satisfies \ref{phi et psi croissent}, \ref{phi psi}, \ref{psi phi} and \ref{def phi}. Then, $\phi_A$ preserves lower bounds up to $\simeq$, so, using \bf{Lemma} \ref{Utile}, for all $a \in A$, $$\phi_A(a) \ust{\simeq}{\ref{celui-ci} (1)} \bigwedge[d \in \dcal] \phi_A(\ul d \+ l(a \+ \ul{1 - d})) \ust{\simeq}{\ref{def phi}} \bigwedge[d \in \dcal] \ul{\phi(d)} \+ l(a \+ \ul{1 - d})) \simeq \bigwedge[d \in \dcal] \ul{d} \+ l(a \+ \ul{1 - \psi(d)}).$$ 
\end{proof} 

\begin{lem} 

For all $\phi \et \psi \colon [0,1] \rightarrow [0,1]$ such that $\phi$ is right adjoint to $\psi$, for all commutative residuated complete lattice $\lcal$, \ref{phi et psi croissent}, \ref{phi psi}, \ref{psi phi} and \ref{def phi} are satisfied by  $(\USC, \phi, \psi)$, where $\phi$ and $\psi$ are defined on $\USC$ as in \bf{Lemma} \ref{def}. 

\end{lem} 

\begin{proof} 

Let $a \et b \colon [0,1] \rightarrow [0,1]$ such that $a$ is right adjoint to $b$, and let $\lcal$ be a commutative residuated complete lattice. \bf{Lemma} \ref{def} already proved \ref{phi et psi croissent}, \ref{phi psi} and \ref{psi phi}. 

For all $f \in \USC$, $d \in \dcal$ and $q \in [0,1]$, according to \bf{Lemmas} \ref{calculatoire} and \ref{0-indicator}, \begin{align*} \phi(\ul d \+ l(f))(q) &= (\ul d \+ 0_{f(1)})(\psi(q)) 
\\ 
&= 0_{f(1)}(\psi(q) \- d) 
\\ 
&= 0_{f(1)}(q \- \phi(d)) 
\\ 
&= (\ul{\phi(d)} \+ l(f))(q).\end{align*} 
\end{proof} 

Thus, for any complete Archimedean model $A$ of $\T_0$ and $\phi\colon [0,1] \rightarrow [0,1]$ admitting a left adjoint $\psi$, let's define $\fct[\phi_A][A, A]{a, \bigwedge[d \in \dcal] \ul d \+ l(a \+ \ul{1 - \psi(d)})}$ and $\fct[\phi_A][A, A]{a, {\bigwedge \left\{b \in A \, |\, a \preceq \phi_A(b)\right\}}}$. 

\begin{thm} \label{Il faut le citer plus bas} 

Let $\phi\colon [0,1] \rightarrow [0,1]$ such that $\phi$ admits a left adjoint $\psi$. Let $A$ be complete Archimedean model of $\T_0$ and $\lcal$ be the \crcl associated to $A$ by \bf{Theorem} \ref{celui-là}. For every $\Phi\colon A \rightarrow A$, the following two statements are equivalent : 

\begin{enumerate} 

\item the quotient of $(A, \Phi)$ by $\simeq$ is isomorphic to $(\USC, \phi_{\USC})$ 

\item $\Phi \simeq \phi_A$. 

\end{enumerate} 

In this case, the isomorphism also preserves $\psi$. 

\end{thm} 

\begin{proof} 

Assume first that $\Phi \simeq \phi_A$. 

Let $i$ and $k$ be as in the proof of \textbf{Theorem} \ref{celui-là}. Let us also denote by $\bar\phi$ and $\bar\psi$ the respective interpretations of $\phi$ and $\psi$ in $USC(\lcal(A)^{op})$. For all $a \in A$, \begin{align*} k(\bar\phi(i(a))) &= \bigwedge[d \in \dcal] \ul d \+ \pi^\ast(\bar\phi(i(a))(d)) = \bigwedge[d \in \dcal] \ul d \+ \pi^\ast(i(a)(\psi(d))) = \bigwedge[d \in \dcal] \ul d \+ l(a \+ \ul{1 - \psi(d)}) 
\\ 
&= \bigwedge[d \in \dcal] \ul{\phi(d)} \+ l(a \+ \ul{1 - d}) = \phi_A(a) \simeq \Phi(a),\end{align*} so $i(\Phi(a)) = \bar\phi(i(a))$. 

Conversely, if the quotient of $(A,\Phi)$ $\simeq$ is isomorphic to $(\USC, \phi_{\USC})$, then the isomorphism preserves the lower bounds, so,by \bf{Lemma} \ref{0-indicator}, it suffices to prove that, for every commutative residuated complete lattice $\lcal$ and $f \in \USC$, $\bar\phi(f) = \bigwedge[q \in {[0,1]}] \ul q \+ 0_{f(\psi(q))}$. Let thus $\lcal$ be a commutative residuated complete lattice and $f \in \USC$. 

For all $p \in [0,1]$, $$\left(\bigwedge[q \in {[0,1]}] \ul q \+ 0_{f(q)}\right) (p) = \bigvee[q \in {[0,1]}] \bigvee[r < p] \ul q (r) \otimes 0_{f(q)}(p - r) = \bigvee[q < r < p] 0_{f(q)}(p - r) = \bigvee[q < p] f(q) = f(p).$$ 

Applying this to $\bar\phi(f) = f \circ \psi$ gives $\bar\phi(f) = \bigwedge[q \in {[0,1]}] \ul q \+ 0_{f(\psi(q))}$. 

Assume now that the quotient of $(A, \Phi)$ by $\simeq$ is isomorphic to $(\USC, \bar\phi)$ and write $i$ the composition of such an isomorphism with the quotient map. 

For all $a \in A$ and $f \in USC(\lcal)$, there exists $b \in A$ such that $f = i(b)$ and thus, \begin{align*} i(\psi_A(a)) \leq f &\Leftrightarrow i(\psi_A(a)) \leq i(b) \\ 
&\Leftrightarrow \psi_A(a) \preceq b \\ 
&\Leftrightarrow a \preceq \phi_A(b) \\ 
&\Leftrightarrow i(a) \leq i(\phi_A(b)) \\ 
&\Leftrightarrow i(a) \leq \bar\phi(i(b)) \\ 
&\Leftrightarrow \bar\psi(i(a)) \leq f.
\end{align*} 
\end{proof} 

\begin{cor} \label{The corollaire} 

For every complete Archimedean model $A$ of $\T_1$, there exists $\lcal$ a commutative residuated complete lattice such that $A$ is isomorphic to $\USC$. 

\end{cor} 

\begin{proof} 

Let $A$ by a complete Archimedean model of $\T_1$. It is a complete Archimedean model of $\T_0$, so, according to \bf{Theorem} \ref{celui-là}, there exists $\lcal$ a \crcl and $i\colon A \rightarrow \USC$ a surjective morphism such that, for all $a \et b \in A$, $a \preceq b \Leftrightarrow i(a) \leq i(b)$. Since the interpretations of $(\USC, \, 2 \cdot, \, \frac{\cdot}{2})$, $(\USC,\, j_\ast,\, j)$, $(\USC,\, \alpha,\, \beta)$ and $(\USC,\, l,\, l^\ast)$ satisfie points 1 to 4 of \bf{Proposition} \ref{la proposition}, $2 \cdot = (2 \cdot)_A$, $j_\ast = {j_\ast}_A$, $\alpha = \alpha_A$ and $l = l_A$. Hence, according to \bf{Theorem} \ref{Il faut le citer plus bas}, $i$ preserves $2 \cdot$, $j_\ast$, $\alpha$, $l$ and their left adjoints $\frac{\cdot}{2}$, $j$, $\beta$ and $l^ \ast$. 

\end{proof} 

\subsection{Models of $\T$} \label{subsection Models of T} 

\subsubsection{Complete models of $\T$} 

Now that we have proven any complete Archimedean model of $\T_1$ to be isomorphic to an AC-algebra, we prove that any complete model of $\T$ can be endowed with a preorder that makes it a complete Archimedean model of $\T_1$. Therefore its quotient by the equivalence relation induced by the preorder is isomorphic to an AC-algebra (\bf{Theorem} \ref{Complétude}). Then, we will deal with non complete non Archimedean models of $\T$. 

In this section, we prove the following theorem. 

\begin{thm} \label{Complétude} 

For all complete model $A$ of $\T$, there exists a commutative residuated complete lattice $\lcal$ such that $\quotient{A}{\simeq}$ is isomorphic to $USC(\lcal)$. 

\end{thm} 

Let us now assume $A$ to be a complete model of $\T$. 

\begin{nota} \label{preceq 0} 

Define $\preceq$ on $A$ by: for all $a \et b \in A$, $a \preceq b \Leftrightarrow \forall n \in \nn \; a \leq b \+ \ul{\frac{1}{2^n}}$ \blue{and $\simeq$ by $a \simeq b \Leftrightarrow a \preceq b \et b \preceq a$.} Let's define $\lcal(A) = \{e \in A \, | \, 2 e \simeq e\}$. 

\end{nota} 

To this purpose, we will prove that, by defining $\fct[l][A, \blue{A}]{a, \bigvee[\ust{e \in \lcal(A)}{e \preceq a}] e}$ and $\fct[\beta][A,A]{a, 2a \wedge j_\ast(a)}$, there exists $l^\ast\colon A \rightarrow A$ such that $(A,\, \preceq,\, \wedge,\, (\ul d)_{d \in \dcal},\, \+,\, \-,\, 2,\, \frac{\cdot}{2},\, j_\ast,\, j,\, \alpha,\, \beta,\, l, l^\ast)$ satisfies $\T_1$. 

\begin{lem} \label{définition dyadiques} $ $ 

\begin{enumerate} 

\item \label{définition dyadiques 1} For all $n \in \nn$, $\alpha^{n-1} \circ j_\ast (\ul 0) = \ul{\frac{1}{2^n}}$. 

\item \label{définition dyadiques 2} For all dyadic $d \in [0,1]$ and $n \in \nn^*$ such that $2^n d \in \nn$, $\ul d = \sum[k = 1][2^n d] \ul{\frac{1}{2^n}}$. 

\end{enumerate} 

\end{lem} 

\begin{proof} \begin{enumerate} 

\item By definition, $\ul{\frac{1}{2}} = j_\ast(\ul 0)$. 

Let $n \in \nn$ such that $\alpha^{n-1} \circ j_\ast(\ul 0) = \ul{\frac{1}{2^n}}$. $\alpha^n \circ j_\ast(\ul 0) = \alpha\left(\ul{\frac{1}{2^n}}\right) = \ul{\frac{1}{2^{n+1}}} \vee j\left(\ul{\frac{1}{2^n}}\right)$. However, $\ul{\frac{1}{2^n}} \leq \ul{\frac{1}{2}} = j_\ast(\ul 0)$, so $j\left(\ul{\frac{1}{2^n}}\right) \leq \ul 0 \leq \ul{\frac{1}{2^{n+1}}}$. Thus $\alpha^n \circ j_\ast(\ul 0) = \ul{\frac{1}{2^{n+1}}}$. 

\item For all $n \in \nn^*$, $\ul{\frac{1}{2^{n+1}}} \+ \ul{\frac{1}{2^{n+1}}} = \alpha^n \circ j_\ast(\ul 0) + \alpha^n \circ j_\ast(\ul 0) \leq \ul{\frac{1}{2^n}}$ and $\ul{\frac{1}{2^{n+1}}} \+ \ul{\frac{1}{2^{n+1}}} \geq 2 \ul{\frac{1}{2^{n+1}}} \geq \ul{\frac{1}{2^n}}$. Hence, for all $n \in \nn^*$, $\ul{\frac{1}{2^{n+1}}} \+ \ul{\frac{1}{2^{n+1}}} = \ul{\frac{1}{2^n}}$. Moreover, $\ul{\frac{1}{2}} \+ \ul{\frac{1}{2}} \geq 2 \ul{\frac{1}{2}} \geq \ul 1$, so $\ul{\frac{1}{2}} \+ \ul{\frac{1}{2}} = \ul 1$. 

Let $n \in \nn^*$ such that $2^n d \in \nn$. Let $n_0$ be the smallest non negative integer such that $2^{n_0} d \in \nn$. 

$$\sum[k = 1][2^n d] \ul{\frac{1}{2^n}} = \sum[k = 1][2^{n_0} d] \; \sum[i = 1][2^{n - n_0}] \frac{\ul{\frac{1}{2^{n_0}}}}{2^{n - n_0}} = \sum[k = 1][2^{n_0} d] \ul{\frac{1}{2^{n_0}}} = \ul d.$$ 

\end{enumerate} 
\end{proof} 






\begin{lem} \label{l is nice} 

Let $a \in A$ and $d \in \dcal$. \begin{enumerate} 

\item $2j_\ast(a) = \ul 1$, $\frac{a}{2} \leq a$, $j(a) \leq a$, $l(a) \preceq a$. 

\item For all $n \geq 0$, ${j_\ast}^n(a) \+ \ul{\frac{1}{2^n}} \geq \ul 1$.

\item $\alpha$ is invertible, \blue{of inverse $\beta$}, and, for all $n \in \nn$, $\alpha^n(a) = \bigvee[k = 0][n] \frac{j^k(a)}{2^{n-k}}$ and $\alpha^{-n}(a) = \bigwedge[k = 0][n] j_\ast^{n-k}(2^k a)$. Hence $\alpha(a) \leq a$ and $\alpha^{-1}(a) \geq a$. 

\item Let $c \in A$. \blue{$\_ \+ c$, $2 \cdot$, $\frac{\cdot}{2}$, $j_\ast$, $j$, $\alpha$ and $\alpha^{-1}$ are non-decreasing for $\preceq$. $\_ \- c$ is non-decreasing and $c \- \_$ is non-increasing for $\preceq$.} 

Moreover, for all $a \preceq b$, $l(a) \leq l(b)$. 

\item For all $n \in \nn$, $\alpha^n(a) \leq j(a) \+ \ul{\frac{1}{2^n}}$ and $2a \leq \alpha^{-n}(a) \+ \ul{\frac{1}{2^n}}$. 

\item \label{simeq equivalences} $j(a) \simeq a \Leftrightarrow \alpha(a) \simeq a \Leftrightarrow \alpha^{-1}(a) \simeq a \Leftrightarrow 2a \simeq a$. 

\item \label{fixed point} \blue{$l(a)$ is a fixed point of $2\cdot, \, \alpha \et \alpha^{-1}$, so $l(a) \in \lcal(A)$ and $l(l(a)) = l(a)$.} 

\item \blue{$\bigwedge[n \in \nn] \alpha^n(a) \leq l(a) \preceq \bigwedge[n \in \nn] j^n(a)$, so $l(a) \simeq \bigwedge[n \in \nn] \alpha^n(a) \simeq \bigwedge[n \in \nn] j^n(a)$.} 

\end{enumerate} 

\end{lem} 

\begin{proof} 

Let $a \in A$ and $d \in \dcal$. \begin{enumerate} 

\item $2j_\ast(a) \geq 2j_\ast(\ul 0) \ust{\geq}{\ref{2j*(0) >= 1}} \ul 1$. 

$a \leq 2a$ by \ref{2 >=}, so, by \ref{(2v)/2<=v}, $\frac{a}{2} \leq a$. 

$a \leq j_\ast(a)$ by \ref{j* >=}, so, by \ref{jj*v<=v}, $j(a) \leq a$. 

For all $n \in \nn$ and $e \in \lcal(A)$ such that $e \preceq a$, $e \leq a \+ \ul{\frac{1}{2^n}}$, so $l(a) \leq a \+ \ul{\frac{1}{2^n}}$, i.e. $l(a) \preceq a$. 

\item We notice that, for all $n \in \nn$, $$2\sum[k = 1][n] \ul{\frac{1}{2^{k+1}}} \leq \sum[k = 1][n] \ul{\frac{1}{2^{k+1}}} \+ \sum[k = 1][n] \ul{\frac{1}{2^{k+1}}} = \sum[k = 1][n] \ul{\frac{1}{2^{k+1}}} \+ \ul{\frac{1}{2^{k+1}}} 
= \sum[k = 1][n] \ul{\frac{1}{2^k}}.$$ 

It suffices to prove, by induction on $n \geq 0$, that, for all $n \in \nn$, ${j_\ast}^n(\ul 0) \geq \sum[k = 1][n] \ul{\frac{1}{2^k}}$ and $\sum[k = 1][n] \ul{\frac{1}{2^{k}}} \+ \ul{\frac{1}{2^n}} = \ul 1$. 

${j_\ast}^0(\ul 0) \geq \ul 0 = \sum[k = 1][0] \ul{\frac{1}{2^k}}$ and $\sum[k = 1][0] \ul{\frac{1}{2^{k}}} \+ \ul{\frac{1}{2^0}} = \ul 1$. 

For all $n \geq 0$, such that ${j_\ast}^n(\ul 0) \geq \sum[k = 1][n] \ul{\frac{1}{2^k}}$ and $\sum[k = 1][n] \ul{\frac{1}{2^{k+1}}} \+ \ul{\frac{1}{2^n}} = \ul 1$, $${j_\ast}^{n+1}(\ul 0) \geq j_\ast\left(\sum[k = 1][n] \ul{\frac{1}{2^k}}\right) \geq j_\ast\left(2\sum[k = 1][n] \ul{\frac{1}{2^{k+1}}}\right) \geq \sum[k = 1][n] \ul{\frac{1}{2^{k+1}}} \+ \ul{\frac{1}{2}} = \sum[k = 1][n + 1] \ul{\frac{1}{2^k}}$$ and $$\sum[k = 1][n + 1] \ul{\frac{1}{2^k}} \+ \ul{\frac{1}{2^{n+1}}} = \sum[k = 1][n] \ul{\frac{1}{2^k}} \+ \ul{\frac{1}{2^{n+1}}} \+ \ul{\frac{1}{2^{n+1}}} \geq \sum[k = 1][n] \ul{\frac{1}{2^k}} + 2 \ul{\frac{1}{2^{n+1}}} \geq  \sum[k = 1][n] \ul{\frac{1}{2^k}} \+ \ul{\frac{1}{2^n}} \geq \ul 1.$$ 

Hence, for all $n \in \nn$, ${j_\ast}^n(a) \+ \ul{\frac{1}{2^n}}\geq j_\ast(\ul 0) \+ \ul{\frac{1}{2^n}} \geq \sum[k = 1][n] \ul{\frac{1}{2^k}} \+ \ul{\frac{1}{2^n}} = \ul 1$. 

\item By induction on $n \geq 0$. 

By \ref{alpha beta v >= v}  to \ref{beta alpha v >= v} , we already know that $\alpha$ is invertible and $\alpha^{-1}(a) = 2a \wedge j_\ast(a) \blue{= \beta(a)}$. 

Let $n \in \nn$ such that $\alpha^{-n}(a) = \bigwedge[k = 0][n] j_\ast^{n-k}(2^k a)$. 

\begin{align*} \alpha^{-(n+1)}(a) &= \alpha^{-1}\left(\bigwedge[k = 0][n] j_\ast^{n-k}(2^k a)\right) = \bigwedge[k = 0][n] 2 j_\ast^{n-k}(2^k a) \wedge \bigwedge[k = 0][n] j_\ast^{n+1-k}(2^k a) 
\\ 
&= 2^{n+1}(a) \wedge \ul 1 \wedge \bigwedge[k = 0][n] j_\ast^{n+1-k}(2^k a) = \bigwedge[k = 0][n + 1] j_\ast^{n+1-k}(2^k a).\end{align*} 

Hence, for all $n \in \nn$, $\alpha^{-n}(a) = \bigwedge[k = 0][n] j_\ast^{n-k}(2^k a)$. 

Thus, for all $n \in \nn$, and $c \in A$, \begin{align*} \alpha^n(a) \leq c &\Leftrightarrow a \leq \alpha^{-n}(c) 
\\ 
&\Leftrightarrow \forall 0 \leq k \leq n \; a \leq j_\ast^{n-k}(2^k c) 
\\ 
&\ust{\Leftrightarrow}{\ref{(2v)/2<=v} - \ref{v<=j*jv}} \forall 0 \leq k \leq n \; \frac{j^{n-k}(a)}{2^k} \leq c 
\\ 
&\Leftrightarrow \bigvee[k = 0][n] \frac{j^{n-k}(a)}{2^k} \leq c,\end{align*} 

so $\alpha^n(a) = \bigvee[k = 0][n] \frac{j^{n-k}(a)}{2^k}$. 

\item Let $a \et b \in A$ and $n \in \nn$ such that $a \leq b \+ \ul{\frac{1}{2^n}}$. 

By \ref{2(+) <= 2 + 2}, $2a \leq 2 \left(b \+ \ul{\frac{1}{2^n}}\right) \leq 2b \+ 2\ul{\frac{1}{2^n}} = 2b \+ \ul{\frac{1}{2^{n-1}}}$. 

By \ref{def j* <=}, $j_\ast(a) \leq j_\ast\left(b \+ \ul{\frac{1}{2^{n+1}}}\right) \leq j_\ast(b) \+ \ul{\frac{1}{2^n}}$. 

\blue{\ref{+ 2 <= 2(+)} is equivalent to $\frac{v \+ u}{2} \leq \frac{v}{2} \+ \frac{u}{2}$, so $\frac{a}{2} \leq \frac{b \+ \ul{\frac{1}{2^n}}}{2} \leq \frac{b}{2} \+ \ul{\frac{1}{2^{n+1}}}$. 

\ref{def j* >=} is equivalent to $j(v \+ u) \leq 2v \+ j(u)$, so $j(a) \leq j\left(b \+ \ul{\frac{1}{2^{n+1}}}\right) \leq j(b) \+ \ul{\frac{1}{2^n}}$. 

$\alpha(a) \leq \alpha\left(b \+ \ul{\frac{1}{2^{n+1}}}\right) = \frac{b \+ \ul{\frac{1}{2^{n+1}}}}{2} \vee j(b \+ \ul{\frac{1}{2^{n+1}}}) \leq \left(\frac{b}{2} \+ \ul{\frac{1}{2^n}}\right) \vee \left(j(b) \+ \ul{\frac{1}{2^n}}\right) \leq \alpha(b) \+ \ul{\frac{1}{2^n}}$.} 

$\alpha^{-1}(a) = 2a \wedge j_\ast(a) \leq \left(2b \+ \ul{\frac{1}{2^n}}\right) \wedge \left(j_\ast(b) \+ \ul{\frac{1}{2^n}}\right) = 2b \wedge j_\ast(b) \+ \ul{\frac{1}{2^n}} = \alpha^{-1}(b) \+ \ul{\frac{1}{2^n}}$. 

\blue{$a \- c \leq \left(b \+ \ul{\frac{1}{2^n}}\right) \- c \leq (b \- c) \+ \ul{\frac{1}{2^n}}$. 

$c \- a \geq c \- \left(b \+ \ul{\frac{1}{2^n}}\right) = c \- b \- \ul{\frac{1}{2^n}}$.} 

For all $e \in \lcal$ such that $e \preceq b$, $e \preceq a$, so $l(b) \leq l(a)$. 

Finally, by commutativity and associativity of $\+$, for all $c \in A$, $b \+ c \leq (a \+ c) \+ \ul{\frac{1}{2^n}}$. 

\item Let $n \in \nn$. For all $1 \leq k \leq n$, by \ref{2 >=} and \ref{j* >=}, $j_\ast^{n-k}(2^k a) \geq 2^k a \geq 2a$, and $j_\ast^n(a) \+ \ul{\frac{1}{2^n}} \geq \ul 1 \geq 2a$, so $\alpha^{-n}(a) \+ \ul{\frac{1}{2^n}} \geq 2a$. For all $0 \leq k \leq n - 1$, by \ref{(2v)/2<=v}, \ref{jj*v<=v}, \ref{2 >=} and \ref{j* >=}, $\frac{j^{n-k}(a)}{2^k} \leq j^{n-k}(a) \leq j(a)$, and $\frac{a}{2^n} \leq \ul{\frac{1}{2^n}} \leq j(a) \+ \ul{\frac{1}{2^n}} $, so $\alpha^n(a) \leq j(a) \+ \ul{\frac{1}{2^n}}$. 

\item $\alpha(a) \leq a$, so, if $j(a) \simeq a$, then, for all $n \geq 0$, $\alpha(a) \+ \ul{\frac{1}{2^n}} = (\frac{a}{2} \vee j(a)) \+ \ul{\frac{1}{2^n}} \geq j(a) \+ \ul{\frac{1}{2^n}} \geq a$, so $\alpha(a) \simeq a$. 

$\alpha^{-1}(a) \geq a$, so if $2a \simeq a$, then, for all $n \in \nn$, $\alpha^{-1}(a) = 2a \wedge j_\ast(a) \leq 2a \leq a \+ \ul{\frac{1}{2^n}}$, so $\alpha^{-1}(a) \simeq a$. 

If $\alpha(a) \simeq a$, then, since $\alpha$ is continuous, for all $n \in \nn$, $a \simeq \alpha^n(a) \leq j(a) \+ \ul{\frac{1}{2^n}}$, so $a \preceq j(a)$, i.e. $a \simeq j(a)$. 

$\alpha^{-1}$ is continuous because $2\cdot$ and $j_\ast$ are, so, if $\alpha^{-1}(a) \simeq a$, then, for all $n \in \nn$, $a \+ \ul{\frac{1}{2^n}} \simeq \alpha^{-n}(a) \+ \ul{\frac{1}{2^n}} \geq 2a$, so $a \succeq 2a$, i.e. $a \simeq 2a$. 

Finally, $a \simeq \alpha(a) \Leftrightarrow \alpha^{-1}(a) \simeq \alpha(a)$, thanks to  of $\alpha$ and $\alpha^{-1}$. 

\item \color{blue} Let $e \in \lcal(A)$ such that $e \preceq a$. $e \leq l(a)$, so $\alpha(e) \leq \alpha(l(a))$. However, according to point 6, $e \simeq \alpha(e)$, so $e \preceq \alpha(l(a))$. Hence, for all $n \in \nn$, $l(a) \leq \alpha(l(a)) \+ \ul{\frac{1}{2^n}}$ and thus $l(a) \preceq \alpha(l(a))$. However, $\alpha(l(a)) \leq l(a)$, so $\alpha(l(a)) \simeq l(a)$. Hence $\beta(l(a)) \simeq l(a)$ and $2l(a) \simeq l(a)$. 

$2l(a) \in \lcal(A)$, so, since $2l(a) \simeq l(a) \preceq a$, $2l(a) \leq l(a)$, i.e. $2l(a) = l(a)$. In the same way, $\beta(l(a)) = l(a)$ and thus $\alpha(l(a)) = l(a)$. 

\color{black} 

\item \blue{From points \ref{simeq equivalences} and \ref{fixed point}, we can deduce that $j(l(a)) \simeq l(a)$. Hence, since $l(a) \preceq a$, for all $n \in \nn$, $l(a) = j^n(l(a)) \preceq j^n(a)$, so $l(a) \preceq \bigwedge[n \geq N] j^n(a)$. Since, for all $a \in A$, $j(a) \leq \frac{a}{2} \vee j(a) = \alpha(a)$, $\bigwedge[n \in \nn] j^n(a) \leq \bigwedge[n \in \nn] \alpha^n(a)$. Since $\alpha\left(\bigwedge[n \in \nn] \alpha^n(a)\right) = \bigwedge[n \in \nn] \alpha^{n+1} (a) = \bigwedge[n \in \nn] \alpha^n (a)$, and $\bigwedge[n \in \nn] \alpha^n (a) \leq a$, $\bigwedge[n \in \nn] \alpha^n (a) \leq l(a)$.} 


\end{enumerate} 
\end{proof} 

\begin{lem} \label{lemme 2} For all $a \et b \in A$ and $d \et d' \in \dcal$, 

\begin{enumerate} 

\item $\ul d \+ \ul{d'} = \ul{d \dotp d'}$. 

\item $2 (\ul d \+ l(a)) = \ul{2 d} \+ l(a)$ 

\item $j_\ast (\ul d \+ l(a)) = \ul{j_\ast(d)} \+ l(a)$ 

\item $\alpha (\ul d \+ l(a)) = \ul{\alpha(d)} \+ l(a)$ 

\item $l (\ul d \+ l(a)) = \ul{l(d)} \+ l(a)$ 

\item $a \leq \ul d \+ l(a \+ \ul{1 - d})$ and for all $n \in \nn$, $\bigwedge[k = 1][2^{n+1}] \ul{1 - \frac{k}{2^{n+1}}} \+ l\left(a \+ \ul{\frac{k}{2^{n+1}}}\right) \blue\preceq a \+ \ul{\frac{1}{2^n}}$ . 

\end{enumerate} 

\end{lem} 

\begin{proof} Let $a \et b \in A$ and $d \et d' \in [0,1]$. 

\begin{enumerate} 

\item Let $n$ be a non negative integer such that $2^n d \et 2^n d' \in \nn$. 

$\ul d \+ \ul{d'} = \sum[k = 1][2^n d] \ul{\frac{1}{2^n}} + \sum[k = 1][2^n d'] \ul{\frac{1}{2^n}} = \sum[k = 1][2^n (d + d')] \ul{\frac{1}{2^n}} = \ul{d \+ d'}$ 

\item For all integer $n \geq 2$ such that $2^n d \in \nn$, $2 \ul{\frac{1}{2^n}} = 2 \frac{\ul{\frac{1}{2^{n-1}}}}{2} = \ul{\frac{1}{2^{n-1}}}$, and \linebreak $2 \ul{\frac{1}{2^n}} \leq \ul{\frac{1}{2^n}} \+ \ul{\frac{1}{2^n}} \leq \ul{\frac{1}{2^{n-1}}}$, so $2 \ul{\frac{1}{2^n}} = \ul{\frac{1}{2^{n-1}}}$. $2 \ul d = 2 \sum[k = 1][2^n d] \ul{\frac{1}{2^n}} = \sum[k = 1][2^n d] 2 \ul{\frac{1}{2^n}} = \sum[k = 1][2^{n-1} 2d] \ul{\frac{1}{2^{n-1}}} = \ul{2d}$. 

Thus, by \ref{+ 2 <= 2(+)}, \ref{2(+) <= 2 + 2} and \textbf{Lemma} \ref{l is nice} (7), $2(\ul d \+ l(a)) = 2 \ul d \+ 2 l(a) = \ul{2d} \+ l(a)$. 

\item $j_\ast (\ul d \+ l(a)) = j_\ast\left(2\left(\ul{\frac{d}{2}} \+ l(a)\right)\right) = \ul{\frac{d}{2}} \+ l(a) \+ \ul{\frac{1}{2}} = \ul{\frac{d}{2} \+ \frac{1}{2}} \+ l(a) = \ul{j_\ast(d)} \+ l(a)$. 

\item Since $\+$ admits a residual (cf. \ref{residuation}), $\+$ preserves lower bounds and so $$2(\ul d \+ l(a)) \wedge j_\ast(\ul d \+ l(a)) = (\ul{2d} \+ l(a)) \wedge (\ul{j_\ast(d)} \+ l(a)) = (\ul{2d} \wedge \ul{j_\ast(d)}) \+ l(a) = \ul{2d \wedge j_\ast(d)} \+ l(a).$$ 

Since, for all $x \in [0, 1]$, $2\alpha (x) \wedge j_\ast \circ \alpha(x) = x$, $$\alpha^{-1}(\ul{\alpha(d)} \+ l(a)) = 2(\ul{\alpha(d)} \+ l(a)) \wedge j_\ast(\ul{\alpha(d)} \+ l(a)) = \ul{2\alpha (d) \wedge j_\ast \circ \alpha(d)} \+ l(a) = \ul{d} \+ l(a),$$ so $\alpha (\ul d \+ l(a)) = \ul{\alpha(d)} \+ l(a)$. 

\item If $d = 1$, then $l(\ul d \+ l(a)) = \ul 1 = \ul{l(1)} \+ l(a)$. Assume that $q < 1$, ie $\bigwedge[n \in \nn] \alpha^n(d) = 0$. By \textbf{Lemma} \ref{l is nice} (8), $l(\ul d \+ l(a)) \simeq \bigwedge[n \in \nn] \alpha^n(\ul d \+ l(a)) \simeq \bigwedge[n \in \nn] \ul{\alpha^n(d)} \+ l(a)$ and thus, for all $k \in \nn$, $l(\ul d \+ l(a)) \leq \ul{\frac{1}{2^k}} \+ l(a) \leq \ul{\frac{1}{2^k}} \+ a$. Hence $l(\ul d \+ l(a)) \preceq a$ and thus $l(\ul d \+ l(a)) \leq l(a) = \ul{l(d)} \+ l(a)$. However, $\ul d \+ l(a) \geq l(a)$, so $l(\ul d \+ l(a)) \geq l(l(a)) = l(a) = \ul{l(d)} \+ l(a)$. 

\item By \ref{v<=d}, for all $n \in \nn$, $a \leq \ul d \+ \alpha^n (a \+ \ul{1 - d})$, so, by \textbf{Lemma} \ref{l is nice} (8), $a \leq \ul d \+ l(a \+ \ul{1 - d})$. 

Similarly, using \ref{v>=d}, for all $n \in \nn$, $$a \+ \ul{\frac{1}{2^n}} \geq \bigwedge[k = 1][2^{n+1} -1] \ul{1 - \frac{k}{2^{n+1}}} \+ \alpha^{n+1} \left(a \+ \ul{\frac{k}{2^{n+1}}}\right) \blue\succeq \bigwedge[k = 1][2^{n+1} -1] \ul{1 - \frac{k}{2^{n+1}}} \+ l\left(a \+ \ul{\frac{k}{2^{n+1}}}\right).$$ 

\end{enumerate} 
\end{proof} 

\begin{lem} \label{lemme 3} 

\color{blue} Let $\fct[l^\ast][A, A]{a, \bigvee[n \in \nn] \alpha^{-n}(a)}$. For all $a \et b \in A$, $l^\ast(a) \preceq b \Leftrightarrow a \preceq l(b)$, and $l(l(a) \+ l(b)) = l(a) \+ l(b)$. \color{black} 

\end{lem} 

\begin{proof} 

\blue{Let $a \et b \in A$. $$l^\ast(a) \preceq b \Leftrightarrow \forall n \in \nn \; \alpha^{-n}(a) \preceq b \Leftrightarrow \forall n \in \nn \; a \preceq \alpha^n(b) \Leftrightarrow a \preceq l(b)$$} For all $a \et b \in A$, $2(l(a) \+ l(b)) = 2l(a) + 2l(b) = l(a) \+ l(b)$, so $l(a) \+ l(b) \in \lcal(A)$, so\blue{, since $l(a) \+ l(b) \preceq l(a) \+ l(b)$,} $l(l(a) \+ l(b)) \geq l(a) \+ l(b)$. However, $l(b) \+ \bigwedge[n \in \nn] \ul{\frac{1}{2^n}} \preceq l(b)$ and \linebreak $l(b) \+ \bigwedge[n \in \nn] \ul{\frac{1}{2^n}} \in \lcal(A)$, so $l(b) \+ \bigwedge[n \in \nn] \ul{\frac{1}{2^n}}  \leq l(b)$. Hence $$l(l(a) \+ l(b)) \leq l(a) \+ l(b) \+ \bigwedge[n \in \nn] \ul{\frac{1}{2^n}} \leq l(a) \+ l(b).$$ 
\end{proof} 

\begin{proof}[Proof of \textbf{Theorem} \ref{Complétude}] 

According to \textbf{Lemmas} \ref{l is nice}, \ref{lemme 2} and \ref{lemme 3}, since for all \linebreak $a \et b \in A$, $a \leq b \Rightarrow a \preceq b$, $\acal = (A,\, \preceq,\, (\ul d)_{d \in \dcal},\, \+,\, \-,\, 2,\, \frac{\cdot}{2},\, j_\ast,\, j,\, \alpha,\, \beta,\, l, l^\ast)$ is a \linebreak model of $\T_1$. Moreover, for all $a \et b \in A$ such that for all $n \in \nn$ $a \preceq b \+ \ul{\frac{1}{2^n}}$, for all $n \in \nn$, $a \leq b \+ \ul{\frac{1}{2^{n+1}}} \+ \ul{\frac{1}{2^{n+1}}} = b \+ \ul{\frac{1}{2^n}}$, so $a \preceq b$. Thus, $\acal$ is an Archimedean model of $\T_1$. Finally, since $\leq$ is a complete order on $A$, we can consider $\bigwedge\colon \{E \subset A\} \rightarrow A$. For all $E \subset A$ and $b \in A$, $b \preceq \bigwedge E \Leftrightarrow \forall n \in \nn \; b \leq \bigwedge E \+ \ul{\frac{1}{2^n}} \Leftrightarrow \forall n \in \nn \; \forall a \in E \; b \leq a \+ \ul{\frac{1}{2^n}} \Leftrightarrow \forall a \in E \; b \preceq a$, so $\acal$ is a complete model of $\T_1$. Hence, according to \textbf{Corollary} \ref{The corollaire}, there exists a residuated commutative complete lattice $\lcal$ such that $\quotient{\acal}{\simeq}$ is $\mathcal{L}_1$-isomorphic to $\USC$. \color{blue} Let us denote by $i$ the isomorphism. We can show that $i$ preserves the binary lower bounds in the same way as we show it preserves the binary upper bounds. For all $a$ $b \et c \in A$, \begin{align*} i(a \vee b) \leq i(c) &\Leftrightarrow a \vee b \preceq c 
	\\
	&\Leftrightarrow \forall n \in \nn\; a \vee b \leq c \+ \ul{\frac{1}{2^n}} 
	\\ 
	&\Leftrightarrow \forall n \in \nn\; a \leq c \+ \ul{\frac{1}{2^n}} \et b \leq c \+ \ul{\frac{1}{2^n}} 
	\\ 
	&\Leftrightarrow i(a) \leq i(c) \et i(b) \leq i(c) 
	\\ 
	&\Leftrightarrow i(a) \vee i(b) \leq i(c). \end{align*} 
	
Hence the $\quotient{A}{\simeq}$ is $\mathcal{L}$-isomorphic to $\USC$. \color{black} 
\end{proof} 

{\remark Since, for all \crcl $\lcal$, $\ul{\frac{1}{2}} = \frac{\ul 1}{2}$, $\ul{\frac{1}{2}} \preceq \frac{\ul 1}{2}$ is a consequence of $\T$.} 

\subsubsection{General models of $\T$} 

Now that we have proven any \blue{complete model of $\T$ to be, up to a quotient,} isomorphic to an AC-algebra, we prove that general models of $\T$ are, up to a quotient, embeddable into AC-algebras (\bf{Theorem} \ref{big theorem}). 

\begin{nota} \label{preceq} 

$u \preceq v \Leftrightarrow \forall n \in \nn \; u \leq v \+ \ul{\frac{1}{2^n}}$ and $u \simeq v \Leftrightarrow u \preceq v \et v \preceq u$. 

\end{nota} 

To complete models of $\T$, we use the Macneille completion of an order whose construction was first given in \cite[Definition 11.4]{macneillePartiallyOrderedSets1937}. 

\begin{defi}[{\cite{banaschewskiHullensystemeUndErweiterung1956}, cf. also \cite{banaschewskiCategoricalCharacterizationMacNeille1967}, and \cite[Theorem 7041]{daveyIntroductionLatticesOrder2002}}] 

The Macneille completion of an ordered set $(X, \leq)$ is a complete ordered set $(\bar X, \leq)$ together with an non-decreasing function \linebreak $\phi\colon X \rightarrow \bar X$ such that for every $x \in \bar X$, there exist $(y_i)_{i \in I} \in X^I$ and $(z_j)_{j \in J} \in X^J$ such that $\bigvee[i \in I] \phi(y_i) = x = \bigwedge[j \in J] \phi(z_j)$. 

\end{defi} 

The aim of this section is to prove the \bf{Theorem} \ref{big theorem}. What we have to prove is that the $\lcal$-structure can be extended to the Macneille completion and that it still satisfies the axioms of $\T$. 

First, a bit of preliminaries about the Macneille completion. 

\begin{lem} \label{opp Macneille} 

The Macneille completion of the opposite of an ordered set is the opposite of its Macneille completion. 

\end{lem} 

The binary case of the following lemma is proven in \cite[Proposition 3.17]{theunissenMacNeilleCompletionsLattice2007}. 

\begin{lem}  \label{Macneille} 

Let $n \in \nn$, $X , X_1 , \ldots , X_n$ be ordered sets and $f\colon \prod[i = 1][n] X_i \rightarrow X$. \begin{enumerate} 

\item There exists $\bar f\colon \prod[i = 1][n] \ol X_i \rightarrow \ol X$ preserving upper bounds in each coordinate such that \linebreak $\ol f \circ (\phi , \ldots , \phi) = \phi \circ f$ if and only if there exist $g_1 , \ldots ,g_n\colon \prod[i = 1][n] X_i \rightarrow \ol X$ such that, for all \linebreak \adjustbox{width = 396pt}{$1 \leq i \leq n$, for all $x \in \prod[i = 1][n] X_i$ and $y \in X$ $f(x) \leq y \Leftrightarrow \phi(x_i) \leq g_i(x_1 , \ldots , x_{i-1} , y , x_{i+1} , \ldots , x_n)$.} 

Moreover, if it exists, $\bar f$ is unique. 

\item There exists $\bar f\colon \prod[i = 1][n] \ol X_i \rightarrow \ol X$ preserving lower bounds in each coordinate such that \linebreak $\ol f \circ (\phi , \ldots , \phi) = \phi \circ f$ if and only if there exist $g_1 , \ldots ,g_n\colon \prod[i = 1][n] X_i \rightarrow \ol X$ such that, for all \linebreak \adjustbox{width = 396pt}{$1 \leq i \leq n$, for all $x \in \prod[i = 1][n] X_i$ and $y \in X$ $f(x) \geq y \Leftrightarrow \phi(x_i) \geq g_i(x_1 , \ldots , x_{i-1} , y , x_{i+1} , \ldots , x_n)$.} 

Moreover, if it exists, $\bar f$ is unique. 

\end{enumerate} 

\end{lem} 

Let now $A$ be a model of $\T$ and denote by $\phi\colon A \rightarrow \bar A$ the canonical morphism from $A$ to its Macneille completion. 

\begin{defi}[\citup{theunissenMacNeilleCompletionsLattice2007}] 

On $\bar A$, we define $\vee$ and $\wedge$ to be the binary upper and lower bounds, $\ul 0$ to be the bottom element and $\ul 1$ to be the top one. Since $\+$, multiplication by $2$, $j_\ast$, and $\alpha$ admit residual and adjoints, thanks to \textbf{Lemmas} \ref{Macneille} and \ref{opp Macneille}, we can define the lower bound preserving functions $\+$, $\-$, $2 \cdot$, $j_\ast$, $\alpha$ and $\beta$ on $\ol A$ as follows: 

for all $(a_i)_{i \in I} \in A^I$ and $(b_j)_{j \in J} \in A^J$, \begin{enumerate} 

\item $\bigwedge[i \in I] \phi(a_i) \+ \bigwedge[j \in J] \phi(b_j) = \bigwedge[(i,j) \in I \times J] \phi(a_i \+ b_j)$ 
\item $\bigvee[i \in I] \phi(a_i) \- \bigwedge[j \in J] \phi(b_j) = \bigvee[(i,j) \in I \times J] \phi(a_i \- b_j)$ 
\item $2 \bigwedge[i \in I] \phi(a_i) = \bigwedge[i \in I] \phi(2 a_i)$ 
\item $j_\ast\left(\bigwedge[i \in I] \phi(a_i)\right) = \bigwedge[i \in I] \phi(j_\ast(a_i))$ 
\item $\alpha\left(\bigwedge[i \in I] \phi(a_i)\right) = \bigwedge[i \in I] \phi(\alpha(a_i))$ 

\end{enumerate} 

We define $\frac{\cdot}{2}$ and $j$ as the respective adjoints of $2$ and $j_\ast$. 

\end{defi} 

\begin{lem} 

$\phi$ is a morphism of $\mathcal{L}$-structures. 

\end{lem} 

\begin{proof} 

It suffices to prove that, $\phi$ is a morphism for $\frac{\cdot}{2}$ and $j$. Both proofs follow the same steps, so we just prove that $\phi$ preserves $\frac{\cdot}{2}$. 

For all $a \in A$ and $b \in \bar A$, \begin{align*} \phi\left(\frac{a}{2}\right) \leq b &\Leftrightarrow \phi\left(\frac{a}{2}\right) \leq \bigwedge[\ust{c \in A}{\phi(c) \leq b}] \phi(c) \\ 
&\Leftrightarrow \forall c \in A \tq \phi(c) \leq b, \; \phi\left(\frac{a}{2}\right) \leq \phi(c) \\ 
&\Leftrightarrow \forall c \in A \tq \phi(c) \leq b, \; \frac{a}{2} \leq c \\ 
&\Leftrightarrow \forall c \in A \tq \phi(c) \leq b, \; a \leq 2 c \\ 
&\Leftrightarrow \forall c \in A \tq \phi(c) \leq b, \; \phi(a) \leq 2 \phi(c) \\ 
&\Leftrightarrow \phi(a) \leq \bigwedge[\ust{c \in A}{\phi(c) \leq b}] 2 \phi(c) \\ 
&\Leftrightarrow \phi(a) \leq 2 \bigwedge[\ust{c \in A}{\phi(c) \leq b}] \phi(c) \\ 
&\Leftrightarrow \frac{\phi(a)}{2} \leq b. \end{align*}

\end{proof} 

{\remark This also proves that the upper bounds preserving operations induced by $j$ and $\frac{\cdot}{2}$ on $\bar A$ are left adjoint to $j_\ast$ and $2\cdot$.} 

To conclude the proof of \bf{Theorem} \ref{big theorem}, we just have to prove next theorem. 

\begin{thm} \label{completion} 

$\bar A$ is a complete model of $\T$. 

\end{thm} 

An immediate corollary of \bf{Theorem} \ref{completion} is the following one. 

\begin{cor} 

There exists a commutative residuated complete lattice $\lcal$ such that the quotient $\quotient{\bar A}{\simeq}$ is isomorphic $\USC$. 

\end{cor} 

Here are two lemmas we need to prove \bf{Theorem} \ref{completion}. 

\begin{lem}[{\cite[Proposition 3.3]{theunissenMacNeilleCompletionsLattice2007}}]  

The adjunction axioms between $\+$ and $\-$ are satisfied by $\bar A$. 

\end{lem} 

\begin{lem} \label{lemme Macneille} 

Let $f$ and $g$ be two non-decreasing functions from $\ol X^n$ to $\ol X$ such that, for all $x \in X^n$, $f(x) \leq g(x)$. 

If $f$ preserves upper bounds in each coordinate or $g$ preserves lower bounds in each coordinate, then $f \leq g$. 

\end{lem} 

\begin{proof} Let $(x_j)_{i \in J} \in X^n$. 

Let us assume that $f$ preserves all upper bounds. Then $$f\left(\bigvee[j \in J] x_j\right) = \bigvee[(j_i)_{1 \leq i \leq n} \in J^n] f(x_{1,j_1} , \ldots , x_{n,j_n}) \leq \bigvee[1 \leq i \leq n] \bigvee[j_i \in J] g(x_{1,j_1} , \ldots , x_{n,j_n}) \leq g\left(\bigvee[j \in J] x_j\right).$$ 

The second case is symmetric. 

\end{proof} 

\begin{proof}[Proof of \textbf{Theorem} \ref{completion}] 

$\bar A$ is a complete bounded lattice and $\+$ admits a residual, so $\bar A$ is a commutative residuated complete lattice. The adjunction axioms \ref{(2v)/2<=v}, \ref{v<=2(v/2)}, \ref{jj*v<=v} and \ref{v<=j*jv} are satisfied by definition. 

To prove that the other axioms of $\T$ are satisfied by $\bar A$, we notice that the axioms and schemes of axioms from \ref{2(+) <= 2 + 2} to \ref{v>=d} are of the form $f \leq g$, where $f$ and $g$ are non-decreasing functions such that $g$ preserves lower bounds in each coordinate. Thus, \textbf{Lemma} \ref{lemme Macneille} enables to conclude that $\bar A$ satisfies $\T$. 
\end{proof} 

\begin{proof}[Proof of \textbf{Theorem} \ref{big theorem}] 

By \textbf{Theorems} \ref{completion} and \ref{Complétude}, there exists a commutative residuated lattice $\lcal$ and an isomorphism $i\colon \quotient{\bar A}{\simeq} \rightarrow \USC$. Let $\pi$ denote the quotient morphism from $\bar A$ to $\quotient{\bar A}{\simeq}$ and $\phi$ the inclusion morphism of $A $ into $\bar A$. The image of $i \circ \pi \circ \phi$ is endowed with an $\mathcal{L}$-structure, which makes it the quotient of $A$ by $\simeq$. Hence, the quotient of $A$ by $\simeq$ is an $\mathcal{L}$-structure that embeds into $\USC$. 
\end{proof} 

\color{blue} 

\bf{Theorem} \ref{complete archimedean} is an immediate corollary of \bf{Theorem} \ref{completion} 

\color{black} 

\begin{cor} \label{complétude des AC-algèbres} 

For all $\mathcal{L}$-terms $\phi \et \psi$, the following assertions are equivalent: 

\begin{enumerate} 

\item \label{complétude des AC-algèbres 1} For all commutative residuated complete lattice $\lcal$, $\USC \models \phi \leq \psi$. 

\item \label{complétude des AC-algèbres 2} For all $n \in \nn$, $\phi \leq \psi \+ \ul{\frac{1}{2^n}}$ is consequence of $\T$. 

\end{enumerate} 

\end{cor} 

\begin{proof} 

We already proved that \blue{the class of $\USC$, for $\lcal$ a commutative residuated complete lattice} is sound for $\T$. By Archimedeanity of every AC-algebra, we have that \ref{complétude des AC-algèbres 2} implies \ref{complétude des AC-algèbres 1}. 

Let $\phi \et \psi$ be $\mathcal{L}$-terms both having $k$ free variables such that, for all commutative residuated complete lattice $\lcal$, $\USC \models \phi \leq \psi$. Let $A$ be a model of $\T$. According to \bf{Theorem} \ref{big theorem}, there exists a commutative residuated lattice $\lcal$ and a morphism $i \colon A \rightarrow \USC$ such that, for all $a \et b \in A$, $i(a) \leq i(b) \Leftrightarrow \forall n \in \nn \, a \leq b \+ \ul{\frac{1}{2^n}}$. For all $a \in A^k$ and $n \in \nn$, since $i(\phi(a)) \leq i(\psi(a))$, $\phi(a) \leq \psi(a) \+ \ul{\frac{1}{2^n}}$. The class of models of $\T$ being complete for $\T$, for all $n \in \nn$ $\phi \leq \psi \+ \ul{\frac{1}{2^n}}$ is consequence of $\T$. 
\end{proof} 

\section{Cut Admissibility for Affine Continuous Logic} \label{section general Cut Admissibility} 

\addtocounter{subsection}{1} 

The \blue{purpose} of this section is to exhibit a deductive system in a sequent-style calculus for the logic described by AC-algebras. For this, we define a language for structures (on the left side of the turnstile) and formulas, and give the correspondence with the language $\mathcal{L}$. Here can be seen the role of $\alpha$, which is to be a substitue to both $\frac{\cdot}{2}$ and $j$ whenever we need them on the left side of a turnstile. \\ The theorems stated in this section rely on \bf{Theorems} \ref{Completeness theorem annexes} and \ref{Cut Admissibility annexes} proven in \bf{Annexes} \ref{Annexes}. In order to talk about cut elimination, we need to define a language for the structures and formulas of our sequent like calculus and give the correspondence with the language $\mathcal{L}$. 

\begin{figure}[!ht] 

$$\begin{array}{|c|c|c||c|c|} 
\hline 
\multicolumn{3}{|c||}{\text{Positive symbols}}& \multicolumn{2}{c|}{\text{Negative correspondent}}\\ 
\hline 
\text{structures}& \text{formulas} & \text{algebraic correspondent} & \text{formulas} & \text{algebraic notation}\\ 
\hline 
,& \+& \+& \-& \-\\ 
\hline 
\epsilon& \ul 0& \ul 0& \ul 1& \ul 1\\ 
\hline 
& \wedge& \wedge& \vee& \vee\\ 
\hline 
\circ_2& 2& 2& \frac{\cdot}{2}& \frac{\cdot}{2}\\ 
\hline 
\bullet_2& j_\ast& j_\ast& j& j\\ 
\hline 
\circ_\alpha& \alpha& \alpha& \blacksquare_\alpha& 2v \wedge j_\ast(v)\\ 
\hline 
\end{array}$$ 

\caption{Correspondence between structure symbols and $\mathcal{L}$} \label{Correspondence between structure symbols and L} 

\end{figure} 

\begin{nota} \label{notations coupures} 

For all $k \in \nn$, let $k \gamma$ denote $\gamma, \ldots, \gamma$ $k$ times, and for every dyadic number in $[0,1)$, $d = \frac{k}{2^n}$ with $k$ odd, let $\epsilon_d$ denote $k \circ_\alpha^{n-1} \bullet_2 \epsilon$. To be noticed, $\epsilon_{\frac{1}{2}} = \bullet_2 \epsilon$. We also define $\epsilon_1$ as $\circ_2 \bullet_2 \epsilon$. 

The justification of these notations is the following: \textbf{Lemma} \ref{définition dyadiques} proves that, with the definition given in \textbf{Notation} \ref{notation dyadiques}, the interpretation of a dyadic $\frac{k}{2^n}$ in our language is $k \, \alpha^{n-1} \circ j_\ast(0)$. Moreover, axiom \ref{2j*(0) >= 1} ensures that $\ul 1 = 2 j_\ast(\ul 0)$ and, for any algebra $A$ in the language $\mathcal{L} \setminus \{\ul 1\}$ satisfying axioms \ref{com monoid} to \ref{v>=d} where $\ul 1$ is replaced by $2 j_\ast(\ul 0)$, by letting $\ul 1 = 2 j_\ast(\ul 0)$, $A$ is now a model of $\T$. 

\end{nota} 

The system \MGL for modal full Lambek calculus can be applied with one or several modalities. We need six modalities and three structural symbols $\circ_2$, $\bullet_2$ and $\circ_\alpha$. We thus obtain a system \MGL($\circ_2, \bullet_2, \circ_\alpha)$ given by \GL (from the \textbf{Appendix} \ref{Annexes}) --- understood with contexts of the extended language --- and \textit{Figure} \ref{Introduction Rules for Modalities}. In addition to these rules, we add the structural rules given in \textit{Figure} \ref{CFLew} and call the total system \CFLew. 

\begin{figure}[!ht] 

$$\begin{array}{|cccc|} 
\hline 

[\target{L \,2}] \seq{\Gamma[\circ_2 A] \vdash B}{{\Gamma[2 A] \vdash B}} &[\target{R \,2}] \seq{\gamma \vdash B}{\circ_2 \gamma \vdash 2 B} &[\target{L \, \frac{\cdot}{2}}] \seq{{\Gamma[A] \vdash B}}{{\Gamma[\circ_2 \frac{A}{2}] \vdash B}} &[\target{R \, \frac{\cdot}{2}}] \seq{{\circ_2 \gamma \vdash A}}{\Gamma[\gamma] \vdash  \frac{A}{2}} \\ 

[\target{L \,j_\ast}] \seq{\Gamma[\bullet_2 A] \vdash B}{{\Gamma[j_\ast(A)] \vdash B}} &[\target{R \,j_\ast}] \seq{\gamma \vdash B}{\bullet_2 \gamma \vdash j_\ast(B)} &[\target{L \, j}] \seq{{\Gamma[A] \vdash B}}{{\Gamma[\bullet_2 j(A)] \vdash B}} &[\target{R \, j}] \seq{{\bullet_2 \gamma \vdash A}}{\Gamma[\gamma \vdash  j(A)} \\ 

[\target{L \,\alpha}] \seq{\Gamma[\circ_\alpha A] \vdash B}{{\Gamma[\alpha(A)] \vdash B}} &[\target{R \,\alpha}] \seq{\gamma \vdash B}{\circ_\alpha \gamma \vdash \alpha(B)} &[\target{L \, \blacksquare_\alpha}] \seq{{\Gamma[A] \vdash B}}{{\Gamma[\circ_\alpha \blacksquare_\alpha A] \vdash B}} &[\target{R \, \blacksquare_\alpha}] \seq{{\circ_\alpha \gamma \vdash A}}{\Gamma[\gamma] \vdash  \blacksquare_\alpha A} \\ 

\hline 
\end{array}$$ 

\caption{Introduction Rules for Modalities} \label{Introduction Rules for Modalities} 

\end{figure} 

\begin{figure}[!ht] 

\setlength{\arraycolsep}{2 pt} \adjustbox{width=\textwidth}{$\begin{array}{|ccc|} 

\hline 

[\ref{com monoid}a] \seq{{\Gamma[\gamma, \delta] \vdash A}}{{\Gamma[\delta, \gamma] \vdash A}} & [\ref{com monoid}b] \seq{{\Gamma[\gamma, (\delta, \pi)] \vdash A}}{{\Gamma[(\gamma, \delta), \pi] \vdash A}} & [\ref{com monoid}c] \seq[=]{{\Gamma[\epsilon, \gamma] \vdash A}}{{\Gamma[\gamma] \vdash A}} \\ 

[\ref{2 >=}] \seq{\Gamma[\gamma] \vdash A}{\Gamma[\circ_2 \gamma] \vdash A} & [\ref{j* >=}] \seq{\Gamma[\gamma] \vdash A}{\Gamma[\bullet_2 \gamma] \vdash A} & [\ref{bounded lattice}] \seq{\Gamma[\epsilon] \vdash A}{\Gamma[\gamma] \vdash A} \\ 

[\ref{2v <=}] \seq{\Gamma[\circ_2 \gamma] \vdash A, \Gamma[\circ_2 \delta] \vdash A}{{\Gamma[\gamma, \delta] \vdash A}} & \blue{[\link{cont de 2}[(4.a)]] \seq[=]{{\Gamma[\circ_2 (\gamma, \delta)] \vdash A}}{{\Gamma[\circ_2 \gamma, \circ_2 \delta] \vdash A}}} & \blue{[\ref{bounded lattice} \et \ref{2j*(0) >= 1}]} \seq{{}}{\epsilon_1 \vdash A} \\ 

& [(\hyperlink{def j*}{4.b})] \blue{\seq[=]{\Gamma[{\bullet_2 (\circ_2 \gamma, \delta)}] \vdash A}{{\Gamma[\gamma , \bullet_2 \delta] \vdash A}}} & \\ 

\hypertarget{4.d}{[\ref{alpha beta v >= v}a]} \seq{\Gamma[\gamma] \vdash A}{\Gamma[\circ_\alpha \circ_2 \gamma] \vdash A} & [\ref{alpha beta v >= v}b] \seq{\Gamma[\gamma] \vdash A}{\Gamma[\circ_\alpha \bullet_2 \gamma] \vdash A} & [\ref{alpha beta v <=  v}] \seq{\Gamma[\circ_\alpha \circ_2 \gamma] \vdash A, \Gamma[\circ_\alpha \bullet_2 \gamma] \vdash A}{\Gamma[\gamma] \vdash A} \\ 

[\ref{beta alpha v >= v}a] \seq{\Gamma[\gamma] \vdash A}{\Gamma[\circ_2 \circ_\alpha \gamma] \vdash A} & [\ref{beta alpha v >= v}b] \seq{\Gamma[\gamma] \vdash A}{\Gamma[\bullet_2 \circ_\alpha \gamma] \vdash A} & [\ref{beta alpha v <=  v}] \seq{\Gamma[\circ_2 \circ_\alpha \gamma] \vdash A, \blue\Gamma[\bullet_2 \circ_\alpha \gamma] \vdash A}{\Gamma[\gamma] \vdash A} \\ 

[\ref{v<=d}] \seq{\Gamma[\gamma] \vdash A}{{\Gamma[\epsilon_d, {\circ_\alpha}^n (\gamma {,}  \epsilon_{1 - d})] \vdash A}} & [\ref{1/2^n + 1/2^n}]\seq{{\Gamma[\epsilon_{\frac{1}{2^n}}, \epsilon_{\frac{1}{2^n}}] \vdash A}}{\Gamma[\epsilon_{\frac{1}{2^{n-1}}}] \vdash A} & [\hypertarget{Archi}{\ref{archimedean}}] \seq{\forall n \in \nn\; \Gamma[\epsilon_{\frac{1}{2^n}}\blue] \vdash A}{\Gamma[\epsilon ]\vdash A} \\ 

\multicolumn{3}{|c|}{[\ref{v>=d}] \seq{{\Gamma\left[\epsilon_{1 - \frac{1}{2^{n+1}}} ,  {\circ_\alpha}^{n+1}\left(\gamma , \epsilon_{\frac{1}{2^{n+1}}}\right)\right] \vdash A}, \ldots, {\Gamma\left[\epsilon_{1 - \frac{2^{n+1} - 2}{2^{n+1}}} ,  {\circ_\alpha}^{n+1}\left(\gamma , \epsilon_{\frac{2^{n+1} - 2}{2^{n+1}}}\right)\right] \vdash A}}{{\Gamma[\gamma , \epsilon_{\frac{1}{2^n}}] \vdash A}}}\\ 

\hline 

\end{array}$} 

\caption{\CFLew} \label{CFLew} 

\end{figure} 

\begin{lem} 

Algebraic models of \CFLew are the Archimedean models of the theory $\T$. 

\end{lem} 

\begin{proof} 

According to the soundnes part of \textbf{Theorem} \ref{Completeness theorem annexes}, the models of \CFLew are the residuated lattices satisfying, for each previous structural rule $r$, the formula in the language $\L$ where each structure variable has been replaced by a fresh new formula variable, $\seq{{}}{{}}$ by $\Rightarrow$, $\vdash$ by $\geq$, $,$ by $\+$, $\circ_2$ by $2 \cdot$, $\bullet_2$ by $j_\ast$, $\circ_\alpha$ by $\alpha$ and each $\epsilon_d$ by $\ul d$ and the contexts have been removed. 

\begin{exe} \label{Exemples d'équivalence} 

Let us take \link{2 adjunction}[\emph{(3.a)}] as an example of axioms the rules of which are introduction rules, and $(\hyperlink{4.d}{\hyperlink{def alpha, 1}{4.d}})$ as an example of axioms the rules of which are structural rules. 

\begin{enumerate} 

\item[{\link{2 adjunction}[\emph{(3.a)}]} : ] $\seq{\Gamma[\circ_2 A] \vdash B}{\Gamma[2A] \vdash B}$ and $\seq{\Gamma[A] \vdash B}{\Gamma[\circ_2 \frac{A}{2}] \vdash B}$ give $\seq{\seq{A \vdash a}{\circ_2 \frac{A}{2} \vdash A}}{2 \frac{A}{2} \vdash A}$ and $\seq{\seq{A \vdash A}{\circ_2 A \vdash 2A}}{A \vdash 2 \frac{A}{2}}$. 

\item[(\hyperlink{4.d}{\hyperlink{def alpha, 1}{4.d}}) : ] First, $\seq{\Gamma[\gamma] \vdash A}{\Gamma[\circ_\alpha \circ_2 \gamma] \vdash A}$ gives the formula $b \geq a \Rightarrow \alpha(2b) \geq a$, which is equivalent to $\alpha(2b) \geq b$. 

Second, $\seq{\Gamma[\gamma] \vdash A}{\Gamma[\circ_\alpha \bullet_2 \gamma] \vdash A}$ gives the formula $b \geq a \Rightarrow \alpha \circ j_\ast(b) \geq a$, which is equivalent to $\alpha \circ j_\ast(b) \geq b$. Thanks to the fact that $\alpha$ admits a left adjoint, these two formulas are in turn equivalent to the following one: $v \leq \alpha(2v \wedge j_\ast(v))$ \ref{alpha beta v >= v}. 

Third, $\seq{\Gamma[\circ_\alpha \circ_2] \gamma \vdash A, \Gamma[\circ_\alpha \bullet_2 \gamma] \vdash A}{\Gamma[\gamma] \vdash A}$ gives the formula \linebreak $\alpha(2b) \geq a \, \& \, \alpha \circ j_\ast(b) \geq a \Rightarrow b \geq a$, which is equivalent to $b \geq \alpha(2b) \wedge \alpha \circ j_\ast(b)$. Since $\alpha$ admits a left adjoint, it is, in turn, equivalent to $\alpha(2v \wedge j_\ast(v)) \leq v$ \ref{alpha beta v <=  v}. 

\end{enumerate} 

\end{exe} 

\ref{croissance} and \ref{aDjunctions} can be treated the same way as \link{2 adjunction}[\emph{(3.a)}] and \ref{bounded lattice}, \ref{Defining axioms} and \ref{infinitesimals} can be treated the same way as (\hyperlink{4.d}{\hyperlink{def alpha, 1}{4.d}}). The only rule to which the previous methodology can't apply directly is \hyperlink{Archi}{(6.b)} because it is infinitary. However, it is immediate that every model of this rule has equivalently the property $\forall n \in \nn$ $\ul{\frac{1}{2^n}} \geq v \Rightarrow \ul 0 \geq v$, which is actually property \ref{archimedean}. Hence models of \CFLew are exactly the Archimedean models of $\T$. 
\end{proof} 

Since all the added rules are analytic (cf. \bf{Definition} \ref{analytic rules annexes}), according to \textbf{Theorems} \ref{Completeness theorem annexes} and \ref{Cut Admissibility annexes}, the following two theorems are true. 

\begin{thm}[Completeness theorem] \label{Completeness theorem} 

The class of $USC(\lcal)$, for $\lcal$ a complete commutative residuated lattice, is sound and complete for $\CFLew$. 

\end{thm} 

\begin{thm}[Cut Admissibility theorem] \label{Cut Admissibility} $ $ 

In the system \CFLew, for all formulas $a_1, \ldots, a_n \et b$ and $\{, \, , \circ_2, \bullet_2, \circ_\alpha, \epsilon\}$-term $G$ such that there exists a deduction of $G(a_1, \ldots, a_n) \vdash b$ using the cut rule, there exists a deduction of \linebreak $G(a_1, \ldots, a_n) \vdash b$ not using the cut rule. 

\end{thm} 

Finally, we prove a correspondence theorem between structural rules in the language of the system $\GL$ (\cite[Table 1]{galatosResiduatedFramesApplications2012}) and structural rules in the language $\CFLew$. 

\begin{defi}[{\cite[p.8 and Definition 4.3]{ciabattoniAlgebraicProofTheory2012}}] 

Let us denote by $\mathcal{L}_\GL$ the language $\{,\,,\, \cdot,\, \epsilon,\, 1,\, /,\, \setminus,\,\}$, which is the language in which $\GL$ is expressed. 

A rule in the language $\mathcal{L}_\GL$ is said \emph{analytic} when it is an $\mathcal{L}_\GL$-rule of the form \linebreak $\seq{\Gamma[\Upsilon_1] \vdash A , \ldots , \Gamma[\Upsilon_n] \vdash A}{\Gamma[\Upsilon_0] \vdash A}$ --- where the $\Upsilon_i$s are $\{, \, , \, \epsilon\}$-terms --- satisfying: 

\begin{enumerate} 

\item[] \hypertarget{Linearity}{Linearity :} $A$ is a formula variable and the variables of $\Upsilon_0$ are distinct. 

\item[] \hypertarget{Separation}{Separation :} $A$ doesn't appear in $\Upsilon_0$. 

\item[] \hypertarget{Inclusion}{Inclusion :} The variables of the $\Upsilon_i$'s are among the ones of $\Upsilon_0$. 

\end{enumerate} 

\end{defi} 

\begin{thm} 

Let $r = \seq{\Gamma[\Upsilon_1] \vdash A , \ldots , \Gamma[\Upsilon_n] \vdash A}{\Gamma[\Upsilon_0] \vdash A}$ be an analytic $\mathcal{L}_\GL$ rule and let's denote by $k_i$ the total number of times \quad $,$ \quad appears in $\Upsilon_i$, $1 \leq i \leq n$. Let $r_c$ be the structural rule obtained from $r$ by adding $\circ_2$ $2^{k_1 \vee \ldots \vee k_n}$ times in front of $\Upsilon_0'$ by $2^{k_1 \vee \ldots \vee k_n}$. 

$r_c$ is analytic and, for all commutative residuated complete lattice $\lcal$, $\lcal$ satisfies $r$ in the sense of \bf{Definition} \ref{def satisfaction} if and only if $USC(\lcal)$ satisfies $r_c$ in the same sense. 

\end{thm} 

\begin{proof} 

The analycity of $r_c$ is obvious. 

Let $\lcal$ be a commutative residuated complete lattice and, for all $\mathcal{L}_\GL$-term $\Upsilon$, let $\Upsilon^\circ$ denote the $\mathcal{L}_{crl}$-term where $\epsilon$ and $,$ are respectively replaced by $\top$ and $\otimes$, let $\Upsilon^\bullet$ denote the $\mathcal{L}$-term where $\epsilon$ and $,$ are respectively replaced by $\ul 0$ and $\+$. 

Assume that $\lcal$ satisfies $r$, i.e. for all $U \colon \vcal \rightarrow \lcal$ and $V \in \lcal$, if $\Upsilon_1^\circ[U] \leq V$ and $\ldots$ and $\Upsilon_1^\circ[U] \leq V$, then $\Upsilon_0^\circ[U] \leq V$. $\lcal \models \Upsilon_1^\circ \vee \ldots \vee \Upsilon_n^\circ \geq \Upsilon_0^\circ$. Since, for all $0 \leq i \leq n$ $\Upsilon_i$ is lax and colax, from \bf{Theorems} \ref{de lcal à USC(lcal)}, \ref{théorème de comparaison} and the remark that, for all $h \et g \in USC(\lcal)$, $\frac{h \+ g}{2} \leq \frac{h}{2} \+ \frac{g}{2}$, we deduce that, for all $f \colon \vcal \rightarrow USC(\lcal)$, \linebreak $\frac{\Upsilon_1^\bullet[f]}{2^{k_1}} \wedge \ldots \wedge \frac{\Upsilon_n^\bullet[f]}{2^{k_n}} \leq \Upsilon_0^\bullet[f]$. Finally, since for all $h \et g \in USC(\lcal)$ and $k \et k' \in \nn$, $\frac{h \wedge g}{2^{k \vee k'}} \leq \frac{h}{2^k} \wedge \frac{g}{2^{k'}}$, for all $f \colon \vcal \rightarrow \lcal$, $\frac{\Upsilon_0^\bullet[f] \wedge \ldots \wedge \Upsilon_n^\bullet[f]}{2^{k_1 \vee \ldots \vee k_n}} \leq \Upsilon_0^\bullet[f]$, i.e. $\Upsilon_0^\bullet[f] \wedge \ldots \wedge \Upsilon_n^\bullet[f] \leq 2^{k_1 \vee \ldots \vee k_n} \Upsilon_0^\bullet[f]$. 

Hence $USC(\lcal)$ satisfies $r_c$. 

Conversely, assume that $USC(\lcal)$ satisfies $r_c$. Let $U \colon \vcal \rightarrow \lcal$. 

Then, $\Upsilon_0^\bullet[0_U] \wedge \ldots \wedge \Upsilon_n^\bullet[0_U] \leq 2^{k_1 \vee \ldots \vee k_n} \Upsilon_0^\bullet[0_U] = \Upsilon_0^\bullet[0_U]$, so, according to \bf{Theorem} \ref{first half}, \linebreak $\Upsilon_0^\circ[U] \vee \ldots \vee \Upsilon_n^\circ[U] \geq \Upsilon_0^\bullet[U]$. Hence $\lcal$ satisfies $r$. 
\end{proof} 

\section{Intuitionistic continuous logic} \label{section Intuitionistic continuous logic} 

We give a study of the case where the underlying commutative residuated complete lattice $\lcal$ is a locale. In this special case, we actually axiomatize algebras whose quotient by the $\simeq$ relation (\bf{Notation} \ref{preceq}) embeds into some $USC(X)$ for some topological space $X$. To do this, we first need some preliminary results about ordered topological spaces. Second, we provide an axiomatisation for the algebras of the form $USC(\lcal)$ for $\lcal$ a locale, and, third, we show that the theory of the $USC(\lcal)$---for $\lcal$ a locale---is the same as the one of the $USC(X)$---for $X$ a topological space. To achieve this task, we build from any locale $\lcal$ a topological space $X$ and an embedding from $USC(\lcal)$ to $USC(X)$. Finally, we give a sequent-style deductive system admitting the cut rule for these algebras. 

\subsection{Topological Preliminaries} \label{subsection Topological Preliminaries} 

This subsection is independent of the preceding sections. We give some results about compact ordered topological spaces. Let $(X, \leq)$ be an ordered topological space. 

\begin{nota} 

For all $A \subset X$, let $\downarrow A = \{x \in X \, | \, \exists a \in A \tq x \leq a\}$ and \linebreak $\uparrow A = \{x \in X \, | \, \exists a \in A \tq x \geq a\}$. 

For all $f \colon X \rightarrow [0,1]$, let's denote by $\downarrow f$ the smallest non-increasing function greater than $f$, $\uparrow f$ the smallest non-decreasing function greater than $f$, $\bar f$ is the smallest lower semi-continuous function greater than $f$ and $\ovst{f}{\circ}$ the greatest upper semi-continuous function smaller than $f$. 

\end{nota} 

The next lemma gives caracterisations of the notion of intuitionistic space (cf. \bf{Definition} \ref{intuitionistic topological space}), while drawing a parallel between being totally disconnected and being Hausdorff, between open sets and \lsc functions, and between clopen sets and continuous functions. 


\begin{lem} \label{gros lemme} $ $ 

\adjustbox{max width = \textwidth}{
\begin{tikzcd}
	&&& {\forall U \subset X \texte{open} \ol{(\uparrow U)^c} \subset (\uparrow U)^c} \\
	{\forall \texte{clopen} C \subset X \; \uparrow C \texte{is clopen} } && {\forall U \subset X \texte{open} \uparrow U \texte{is open}} && {\forall A \subset X \; \downarrow \bar A \subset \ol{\downarrow A}} \\
	\\
	\\
	\\
	{\forall f \in C^0(X) \; \uparrow f \in C^0(X) } && {\forall f \in LSC(X) \; \uparrow f \in LSC(X) } && {\forall f \colon X \rightarrow [0,1] \; \downarrow \bar f \leq \ol{\downarrow f} } \\
	&&& {\forall f \in LSC(X) \; \ol{1 \- \uparrow f} \leq 1 \- \uparrow f}
	\arrow["\target{(1)}", Rightarrow, from=2-3, to=2-5]
	\arrow["\target{(1 \; bis)}"', Rightarrow, from=2-5, to=1-4]
	\arrow["\target{(1 \; ter)}"', Rightarrow, from=1-4, to=2-3]
	\arrow["\texte{If} X \texte{is totally disconnected}", "\target{(2)}"', curve={height=-30pt}, Rightarrow, from=2-1, to=2-3]
	\arrow["\texte{If} X \texte{is compact Hausdorff}", "\target{(2 \; bis)}"', curve={height=-30pt}, Rightarrow, from=2-3, to=2-1]
	\arrow["\target{(3)}", Rightarrow, 2tail reversed, from=2-5, to=6-5]
	\arrow["\target{(4)}", Rightarrow, 2tail reversed, from=6-3, to=2-3]
	\arrow["\target{(5)}"', Rightarrow, from=6-3, to=6-5]
	\arrow["\target{(5 \; bis)}", Rightarrow, from=6-5, to=7-4]
	\arrow["\target{(5 \; ter)}", Rightarrow, from=7-4, to=6-3]
	\arrow["\texte{if} X \texte{is Hausdorff}", "\target{(6)}"', curve={height=-30pt}, Rightarrow, from=6-1, to=6-3]
	\arrow["\target{(6 \; bis)}"', "\texte{if} X \texte{is compact Hausdorff}", curve={height=-30pt}, Rightarrow, from=6-3, to=6-1]
\end{tikzcd} 
} 
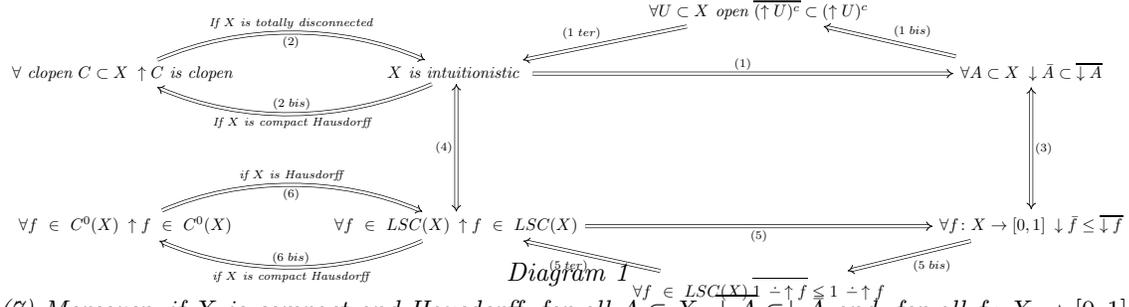
\begin{center} \target{\textit{Diagram 1}} \end{center} 

\link{(7)} Moreover, if $X$ is compact and Hausdorff, for all $A \subset X$, $\ol{\downarrow A} \subset \downarrow \bar A$ and, for all $f\colon X \rightarrow [0,1]$, $\ol{\downarrow f} \leq \downarrow \bar f$. 

\end{lem} 

To prove \textbf{Lemma} \ref{gros lemme}, we need to prove some lemmas first. 

First, from \textbf{Lemma} \ref{Adjunction morphisms}, we can deduce that \link{(1)}, \link{(1 \; bis)}, \link{(1 \; ter)} and \link{(2 \; bis)} are respectively equivalent to \link{(5)}, \link{(5 \; bis)}, \link{(5 \; ter)} and \link{(6 \; bis)}, and we can prove \link{(3)} and \link{(4)}. 

\begin{lem} \label{Adjunction morphisms} $ $ 

\begin{enumerate} 

\item \label{adjunction morphisms 1} For all $A \subset X$, $\ol{\1_A }= \1_{\bar A}$, $\mathring{\1_A} = \1_{\mathring A}$, $\downarrow \1_A = \1_{\downarrow A}$ and $\uparrow \1_\A = \1_{\uparrow A}$, so: 

\begin{enumerate}[label = (\arabic{enumi}.\alph{enumii})] 

\item $A$ is open \ssi $\1_A$ is \lsc, 

\item $A$ is closed \ssi $\1_A$ is \usc, 

\item $A$ is upward closed \ssi $\1_A$ is non-decreasing, 

\item $A$ is downward closed \ssi $\1_A$ is non-increasing. 

\end{enumerate} 

\item \label{adjunction morphisms 2} For all $f \colon X \rightarrow [0,1]$, $f = \bigvee[q \in {[0,1]}] \1_{f^{-1}( (q,1] )} \;\dotm (1 \- q)$. 

\end{enumerate} 

\end{lem} 

\begin{proof} 

Point \ref{adjunction morphisms 1} is obvious. 

Let $f\colon X \rightarrow [0,1]$. 

For all $q \in [0,1]$ and $x \in X$, if $\1_{f^{-1}( (q,1] )}(x) = 1$, then $x \in f^{-1}( (q,1] )$, so \linebreak $f(x) > q = \1_{f^{-1}( (q,1] )} \dotm (1 \- q)$, so $f \geq \1_{f^{-1}( (q,1] )} \;\dotm (1 \- q)$. For all $x \in X$ and $q < f(x)$, $\1_{f^{-1}( (q,1] )}(x) \dotm (1 \- f(x)) = q$, so $f(x) = \bigvee[q \in {[0,1]}] \1_{f^{-1}( (q,1] )} \;\dotm (1 \- q)$. 
\end{proof}

Second, we prove \link{(1)}, \link{(1 \; bis)}, \link{(1 \; ter)}, \link{(2)}, \link{(2 \; bis)} and \link{(6)} thanks to \textbf{Lemma} \ref{Outils}. 

\begin{lem} \label{Outils} $ $ 

\begin{enumerate} 

\item \label{outils 1} For all $A \et B \subset X$, $\uparrow A \subset B^c \Leftrightarrow A \subset (\downarrow B)^c$. 

\item \label{outils 2} If $X$ is compact Hausdorff, for all $F \subset X$ closed, $\uparrow F$ is closed. 

\item \label{outils 3} if $X$ is compact and Hausdorff, $LSC(X) = \{\bigvee A, A \subset C^0(X)\}$. 

\end{enumerate} 

\end{lem} 

\begin{proof} \begin{enumerate} 

\item Assume that $\uparrow A \subset B^c$ and let $x \in A$. For all $y \in B$, if $x \leq y$, then $y \in \uparrow A$ and so $y \in B^c$. Hence $x \in (\downarrow B)^c$. 

Assume that $A \subset (\downarrow B)^c$ and let $x \in \uparrow A$. There exists $y \in A$ \tq $x \geq y$. If $x \in B$, then $y \in \downarrow A$ and then $y \not\!\!\in A$. Hence $x \in B^c$. 

\item Let $F \subset X$ closed. Denote by $\pi_1$ the projection on the first coordinate from $X^2$ to $X$. 

$\uparrow F = \pi_1(\geq \cap (X \times F))$. $\geq$ is closed in $X^2$ which is compact so $\geq$ is itself compact, and thus $X \times F$ is also compact. Since $\pi_1$ is continuous, $\uparrow F$ is compact and thus closed in $X$. 

\item According to \bf{Lemma} \ref{cor Urysohn croissant}, if we denote by $X^{op}$ the opposite order of $X$, \linebreak $USC(X^{op}) = \{\bigwedge A, A \subset C^0(X)\}$, which is actually $LSC(X) = \{\bigvee A, A \subset C^0(X)\}$. 

\end{enumerate} 
\end{proof} 

\begin{proof}[Proof of \textbf{Lemma} \ref{gros lemme}] We only have to prove \link{(1)}, \link{(1 \; bis)}, \link{(1 \; ter)}, \link{(2)}, \link{(2 \; bis)}, \link{(6)} and  \link{(7)} thanks to \textbf{Lemma} \ref{Adjunction morphisms}. As an example of application of \bf{Lemma} \ref{Adjunction morphisms}, we also prove \link{(3)}. 

\link{(1)}: Let $A \subset X$. Let $U = (\ol{\downarrow A})^c$. 

$\downarrow A \subset \ol{\downarrow A}$, so $U = (\ol{\downarrow A})^c \subset (\downarrow A)^c$, which gives, according to \textbf{Lemma} \ref{Outils} \ref{outils 1}, $\uparrow U \subset A^c$. By hypothesis, $\uparrow U \subset \round{(A^c)} = (\bar A)^c$. Thanks to \textbf{Lemma} \ref{Outils} \ref{outils 1}, $U \subset (\downarrow \bar A)^c$, and so $\downarrow \bar A \subset \ol{\downarrow A}$. 

\link{(1 \; bis)}: Let $U \subset X$ be open. 

$(\uparrow U)^c \subset (\uparrow U)^c$, so, according to \textbf{Lemma} \ref{Outils} \ref{outils 1} applied to the opposite order, $\downarrow((\uparrow U)^c) \subset U^c$, and thus $\downarrow\left(\ol{(\uparrow U)^c}\right) \subset \ol{\downarrow((\uparrow U)^c)} \subset U^c$, by hypothesis. Hence $U \subset \left(\downarrow\ol{(\uparrow U)^c}\right)^c$. Again by \textbf{Lemma} \ref{Outils} \ref{outils 1}, $\uparrow U \subset \left(\ol{(\uparrow U)^c}\right)^c$, i.e. $\ol{(\uparrow U)^c} \subset (\uparrow U)^c$.

\link{(1 \; ter)}: Let $U \subset X$ be open. $\ol{(\uparrow U)^c} \subset (\uparrow U)^c$, i.e. $(\uparrow U)^c$ is closed, i.e. $\uparrow U$ is open. 

\link{(2)}: Assume that $X$ is totally disconnected and that, for all clopen $C$, $\uparrow C$ is clopen. Then, for all $U \subset X$ open, $\uparrow U = \bigcup[C \subset U \texte{clopen}] \uparrow C$, which is open. 

\link{(2 \; bis)}: Assume that $X$ is compact, Hausdorff and intuitionistic. For all clopen $C$ of $X$, $\uparrow C$ is open, and, thanks to \textbf{Lemma} \ref{Outils} \ref{outils 2}, $\uparrow C$ is also closed. 

\link{(3)}: If, for all $f \colon X \rightarrow [0,1]$, $\downarrow \bar f \leq \ol{\downarrow f}$, then, for all $A \subset X$, $\1_{\downarrow \bar A} = \downarrow \ol{\1_A} \leq \ol{\downarrow \1_A} = \1_{\ol{\downarrow A}}$, so $\downarrow \bar A \subset \1_{\ol{\downarrow A}}$. 

Assume now that, for all $A \subset X$, $\downarrow \bar A \subset \1_{\ol{\downarrow A}}$ and let $f \colon X \rightarrow [0,1]$. Since $\downarrow$ and $\bar{\white{A}}$ preserve upper bounds, and commute with subtraction of constants, 

\begin{align*} 
\downarrow \bar f &= \bigvee[q \in {[0,1]}] \downarrow \ol{\1_{f^{-1}( (q,1] ) \dotm (1 \- q)}} \\ 
&= \bigvee[q \in {[0,1]}] \downarrow \ol{\1_{f^{-1}( (q,1] )}} \;\dotm (1 \- q) \\ 
&=  \bigvee[q \in {[0,1]}] \1_{\downarrow \ol{f^{-1}( (q,1] )}} \;\dotm (1 \- q) \\ 
&\leq \bigvee[q \in {[0,1]}] \1_{\ol{\downarrow f^{-1}( (q,1] )}} \;\dotm (1 \- q) \\ 
&=  \bigvee[q \in {[0,1]}] \ol{\downarrow\1_{f^{-1}( (q,1] )}} \;\dotm (1 \- q) \\ 
&=  \bigvee[q \in {[0,1]}] \ol{\downarrow\1_{f^{-1}( (q,1] )} \;\dotm (1 \- q)} \\ 
&=  \ol{\downarrow f} 
\end{align*} 

\link{(6)}: According to \textbf{Lemma} \ref{Outils} \ref{outils 3}, for all $f \in LSC(X)$, $f = \bigvee[\ust{g \in C^0(X)}{g \leq f}] g$, so $\uparrow f = \bigvee[\ust{g \in C^0(X)}{g \leq f}] \uparrow g$, which is \lsc by \bf{Lemma} \ref{Outils} \ref{outils 3} and because, for all $g \in C^0(X)$, \linebreak $\uparrow g \in C^0(X) \subset LSC(X)$. 

\link{(7)}: The two conclusions of \link{(7)} are equivalent, thanks to \textbf{Lemma} \ref{Adjunction morphisms}. Let us prove the first one. 

Assume $X$ is compact and Hausdorff. Let $A \subset X$. 

$\bar A$ is closed, so, according to \textbf{Lemma} \ref{Outils} \ref{outils 2}, $\uparrow \bar A$ is also closed. Since $\uparrow A \subset \uparrow \bar A$, $\ol{\uparrow A} \subset \uparrow \bar A$. 
\end{proof} 

\begin{cor} \label{the smallest non-decreasing usc function greater} 

For all compact Hausdorff ordered topological space $X$ such that the upward closure of an open is open, and $f \in C^0(X)$, $\uparrow f$ is the smallest non-decreasing \usc function greater than $f$. 

\end{cor} 

\begin{proof} 

Let $X$ be an compact Hausdorff ordered topological space $X$ such that the upward closure of an open is open and $f \in C^0(X)$. According to \textbf{Lemma} \ref{gros lemme} 4 and 6 bis, $\uparrow f$ is continuous and thus upper semi-continuous. Moreover, for all non-decreasing upper semi-continuous function $g \geq f$, since $g$ is non-decreasing, $g \geq \uparrow f$, which proves the corollary. 
\end{proof} 

\subsection{Algebraic axiomatisation of IC-algebras} \label{subsection algebraic axiomatisation of IC-algebras} 

We are here aiming at an axiomatisation of the class of $\USC$ for $\lcal$ a locale. 

\begin{nota} 

We remind the reader of \bf{Notation} \ref{topological space} : for all ordered topological space $(X, \leq)$, we denote by $C_\nearrow^0(X)$ the set of continuous increasing functions from $X$ to $[0,1]$, by $C^0(X)$ the set of continuous functions from $X$ to $[0,1]$ and by $USC_\nearrow(X)$ the set of non-decreasing upper semi-continuous functions from $X$ to $[0,1]$. 

\end{nota} 

\begin{defi} 

We denote by $IC$ the class whose elements are the $\USC$ for $\lcal$ a locale and call these algebras \emph{Ituitionnistic Continuous Algebras}. 

\end{defi} 

\subsubsection{Review of the algebra $USC(\lcal)$} 

In the case where $\lcal$ is a locale, $\USC$ has nicer properties (\bf{Theorem} \ref{int formule formules}). 

Let $\lcal$ be a locale. $\lcal$ is a commutative residuated complete lattice, so $USC(\lcal)$ has an $\mathcal{L}$-structure. The lower bounds of two upper semi-continuous functions is simply taken pointwise and the sum and difference of two upper semi-continuous functions $f$ and $g$ from $\lcal$ to $[0,1]$ are $$\fct[f \+ g][{{{[0,1]}}}, \lcal]{q, \underset{p < q}{\bigvee} f(\< p) \wedge g(q \- p)} \et \fct[f \- g][{{{[0,1]}}}, \lcal]{q, \underset{p < q}{\bigvee} \, \underset{r \in [0,1]}{\bigwedge} g(r \- p) \nrightarrow f(\< r)}.$$ 

The supremum of $f \et g \in \USC$ is given by, for all $q \in [0,1]$, $(f \vee g)(q) = \bigvee[p < q] f(p) \wedge g(p)$. As in \bf{Lemma} \ref{calculatoire}, for all $f \in \USC$, $2 f = f \circ \frac{\cdot}{2}$ and $j_\ast(f) = f \circ j$. $f \mapsto \frac{f}{2}$ and $f \mapsto j(f)$ are respectively left adjoint to $j_\ast$ and $2 \cdot$. 

\begin{thm} \label{int formule formules} 

Let $\phi[v]$ and $\psi[v]$ be terms in the language $\Lusc$. 

Then, if $[0,1] \models \phi \leq \psi$, $\USC \models \phi \leq \psi$. 

\end{thm} 

\begin{proof} 

Assume $[0,1] \models \phi \leq \psi$ and let $f\colon \vcal \rightarrow \USC$. 

\it{Remark} \ref{Rq tense = inf} implies that for all $p \colon \vcal \rightarrow [0,1]$, $a_\phi(p_v) = \phi[p] \leq \psi[p] = a_\psi(p_v)$, so $a_\phi \leq a_\psi$. According to \it{Remark} \ref{Rq tense = inf} and \bf{Lemma} \ref{monotonie de l'action}, $$\phi[f] = a_\phi(f_v) \leq a_\psi(f_v) = \psi[f].$$ 
\end{proof} 

\begin{cor} \label{2+} 

Let $\lcal$ be a commutative residuated complete lattice. The following assertions are equivalent : 

\begin{enumerate} 

\item For all $f \in \USC$, $2 f = f \+ f$. 

\item $\lcal$ is a locale. 

\end{enumerate} 

\end{cor} 

\begin{proof} 

We remind the reader that the order on $\USC$ is the reverse pointwise order. 

If $\USC \models 2v \geq v \+ v$, then, according to \bf{Theorem} \ref{first half}, $\lcal \models v \leq v \otimes v$ and, if $\lcal \models v \leq v \otimes v$, then, according to \bf{Theorems} \ref{de lcal à USC(lcal)} and \ref{théorème de comparaison}, $\USC \models v \geq \frac{v \+ v}{2}$, which is equivalent to $\USC \models 2v \geq v \+ v$. 

Moreover, in a \crcl, $v \leq v \otimes v$ and $u \wedge v = u \otimes v$ are equivalent. Indeed, the later obviously implies the former, and, if for all $V \in \lcal$, $V \leq V \otimes V$, then for all $U \et V \in \lcal$, $U \wedge V \leq (U \wedge V) \otimes (U \wedge V) \leq U \otimes V$, so, since $U \otimes V \leq U \otimes \top = U$ and $U \otimes V \leq \top \otimes V = V$, $U \otimes V = U \wedge V$. 

\end{proof} 

\begin{remark} 

We end this subsection by noticing that, in the case where $\lcal$ is the topology $\Tcal$ of a topological space $X$, this structure is the one induced by the $\mathcal{L}$-structure of $USC(X)$ through the bijection $USC(\Tcal) \simeq USC(X)$. 

\end{remark} 

\subsubsection{The caracterisation} 

Here comes the list of axioms that are satisfied by every $IC$-algebra, by \bf{Theorem} \ref{int formule formules}. 

\begin{enumerate} 

\item \textbf{$\blue{(\vee,\, \wedge, \+, \-, \ul 0, \ul 1)}$ is a bounded commutative residuated lattice structure, as in \ref{bounded commutative residuated lattice}.} 

\item \label{int croissance} \textbf{$2$, $\frac{\cdot}{2}$, $j_\ast$, $j$ and $\alpha$ don't decrease.} 

\item \textbf{The adjunctions, as in \ref{aDjunctions}.} 

\item \textbf{Defining axioms:} 

\begin{enumerate*}[label=] 

\item \color{blue} \hypertarget{int def 2}{} \begin{enumerate*}[label=(4.f.2') , itemjoin={\quad}] \item[\ref{2v <=}] \; and \item \label{int 2v >=} $v \+ v \leq 2 v$, \end{enumerate*} \color{black} 

\item \blue{\ref{def j* <=} and \ref{def j* >=}}, 

\item \hypertarget{>=}{} \ref{2 >=} \; and \; \ref{j* >=}, 

\item \hypertarget{int def alpha, 1}{} \ref{alpha beta v >= v} \; and \; \ref{alpha beta v <=  v}, 

\item \hypertarget{int def alpha, 2}{} \ref{beta alpha v <=  v} \; and \; \ref{beta alpha v >= v}, 

\item \blue{\ref{2j*(0) >= 1},} 

\item \red{\ref{1/2^n + 1/2^n} LÀ (et les règles de déduction)} 

\end{enumerate*} 

\color{blue} \item \color{black} \textbf{And the algebra of values is $[0,1]$ to infinitesimals, as in \ref{Values}.} 

\end{enumerate} 

We denote by $\T_\int$ the previous theory. 

\begin{thm} \label{Equivalent theories} 

$\T_\int$ and $\T \cup \{2 v \geq v \+ v\}$ are equivalent theories. 

\end{thm} 

\begin{proof} 

$\T_\int = (T \setminus \{2(u \+ v) \leq 2u \+ 2v, 2u \+ 2v \leq 2(u \+ v)\}) \cup \{2v = v \+ v\}$, so $\T \cup \{2v \geq v \+ v\}$ implies $\T_\int$. \\ Moreover, $(T \setminus \{2(u \+ v) \leq 2u \+ 2v, 2u \+ 2v \leq 2(u \+ v)\}) \vdash 2v = v \+ v \Rightarrow 2(u \+ v) = 2u \+ 2v$, so $\T_\int \vdash 2(u \+ v) = 2u \+ 2v$, so $\T_\int$ implies $\T$. 
\end{proof} 

\begin{cor} \label{int big corollary} 

For all model $A$ of $\T_\int$, there exists a locale $\lcal$ such that the quotient of the Macneille completion of $A$ by $\simeq$ is isomorphic to $\USC$. 

For all model $A$ of $\T_\int$, there exists a locale $\lcal$ such that the quotient of $A$ by $\simeq$ embeds into $\USC$. 

\end{cor} 

\begin{proof} 

Let $A$ be a model of $\T_\int$. $A$ is a model of $\T$, so, according to \bf{Theorem} \ref{big theorem}, there exists a commutative residuated complete lattice $\lcal$ such that the quotient of the Macneille completion of $A$ by $\simeq$ is isomorphic to $\USC$. Since, in $A$, $\otimes = \vee$, in its Macneille completion, this identity is still true. Hence it is also true in the quotient and thus in $\USC$. \bf{Theorem} \ref{first half} ensures that, for all $U \et V \in \lcal$, $U \otimes V =  U \wedge V$. Hence $\lcal$ is a locale. 

This construction gives us a function $i$ from $A$ to $\USC$ such that, for all $a \et b \in A$, $i(a) = i(b) \Leftrightarrow a \simeq b$. Hence, the image of $i$ is the quotient of $A$ by $\simeq$, which proves that the quotient of $A$ by $\simeq$ embeds into $\USC$. 
\end{proof} 

\begin{cor} \label{complétude des IC-algèbres} 

For all $\mathcal{L}$-terms $\phi \et \psi$, the following assertions are equivalent: 

\begin{enumerate} 

\item \label{complétude des IC-algèbres 1} For all locale $\lcal$, $\USC \models \phi \leq \psi$. 

\item \label{complétude des IC-algèbres 2} For all $n \in \nn$, $\phi \leq \psi \+ \ul{\frac{1}{2^n}}$ is consequence of $\T_\int$. 

\end{enumerate} 

\end{cor} 

\begin{proof} 

For all locale $\lcal$, $\USC$ is a model of $\T \cup \{2v \geq v \+ v\}$, which is equivalent to $\T_\int$, so $\USC$ is a model of $\T_\int$. By Archimedeanity of every IC-algebra, we have that \ref{complétude des IC-algèbres 2} implies \ref{complétude des IC-algèbres 1}. 

Let $\phi \et \psi$ be $\mathcal{L}$-terms both having $k$ free variables such that, for all locale $\lcal$, $\USC \models \phi \leq \psi$. Let $n \in \nn$. Let $A$ be a model of $\T$. According to \bf{Corollary} \ref{int big corollary}, there exists a commutative residuated lattice $\lcal$ and a morphism $i \colon A \rightarrow \USC$ such that, for all $a \et b \in A$, \linebreak $i(a) \leq i(b) \Leftrightarrow \forall n \in \nn \, a \leq b \+ \ul{\frac{1}{2^n}}$. For all $a \in A^k$ and $n \in \nn$, since $i(\phi(a)) \leq i(\psi(a))$, $\phi(a) \leq \psi(a) \+ \ul{\frac{1}{2^n}}$. 

The class of models of $\T$ being complete for $\T$, $\phi \leq \psi \+ \ul{\frac{1}{2^n}}$ is consequence of $\T$. 
\end{proof} 

\begin{cor} 

Let $\mathcal{L}' = \{\+, \-, \wedge, \vee, \frac{\cdot}{2}, \ul 1, \ul 0\}$ and write $2v$ for $v \+ v$, $j(v)$ for $\left(v \- \frac{\ul 1}{2}\right) \+ \left(v \- \frac{\ul 1}{2}\right)$, $j_\ast(v)$ for $v \+ \frac{\ul 1}{2}$ and $\alpha(v)$ for $\frac{v}{2} \vee \left(\left(v \- \frac{\ul 1}{2}\right) \+ \left(v \- \frac{\ul 1}{2}\right)\right)$. Let $\T_\int'$ be the following theory in the language $\Lcal'$. 

\begin{enumerate} 

\item \textbf{$\blue{(\vee,\, \wedge, \+, \-, \ul 0, \ul 1)}$ is a bounded commutative residuated lattice structure, as in \ref{bounded commutative residuated lattice}.} 

\item $\frac{\cdot}{2}$ \textbf{is non-decreasing.} 

\item \textbf{The adjunctions:} \ref{aDjunctions} and \hypertarget{2(v/2)<=v int}{(3.a.2')} $\frac{v}{2} \+ \frac{v}{2} \leq v$ 

\item \textbf{Defining axioms:} 

\begin{enumerate*}[label=] 




\item \hypertarget{int def alpha, 1}{} \ref{alpha beta v >= v}, 

\item \hypertarget{int def alpha, 2}{} \ref{beta alpha v <=  v} 



\end{enumerate*} 

\color{blue} \item \color{black} \textbf{And the algebra of values is $[0,1]$ to infinitesimals, as in \ref{Values}.} 

\end{enumerate} 

For all locale $\lcal$, $\USC$ is a model of $\T_\int'$. 

$a \preceq b \Leftrightarrow \forall n \in \nn \; a \leq b \+ \ul{\frac{1}{2^n}}$ defines a preorder on every model of $\T_\int'$. 

For all model $A$ of $\T_\int'$, there exists a locale $\lcal$ such that the quotient of the Macneille completion of $A$ by the equivalence relation induced by the preorder $\preceq$ is isomorphic to $\USC$. 

For all model $A$ of $\T_\int'$, there exists a locale $\lcal$ such that the quotient of $A$ by the equivalence relation induced by the preorder $\preceq$ embeds into $\USC$. 

For all $\mathcal{L}$-terms $\phi \et \psi$, if, for all locale $\lcal$, $\USC \models \phi \leq \psi$, then for all $n \in \nn$, $\phi \leq \psi \+ \ul{\frac{1}{2^n}}$ is consequence of $\T_\int'$. 

\end{cor} 

\begin{proof} 

First, for all locale $\lcal$, by \bf{Theorem} \ref{int formule formules}, $\USC$ is a model of $\T_\int'$. 

Second, let $A$ be a model of $\T_\int'$. 

\ul{\bf{Claim :}} By defining $2$, $j$, $j_\ast$ and $\alpha$ by $2v = v \+ v$, $j(v) = \left(v \- \frac{\ul 1}{2}\right) \+ \left(v \- \frac{\ul 1}{2}\right)$, $j_\ast(v) = v \+ \frac{\ul 1}{2}$ and $\alpha(v) = \frac{v}{2} \vee \left(\left(v \- \frac{\ul 1}{2}\right) \+ \left(v \- \frac{\ul 1}{2}\right)\right)$, $A$ is a model of $\T_\int$. 

Since every model of $\T_\int'$ is a model of $\T_\int$, according to \textbf{Corollary} \ref{int big corollary}: 

\begin{enumerate}[label=] 

\item For all model $A$ of $\T_\int'$, there exists a locale $\lcal$ such that the quotient of the Macneille completion of $A$ by the equivalence relation induced by the preorder $\preceq$ is isomorphic to $\USC$. 

\item For all model $A$ of $\T_\int'$, there exists a locale $\lcal$ such that the quotient of $A$ by the equivalence relation induced by the preorder $\preceq$ embeds into $\USC$. 
\end{enumerate} 

Since $\T_\int$ and $\T_\int'$ have same Archimedean models, from \bf{Corollary} \ref{complétude des IC-algèbres}, we deduce that, for all $\mathcal{L}$-terms $\phi \et \psi$, if, for all locale $\lcal$, $\USC \models \phi \leq \psi$, then for all $n \in \nn$, $\phi \leq \psi \+ \ul{\frac{1}{2^n}}$ is consequence of $\T_\int'$. 

\ul{\bf{Proof of the claim :}} The axioms of $\T_\int$ that aren't in $\T_\int'$ are the following ones: 

\begin{enumerate} 
	\item $2$, $j$, $j_\ast$ and $\alpha$ are non-decreasing. \ref{croissance} 
	\item \link{cont de 2}[(4.a)], \link{>=}[(4.c)] and \link{fin def 2 et j*}[(4.f)]. 
	\item \link{def j*}[(4.b)], \ref{alpha beta v <=  v}, \ref{beta alpha v >= v} and \ref{1/2^n + 1/2^n}. 
\end{enumerate} 

The first two points are obviously satisfied by $A$. Let us proove the third one. Let $a \et b \in A$. 

\link{def j*}[(4.b)] : Notice first that, for all $x \et y \in A$, since $2\left(\frac{x}{2} \+ \frac{y}{2}\right) = x \+ y$, $\frac{x \+ y}{2} \leq \frac{x}{2} \+ \frac{y}{2}$. Hence, $j_\ast(2a \+ b) = \frac{2a \+ b}{2} \+ \frac{\ul 1}{2} \leq \frac{2a}{2} \+ \frac{b}{2} \+ \frac{\ul 1}{2} \leq a \+ \frac{b}{2} \+ \frac{\ul 1}{2} = a \+ j_\ast(b)$. 

Then, notice that, thanks to \ref{jj*v<=v} and \ref{v<=j*jv}, $\forall x \et y \in A\; j_\ast(2x \+ y) \geq x \+ j_\ast(y)$ is equivalent to $\forall x \et y \in A\; 2x \+ j(y) \geq j(x \+ y)$. We prove that $j(a \+ b) \leq 2a \+ j(b)$: $$j(a \+ b) = 2\left((a \+ b) \- \frac{\ul 1}{2}\right) \leq 2\left(\left(b \- \frac{\ul 1}{2}\right) \+ a \right) = j(b) \+ 2a.$$ 

\ref{alpha beta v <=  v} : $\alpha(2a \wedge j_\ast(a)) = \frac{2a \wedge j_\ast(a)}{2} \vee j(2a \wedge j_\ast(a)) \leq \frac{2a}{2} \vee j \circ j_\ast(a) \leq a$. 

\ref{beta alpha v >= v} : $2\alpha(a) \geq 2 \frac{a}{2} \geq a$ and $j_\ast \circ \alpha(a) \leq j_\ast \circ j(a) \geq a$. 

\ref{1/2^n + 1/2^n} : First, let us notice that $j_\ast(\ul 0) = \frac{\ul 1}{2}$ and, for all $x \in A$ such that $x \leq \frac{\ul 1}{2}$, $\alpha(x) = \frac{x}{2} \vee 2\left(x \- \frac{\ul 1}{2}\right) = \frac{x}{2}$, so, for all $n \in \nn$, $\alpha^n \circ j_\ast(\ul 0) = \frac{\ul 1}{2^n}$. Hence, for all $n \in \nn*$, using \hyperlink{2(v/2)<=v int}{(3.a.2')}, $\alpha^n \circ j_\ast(\ul 0) \+ \alpha^n \circ j_\ast(\ul 0) = 2 \frac{\ul 1}{2^n} \leq \frac{\ul 1}{2^{n-1}} = \alpha^{n-1} \circ j_\ast(\ul 0)$. 
\end{proof} 

\subsection{Equivalence between IC-algebras and MC-algebras} 

The aim of this section is to prove \textbf{Theorem} \ref{gros thm} and its \bf{Corollary} \ref{corollaire Macneille}. 

\begin{nota} 

First of all, let's denote by $\LMarco$ the language of \cite[pp. 1404-1405]{abbadiniDualCompactOrdered2019}, in which $\oplus$ is replace by $\+$ and to which is added the symbol $\-$. Let us also denote by $\Marco$ the theory of \cite[pp. 1404-1405]{abbadiniDualCompactOrdered2019}, to which are added the following two axioms: $$a \leq (a \- b) \+ b \et (a \+ b) \- b \leq a.$$ 

\end{nota} 

Following \cite{marquesCategoricalLogicPerspective2023}, we give the following definitions. 

\begin{defi} \label{intuitionistic topological space} 

Let $(X, \leq)$ be an ordered topological space. We recall that $X$ is \emph{Hausdorff} if $\leq$ is closed (cf. \bf{Definition} \ref{Hausdorff}). 

$X$ is \emph{intuitionistic} if, for all open subset $U$ of $X$, $\uparrow U$ is open. 

\end{defi} 

\begin{defi} 

MC-algebras are defined in \cite[pp. 1404-1405]{abbadiniDualCompactOrdered2019} and we define $d$ on any MC-algebra $A$, as in \cite[Definition 6.1]{abbadiniDualCompactOrdered2019}, by $$\forall f \et g \in A\quad d(f,g) = \left(\bigwedge \{q \in [0,1] \, | \, f \leq g \+ q\}\right) \vee \left(\bigwedge \{q \in [0,1] \, | \, g \leq f \+ q\}\right).$$ 

Let us recall the definiton of Archimedeanity and Cauchy-completeness for MC-algebras (\cite[Definitions 6.2 and 6.4]{abbadiniDualCompactOrdered2019}). 

Let $A \in MC$. We say that $A$ is Archimedean if, for all $f \et g \in A$, $d(f,g) = 0 \Rightarrow f = g$. 

Let $(a_n)_{n \in \nn}$ be a sequence in $A$, and let $a \in A$. The sequence $(a_n)_{n \in \nn}$ is called a \emph{Cauchy sequence} if, for every $\epsilon > 0$, there exists $k \in \nn$ such that for all $n, m \geq k$,$d(a_n, a_m) < \epsilon$. The sequence $(a_n)_{n \in \nn}$ is said to \emph{converge} to $a \in A$, or that $a$ is a \emph{limit} of $(a_n)_{n \in \nn}$, if for every $\epsilon > 0$, there exists $n_0 \in \nn$ such that for all $m \geq n_0$, $d(a_m, a) < \epsilon$. The sequence $(a_n)_{n \in \nn}$ is said to \emph{converge} if there exists $b \in A$ such that $(a_n)_{n \in \nn}$ converges to $b$. The set $A$ is said to be \emph{Cauchy-complete} if every Cauchy sequence in $A$ converges. 

An MC-algebra $A$ is said \emph{intuitionistic} when $\dotp$ admits a residual. 

\end{defi} 

Intuitionistic ordered spaces and MC-algebras are linked by the following duality theorem. 

\begin{thm}[{\cite[Theorem 8.5.]{abbadiniDualCompactOrdered2019}}] \label{isomorphie MC-algèbres} 

The map that assigns to every ordered topological space $X$ the algebra $\Cn{X}$ gives rise to an equivalence of categories between the opposite of the category of compact Hausdorff ordered topological spaces with monotone continuous maps and the category of Cauchy-complete Archimedean MC-algebras. 

We denote by $Sp(A)$ the compact Hausdorff ordered topological space thus associated to an MC-algebra $A$. 

\end{thm} 

\begin{thm}[{\cite[Proposition 1.3.23.]{marquesCategoricalLogicPerspective2023}}] \label{équivalence} 

The category of intuitionistic Archimedean Cauchy-complete MC-algebras is equivalent to the opposite of the category of compact Hausdorff intuitionistic ordered spaces. 

\end{thm} 

\begin{lem} 

Let $A$ be an MC-algebra. 

$A$ is Archimedean if and only if, for all $f \et g \in A$ and $\forall n \in \nn\; f \leq g \+ \frac{1}{2^n} \Rightarrow f \leq g$. 

\end{lem} 

\begin{proof} 

Let $f \et g \in A$. 

Assume that $A$ is Archimedean and, for all $n \in \nn$, $f \leq g \+ \frac{1}{2^n}$. Then, \linebreak $\bigwedge \{q \in [0,1] \, | \, f \leq f \wedge g \+ q\} = \bigwedge \{q \in [0,1] \, | \, f \wedge g \leq f \+ q\} = 0$, so $d(f, f \wedge g) = 0$ \linebreak and thus, since $A$ is Archimedean, $f = f \wedge g$, i.e. $f \leq g$. 

Conversely, assume that, for all $h \et k \in A$, $\forall n \in \nn\; h \leq k \+ \frac{1}{2^n} \Rightarrow h \leq k$ and that $d(f,g) = 0$. Then, for all $n \in \nn$, $f \leq g \+ \frac{1}{2^n}$ and $g \leq f \+ \frac{1}{2^n}$, so $f \leq g$ and $g \leq f$, i.e. $f = g$. 
\end{proof} 

\begin{defi} \label{def odot} 

For all locale $\lcal$, let's define $\odot$ by, for all $a \et b \in \USC$, $a \odot b = j\left(\frac{a}{2} \+ \frac{b}{2}\right)$. 

We also define the theory $\Marco_\int$ as the theory $\Marco$ in which all terms of the form $u \odot v$ have been replaced by $j\left(\frac{u}{2} \+ \frac{v}{2}\right)$. 

\end{defi} 

\begin{thm} \label{gros thm} $ $ 

\begin{enumerate} 

\item \label{gros thm 1} For all locale $\lcal$, $\USC$ is an MC-algebra. 

\item \label{gros thm 2} Let $X$ be an ordered topological space. There exists a topological space $Y$ and an embedding of $\Cn{X}$ in $USC(Y)$ for the language $\{\leq, \+, \odot, \wedge, \vee, (\ul q)_{q \in [0,1]}, \frac{\cdot}{2}\}$. 

Moreover, if $X$ is intuitionistic, then $\Cn{X}$ is intuitionistic and this embedding also preserves $\-$. 

\end{enumerate} 

\end{thm} 

\begin{proof} \begin{enumerate} 

\item The theory of MC-algebras is satisfied by $[0,1]$ (\cite[p. 1404]{abbadiniDualCompactOrdered2019}). Moreover, this theory is given in a language only having \usc functions and every axiom is of the form $\phi \leq \psi$ with $\phi$ a term of $\Lusc$. Thus, by \textbf{Theorem} \ref{int formule formules}, every axiom of the theory of MC-algebras is also satisfied by every IC-algebra. 

\item Let $(X, \Tcal)$ be an ordered topological space. 

We recall that, we denote by $X_u$ the set $X$ endowed with the topology \linebreak $\{U \in \Tcal \, | \, U \texte{is downward closed}\}$, and that $USC_\nearrow(X) = USC(X_u)$ (\bf{Lemma} \ref{X_u}). Therefore, we prove that $C^0_\nearrow(X)$ embeds into $USC_\nearrow(X)$. 

For all $f \in \Cn{X}$, $f$ is non-decreasing and upper semi-continuous, so $f \in USC_\nearrow(X)$. That gives a natural order embedding of $\Cn{X}$ into $USC_\nearrow(X)$, and this injection preserves all the punctual operations, namely $\dotp$, $\odot$, $\vee$, $\wedge$, $\frac{\cdot}{2}$, $j_\ast$, $j$ and $\alpha$, and the constants. 

Assume now that $X$ is intuitionistic and let $f \et g \in \Cn{X}$. \\ Let $\fct[f \dotm g][X, {[0,1]}]{x, f(x) \dotm g(x)}$ and let's denote by $\ominus$ the residual of $\+$ in $USC_\nearrow(X)$. 

According to \textbf{Lemma} \ref{gros lemme} \link{(4)} and \link{(6 bis)}, $\uparrow (f \dotm g) \in \Cn{X}$, and, for all $h \in \Cn{X}$, $\uparrow (f \dotm g) \leq h \Leftrightarrow f \dotm g \leq h \Leftrightarrow f \leq g \+ h$. 

Hence $(f, g) \mapsto \uparrow (f \dotm g)$ is the residual of $\+$. 

However, according to \bf{Corollary} \ref{the smallest non-decreasing usc function greater},  $\uparrow (f \dotm g)$ is the smallest upper-semicontinuous non-decreasing function from $X$ to $[0,1]$ greater than $f \dotm g$, that-is-to-say, for all \linebreak $h \in USC_\nearrow(X)$, $\uparrow (f \dotm g) \leq h \Leftrightarrow f \dotm g \leq h$ and $\uparrow (f \dotm g) \in USC_\nearrow(X)$. Hence $\uparrow (f \dotm g) \leq h \Leftrightarrow f \leq g \+ h$ and $\uparrow (f \dotm g) \in USC_\nearrow(X)$. Hence $\uparrow (f \dotm g) = f \ominus g$. Thus, the embedding preserves the residuals. 

\end{enumerate} 
\end{proof} 

\begin{cor} \label{corollaire Macneille} 

The class of intuitionistic MC-algebras is stable under Macneille completion. 

\end{cor} 

\begin{cor} \label{corollaire extension conservative} 

$\T_\int$ is a conservative extension of $\Marco_\int$ in the following sense: 

For all terms $\phi \et \psi$ in the language $\LMarco$, for all $n \in \nn$ $\phi \leq \psi \+ \ul{\frac{1}{2^n}}$ is a theorem of $\T_\int$, if and only if for all $n \in \nn$ $\phi \leq \psi \+ \ul{\frac{1}{2^n}}$ is a theorem of $\Marco_\int$. 

\end{cor} 

\begin{proof} 

Let $X$ be an ordered topological space. Let $X_u$ be the topological space of the \textbf{proof} of \textbf{Theorem} \ref{gros thm}. 

According to \textbf{Theorem} \ref{gros thm} \ref{gros thm 2}, $\fct[\Cn{X}, USC(Y)]{f, f}$ is an embedding for the language \linebreak $\{\+, \odot, \wedge, \vee, (\ul q)_{q \in [0,1]}, \frac{\cdot}{2}\}$. Hence, for all terms of arity $k$ $\phi[v] \et \psi[v]$ in the language $\LMarco$, if for all $n \in \nn$ $\phi \leq \psi \+ \ul{\frac{1}{2^n}}$ is a theorem of $\T_\int$, then, for all $f \in \Cn{X}^k$, since $\phi[f] \leq \psi[f]$ in $USC(Y)$, $\phi[f] \leq \psi[f]$ in $\Cn{X}$, so, for all $n \in \nn$ $\phi \leq \psi \+ \ul{\frac{1}{2^n}}$ is a theorem of $\Marco_\int$. 

Conversely, for all locale $\lcal$, for all terms of arity $k$ $\phi[v] \et \psi[v]$ in the language $\LMarco$, if for all $n \in \nn$ $\phi \leq \psi \+ \ul{\frac{1}{2^n}}$ is a theorem of $\Marco_\int$, then, since $\USC \models \Marco_\int$, for all $f \in \USC^k$, $\phi[f] \leq \psi[f]$, so, for all $n \in \nn$ $\phi \leq \psi \+ \ul{\frac{1}{2^n}}$ is a theorem of $\T_\int$. 
\end{proof} 

Hence we can finally conclude by stating the following theorem. 

\begin{thm} \label{complétude des MC-algèbres} 

Both the classes of intuitionistic MC-algebras and $IC$ are sound for both theories $\T_\int$ and $\Marco$. 

Moreover, for all $\mathcal{L}$-terms $\phi \et \psi$, the following assertions are equivalent: 

\begin{enumerate} 

\item \label{complétude des MC-algèbres 1} For all $n \in \nn$, $\phi \leq \psi \+ \ul{\frac{1}{2^n}}$ is consequence of $\Marco_\int$. 

\item \label{complétude des MC-algèbres 2} For all $n \in \nn$, $\phi \leq \psi \+ \ul{\frac{1}{2^n}}$ is consequence of $\T_\int$. 

\item \label{complétude des MC-algèbres 3} For all locale $\lcal$, $\USC \models \phi \leq \psi$. 

\item \label{complétude des MC-algèbres 4} For all topological space $X$, $\Cn{X} \models \phi \leq \psi$. 

\end{enumerate} 

\end{thm} 

\subsection{Reduction of the axiomatisation of $USC(\lcal)$ to the axiomatisation of $USC(X)$} \label{subsection reduction} 

This subsection aims at proving that the class of IC-algebras and the class of all the $USC(X)$ for $X$ a topological space have same theory in the language $\Lusc$ and that the class of IC-algebras and the class of Cauchy-complete Archimedean MC-algebra have same theory in the language $(\conc)_{n \in \nn}$. The first subsubsection is dedicated to defining the action of $(\conc)_{n \in \nn}$ on every Cauchy-complete Archimedean MC-algebra. 

\subsubsection{The action of all continuous functions from $[0,1]$ to $[0,1]$ on $USC(\lcal)$} 

Before getting to the heart of the matter, we need some notations. 

\begin{nota} 

We recall here that the topology of $[0,1]$ is denoted $\Tcal_u$ and the topology of $[0,1]^n$ is denoted $\Tcal_{u,n}$, for all $n \in \nn$ (\textup{\ref{locale}}). 

\end{nota} 

\begin{thm} \emph{(MC version of \textbf{Theorem} \ref{PositiveAction})}. \label{MC-action} 

Let $A$ be a Cauchy-complete Archimedean MC-algebra. 

There exists a unique family of continuous functions $\left(\cdot\colon \conc[n] \rightarrow C^0(A^n,A)\right)_{n \in \nn}$ that is associative in the sense that, for all $a \in \conc[n]$ and $(b_1, \ldots, b_n) \in \prod[i = 1][n] \conc[k_i]$ and $(f_{i,1}, \ldots, f_{i,k_i})_{i \in \dc[1,n]} \in A^{\sum[i = 1][n] k_i}$, 

$$(a \circ (b_1, \ldots, b_n)) \cdot (f_{1,1}, \ldots, f_{n,k_n}) = a \cdot (b_1 \cdot (f_{1,1}, \ldots, f_{1,k_1}), \ldots, b_n \cdot (f_{n,1}, \ldots, f_{n,k_n}))$$ 

and, for all $(f,g) \in A^2$: \begin{enumerate} 

\item $\vee \cdot (f , g) = f \vee g$ 
\item $\wedge \cdot (f , g) = f \wedge g$ 
\item $\dotp \cdot (f , g) = f \+ g$ 
\item $\odot  \cdot (f , g) = f \odot g$

\end{enumerate} 

Moreover, $\cdot$ are isometries, and, if there exists an ordered topological space $X$ such that $A = C^0_\nearrow(X)$, then for every $n \in \nn$, $a \in \conc[n]$ and $f \in C^0_\nearrow(X)^n$ $a \cdot f = a \circ f$. 

\end{thm} 

\remark{The previous theorem means that there is a unique structure of module on any MC-algebra over the operad $\left(\conc[n]\right)_{n \in \nn}$.} 

The uniqueness of such a family follows from a Stone-Weierstrass-type theorem, which can itself be seen as a corollary of \bf{Theorem} \ref{IncreasingSW}. 

\begin{thm}[Stone-Weierstrass Theorem for MC-algebras, {\cite[Theorem 8.3]{abbadiniDualCompactOrdered2019}}] \label{MC-SW} 

Let $X$ be a preordered topological space, let $L$ be an MC-subalgebra of $C_\nearrow^0(X)$, and suppose that, for every $x, y \in X$, if $x \not\geq y$ then there exists $\phi \in L$ such that $\phi(x) < \phi(y)$. 

If $X$ is compact, then, for every $\psi \in C_\nearrow^0(X)$, there exists a sequence $(\psi_n)_{n \in \nn}$ in $L$ uniformely converging to $\psi$. 

\end{thm} 

\begin{cor} \label{corMC-SW} Let $n \in \nn$. 

$L_n = \{a \in \conc[n] \, | \, a \text{ is a composition of } \vee, \, \wedge, \+, \, \odot, \, (d)_{d \in [0,1] \texte{dyadic}} \et \text{the projections}\}$ is dense in $\conc[n]$. 

\end{cor} 

\begin{proof} 

Let $n \in \nn$. $L_n$ is stable by $\+$, $\vee$, $\wedge$, and $\odot$. Let $\ol{L_n}$ denote its uniforme closure. By continuity of $\+$, $\vee$, $\wedge$, and $\odot$, $\ol{L_n}$ is still stable by $\+$, $\vee$, $\wedge$, and $\odot$. Moreover, it contains all the constant functions. Therefore, $\ol{L_n}$ is an MC-algebra. 

For all $x \not \leq y \in [0,1]^n$, there exists $i \in \dc[1,n]$ such that $x_i < y_i$, and so the projection on the $i$-th coordinate, $\pi$ satisfies $\pi(x) < \pi(y)$. $\pi \in L_n$. Thanks to \textbf{Theorem} \ref{MC-SW}, $[0,1]^n$ being compact, $\ol{L_n}$ is dense in $\conc[n]$, i.e. $L_n$ is dense in $\conc[n]$. 
\end{proof} 

\begin{proof} Proof of \textbf{Theorem} \ref{MC-action}. 

Let $A$ be a Cauchy-complete and Archimedean MC-algebra. \bf{Theorem} \ref{isomorphie MC-algèbres} allows to assume that there exists $X$ a compact Hausdorff ordered topological space such that $A = C_\nearrow^0(X)$. Action of $\Lcinc$ on $C_\nearrow^0(X)$ by composition satisfies points 1 to 4 of \bf{Theorem} \ref{MC-action}. 

%
To prove the uniqueness part of \textbf{Theorem} \ref{MC-action}, let $\cdot$ be a family of functions as in \textup{\textbf{Theorem} \ref{MC-action}}, $n \in \nn$ and $a \in \conc[n]$. There exists $(a_k)_{k \in \nn} \in L_n^\nn$ such that $a_k \rightarrow a$. 

For all $k \in \nn$, $a_k$ being a composition of elements of $\{\vee, \, \wedge, \, \+, \, \odot\}$ and of projections, since $\cdot$ preserves the composition, for all $f \in C_\nearrow^0(X)^n$ $a_k \cdot f = a_k \circ f$. $\cdot$ being continuous, for all $f \in C_\nearrow^0(X)^n$, $(a \cdot f) = \lim a_k \cdot f = \lim a_k \circ f = a \circ f$. 
\end{proof} 

\subsubsection{The reduction} 

The aim of this subsection is to prove the next theorem. 

\begin{thm} \label{Reduction} 

For every $n \in \nn$ and $a \in \Lusc[n]$, if for all topological space $X$ and $f \in USC(X)^n$ $a \cdot f = \ul 0$, then for all locale $\lcal$ and $f \in USC(\lcal)^n$, $a \cdot f = \ul 0$. 

\end{thm} 

An immediate corollary is the following one. 

\begin{cor} 

For every $n \in \nn$ and $a \in \conc[n]$, if for all compact Hausdorff intuitionistic topological space $X$ and $f \in \Cn{X}^n$ $a \cdot f = \ul 0$, then for all locale $\lcal$ and $f \in USC(\lcal)^n$, $a \cdot f = \ul 0$. 

\end{cor} 

A tool lemma for this purpose is the following one. 

\begin{lem} \label{L-morphisms} 

Any $\mathcal{L}$-morphism $F\colon USC(\lcal) \rightarrow USC(\lcal')$, where both $\lcal$ and $\lcal'$ are locales, is $1$-Lipschitzian, and if $F$ is an embedding, $F$ is an isometry. 

Thus, for all Cauchy-complete Archimedean MC-algebras $A$ and $B$ and all morphism of MC-algebras $F\colon A \rightarrow B$, $F$ is $1$-Lipschitzian, and, if $F$ is an embedding, $F$ is an isometry. 

\end{lem} 

\begin{proof} 

For all $(f,g) \in USC(\lcal)^2$ and $d \in [0,1]$ dyadic, if $\norm[f \- g] \leq d$, then $f \- g \leq \ul d$, and thus $F(f) \- F(g) = F(f \- g) \leq F(\ul d) = \ul d$, which gives $\norm[F(f) \- F(g)] \leq d$. Thus, for all $(f,g) \in USC(\lcal)^2$, $\norm[F(f) \- F(g)] \leq \norm[f \- g]$, and then $$d(F(f),F(g)) = \norm[F(f) \- F(g)] \wedge \norm[F(g) \- F(f)] \leq \norm[f \- g] \wedge \norm[g \- f] = d(f,g).$$ 

If $F$ is an embedding, then, for all $(f,g) \in USC(\lcal)^2$ and $d \in [0,1]$ dyadic, 

\begin{align*} \norm[F(f) \- F(g)] \leq d &\Leftrightarrow F(f) \- F(g) \leq \ul d \\ 
& \Leftrightarrow F(f \- g) \leq F(\ul d) \\ 
& \Leftrightarrow f \- g \leq \ul d \\ 
& \Leftrightarrow \norm[f \- g] \leq d 
\end{align*} 

which amounts to $\norm[F(f) \- F(g)] = \norm[f \- g]$. 

Thus, for all $(f,g) \in USC(\lcal)^2$, $$d(F(f),F(g)) = \norm[F(f) \- F(g)] \wedge \norm[F(g) \- F(f)] = \norm[f \- g] \wedge \norm[g \- f] = d(f,g).$$ 
To prove the result for MC-algebras, thanks to \bf{Theorem} \ref{isomorphie MC-algèbres}, we just have to prove it for $C_\nearrow^0(X)$ and $C_\nearrow^0(Y)$, for $X$ and $Y$ compact ordered topological spaces. Let thus $F\colon C_\nearrow^0(X) \rightarrow C_\nearrow^0(Y)$ be a morphism of MC-algebras. By duality, it corresponds to a homeorphism from $Y$ to $X$ and thus extends to a morphism of AC-algebras $\tilde F \colon USC_\nearrow(X) \rightarrow USC_\nearrow(Y)$. It is thus 1-Lipschitzian. Moreover, if $F$ is an embedding, then so is $\tilde F$ and thus $\tilde F$ is an isometry and so is $F$. 

\end{proof} 

\begin{cor} \label{MC morphisms} 

Let $A$ and $B$ be two Cauchy-complete Archimedean MC-algebras and $F\colon A \rightarrow B$ be a morphism of MC-algebras. 

For all $a \in \conc[n]$ and $f \in A$, $F(a \cdot f) = a \cdot F(f)$. 

\end{cor} 

\begin{proof} 

Let $a \in \conc[n]$ and $f \in A$. There exists $(a_k)_{k \in \nn} \in L_n^\nn$ such that $a_k \rightarrow a$. For all $k \in \nn$, $F(a_k \cdot f) = a_k \cdot F(f)$, so, since $F$ is continuous, $F(a \cdot f) = a \cdot F(f)$. 
\end{proof} 

\begin{proof} Proof of \textbf{Theorem} \ref{Reduction}. 

Let $\lcal$ be a locale. 

Since $USC(\lcal)$ is a Cauchy-complete Archimedean MC-algebra, there exists a compact Hausdorff ordered topological space $X$ and an isomorphism of MC-algebras $i\colon USC(\lcal) \simeq C^0_\nearrow(X)$ (\bf{Theorem} \ref{isomorphie MC-algèbres}). Moreover, there exists a topological space $Y$ and $\fct[i'][C_\nearrow^0(X), {USC(Y)}]{f,f}$ an embedding of MC-algebras. According to \bf{Corollary} \ref{MC morphisms}, $i$ and $i'$ preserves $b \cdot \_$, for all $b \in \con[n]$. The lower bound of a family of continuous functions, if it exists, is the punctual lower bound. Hence, $i'$, as $i$, also preserves lower bounds. Thus, for all $a \in \Lusc[n]$ and $f \in \USC^n$, $$i' \circ i(a(f)) = i' \circ i\left(\bigwedge[\ust{b \in C^0_\nearrow(X)}{b \geq a}] b(f)\right) = \left(\bigwedge[\ust{b \in C^0_\nearrow(X)}{b \geq a}] b(i' \circ i(f))\right) = a(i' \circ i(f)).$$ 

Thus, for all $a \in \Lusc[n]$, if for all topological space $W$ and $f \in USC(W)$, $a \cdot f = \ul 0$, then, for all $f \in USC(\lcal)^n$, $i' \circ i(a(f)) = a(i' \circ i(f)) = \ul 0$, which, by injectivity of $i' \circ i$, gives $a \cdot f = \ul 0$. 
\end{proof} 

\subsection{Cut Admissibility} \label{subsection intuitionistic Cut Admissibility} 

The language for sequent calculus in the intuitionistic case is the same as the one for \CFLew, recalled in the following table. We keep Notation \ref{notations coupures}. 

\begin{figure}[!ht] 

$$\begin{array}{|c|c|c||c|c|} 
\hline 
\multicolumn{3}{|c||}{\text{Positive symbols}}& \multicolumn{2}{c|}{\text{Negative correspondent}}\\ 
\hline 
\text{structures}& \text{formulas} & \text{algebraic correspondent} & \text{formulas} & \text{algebraic notation}\\ 
\hline 
,& +& +& -& -\\ 
\hline 
\epsilon& \ul 0& \ul 0& \ul 1& \ul 1\\ 
\hline 
& \wedge& \wedge& \vee& \vee\\ 
\hline 
\circ_2& 2& 2& \frac{\cdot}{2}& \frac{\cdot}{2}\\ 
\hline 
\bullet_2& j_\ast& j_\ast& j& j\\ 
\hline 
\circ_\alpha& \alpha& \alpha& \blacksquare_\alpha& 2v \wedge j_\ast(v)\\ 
\hline 
\end{array}$$ 

\caption{Correspondence between structure symbols and $\mathcal{L}$} 

\end{figure} 

\begin{figure}[!ht] 

\adjustbox{max width = \textwidth}{

$\begin{array}{|ccc|} 

\hline 

[\ref{com monoid}a] \seq{{\Gamma[\gamma, \delta] \vdash A}}{{\Gamma[\delta, \gamma] \vdash A}} & [\ref{com monoid}b] \seq{{\Gamma[\gamma, (\delta, \pi)] \vdash A}}{{\Gamma[(\gamma, \delta), \pi] \vdash A}} & [\ref{com monoid}c] \seq[=]{{\Gamma[\epsilon, \gamma] \vdash A}}{{\Gamma[\gamma] \vdash A}} \\ 

[\ref{2 >=}] \seq{\Gamma[\gamma] \vdash A}{\Gamma[\circ_2 \gamma] \vdash A} & [\ref{j* >=}] \seq{\Gamma[\gamma] \vdash A}{\Gamma[\bullet_2 \gamma] \vdash A} & [\ref{bounded lattice}] \seq{\Gamma[\epsilon] \vdash A}{\Gamma[\gamma] \vdash A} \\ 

[\ref{2v <=}] \seq{\Gamma[\circ_2 \gamma] \vdash A, \Gamma[\circ_2 \delta] \vdash A}{{\Gamma[\gamma, \delta] \vdash A}} & \blue{[\ref{int 2v >=}] \seq{{\Gamma[\gamma, \gamma] \vdash A}}{\Gamma[\circ_2 \gamma] \vdash A}} & \blue{[\ref{bounded lattice} \et \ref{2j*(0) >= 1}]} \seq{{}}{\epsilon_1 \vdash A} \\ 

& [(\hyperlink{def j*}{4.b})] \blue{\seq[=]{\Gamma[{\bullet_2 (\circ_2 \gamma, \delta)}] \vdash A}{{\Gamma[\gamma , \bullet_2 \delta] \vdash A}}} & \\ 

\hypertarget{4.d}{[\ref{alpha beta v >= v}a]} \seq{\Gamma[\gamma] \vdash A}{\Gamma[\circ_\alpha \circ_2 \gamma] \vdash A} & [\ref{alpha beta v >= v}b] \seq{\Gamma[\gamma] \vdash A}{\Gamma[\circ_\alpha \bullet_2 \gamma] \vdash A} & [\ref{alpha beta v <=  v}] \seq{\Gamma[\circ_\alpha \circ_2 \gamma] \vdash A, \Gamma[\circ_\alpha \bullet_2 \gamma] \vdash A}{\Gamma[\gamma] \vdash A} \\ 

[\ref{beta alpha v >= v}a] \seq{\Gamma[\gamma] \vdash A}{\Gamma[\circ_2 \circ_\alpha \gamma] \vdash A} & [\ref{beta alpha v >= v}b] \seq{\Gamma[\gamma] \vdash A}{\Gamma[\bullet_2 \circ_\alpha \gamma] \vdash A} & [\ref{beta alpha v <=  v}] \seq{\Gamma[\circ_2 \circ_\alpha \gamma] \vdash A, \blue\Gamma[\bullet_2 \circ_\alpha \gamma] \vdash A}{\Gamma[\gamma] \vdash A} \\ 

[\ref{v<=d}] \seq{\Gamma[\gamma] \vdash A}{{\Gamma[\epsilon_d, {\circ_\alpha}^n (\gamma {,} \epsilon_{1 \- d})] \vdash A}} & [\ref{1/2^n + 1/2^n}] \seq{{\Gamma[\gamma] \vdash A}}{\Gamma[\circ_2] \vdash A} & [\hypertarget{Archi}{\ref{archimedean}}] \seq{\forall n \in \nn\; \Gamma[\epsilon_{\frac{1}{2^n}}\blue] \vdash A}{\Gamma[\epsilon ]\vdash A} \\ 

\multicolumn{3}{|c|}{[\ref{v>=d}] \seq{{\Gamma\left[\epsilon_{1 \- \frac{1}{2^{n+1}}} ,  {\circ_\alpha}^{n+1}\left(\gamma , \epsilon_{\frac{1}{2^{n+1}}}\right)\right] \vdash A}, \ldots, {\Gamma\left[\epsilon_{1 \- \frac{2^{n+1} \- 2}{2^{n+1}}} ,  {\circ_\alpha}^{n+1}\left(\gamma , \epsilon_{\frac{2^{n+1} \- 2}{2^{n+1}}}\right)\right] \vdash A}}{{\Gamma[\gamma , \epsilon_{\frac{1}{2^n}}] \vdash A}}}\\ 

\hline 

\end{array}$ 

}

\caption{\LJK} \label{LJK} 

\end{figure} 

Here again, we use a system similar to \MGL with several modalities. We need six modalities and three structural symbols $\circ_2$, $\bullet_2$ and $\circ_\alpha$. We thus obtain a system \MGL($\circ_2, \bullet_2, \circ_\alpha)$ given by \GL (from the \textbf{Appendix} \ref{Annexes}) understood with contexts of the extended language and \textit{Figure} \ref{Introduction Rules for Modalities}. In addition to these rules, we add the structural rules given in \textit{Figure} \ref{LJK} and call the total system \LJK. 

\vspace{\baselineskip} 

Since all the added rules are analytic (cf. \bf{Definition} \ref{analytic rules annexes}), according to \textbf{Theorems} \ref{Completeness theorem annexes}, \ref{Cut Admissibility annexes}, \ref{Equivalent theories} and \ref{complétude des MC-algèbres}, the \bf{Theorems} \ref{Completeness theorem int} and \ref{Cut Admissibility theorem int} are true. 

\section{Involutive case} \label{Involutive case} 

In this section, we study the involutivity of the negation. We first prove that the negation of an AC-algebra is involutive if and only if this is the case for the negation of the underlying commutative residuated complete lattice (\bf{Theorem} \ref{involutive correspondence}), which leads to an axiomatisation of these involutive AC-algebras. Finally, we give a sequent-style cut-free deductive system admitting the cut rule that describes involutive AC-algebras. 

\subsection{Involutive Continuous Algebras} 

\begin{thm} \label{involutive correspondence} 

Let $\lcal$ be a commutative residuated complete lattice. 

The negation of $USC(\lcal)$ is involutive if and only if so is the negation of $\lcal$. 

\end{thm} 

\begin{proof} 

\color{blue} 
If the negation on $USC(\lcal)$ is involutive, then, by \bf{Lemma} \ref{first half}, so is the negation on $\lcal$. 

Assume now that the negation on $\lcal$ is involutive. For all $f \in \USC$, \linebreak $(1 \- f)(q) = \bigvee[p < q] \bigwedge[1 > r \geq p] (f(r - p) \nrightarrow \bot) = \bigvee[p < q] f(1 - p) \nrightarrow \bot$, so \begin{align*}(1 \- (1 \- f))(q) &= \bigvee[p < q] (1 \- f)(1 - p) \nrightarrow \bot \\ &= \bigvee[p < q] \left(\bigvee[r < 1 - p] f(1 - r) \nrightarrow \bot\right) \nrightarrow \bot \\ &= \bigvee[p < q] (f(p) \nrightarrow \bot) \nrightarrow \bot \\ &= \bigvee[p < q] f(p) \\ &= f(q).\end{align*} 
\color{black} 
\end{proof} 

\begin{defi} 

We denote by $T_\inv$ the theory $T \cup \{v \leq (1 \- (1 \- v))\}$. 

We denote by $InAC$ the class whose elements are the $USC(\lcal)$ for $\lcal$ \blue{an} involutive commutative residuated complete lattice and call these algebras \emph{Involutive Affine Continuous Algebras}. 

\end{defi} 

\begin{thm} \label{inv big theorem} 

For all model $A$ of $\T_\inv$, there exists an involutive commutative residuated complete lattice $\lcal$ such that the quotient of the Macneille completion of $A$ by $\simeq$ is isomorphic to $USC(\lcal)$. 

For all model $A$ of $\T_\inv$, there exists an involutive commutative residuated complete lattice $\lcal$ such that the quotient of $A$ by $\simeq$ embeds into $USC(\lcal)$. 

\end{thm} 

\begin{proof} 

Let $A$ be a model of $\T_\inv$. $A$ is a model of $\T$, so, according to \bf{Theorem} \ref{big theorem}, there exists a commutative residuated complete lattice $\lcal$ such that the quotient of the Macneille completion of $A$ by $\simeq$ is isomorphic to $USC(\lcal)$. According to this same theorem, the quotient of $A$ by $\simeq$ also embeds into $USC(\lcal)$. 

Since, $A$ satisfies $v \leq 1 \- (1 \- v)$ and the function $a \mapsto a$ preserves upper bounds, \bf{Lemma} \ref{lemme Macneille} ensures that the Macneille completion of $A$, this identity is still true. Hence it is also true in the quotient and thus in $USC(\lcal)$. \bf{Theorem} \ref{involutive correspondence} then ensures that the negation on $\lcal$ is involutive. 

\end{proof} 

\begin{cor} \label{complétude des InAC-algèbres} 

For all $\mathcal{L}$-terms $\phi \et \psi$, the following assertions are equivalent: 

\begin{enumerate} 

\item \label{complétude des InAC-algèbres 1} For all involutive commutative residuated complete lattice $\lcal$, $\USC \models \phi \leq \psi$. 

\item \label{complétude des InAC-algèbres 2} For all $n \in \nn$, $\phi \leq \psi \+ \ul{\frac{1}{2^n}}$ is consequence of $\T_\inv$. 

\end{enumerate} 

\end{cor} 

\begin{proof} 

For all involutive commutative residuated complete lattice $\lcal$, $\USC$ is a model of $\T_\inv$. By Archimedeanity of every InAC-algebra, we have that \ref{complétude des InAC-algèbres 2} implies \ref{complétude des InAC-algèbres 1}. 

Let $\phi \et \psi$ be $\mathcal{L}$-terms both having $k$ free variables such that, for all locale $\lcal$, $\USC \models \phi \leq \psi$. Let $n \in \nn$. 

Let $A$ be a model of $\T_\inv$. According to \bf{Corollary} \ref{inv big theorem}, there exists an involutive commutative residuated lattice $\lcal$ and a morphism $i \colon A \rightarrow \USC$ such that, for all $a \et b \in A$, \linebreak $i(a) \leq i(b) \Leftrightarrow \forall n \in \nn \; a \leq b \+ \ul{\frac{1}{2^n}}$. For all $a \in A^k$ and $n \in \nn$, since $i(\phi(a)) \leq i(\psi(a))$, $\phi(a) \leq \psi(a) \+ \ul{\frac{1}{2^n}}$. 

The class of models of $\T$ being complete for $\T$, $\phi \leq \psi \+ \ul{\frac{1}{2^n}}$ is consequence of $\T$. 

\end{proof} 

\subsection{Cut Admissibility} \label{subsection Cut Admissibility} 

In the involutive case, since, contrary to \textbf{Appendix} \ref{Annexes}, '$,$' is commutative, there is only one negation symbol $\neg$. Moreover, since $\alpha$ is bijective, $\circ_\alpha = \bullet_\alpha$. Thus, we can forget $\bullet_\alpha$ and add the rule $\seq[=]{\circ_\alpha \Gamma \vdash \Delta}{\Gamma \vdash \circ_\alpha \Delta}$. However, everything works as in \textbf{Appendix} \ref{Annexes}. The language for the structures and formulas of sequent calculus in the involutive case, and its correspondence with the language $\mathcal{L}_\inv$, are given in the table of \it{Figure} \ref{Inv Correspondence between structure symbols and L}. We will also consider the language $\mathcal{L}_\inv'$ which is $\mathcal{L}_\inv$ with the symbols of the table of \it{Figure} \ref{Correspondence between the new structure symbols and L}. We keep Notation \ref{notations coupures}. 

\begin{figure}[!ht] 

$$\begin{array}{|c|c|c||c|c|} 
\hline 
\multicolumn{3}{|c||}{\text{Left interpretation}}& \multicolumn{2}{c|}{\text{Right interpretation}}\\ 
\hline 
\text{structures}& \text{formulas} & \text{algebraic correspondent} & \text{formulas} & \text{algebraic correspondent}\\ 
\hline 
,& \+& \+& \odot& j\left(\frac{\_}{2} \+ \frac{\_}{2}\right)\\ 
\hline 
\epsilon& \ul 0& \ul 0& \ul 1& \ul 1\\ 
\hline 
\circ_2& 2& 2& \frac{\cdot}{2}& \frac{\cdot}{2}\\ 
\hline 
\bullet_2& j_\ast& j_\ast& j& j\\ 
\hline 
\circ_\alpha& \alpha& \alpha& \blacksquare_\alpha& 2v \wedge j_\ast(v)\\ 
\hline 
\ng[\blank]& \ng& \ng& \ng& \ng\\ 
\hline 
\end{array}$$ 

\caption{Correspondence between structure symbols and $\mathcal{L}$} \label{Inv Correspondence between structure symbols and L} 

\end{figure} 

\begin{figure}[!ht] 

$$\begin{array}{|c|c|c||c|c|} 
\hline 
\multicolumn{3}{|c||}{\text{Left interpretation}}& \multicolumn{2}{c|}{\text{Right interpretation}}\\ 
\hline 
\text{structures}& \text{formulas} & \text{algebraic correspondent} & \text{formulas} & \text{algebraic correspondent}\\ 
\hline 
\bullet_\alpha& \alpha& \alpha& \blacksquare_\alpha& 2v \wedge j_\ast(v)\\ 
\hline 
\bs[\blank]& \bs& \bs& \bs& \bs\\ 
\hline 
\end{array}$$ 

\caption{Correspondence between the new structure symbols and $\mathcal{L}$} \label{Correspondence between the new structure symbols and L} 

\end{figure} 

The system \InMGL for involutive modal full Lambek calculus can be applied with one or several modalities and one or two negations. Here, we also have only one negation, a difference that \textbf{Lemma} \ref{only one negation} tacles. We thus obtain a system \InMGL($\circ_2, \bullet_2, \circ_\alpha$) in the language $\mathcal{L}_\inv$ given by \InGL (from the \textbf{Appendix} \ref{Annexes}) and \textit{Figure} \ref{inv Introduction Rules for Modalities} and a system \InMGL($\circ_2, \bullet_2, \circ_\alpha, \bullet_\alpha$) in the language $\mathcal{L}_\inv'$ given by \InGL and \textit{Figure} \ref{inv Introduction Rules for Expanded Modalities}. In addition to the rules of \InMGL($\circ_2, \bullet_2, \circ_\alpha$), we add the structural rules given in \textit{Figure} \ref{InCFLew} and call the total system \InCFLew. Define \InCFLew' as the system in the language $\mathcal{L}_\inv'$ having for rules the ones of \InMGL($\circ_2, \bullet_2, \circ_\alpha, \bullet_\alpha$) and \it{Figure} \ref{InCFLew}. 

\begin{figure}[!ht] 

$$\begin{array}{|cccc|} 
\hline 

[\target{L \,2}] \seq{\circ_2 A \vdash \Delta}{{2 A \vdash \Delta}} &[\target{L \,j_\ast}] \seq{\bullet_2 A \vdash \Delta}{{j_\ast(A) \vdash \Delta}} &[\target{R \,2}] \seq{\Gamma \vdash A}{\circ_2 \Gamma \vdash 2 A} &[\target{R \,j_\ast}] \seq{\Gamma \vdash A}{\bullet_2 \Gamma \vdash j_\ast(A)} \\ 

[\target{\circ_2 / \bullet_2}] \seq[=]{{\circ_2 \Gamma \vdash \Delta}}{\Gamma \vdash \bullet_2 \Delta} &[\target{L \,\alpha}] \seq{\circ_\alpha A \vdash \Delta}{{\alpha(A) \vdash \Delta}} &[\target{R \,\alpha}] \seq{\Gamma \vdash A}{\circ_\alpha \Gamma \vdash \alpha(A)} &[\target{\circ_\alpha / \circ_\alpha}] \seq[=]{{\circ_\alpha \Gamma \vdash \Delta}}{\Gamma \vdash \circ_\alpha \Delta} \\ 

\target{Invng L}[[\ng L]] \seq[=]{{\gamma, \delta \vdash \beta}}{{\delta \vdash \ng[\gamma], \beta}} &\target{Invng R}[[\ng R]] \seq[=]{{\gamma, \delta \vdash \beta}}{{\gamma \vdash \beta, \ng[\delta]}} &\target{[Lneg]}[[L\ng]] \seq{{\ng[a] \vdash \delta}}{\ng a \vdash \delta} &\target{[Rneg]}[[R\ng]] \seq{{\gamma \vdash \ng[a]}}{\gamma \vdash \ng a} \\ 

\hline 
\end{array}$$ 

\caption{Introduction Rules \InMGL($\circ_2, \bullet_2, \circ_\alpha$)} \label{inv Introduction Rules for Modalities} 

\end{figure} 

\begin{figure}[!ht] 

$$\begin{array}{|cccc|} 
\hline 

[\target{L \,2}] \seq{\circ_2 A \vdash \Delta}{{2 A \vdash \Delta}} &[\target{L \,j_\ast}] \seq{\bullet_2 A \vdash \Delta}{{j_\ast(A) \vdash \Delta}} &[\target{R \,2}] \seq{\Gamma \vdash A}{\circ_2 \Gamma \vdash 2 A} &[\target{R \,j_\ast}] \seq{\Gamma \vdash A}{\bullet_2 \Gamma \vdash j_\ast(A)} \\ 

[\target{\circ_2 / \bullet_2}] \seq[=]{{\circ_2 \Gamma \vdash \Delta}}{\Gamma \vdash \bullet_2 \Delta} &[\target{L \,\alpha}] \seq{\circ_\alpha A \vdash \Delta}{{\alpha(A) \vdash \Delta}} &[\target{R \,\alpha}] \seq{\Gamma \vdash A}{\circ_\alpha \Gamma \vdash \alpha(A)} &[\target{\circ_\alpha / \bullet_\alpha}] \seq[=]{{\circ_\alpha \Gamma \vdash \Delta}}{\Gamma \vdash \bullet_\alpha \Delta} \\ 

\target{[Invtild L]}[[\bs L]] \seq[=]{{\gamma, \delta \vdash \beta}}{{\delta \vdash \bs[\gamma], \beta}} &\target{[Invneg R]}[[\ng R]] \seq[=]{{\gamma, \delta \vdash \beta}}{{\gamma \vdash \beta, \ng[\delta]}} &[\target{L \,\alpha \,2}] \seq{\bullet_\alpha A \vdash \Delta}{{\alpha(A) \vdash \Delta}} & [\target{R \,\alpha \,2}] \seq{\Gamma \vdash A}{\bullet_\alpha \Gamma \vdash \alpha(A)} \\ 

\target{[Ltild]}[[L\bs]] \seq{{\bs[a] \vdash \delta}}{\bs a \vdash \delta} &\target{[Rtild]}[[R\bs]] \seq{{\gamma \vdash \bs[a]}}{\gamma \vdash \bs a} &\target{[Lneg]}[[L\ng]] \seq{{\ng[a] \vdash \delta}}{\ng a \vdash \delta} &\target{[Rneg]}[[R\ng]] \seq{{\gamma \vdash \ng[a]}}{\gamma \vdash \ng a} \\ 

\hline 
\end{array}$$ 

\caption{Introduction Rules for \InMGL($\circ_2, \bullet_2, \circ_\alpha, \bullet_\alpha$)} \label{inv Introduction Rules for Expanded Modalities} 

\end{figure} 

\begin{figure}[!ht] 

\adjusttopage{$\begin{array}{|ccc|} 

\hline 

[\ref{com monoid}] \seq{{\Gamma, \Delta \vdash \Theta}}{{\Delta, \Gamma \vdash \Theta}} & [\ref{com monoid}b] \seq{{\Gamma, (\Delta, \Pi) \vdash \Theta}}{{(\Gamma, \Delta), \Pi \vdash \Theta}} & [\ref{com monoid}c] \seq[=]{{\epsilon, \Gamma \vdash \Theta}}{{\Gamma \vdash \Theta}} \\ 

[\ref{2 >=}] \seq{\Gamma \vdash \Theta}{\circ_2 \Gamma \vdash \Theta} & [\ref{j* >=}] \seq{\Gamma \vdash \Theta}{\bullet_2 \Gamma \vdash \Theta} & [\ref{bounded lattice}] \seq{}{\Gamma \vdash \ul 0} \\ 

[\ref{2v <=}] \seq{\circ_2 \Gamma \vdash \Theta, \circ_2 \Delta \vdash \Theta}{{\Gamma, \Delta \vdash \Theta}} & \blue{[\link{cont de 2}[(4.a)]] \seq[=]{{\circ_2 (\Gamma, \Delta) \vdash \Theta}}{{\circ_2 \Gamma, \circ_2 \Delta \vdash \Theta}}} & \blue{[\ref{bounded lattice} \et \ref{2j*(0) >= 1}]} \seq{{}}{\epsilon_1 \vdash \Theta} \\ 

& [(\hyperlink{def j*}{4.b})] \blue{\seq[=]{{\bullet_2 (\circ_2 \Gamma , \Delta)\vdash \Theta}}{{\Gamma , \bullet_2 \Delta \vdash \Theta}}} & \\ 

\hypertarget{4.d}{[\ref{alpha beta v >= v}a]} \seq{\Gamma \vdash \Theta}{\circ_\alpha \circ_2 \Gamma \vdash \Theta} & [\ref{alpha beta v >= v}b] \seq{\Gamma \vdash \Theta}{\circ_\alpha \bullet_2 \Gamma \vdash \Theta} & [\ref{alpha beta v <=  v}] \seq{\circ_\alpha \circ_2 \Gamma \vdash \Theta, \circ_\alpha \bullet_2 \Gamma \vdash \Theta}{\Gamma \vdash \Theta} \\ 

[\hyperlink{def alpha, 2}{4.e}a] \seq{\Gamma \vdash \Theta}{\circ_2 \circ_\alpha \Gamma \vdash \Theta} & [\ref{beta alpha v >= v}b] \seq{\Gamma \vdash \Theta}{\bullet_2 \circ_\alpha \Gamma \vdash \Theta} & [\ref{beta alpha v <=  v}] \seq{\circ_2 \circ_\alpha \Gamma \vdash \Theta, \bullet_2 \circ_\alpha \Gamma \vdash \Theta}{\Gamma \vdash \Theta} \\ 

[\ref{v<=d}] \seq{\Gamma \vdash \Theta}{{\epsilon_d, {\circ_j}^n (\Gamma {,} \epsilon_{1 - d}) \vdash \Theta}} & [\ref{1/2^n + 1/2^n}]\seq{{\epsilon_{\frac{1}{2^n}}, \epsilon_{\frac{1}{2^n}} \vdash \Theta}}{\epsilon_{\frac{1}{2^{n-1}}} \vdash \Theta} & [\ref{archimedean}] \seq{\forall n \in \nn\; \epsilon_{\frac{1}{2^n}} \vdash \Theta}{\epsilon \vdash \Theta} \\ 

\multicolumn{3}{|c|}{[\ref{v>=d}] \seq{{\epsilon_{1 - \frac{1}{2^{n+1}}} ,  {\circ_j}^{n+1}\left(\Gamma , \epsilon_{\frac{1}{2^{n+1}}}\right) \vdash \Theta}, \ldots, {\epsilon_{1 - \frac{2^{n+1} - 2}{2^{n+1}}} ,  {\circ_j}^{n+1}\left(\Gamma , \epsilon_{\frac{2^{n+1} - 2}{2^{n+1}}}\right) \vdash \Theta}}{{\Gamma , \epsilon_{\frac{1}{2^n}} \vdash \Theta}}}\\ 

\hline 

\end{array}$} 

\caption{\InCFLew} \label{InCFLew}   

\end{figure} 

\pagebreak 

Since $\circ_\alpha$ is self dual, we reduce \InCFLew to the setting of \textbf{Appendix} \ref{Annexes} thanks to the next lemma. 

\begin{lem} \label{only one negation} 

The systems \target{one}[(1)] \InMGL($\circ_2, \bullet_2, \circ_\alpha$) $\cup \left\{\seq{{\Gamma, \Delta \vdash \Theta}}{{\Delta, \Gamma \vdash \Theta}}\right\}$ and \target{two}[(2)] \linebreak \InMGL($\circ_2, \bullet_2, \circ_\alpha, \bullet_\alpha$) $\cup \left\{\seq{{\Gamma, \Delta \vdash \Theta}}{{\Delta, \Gamma \vdash \Theta}}, \; \seq[=]{\circ_\alpha \Gamma \vdash \Delta}{\bullet_\alpha \Gamma \vdash \Delta}\right\}$ are equivalent in the following sense: 

For every deduction in \link{one}[(1)] there exists a deduction in \link{two}[(2)] with same premisses and conclusion. 

By letting, in \link{one}[(1)], $\bullet_\alpha = \circ_\alpha$ and $\bs = \ng$, every deduction in \link{two}[(2)] is a deduction in \link{one}[(1)]. 

Hence, in this sense, \InCFLew and \InCFLew' are equivalent. 

\end{lem} 

\begin{proof} 

We just have to prove that the rules of each system are deducible in the other one. 

\link{one}[(1)] $\rightarrow$ \link{two}[(2)] : The two rules of \link{one}[(1)] that are not rules of \link{two}[(2)] are $[\link{\circ_\alpha / \circ_\alpha}]$ and $\link{Invng L}[[\ng L]]$. 

Here are the proof trees of $\seq[=]{{\Gamma \vdash \circ_\alpha \Delta}}{\bullet_\alpha \Gamma \vdash \Delta}$, $[\link{\circ_\alpha / \circ_\alpha}]$ and $\link{Invng L}[[\ng L]]$in \link{two}[(2)]. We recall the validity of the deduction $\seq[=]{\gamma \vdash \ng[\delta]}{\delta \vdash \bs[\gamma]}\link{[\bs/\ng]}$ and that $\ng[(\gamma, \delta)] = \ng[\delta], \ng[\gamma] \textup{\target{[neg]}}$. 

\begin{center} 

\begin{tabular}{ccccc} 

\normalAlignProof 
\alwaysDoubleLine 
\Axiom$\gamma \fCenter \circ_\alpha \delta$ 
\UnaryInf$\gamma \fCenter \bs[\ng[(\circ_\alpha \delta)]]$ 
\RightLabel{$\link{[\bs/\ng]}$} 
\UnaryInf$\ng[(\circ_\alpha \delta)] \fCenter \ng[\gamma]$ 
\UnaryInf$\circ_\alpha (\ng[\delta]) \fCenter \ng[\gamma]$ 
\UnaryInf$\ng[\delta] \fCenter \bullet_\alpha (\ng[\gamma])$ 
\UnaryInf$\ng[\delta] \fCenter \ng[(\bullet_\alpha \gamma)]$ 
\RightLabel{$\link{[\bs/\ng]}$} 
\UnaryInf$\bullet_\alpha \gamma \fCenter \bs[\ng[\delta]]$ 
\UnaryInf$\bullet_\alpha \gamma \fCenter \delta$ 
\DisplayProof{} & & 

\normalAlignProof 
\alwaysDoubleLine 
\Axiom$\circ_\alpha \gamma \fCenter \delta$ 
\UnaryInf$\bullet_\alpha \gamma \fCenter \delta$ 
\RightLabel{$[\link{\circ_\alpha / \bullet_\alpha}]$} 
\UnaryInf$\gamma \fCenter \circ_\alpha \delta$ 
\DisplayProof{} & & 

\normalAlignProof 
\alwaysDoubleLine 
\Axiom$\gamma, \delta \fCenter \beta$ 
\RightLabel{\ref{com monoid}} 
\UnaryInf$\delta, \gamma\fCenter \beta$ 
\RightLabel{$\link{[Invneg R]}[[\ng R]]$} 
\UnaryInf$\delta \fCenter \beta, \ng[\gamma]$ 
\UnaryInf$\delta \fCenter \bs[\ng[(\beta, \ng[\gamma])]]$ 
\RightLabel{$\link{[\bs/\ng]}$} 
\UnaryInf$\ng[(\beta, \ng[\gamma])] \fCenter \ng[\delta]$ 
\RightLabel{\link{[neg]}} 
\UnaryInf$\ng[\ng[\gamma]], \ng[\beta] \fCenter \ng[\delta]$ 
\RightLabel{$\ref{com monoid}$} 
\UnaryInf$\ng[\beta], \ng[\ng[\gamma]] \fCenter \ng[\delta]$ 
\RightLabel{$\link{[\bs/\ng]}$} 
\UnaryInf$\ng[(\ng[\gamma], \beta)] \fCenter \ng[\delta]$ 
\RightLabel{\link{[neg]}} 
\UnaryInf$\delta \fCenter \bs[\ng[(\ng[\gamma], \beta)]]$ 
\UnaryInf$\delta \fCenter \ng[\gamma], \beta$ 
\DisplayProof{}

\end{tabular} 

\end{center} 

\link{two}[(2)] $\rightarrow$ \link{one}[(1)]: By letting $\bullet_\alpha = \circ_\alpha$ and $\bs = \ng$, every rule of \link{two}[(2)] is a rule of \link{one}[(1)], so every deduction in \link{two}[(2)] becomes a deduction in \link{one}[(1)]. 

\end{proof} 

\begin{lem} 

Algebraic models of \InCFLew are the ones of the theory $\T_\inv$ satisfying property \ref{archimedean}. 

\end{lem} 

\begin{proof} 

According to \textbf{Theorem} \ref{Completeness theorem annexes} and \textbf{Lemma} \ref{only one negation}, the models of \InCFLew are the residuated lattices satisfying, for each previous structural rule $r$, the formula in the language $\mathcal{L}_\inv'$ where each structure variable has been replaced by a fresh new formula variable, $\seq{{}}{{}}$ by $\Rightarrow$, $\vdash$ by $\geq$, $,$ by $\+$, $\circ_2$ by $2 \cdot$, $\bullet_2$ by $j_\ast$, $\circ_\alpha$ by $\alpha$ and each $\epsilon_d$ by $\ul d$. 

Then, as in \textbf{Example} \ref{Exemples d'équivalence} we can see that these axioms are equivalent to axioms of $\T_\inv$ as refered to in the system \hyperlink{CFLew}{\InCFLew}. \link{croissance}[(2)] can be treated the same way as \link{2 adjunction}[\text{(3.a)}] and \ref{bounded lattice}, \link{Defining axioms}[(4)] and \link{infinitesimals}[(5)] can be treated the same way as \hyperlink{4.d}{(\hyperlink{def alpha, 1}{4.d})}. The only rule to which the previous methodology can't apply directly is \hyperlink{Archi}{\ref{archimedean}} because it is infinitary. However, it is immediate that every model of this rule has equivalently the property $\forall n \in \nn \ul{\frac{1}{2^n}} \geq v \Rightarrow \ul 0 \geq v$, which is actually property \ref{archimedean}. 

Hence the models of \InCFLew are exactly the models of $\T_\inv$. 
\end{proof} 

Since all the rules added to obtain \InCFLew are analytic (cf. \bf{Definition} \ref{analytic rules annexes}), according to \textbf{Theorems} \ref{Completeness theorem annexes} and \ref{Cut Admissibility annexes}, the following theorem is true. 

\begin{thm}[Completeness theorem] 

The class of $USC(\lcal)$, for $\lcal$ a complete commutative involutive residuated lattice, is sound and complete for $\InCFLew$. 

\end{thm} 

\begin{proof} 

Since \InCFLew' and \InCFLew are equivalent (\bf{Lemma} \ref{only one negation}), according to \bf{Theorem} \ref{Completeness theorem annexes}, the class of $USC(\lcal)$ for $\lcal$ a complete commutative involutive residuated lattice is sound and complete for \InCFLew', this class is also sound and complete for $\InCFLew$. 
\end{proof} 

\begin{thm}[Cut Admissibility theorem]  

In the system \InCFLew, for all formulas $a_1, \ldots, a_n$ and $b$ $\{, \, , \circ_2, \bullet_2, \circ_\alpha, \epsilon\}$-term $G$ such that there exists a deduction of $G(a_1, \ldots, a_n) \vdash b$ using the cut rule, there exists a deduction of $G(a_1, \ldots, a_n) \vdash b$ not using the cut rule. 

\end{thm} 

\begin{proof} 

\ul{\bf{Claim :}} In the system \InCFLew', for all formulas $a_1, \ldots, a_n \et b$ and $\{, \, , \circ_2, \bullet_2, \circ_\alpha, \bullet_\alpha, \epsilon\}$-term $G$ such that there exists a deduction of \linebreak $G(a_1, \ldots, a_n) \vdash b$ using the cut rule, there exists a deduction of $G(a_1, \ldots, a_n) \vdash b$ not using the cut rule. 

The claim is a direct consequence of \bf{Theorem} \ref{Cut Admissibility annexes}. 

For each deduction $D$ in \InCFLew $\cup \{\text{cut}\}$, according to \textbf{Lemma} \ref{only one negation} there exists a deduction in \InCFLew' $\cup \{\text{cut}\}$ having same premisses and conclusion. Since, according to \bf{Theorem} \ref{Cut Admissibility}, \InCFLew' admits the cut rule, there exists a deduction with same premisses and conclusion in \InCFLew', which gives, according to \textbf{Lemma} \ref{only one negation}, a deduction in \InCFLew with same premisses and conclusion as $D$. 
\end{proof} 

\section{Boolean case} \label{section Boolean case} 

Here, we deal with the case which cumulates the properties of the intuitionistic and involutive cases. It is the analogue of the study of Boolean logic in the continuous setting. We first prove that the ordered topological space associated to any IC-algebra (\bf{Theorem} \ref{isomorphie MC-algèbres}) is actually just a topological space. Second, we show that the theory obtained to describe this Boolean Continuous Logic is equivalent to the theory of classical continuous logic (\bf{Theorem} \ref{equivalence to classical logic}). Finally, we exhibit a sequent-style deductive system admitting the cut rule that describes this logic. 

\subsection{Boolean Continuous Algebras} 

\begin{defi}[{\cite[p. 22]{galatos2007residuated}}] 

A Boolean algebra is a Heyting algebra $\bcal$ such that, for all $x \in \bcal$, $\ng \ng x = x$. 

\end{defi} 

\begin{thm} 

Let $\lcal$ be a commutative residuated complete lattice. 

$\lcal$ is a Boolean algebra if and only if the negation on $USC(\lcal)$ is involutive and \linebreak $USC(\lcal) \models 2v \geq v \+ v$. 

\end{thm} 

\begin{proof} 

The claim is a consequence of the following one : $\lcal$ is a Boolean algebra if and only if the negation on $\lcal$ is involutive and $\otimes = \wedge$. The claim is true since, for all commutative residuated lattice $\lcal$, $\lcal$ is a Heyting algebra \ssi $\otimes = \wedge$ and for every Heyting algebra $\lcal$, $\lcal$ is a Boolean algebra \ssi its negation is involutive. 
\end{proof} 

\begin{defi} 

We denote by $\T_\class$ the theory $\T_\int \cup \{v \leq (1 \- (1 \- v))\}$. 

We denote by $BC$ the class whose elements are the $USC(\bcal)$ for $\bcal$ a complete Boolean algebra and call these algebras \emph{Boolean Continuous Algebras}. 

\end{defi} 

\begin{thm} \label{class big theorem} 

For all model $A$ of $\T_\class$, there exists a complete Boolean algebra $\bcal$ such that the quotient of the Macneille completion of $A$ by $\simeq$ is isomorphic to $USC(\bcal)$. 

For all model $A$ of $\T_\class$, there exists a complete Boolean algebra $\bcal$ such that the quotient of $A$ by $\simeq$ embeds into $USC(\bcal)$. 

\end{thm} 

\begin{proof} 

Let $A$ be a model of $\T_\class$. $A$ is a model of $\T_\int$, so, according to \bf{Corollary} \ref{int big corollary}, there exists a locale $\bcal$ such that the quotient of the Macneille completion of $A$ by $\simeq$ is isomorphic to $USC(\bcal)$. According to this same theorem, the quotient of $A$ by $\simeq$ also embeds into $USC(\bcal)$. Finally, according to \bf{Theorem} \ref{inv big theorem}, this same locale is involutive, so $\bcal$ is a complete Boolean algebra. 
\end{proof} 

\begin{cor} \label{complétude des BC-algèbres} 

For all $\mathcal{L}$-terms $\phi \et \psi$, the following assertions are equivalent: 

\begin{enumerate} 

\item \label{complétude des BC-algèbres 1} For all complete Boolean algebra $\bcal$, $USC(\bcal) \models \phi \leq \psi$. 

\item \label{complétude des BC-algèbres 2} For all $n \in \nn$, $\phi \leq \psi \+ \ul{\frac{1}{2^n}}$ is consequence of $\T_\class$. 

\end{enumerate} 

\end{cor} 

\begin{proof} 

For all complete Boolean algebra $\bcal$, $USC(\bcal)$ is a model of $\T_\int$ and satisfies \linebreak $v \leq 1 \- (1 \- v)$, so $USC(\bcal)$ is a model of $\T_\class$. By Archimedeanity of every BC-algebra, we have that \ref{complétude des BC-algèbres 2} implies \ref{complétude des BC-algèbres 1}. 

Let $\phi \et \psi$ be $\mathcal{L}$-terms both having $k$ free variables such that, for all complete Boolean algebra $\bcal$, \linebreak $USC(\bcal) \models \phi \leq \psi$. Let $A$ be a model of $\T_\class$. According to \bf{Theorem} \ref{class big theorem}, there exists a complete Boolean algebra $\bcal$ and a morphism $i \colon A \rightarrow USC(\bcal)$ such that, for all $a \et b \in A$, $i(a) \leq i(b) \Leftrightarrow \forall n \in \nn \, a \leq b \+ \ul{\frac{1}{2^n}}$. For all $a \in A^k$ and $n \in \nn$, since $i(\phi(a)) \leq i(\psi(a))$, $\phi(a) \leq \psi(a) \+ \ul{\frac{1}{2^n}}$. The class of all models of $\T_\class$ being complete for $\T_\class$, $\phi \leq \psi \+ \ul{\frac{1}{2^n}}$ is consequence of $T$. 
\end{proof} 

\begin{lem} \label{Y fermé} 

Let $(X, \leq)$ be an intuitionnnistic Hausdorff ordered space. The set $Y$ of minimal elements of $X$ is closed. 

\end{lem} 

\begin{proof} 

Let $y \in \ol Y$. 

Since $X$ is Hausdorff, $\{y\} = \bigcap[\ust{U \texte{open}}{\tq y \in U}] U$. Moreover, for all open $U$ such that $y \in U$, $y \in \ol{U \cap Y}$, so \linebreak $\{y\} = \bigcap[\ust{U \texte{open}}{\tq y \in U}] U \cap Y$. However, for all open $U$ such that $y \in U$, $U \cap Y \subset Y$, so $U \cap Y$ is downward closed. Hence $\{y\}$ is downward closed, which proves that $y \in Y$. 

\end{proof} 

\begin{lem} \label{lemme Y non vide} 

Let $(X, \leq)$ be a non-empty compact ordered space. The set $Y$ of minimal elements of $X$ is not empty. 

\end{lem} 

\begin{proof} 

Let $(X, \leq)$ be a non-empty compact ordered space. Let $E$ be the set of non empty chains on $X$. 

$E$ is not empty because $X$ is not empty. 

Let $\ccal \subset E$ be totally ordered by inclusion and assume $\ccal$ is not empty. Let $A = \bigcup[C \in \ccal] C$. For all $C \in \ccal$, $C \subset A$. For all $x \et y \in A$, there exists $C \in \ccal$ such that $x \et y \in C$, so either $x \leq y$ or $y \leq X$. Thus, $A \in E$. 

According to Zorn's lemma, $E$ admits a maximal element $C$. 

Let $\fcal = \{A \subset X \,|\, A \cap C \neq \emptyset \et \forall x < y \in C \; (y \in A \Rightarrow x \in A)\}$. For all $A \et B \in \fcal$, there exists $x_1 \in A \cap C$ and $x_2 \in B \cap C$, so, since $C$ is a chain, $x_1 \in (A \cap B) \cap C$ or $x_2 \in A \cap B$ and thus $A \cap B \in \fcal$. Hence $\fcal$ is a prefilter, so, since $X$ is compact, the exists $x_0 \in \bigcap[A \in \fcal] \bar A$. For all $x \in X$, if $x \leq x_0$, by maximality of $C$, $x \in C$, so $\downarrow\{x\} \in \fcal$ and thus $x_0 \leq x$. 

$x$ is thus a minimal element of $X$. 

\end{proof} 

\begin{thm} \label{thm  order is  equality} 

For all intuitionistic compact Hausdorff ordered topological space $X$, \linebreak $\Cn{X} \models v \leq 1 \- (1 \- v)$ if and only if the order is the equality on $X$. 

\end{thm} 

\begin{proof} 

Let $(X, \leq)$ be an intuitionistic compact Hausdorff ordered topological space. 

If the order is equality on $X$, then $\Cn{X} = C^0(X)$ and thus $\Cn{X} \models v \leq 1 \- (1 \- v)$. 

Assume now that $\Cn{X} \models v \leq 1 \- (1 \- v)$. Let $Y$ be the set of minimal elements of $X$ and $x \in X$. Let us prove that $x \in Y$, and thus that every element of $X$ is minimal. 

For all $f \in \Cn{X}$ and all $z \leq x$, since $\downarrow\{z\}$ admits minimal elements (\bf{Lemma} \ref{lemme Y non vide}), $$(1 \- f)(z) = \bigvee[y \leq z] 1 \- f(y) \geq \bigvee[\ust{y \in Y}{y \leq z}] 1 \- f(y) = 1,$$ so $f(x) = (1 \- (1 \- f))(x) = \bigvee[z \leq x] 1 \- (1 \- f)(z) = 0$. By \bf{Lemma} \ref{Y fermé}, $Y$ is closed. Thus $Y$ is a downward closed closed subset and $\uparrow\{x\}$ is an upward closed closed subset, so, according to \bf{Lemma} \ref{Urysohn croissant}, $Y \cap \uparrow\{x\} \neq \emptyset$. Hence $x \in Y$. 

\end{proof} 

\begin{cor} \label{corollary order is equality} 

For all BC-algebra $A$, the order on $Sp(A)$ is the equality. 

\end{cor} 

\begin{proof} 

For all BC-algebra $A$, since $A$ is a Cauchy-complete Archimedean MC-algebra, according to \bf{Theorem} \ref{isomorphie MC-algèbres}, $A \simeq \Cn{Sp(A)}$ and $Sp(A)$ is compact and Hausdorff, so, according to \bf{Theorem} \ref{thm  order is  equality}, since $A \models v \leq 1 \- (1 \- v)$, the order on $Sp(A)$ is the equality. 
\end{proof} 

\color{blue} We now recall the theory $T_c$ of classical logic, given in \cite[p. 5]{benyaacovProofCompletenessContinuous2010}: 

\begin{enumerate}[label = (A\arabic*)] 

\item \label{A1} $a \- b \leq a$ 

\item \label{A2} $(c \- a) \- (c \- b) \leq b \- a$ 

\item \label{A3} $a \- (a \- b) \leq b \- (b \- a)$ 

\item \label{A4} $a \- b \leq \ng b \- \ng a$ 

\item \label{A5} $\frac{a}{2} \leq a \- \frac{a}{2}$ 

\item \label{A6} $a \- \frac{a}{2} \leq \frac{a}{2}$

\end{enumerate} 

\begin{thm} \label{equivalence to classical logic} 

The Archimedean models of $T_c$ are exactly the Archimedean models of $\T_\class$. 

\end{thm} 

\begin{proof} 

Since $\T_\class$ axiomatizes $[0,1] = USC(\{\emptyset\})$, and, according to \cite[Fact 4.4]{benyaacovProofCompletenessContinuous2010}, $\{[0,1]\}$ is complete for $T_c$, $T_c$ is a consequence of $\T_\class$ with Archimedeanity. 

Conversely, let $A$ be an Archimedean model of $\T_\class$. According to \bf{Corollary} \ref{corollary order is equality}, there exists an embedding $i\colon A \rightarrow \Cn{Sp(A)}$. We will prove that subtraction is calculated pointwise in $\Cn{Sp(A)}$. From this, we can easily deduce that \ref{A1}, \ref{A2}, \ref{A3}, \ref{A4}, \ref{A5} et \ref{A6} are satisfied by $\Cn{\Sp(A)}$, from which we can conlude that they are also satisfied by $A$. 

Let us now prove that subtraction is calculated pointwise in $\Cn{Sp(A)}$. For all $f$ and \linebreak $g \in \Cn{Sp(A)}$, $f \- g$ is the smallest non decreasing function above $x \mapsto f(x) \- g(x)$, which is non decreasing since the order on $Sp(A)$ is the equality, so $f \- g\colon x \mapsto f(x) \- g(x)$. 
\end{proof} 

\begin{cor} 

For all $n \in \nn$, for all terms $\phi$ and $\psi$ in  the language $\mathcal{L}$, $[0,1] \models \phi \leq \psi$ if and only if for all complete Boolean algebra $\bcal$ $USC(\bcal) \models \phi \leq \psi$. 

\end{cor} 

\begin{proof} 

Indeed, the class $\{[0,1]\}$ is complete for $T_c$ with Archimedeanity, so is complete for $\T_\class$ with Archimedeanity too, and thus, for all complete Boolean algebra $\bcal$ $USC(\bcal) \models \phi \leq \psi$. Reciprocally, $[0,1]$ is itself a BC-algebra. 
\end{proof} 

\begin{cor} 

The definition of $\odot$ given in \bf{Definition} \ref{def odot} coincide with \linebreak the naive one, i.e., for all complete Boolean algebra $\bcal$ and all $f$ and $g \in USC(\bcal)$, \linebreak $j\left(\frac{f}{2} \+ \frac{g}{2}\right) = 1 \- ((1 \- f) \+ (1 \- g))$. 

\end{cor} \color{black} 

\subsection{Cut Admissibility} 

The language for sequent calculus in the intuitionistic case is the same as the one for \InCFLew, recalled in \it{Figure} \ref{InLJK}. The language for the structures and formulas of sequent calculus in the involutive case, and its correspondence with the language $\mathcal{L}$ are given in \it{Figure} \ref{Bool Correspondence between structure symbols and L}. As for the involutive case, we also consider the language $\mathcal{L}_\inv'$ which is $\mathcal{L}_\inv$ with the symbols of the \it{Figure} \ref{Bool Correspondence between the new structure symbols and L}. We keep Notation \ref{notations coupures}. 

\begin{figure}[!ht] 

$$\begin{array}{|c|c|c||c|c|} 
\hline 
\multicolumn{3}{|c||}{\text{Left interpretation}}& \multicolumn{2}{c|}{\text{Right interpretation}}\\ 
\hline 
\text{structures}& \text{formulas} & \text{algebraic correspondent} & \text{formulas} & \text{algebraic correspondent}\\ 
\hline 
,& \+& \+& \odot& 1 - ((1 - \_) \+ (1 - \_))\\ 
\hline 
\epsilon& \ul 0& \ul 0& \ul 1& \ul 1\\ 
\hline 
\circ_2& 2& 2& \frac{\cdot}{2}& \frac{\cdot}{2}\\ 
\hline 
\bullet_2& j_\ast& j_\ast& j& j\\ 
\hline 
\circ_\alpha& \alpha& \alpha& \blacksquare_\alpha& 2v \wedge j_\ast(v)\\ 
\hline 
\ng[\blank]& \ng& \ng& \ng& \ng\\ 
\hline 
\end{array}$$ 

\caption{Correspondence between structure symbols and $\mathcal{L}$} \label{Bool Correspondence between structure symbols and L} 

\end{figure} 

\begin{figure}[!ht] 

$$\begin{array}{|c|c|c||c|c|} 
\hline 
\multicolumn{3}{|c||}{\text{Left interpretation}}& \multicolumn{2}{c|}{\text{Right interpretation}}\\ 
\hline 
\text{structures}& \text{formulas} & \text{algebraic correspondent} & \text{formulas} & \text{algebraic correspondent}\\ 
\hline 
\bullet_\alpha& \alpha& \alpha& \blacksquare_\alpha& 2v \wedge j_\ast(v)\\ 
\hline 
\bs[\blank]& \bs& \bs& \bs& \bs\\ 
\hline 
\end{array}$$ 

\caption{Correspondence between the new structure symbols and $\mathcal{L}$} \label{Bool Correspondence between the new structure symbols and L} 

\end{figure} 

Here again, we use a system similar to \InMGL with several modalities but only one negation. We still work with the system \InMGL($\circ_2, \bullet_2, \circ_\alpha)$ (\it{Figure} \ref{inv Introduction Rules for Modalities}). In addition to the rules of $\InMGL(\circ_2,\, \bullet_2,\, \circ_\alpha)$ (\it{Figure} \ref{inv Introduction Rules for Modalities}), we add the structural rules given in \textit{Figure} \ref{InLJK} and call the total system \InLJK. 

\begin{figure}[!ht] 

\setlength{\arraycolsep}{5 pt} 

\adjusttopage{$\begin{array}{|ccc|} 

\hline 

[\ref{com monoid}a] \seq{{\Gamma, \Delta \vdash \Theta}}{{\Delta, \Gamma \vdash \Theta}} & [\ref{com monoid}b] \seq{{\Gamma, (\Delta, \Pi) \vdash \Theta}}{{(\Gamma, \Delta), \Pi \vdash \Theta}} & [\ref{com monoid}c] \seq[=]{{\epsilon, \Gamma \vdash \Theta}}{{\Gamma \vdash \Theta}} \\ 

[\ref{2 >=}] \seq{\Gamma \vdash \Theta}{\circ_2 \Gamma \vdash \Theta} & [\ref{j* >=}] \seq{\Gamma \vdash \Theta}{\bullet_2 \Gamma \vdash \Theta} & [\ref{bounded lattice}] \seq{\epsilon \vdash \Theta}{\Gamma \vdash \Theta} \\ 

[\ref{2v <=}] \seq{\circ_2 \Gamma \vdash \Theta, \circ_2 \Delta \vdash \Theta}{{\Gamma, \Delta \vdash \Theta}} & \blue{[\ref{int 2v >=}] \seq{{\Gamma, \Gamma \vdash \Theta}}{{\circ_2 \Gamma \vdash \Theta}}} & \blue{[\ref{bounded lattice} \et \ref{2j*(0) >= 1}]} \seq{{}}{\epsilon_1 \vdash \Theta} \\ 

& [(\hyperlink{def j*}{4.b})] \blue{\seq[=]{\bullet_2 (\circ_2 \Gamma {,} \Delta) \vdash \Theta}{{\Gamma , \bullet_2 \Delta \vdash \Theta}}} & \\ 

\hypertarget{4.d}{[\ref{alpha beta v >= v}a]} \seq{\Gamma \vdash \Theta}{\circ_\alpha \circ_2 \Gamma \vdash \Theta} & [\ref{alpha beta v >= v}b] \seq{\Gamma \vdash \Theta}{\circ_\alpha \bullet_2 \Gamma \vdash \Theta} & [\ref{alpha beta v <=  v}] \seq{\circ_\alpha \circ_2 \Gamma \vdash \Theta, \circ_\alpha \bullet_2 \Gamma \vdash \Theta}{\Gamma \vdash \Theta} \\ 

[\ref{beta alpha v >= v}a] \seq{\Gamma \vdash \Theta}{\circ_2 \circ_\alpha \Gamma \vdash \Theta} & [\ref{beta alpha v >= v}b] \seq{\Gamma \vdash \Theta}{\bullet_2 \circ_\alpha \Gamma \vdash \Theta} & [\ref{beta alpha v <=  v}] \seq{\circ_2 \circ_\alpha \Gamma \vdash \Theta, \bullet_2 \circ_\alpha \Gamma \vdash \Theta}{\Gamma \vdash \Theta} \\ 

[\ref{v<=d}] \seq{\Gamma \vdash \Theta}{{\epsilon_d, {\circ_j}^n (\Gamma {,} \epsilon_{1 - d}) \vdash \Theta}} & [\ref{1/2^n + 1/2^n}] \seq{{\epsilon_{\frac{1}{2^n}}, \epsilon_{\frac{1}{2^n}} \vdash \Theta}}{\epsilon_{\frac{1}{2^{n-1}}} \vdash \Theta} & [\hypertarget{Archi}{\ref{archimedean}}] \seq{\forall n \in \nn\; \epsilon_{\frac{1}{2^n}} \vdash \Theta}{\epsilon \vdash \Theta} \\ 

\multicolumn{3}{|c|}{[\ref{v>=d}] \seq{{\epsilon_{1 - \frac{1}{2^{n+1}}} ,  {\circ_j}^{n+1}\left(\Gamma , \epsilon_{\frac{1}{2^{n+1}}}\right) \vdash \Theta}, \ldots, {\epsilon_{1 - \frac{2^{n+1} - 2}{2^{n+1}}} ,  {\circ_j}^{n+1}\left(\Gamma , \epsilon_{\frac{2^{n+1} - 2}{2^{n+1}}}\right) \vdash \Theta}}{{\Gamma , \epsilon_{\frac{1}{2^n}} \vdash \Theta}}}\\ 

\hline 

\end{array}$} 

\caption{\InLJK} \label{InLJK}  

\end{figure} 

By following the same steps as in subsection \ref{subsection Cut Admissibility}, since all the rules added to obtain \InLJK are analytic (cf. \bf{Definition} \ref{analytic rules annexes}), according to \textbf{Theorems} \ref{Completeness theorem annexes} and \ref{Cut Admissibility annexes}, the following two theorems are true. 

\begin{thm}[Completeness theorem] 

The class $BC$ is sound and complete for $\InLJK$. 

\end{thm} 

\begin{thm}[Cut Admissibility theorem] 

In the system \InLJK, for all formulas $a_1, \ldots, a_n \et b$ and $\{, \, , \circ_2, \bullet_2, \circ_\alpha, \epsilon\}$-term $G$ such that there exists a deduction of $G(a_1, \ldots, a_n) \vdash b$ using the cut rule, there exists a deduction of $G(a_1, \ldots, a_n) \vdash b$ not using the cut rule. 

\end{thm} 

\section{Annexes} \label{Annexes} 

\addtocounter{subsection}{1} 

The aim of the \bf{Annexes} is to prove the cut-admissibility theorem \ref{Cut Admissibility annexes} for it is useful to prove all other cut-admissibility theorems of this paper. 

Let $\F_m$ (resp. $\InFm$) be the set of formulas of the language $\L = \{\ul 1, \cdot, /, \setminus, \vee, \wedge, \lozenge, \blacksquare\}$ (resp. $\InL$ = $\{\cdot, \vee, \wedge, \ul 1, \lozenge, \blacklozenge, \bs, \neg\}$), where $\cdot$, $\vee$, $\wedge$, $\setminus$ and $/$ are binary function symbols, $\lozenge \et \square$ are unary function symbols, $\ul 1$ is a constant symbol and $\bs \et \neg$ are unary function symbols. We call variables in formulas \emph{propositional} variables. 

\begin{nota} 

We recall that the \emph{sequent space} $\Sqt$ is the algebra of $\{, , \circ\}$-terms over $\F_m$. We now see $\F_m$ as embedded in $\Sqt$. We denote by small greek letters the elements of $\Sqt$ and call them \emph{\sequences} and by latin letters the ones of $\F_m$. By capital greek letters, we mean a context. 

The \emph{involutive sequent space} $\InSqt$ is the quotient of the algebra of $\{, , \circ, \bullet, \bs[], \ng[], \epsilon\}$-terms over $\InFm$, where $,$ is a binary function symbol, $\circ$, $\bullet$, $\bs$, and $\ng$ are unary function symbols and $\epsilon$ is a constant symbol by the following relations: 

\begin{enumerate} 

\item $\ng[\bs[\gamma]] = \bs[\ng[\gamma]] = \gamma$ 

\item $\bs[(\odot \gamma)] = \odot \left(\bs[\gamma]\right)$ 

\item $\ng[(\odot \gamma)] = \odot \left(\ng[\gamma]\right)$ 

\item $\bs[(\gamma, \delta)] = \bs[\delta], \bs[\gamma]$ et $\ng[(\gamma, \delta)] = \ng[\delta], \ng[\gamma]$ 

\end{enumerate} 

where $\odot$ denotes either $\circ$ or $\bullet$. We finally define $\square a = \bs \lozenge \neg a$ and $\blacksquare = \bs \blacklozenge \neg a$, for all $a \in \InFm$ and the notations $\dotlozenge$ either $\lozenge$ or $\blacklozenge$ and $\dotblacksquare$ denote $\blacksquare$ when $\odot$ denotes $\circ$ and $\square$ when $\odot$ denotes $\bullet$. 

\end{nota} 

We give the rules of \MGL and \InMGL, following respectively \cite{galatosCutEliminationStrong2010} and \cite{galatosResiduatedFramesApplications2012}. They consist respectively of the rules of \GL and \InGL, which are recalled in \it{Figures} \ref{GL Rules} and \ref{InGL Rules}, to which are added the rules of \it{Figures} \ref{MGL Rules} and \it{Figures} \ref{InMGL Rules}. 

\begin{figure}[!ht] 

\begin{center} 

\adjusttopage{$\begin{array}{|cccc|} 

\hline 

\seq{{{\Gamma[a,b] \vdash c}}}{{{\Gamma[a \cdot b] \vdash c}}}\target{[L.]} &\seq{{\gamma \vdash a}, {\delta \vdash b}}{{\gamma, \delta \vdash a \cdot b}}\target{[R.]} &\seq{{\Gamma[b] \vdash c}, {\gamma \vdash a}}{{\Gamma[b / a, \gamma] \vdash c}}\target{[L/]} &\seq{{{\gamma, a \vdash b}}}{\gamma \vdash b / a}\target{[R/]} \\ 

\seq{{\Gamma[\epsilon] \vdash a}}{{\Gamma[1] \vdash a}}\target{[L1]} &\seq{{}}{\epsilon \vdash \ul 1}\target{[R1]} &\seq{{\gamma \vdash a}, {\Gamma[b] \vdash c}}{{\Gamma[\gamma, a \setminus b] \vdash c}}\target{[L\setminus]} &\seq{{{a, \gamma \vdash b}}}{\gamma \vdash a \setminus b}\target{[R\setminus]} \\ 

\seq{{\Gamma[a] \vdash c}, {\Gamma[b] \vdash c}}{\Gamma[a \vee b] \vdash c}\target{[L\vee]} &\seq{{\gamma \vdash a}}{\gamma \vdash a \vee b}\target{[R\vee_1]} &\seq{{\gamma \vdash b}}{\gamma \vdash a \vee b}\target{[R\vee_2]} & \seq{{}}{a \vdash a} \target{[Id]}\\ 

\seq{{\Gamma[a] \vdash c}}{\Gamma[a \wedge b] \vdash c}\target{[L\wedge]} &\seq{{\Gamma[b] \vdash c}}{\Gamma[a \wedge b] \vdash c}\target{[R\wedge_1]} &\seq{\gamma \vdash a, \gamma \vdash b}{\gamma \vdash a \wedge b}\target{[R\wedge_2]} & \\ 

\hline 

\end{array}$} 

\end{center} 

\caption{\GL rules} \label{GL Rules} 

\end{figure} 

\begin{figure}[!ht]

\begin{center} 

$\begin{array}{|cccc|} 

\hline 

\seq{\Gamma[\circ a] \vdash b}{{\Gamma[\lozenge a] \vdash b}}\target{[L\lozenge]} &\seq{\gamma \vdash a}{\circ \gamma \vdash \lozenge a}\target{[R\lozenge]} &\seq{{\Gamma[a] \vdash b}}{{\Gamma[\circ \blacksquare a] \vdash b}}\target{[L\blacksquare]} &\seq{{\circ \gamma \vdash a}}{\gamma \vdash \blacksquare a}\target{[R\blacksquare]} \\ 

\hline 

\end{array}$ 

\end{center} 

\caption{\MGL rules} \label{MGL Rules} 

\end{figure} 

\begin{figure}[!ht] 

\begin{center} 

\adjusttopage{$\begin{array}{|cccc|} 

\hline 

\seq{{{a, b \vdash \delta}}}{{{a \cdot b \vdash \delta}}} \target{[InvL.]}[[L.]]&\seq{{\gamma \vdash a}, {\delta \vdash b}}{{\gamma, \delta \vdash a \cdot b}} \target{[InvR.]}[[R.]]&\seq{{}}{\gamma \vdash \ul 1} \target{[InvR1]}[[R1]]&\seq{{\epsilon \vdash \delta}}{{1 \vdash \delta}} \target{[InvL1]}[[L1]]\\ 

\seq{{a \vdash \delta}, {b \vdash \delta}}{a \vee b \vdash \delta} \target{[InvLv]}[[L\vee]]&\seq{{\gamma \vdash a}}{\gamma \vdash a \vee b} \target{[InvRv1]}[[R\vee_1]]&\seq{{\gamma \vdash b}}{\gamma \vdash a \vee b} \target{[InvRv2]}[[R\vee_2]]&\seq[=]{{\gamma, \delta \vdash \beta}}{{\delta \vdash \bs[\gamma], \beta}} \target{[Invtild]}[[\bs]]\\ 

\seq{{a \vdash \delta}}{a \wedge b \vdash \delta} \target{[InvLw]}[[L\wedge]]&\seq{{b \vdash \delta}}{a \wedge b \vdash \delta} \target{[InvRw1]}[[R\wedge_1]]&\seq{\gamma \vdash a, \gamma \vdash b}{\gamma \vdash a \wedge b} \target{[InvRw2]}[[R\wedge_2]]&\seq[=]{{\gamma, \delta \vdash \beta}}{{\gamma \vdash \beta, \ng[\delta]}} \target{[Invneg]}[[\ng]]\\ 

\seq{{\bs[a] \vdash \delta}}{\bs a \vdash \delta} \target{[Ltild]}[[L\bs]]&\seq{{\gamma \vdash \bs[a]}}{\gamma \vdash \bs a} \target{[Rtild]}[[R\bs]]&\seq{{\ng[a] \vdash \delta}}{\ng a \vdash \delta} \target{[Lneg]}[[L\ng]]&\seq{{\gamma \vdash \ng[a]}}{\gamma \vdash \ng a} \target{[Rneg]}[[R\ng]]\\ 

\multicolumn{4}{|c|}{\seq{{}}{a \vdash a} \target{[InvId]}[[Id]]} \\ 

\hline 

\end{array}$} 

\end{center} 

\caption{\InGL rules} \label{InGL Rules} 

\end{figure} 

\begin{figure}[!ht] 

\begin{center} 

$\begin{array}{|ccc|} 

\hline 

\seq{\odot a \vdash \delta}{{\dotlozenge a \vdash \delta}} \target{[L\dotlozenge]}&\seq{\gamma \vdash a}{\odot \gamma \vdash \dotlozenge a} \target{[R\dotlozenge]}&\seq[=]{{\circ \gamma \vdash \delta}}{\gamma \vdash \bullet \delta} \target{[\circ/\bullet]}\\ 

\hline 

\end{array}$ 

\end{center} 

\caption{\InMGL rules} \label{InMGL Rules}  

\end{figure} 

{\remark the following rules can easily be deduced from the ones of \hyperlink{InMGL rules}{\InMGL}: 

\begin{center} $\begin{array}{ccccc} 
\seq[=]{\gamma \vdash \ng[\delta]}{\delta \vdash \bs[\gamma]}\target{[\bs/\ng]} &\seq{{a \vdash \delta}}{{\square a \vdash \circ \gamma}} &\seq{{\gamma \vdash \circ a}}{\gamma \vdash \square a} &\seq{{a \vdash \delta}}{{\blacksquare a \vdash \bullet \gamma}} &\seq{{\gamma \vdash \bullet a}}{\gamma \vdash \blacksquare a} 
\end{array}$ \end{center} 

To these rules can be added the cut rules for $\MGL$ and $\InMGL$: $$\seq{\gamma \vdash a, \Gamma[a] \vdash b}{\Gamma[\gamma] \vdash b} \et \seq{\gamma \vdash \delta, \delta \vdash \beta}{\gamma \vdash \beta} \hypertarget{Cut}{\texte{(Cut)}}$$}

We now discuss in parallel the involutive and non involutive cases. 

\begin{defi} \label{analytic rules annexes} \it{Non involutive structural and analytic rules:} 

We denote by $\Upsilon$ some terms on $\Sqt$ for the language $\L$ and $A$ terms for $\L$ on $\F_m$. 

A \emph{structural rule} is a rule $r = \seq{\Gamma[\Upsilon_1] \vdash A_1 , \ldots , \Gamma[\Upsilon_n] \vdash A_n}{\Gamma[\Upsilon_0] \vdash A_0}$ such that no symbol other than "$,$", "$\circ$", or "$\epsilon$" appears in any $\Upsilon_i$ nor $A_i$. 

A rule be said \emph{analytic} when it is a structural rule of the form \linebreak $\seq{\Gamma[\Upsilon_1] \vdash A , \ldots , \Gamma[\Upsilon_n] \vdash A}{\Gamma[\Upsilon_0] \vdash A}$ satisfying: 

\begin{enumerate} 

\item[] \hypertarget{Linearity}{Linearity :} $A$ is a formula variable and the variables of $\Upsilon_0$ are distinct. 

\item[] \hypertarget{Separation}{Separation :} $A$ doesn't appear in $\Upsilon_0$. 

\item[] \hypertarget{Inclusion}{Inclusion :} The variables of the $\Upsilon_i$'s are among the ones of $\Upsilon_0$. 

\end{enumerate} 

\it{Involutive structural and analytic rules:} 

We now denote by $\Upsilon$ and $\Psi$ some terms on $\InSqt$ for the language $\InL$. 

A \emph{structural rule} is a rule $r = \seq{\Upsilon_1 \vdash \Psi_1 , \ldots , \Upsilon_n \vdash \Psi_n}{\Upsilon_0 \vdash \Psi_0}$ such that no symbol other than "$,$", "$\circ$", "$\bullet$", $\epsilon$, $\bs[\_]$ and $\ng[\_]$ appears in any $\Upsilon_i$ nor $\Psi_i$. 

A rule be said \emph{analytic} when it is a structural rule of the form $\seq{\Upsilon_1 \vdash \Psi , \ldots , \Upsilon_n \vdash \Psi}{\Upsilon_0 \vdash \Psi}$ satisfying: 

\begin{enumerate} 

\item[] \hypertarget{Linearity}{Linearity :} $\Psi$ is a structure variable and the variables of $\Upsilon_0$ are distinct. 

\item[] \hypertarget{Separation}{Separation :} $\Psi$ don't appear in $\Upsilon_0$. 

\item[] \hypertarget{Inclusion}{Inclusion :} The variables of the $\Upsilon_i$'s are among the ones of $\Upsilon_0$. 

\item[] \hypertarget{Positivity}{Positivity :} No negation symbol appears in the rule. 

\end{enumerate} 

\end{defi} 

{\remark Following \cite{galatosResiduatedFramesApplications2012}, we could replace the positivity condition as follows: any $\Upsilon_i$ is ${,\, , \circ, \bullet, \epsilon}$-terms on variables negated an even numbers of times. Since we won't need this extension, we leave the minor changes in the proofs to the reader.} 

Let $R$ be a set of analytic structural rules of \MGL (\InMGL). 

\begin{defi} 

From here on, we can endow $\Sqt$ $(\InSqt)$ with a binary relation \linebreak $\vdash \subset \Sqt \times \F_m$ $(\vdash \subset \InSqt \times  \InSqt)$ which is the smallest relation satisfying the rules of \it{Figures} \ref{GL Rules} and~\ref{MGL Rules} (\it{Figures} \ref{InGL Rules} and ~\ref{InMGL Rules}) and the rules of $R$ and call the resulting structure $\Sqt_R$ $(\InSqt_R$). We denote by $\Sqt_{\cut,R}$ $(\InSqt_{\cut,R})$ the sets $\Sqt$ $(\InSqt)$ endowd with the smallest relation $\vdash_\cut$ satisfying the previous rules and \hyperlink{Cut}{(Cut)}. 

\end{defi} 

\begin{defi} 

We can define a preorder $\leq$ on $\Sqt_R$ by $\gamma \leq \delta$ if and only if $\forall a \in \F_m\; \forall \text{ context } \Gamma \linebreak \Gamma[\delta] \vdash a \Rightarrow \Gamma[\gamma] \vdash a$ and a preorder $\leq$ on $\InSqt_R$ by $\gamma \leq \delta$ if and only if $\forall \beta \in \InSqt \linebreak \delta \vdash \beta \Rightarrow \gamma \vdash \beta$. 

For all $X \subset \Sqt$, let $$j(X) = \{\gamma \in \Sqt \;|\; \forall a \in \F_m \et \text{context } \Gamma, \; (\forall \delta \in X \; \Gamma[\delta] \vdash a) \Rightarrow \Gamma[\gamma] \vdash a\}.$$ 

For all $X \subset \InSqt$, let $$j(X) = \{\gamma \in \InSqt \;|\; \forall \beta \in \InSqt, \; (\forall \delta \in X \; \delta \vdash \beta) \Rightarrow \gamma \vdash \beta\}.$$ 

\end{defi} 

\begin{lem}[{\cite[p. 279]{ciabattoniAlgebraicProofTheory2012}, \cite[p.1221]{galatosResiduatedFramesApplications2012}}] 

$j$ is a closure operator on the set of subsets of $\Sqt$ $(\InSqt)$. 

\end{lem} 

\begin{lem} \label{ce lemme} 

Every $j$-closed subset of $\Sqt_R$ is downward closed. 

\end{lem} 

We then define $\Sqt_{R,+}$ (resp. $\InSqt_{R,+}$) as the set of $j$-closed subsets of $\Sqt$ (resp. $\InSqt$) and $\fct[g][\F_m, \Sqt_{R,+}]{a, \{\gamma \in \Sqt \; | \; \gamma \vdash a\}}$ $\left(\texte{resp.} \fct[g][\F_m, \InSqt_{R,+}]{\delta, \{\gamma \in \Sqt \; | \; \gamma \vdash \delta\}}\right)$. We endow the set of subsets of $\Sqt$ with the following structure: \begin{enumerate} 

\item For all $X \et Y \subset \Sqt$, $X,Y = \{\gamma,\delta\; \gamma \in X \et \delta \in Y\}$. 

\item For all $X \subset \Sqt$, $\circ X = \{\circ \gamma, \; \gamma \in X\}$. 

\item $\epsilon = \{\epsilon\}$. 

\item For all $X \et Y \subset \Sqt$, $X \wedge Y = X \cap Y$, $X \vee Y = j(X \cup Y)$. 

\item For all $X \et Y \subset \Sqt$, $X \cdot Y = j(X,Y)$. 

\item For all $X \subset \Sqt$, $\lozenge X = j(\circ X)$. 

\item $\ul 1 = j(\epsilon)$.

\item For all $X \et Y \subset \Sqt$, $X \setminus Y = \{\gamma \in \Sqt \; | \; \forall \delta \in X \; \delta \cdot \gamma \in Y\}$, and \linebreak $X / Y = \{\gamma \in \Sqt \; | \; \forall \delta \in X \; \gamma \cdot \delta \in Y\}$ 

\item For all $X \subset \Sqt$, $\blacksquare X = \{\gamma \in \Sqt \; | \; \circ \gamma \in X\}$. 

\end{enumerate} 

This induces an $\L$-structure on $\Sqt_{R,+}$. 

\begin{lem} \label{remarque} 

$j$ is a $\{, , \circ, \epsilon\}$-morphism from the subsets of $\Sqt$ to $\Sqt_{R,+}$. 

\end{lem} 

\begin{proof} 

Claim: $j$ is non-decreasing and, for all $X \et Y \subset \Sqt$, $X \subset j(X)$, $j^2(X) = j(X)$, \linebreak $\circ j(X) \subset j(\circ X)$ and $j(X),j(Y) \subset j(X,Y)$. 

The facts that $j$ is non-decreasing and, for all $X \subset \Sqt$, $X \subset j(X)$ are clear. 

Let $X \et Y \subset \Sqt$. Let $\gamma \in j^2(X)$.  For all context $\Gamma$ and $a \in \F_m$ such that for all $\delta \in X$ $\Gamma[\delta] \vdash a$, and for all $\beta \in j(X)$, by definition of $j$, $\Gamma[\beta] \vdash a$, so, by definition of $j$ again, $\Gamma[\gamma] \vdash a$. Thus $\gamma \in j(X)$. 

Let $\gamma \in j(X)$. For all context $\Gamma$ and $a \in \F_m$ such that for all $\delta \in \circ X$ $\Gamma[\delta] \vdash a$, and for all $\beta \in X$, $\Gamma[\circ \beta] \vdash a$, so, by definition of $j$, taking $\Gamma[\circ \_]$ as context, $\Gamma[\circ \gamma] \vdash a$. Thus $\circ \gamma \in j(\circ X)$. 

Let $\gamma \in j(X) \et \delta \in j(Y$). For all context $\Gamma$ and $a \in \F_m$ such that for all $\eta \in X,Y$ $\Gamma[\eta] \vdash a$, and for all $\alpha \in X \et \beta \in Y$, $\Gamma[\alpha,\beta] \vdash a$, so, by definition of $j$, taking $\Gamma[\_,\beta]$ as a context $\Gamma[\gamma,\beta] \vdash a$ and then, taking $\Gamma[\gamma,\_]$ as a context, $\Gamma[\gamma,\delta] \vdash a$. Thus $\gamma,\delta \; \in j(X,Y)$, which concludes the proof of the claim. 

Thanks to the claim: 

$j(X,Y) \subset j(j(X),j(Y)) = j(X) \cdot j(Y)$ and $j(X) \cdot j(Y) = j(j(X),j(Y)) \subset j^2(X,Y) = j(X,Y)$, so \linebreak $j(X,Y) = j(X) \cdot j(Y)$. $j(\circ X) \subset j(\circ j(X)) = \lozenge j(X)$ and $\lozenge j(X) = j(\circ j(X)) \subset j^2(\circ X) = j(\circ X)$, so $j(\circ X) = \lozenge j(X)$. Finally, $\ul 1 = j(\epsilon)$. 
\end{proof} 

We define the operators on $\InSqt_{R,+}$ $\vee$, $\wedge$, $\ul 1$, $\cdot$ and $\lozenge$ and $\blacklozenge$ in the same way. We \linebreak also define, for all $X \in \InSqt_{R,+}$, $\bs X = \{\gamma \in \InSqt \;|\; \forall \delta \in X\; \gamma \vdash \bs[\delta]\}$ and \linebreak $\ng X = \{\gamma \in \InSqt \;|\; \forall \delta \in X\; \gamma \vdash \ng[\delta]\}$. 

\begin{lem} 

The set of $j$-closed subsets of $\Sqt$ $(\InSqt)$ is closed for these operations. Moreover, $\blacksquare$ is right adjoint to $\lozenge$ on $\Sqt_{R,+}$, for all $X \in \InSqt_{R,+}$, $\bs \ng X = \ng \bs X = X$ and $\bs \lozenge \ng$ is right adjoint to $\blacklozenge$ and $\bs \blacklozenge \ng$ is right adjoint to $\lozenge$. 

\end{lem} 

\begin{proof} According to \cite[Lemma 5.4]{ciabattoniAlgebraicProofTheory2012} and \cite[Corollary 4.3]{galatosResiduatedFramesApplications2012}, we only have to prove that : 

\begin{enumerate} 

\item $\blacksquare$ is well defined on $\Sqt_{R,+}$ and right adjoint to $\lozenge$ on it 

\item $\bs \lozenge \ng$ is right adjoint to $\blacklozenge$ and $\bs \blacklozenge \ng$ is right adjoint to $\lozenge$. 

\end{enumerate} 

\begin{enumerate} 

\item Let $X \subset \Sqt$ be a $j$-closed set and let us show that $j(\blacksquare X) \subset \blacksquare X$. Let $\gamma \in j(\blacksquare X)$. 

By definition of $j$, for all $a \in \F_m \et \text{context } \Delta$, if for all $\delta \in \Sqt$ such that $\circ \delta \in X$ $\Delta[\delta] \vdash b$, then $\Delta[\gamma] \vdash a$. Thus, for all $a \in \F_m \et \texte{context} \Gamma$, for all $\delta \in \Sqt$ such that $\circ \delta \in X$, since $\Gamma[\circ \delta] \vdash a$, taking $\Delta = \Gamma[\circ \_]$, $\Gamma[\circ \gamma] \vdash a$. Hence $\circ \gamma \in j(X) = X$, ie $\gamma \in \blacksquare X$. 

\item Since, by \cite[Lemma 4.2 (i)]{galatosResiduatedFramesApplications2012}, for all $W \et Z \in \Sqt_{R,+}$, $W \subset \bs Z \Leftrightarrow Z \subset \ng W$, for all $X \et Y \in \InSqt_{R,+}$, \begin{align*} X \subset \bs \lozenge \ng Y &\Leftrightarrow \lozenge \ng Y \subset \ng X\\ &\Leftrightarrow \{\circ \gamma, \; \gamma \in \ng Y\} \subset \ng X&& (\texte{def. of} \lozenge)\\ 
&\Leftrightarrow \forall \gamma \in \ng Y\; \forall \delta \in X\; \circ \gamma \vdash \ng[\delta]&& (\texte{def. of} \ng)\\ 
&\Leftrightarrow \forall \delta \in X\; \forall \gamma \in \InSqt\; (\forall \alpha \in Y\; \gamma \vdash \ng[\alpha] \Rightarrow \circ \gamma \vdash \ng[\delta])&& (\texte{def. of} \ng)\\ 
&\Leftrightarrow \forall \delta \in X\; \forall \gamma \in \InSqt\; (\forall \alpha \in Y\; \alpha \vdash \bs[\gamma] \Rightarrow \delta \vdash \bs[(\circ \gamma)])&& \link{[\bs/\ng]}\\ 
&\Leftrightarrow \forall \delta \in X\; \forall \gamma \in \InSqt\; (\forall \alpha \in Y\; \alpha \vdash \bs[\gamma] \Rightarrow \bullet \delta \vdash \bs[\gamma])&& \link{[\circ/\bullet]}\\ 
&\Leftrightarrow \forall \delta \in X\; \forall \gamma \in \InSqt\; (\forall \alpha \in Y\; \alpha \vdash \gamma \Rightarrow \bullet \delta \vdash \gamma)\\ 
&\Leftrightarrow \blacklozenge X \subset Y&& (\texte{def. of} \blacksquare).\end{align*}

The adjunction between $\lozenge$ and $\bs \blacklozenge \ng$ can be proven in the same way. 

\end{enumerate} 
\end{proof} 

We define $f$ as the unique $\L$-morphism from $\F_m$ to $\Sqt_{R,+}$ $(\InSqt_{R,+})$ such that, for every propositional variable $p$, $f(p) = g(p)$. 

\begin{lem} \label{f is nice} 

For all $a \in \F_m$, $a \in f(a) \subset g(a)$. 

\end{lem} 

\begin{proof} The \MGL case: 

We prove it by induction over the formulas. For all propositional variable $p$, \linebreak $p \in f(p) = \{\gamma \in \Sqt \;|\; \gamma \vdash p\}$. Using \cite[Theorem 5.11]{ciabattoniAlgebraicProofTheory2012}, there remains to prove that for all $a \in \F_m$, if $a \in f(a) \subset \{\gamma \in \Sqt \;|\; \gamma \vdash a\}$, then $\lozenge a \in f(\lozenge a) \subset \{\gamma \in \Sqt \;|\; \gamma \vdash \lozenge a\}$ and $\square a \in f(\square a) \subset \{\gamma \in \Sqt \;|\; \gamma \vdash \square a\}$. 

Let $a \in \F_m$ such that $a \in f(a) \subset \{\gamma \in \Sqt \;|\; \gamma \vdash a\}$. Since $a \in f(a)$, by $\link{[L\lozenge]}$, $\lozenge a \leq \circ a \in \lozenge f(a) = f(\lozenge a)$, so, by \bf{Lemma} \ref{ce lemme}, $\lozenge a \in f(\lozenge a)$. Moreover, let $\gamma \in f(\lozenge a) = \lozenge f(a)$. By $\link{[R\lozenge]}$, for all $\delta \in f(a) \subset g(a)$ we deduce $\circ \delta \vdash \lozenge a$, so, by definition of $\lozenge f(a)$, $\gamma \vdash \lozenge a$, i.e. $\gamma \in g(\lozenge a)$. 

For all $b \in \F_m$ and $\Gamma$ such that $\Gamma[a] \vdash b$, by $\link{[L\blacksquare]}$, $\Gamma[\circ \blacksquare a] \vdash b$, so $\circ \blacksquare a \leq a$, which gives, since $f(a)$ is downward closed and $a \in f(a)$, $\circ \blacksquare a \in f(a)$ and thus $\blacksquare a \in \blacksquare f(a) = f(\blacksquare a)$. Finally, for all $\gamma \in f(\blacksquare a) = \blacksquare f(a) \subset \blacksquare\{\delta \in \Sqt \;|\; \delta \vdash a\}$, $\circ \gamma \vdash a$, so, by $\link{[R\blacksquare]}$, $\gamma \vdash \blacksquare a$. Hence $f(\blacksquare a) \subset \{\gamma \in \Sqt \;|\; \gamma \vdash \blacksquare a\}$. 

The \InMGL case: 

Note that the two statements of the \MGL case remain valid with, mutatis mutandis, the same proof. So there only remains to prove that, for all $a \in \F_m$, if $a \in f(a) \subset g(a)$, then $\ng a \in \ng f(a) \subset g(\ng a)$ and $\bs a \in \bs f(a) \subset g(\bs a)$. 

This has been proven in \cite[Theorem 4.4]{galatosResiduatedFramesApplications2012}. 
\end{proof} 

For every $\{, , \circ, \bullet, \epsilon\}$-term $G$ (resp. $\{, , \circ, \bullet, \bs[], \ng[], \epsilon\}$-term), let's write $\ol G$ the corresponding $\{\cdot, \lozenge, \blacklozenge, \epsilon\}$-term (resp. $\{\cdot, \lozenge, \blacklozenge, \bs, \ng, \epsilon\}$-term), and for every $\gamma \in \Sqt$ (resp. $\InSqt$), $\ol \gamma$ the corresponding formula. 

\begin{defi} \label{def satisfaction} 

We say that an $\L$-structure (resp.$\InL$-structure) $A$ satisfies an analytic structural rule $r = \seq{\Gamma[\Upsilon_1] \vdash a , \ldots , \Gamma[\Upsilon_n] \vdash a}{\Gamma[\Upsilon_0] \vdash a} \in R$ if and only if $A$ satisfies the formula $r^\bullet$ where $\seq{{}}{}$ is replaced by $\Rightarrow$, $\vdash$ by $\leq$, variables for \sequences by propositional variables and $\epsilon$, $,$ and $\circ$ are respectively replaced by $\ul 1$, $\cdot$ and $\lozenge$ (and $\bullet$ by $\blacklozenge$). Define $R^\bullet = \{r^\bullet, \, r \in R\}$. 

\end{defi} 

We denote by $\T_\MGL$ the following theory: 

\begin{enumerate} 

\item $\{\ul 1, \cdot, /, \setminus, \wedge, \vee\}$ is a residuated lattice structure 

\item $\lozenge$ is left adjoint to $\blacksquare$. 

\end{enumerate} 

We denote by $\T_\InMGL$ the following theory: 

\begin{enumerate} 

\item $\{\ul 1, \cdot, /, \setminus, \wedge, \vee\}$ is a residuated lattice structure 

\item $\bs \lozenge \ng$ is right adjoint to $\blacklozenge$ and $\bs \blacklozenge \ng$ is right adjoint to $\lozenge$. 

\item $\bs \ng x = \ng \bs x = x$. 

\end{enumerate} 

\begin{lem}[(cf. \citup{ciabattoniAlgebraicProofTheory2012} \textbf{Lemma 5.20})] 

$\Sqt_{R,+}$ $(\InSqt_{R,+})$ is a model of $\T_\MGL$ \linebreak $(\T_\InMGL)$ that satisfies every $r \in R$. 

\end{lem} 

\begin{proof} 

According to \textbf{Lemma 5.20} of \citup{ciabattoniAlgebraicProofTheory2012}, $\Sqt_{R,+}$ is a model of every axiom of $\T_\MGL \cup R^\bullet$ where no $\lozenge \ou \square$ appears. Moreover, $\lozenge$ is left adjoint to $\blacksquare$ thanks to $\link{[L\blacksquare]}$ and $\link{[R\blacksquare]}$. 

Let $r = \seq{{\Gamma[G_1(\gamma_1, \ldots, \gamma_n)] \vdash a}, \ldots, {\Gamma[G_k(\gamma_1, \ldots, \gamma_n)] \vdash a}}{{\Gamma[G_0(\gamma_1, \ldots, \gamma_n)] \vdash a}}$ be an analytic rule, where the $G_i$s are $\{, , \circ, \epsilon\}$-terms satisfied by $\Sqt_R$. Let $X_0, \ldots, X_n \in \Sqt_{R,+}$, such that \linebreak $\ol{G_1}(X_1, \ldots, X_n) \subset X_0, \ldots, \ol{G_k}(X_1, \ldots, X_n) \subset X_0$ and $(\gamma_i)_{1 \leq i \leq n} \in \prod[i = 1][n] X_i$. 

For all $1 \leq i \leq k$, $G_i(X_1, \ldots, X_n) \subset j(X_0)$, so, for all $a \in \F_m$ context $\Gamma$ such that for all $\delta \in X_0$ $\Gamma[\delta] \vdash a$, $\Gamma[G_1(\gamma_1, \ldots, \gamma_n)] \vdash a, \ldots, \Gamma[G_k(\gamma_1, \ldots, \gamma_n)] \vdash a$, so, by $r$, $\Gamma[G_0(\gamma_1, \ldots, \gamma_n)] \vdash a$. Hence \linebreak $G_0(\gamma_1, \ldots, \gamma_n) \in j(X_0) = X_0$. Since the variables $\gamma_1, \ldots, \gamma_n$ are distinct, \linebreak \adjusttopage{$\{G_0(\gamma_1, \ldots, \gamma_n),\; \gamma_i \in X_i\} \subset X_0$. By \textbf{Lemma} \ref{remarque}, $\ol{G_0}(X_1, \ldots, X_n) = j(\{G_0(\gamma_1, \ldots, \gamma_n),\; \gamma_i \in X_i\})$,} so we can conclude that $\ol{G_0}(X_1, \ldots, X_n) \subset X_0$. 

Since the analytic rules of $\InMGL$ do not contain $\ng$ or $\bs$, the proof also works in the $\InMGL$ case. 
\end{proof} 

\begin{thm}[Completeness theorem] \label{Completeness theorem annexes} 

The class of models of $\T_\MGL \cup R^\bullet$ (resp. $\T_\InMGL \cup R^\bullet$) is sound and complete both for $\MGL \cup R$ and $\MGL \cup R \cup \{\hyperlink{Cut}{\text{(Cut)}}\}$ (resp. $\InMGL \cup R$ and $\InMGL \cup R \cup \{\hyperlink{Cut}{\text{(Cut)}}\}$). 

\end{thm} 

\begin{proof} 

\ul{Soundness:} 

Let $A$ be a model of $\T_\MGL \cup R^\bullet$, $\L$-morphism $v\colon \F_m \rightarrow A$. Let us define $\vdash_v$ as follows: for all $\gamma \in \Sqt$ and $b \in \F_m$, $\gamma \vdash_v b \Leftrightarrow v(\ol \gamma) \leq v(b)$. We show that $\vdash_v$ satisfies all the rules satisfied by $\vdash_\cut$ and so $\vdash_\cut \subset \vdash_v$, ie, for all $\gamma \in \Sqt$ and $b \in \F_m$, if $\gamma \vdash b$, then $v(\ol \gamma) \leq v(b)$. 

The fact that $\vdash_v$ satisfies the left and right introduction rules $\link{[L.]}$, $\link{[R.]}$, $\link{[L\vee]}$, $\link{[R\vee_1]}$, $\link{[R\vee_2]}$, $\link{[L\wedge]}$, $\link{[R\wedge_1]}$, $\link{[R\wedge_2]}$, $\link{[L1]}$ and $\link{[R1]}$ is an immediate consequence of $v(\bar\_)$ sending $,$ on $\cdot$, $\vee$ on $\vee$, $\wedge$ on $\wedge$ and $\epsilon$ on $\ul 1$. Let us prove that $\vdash_v$ satisfies $\link{[L.]}$ and $\link{[R.]}$ as examples. 

For all $a$, $b \et c \in \F_m$ and all context $\Gamma$, if $\Gamma[a,b] \vdash_v c$, then $v\left(\ol{\Gamma[a,b]}\right) \leq v(c)$, ie $v(\bar\Gamma[a \cdot b]) \leq v(c)$, ie $v(\Gamma[a \cdot b]) \leq v(c)$, ie $\Gamma[a \cdot b] \vdash_v c$. Hence $\vdash_v$ satisfies $\link{[L.]}$. 

For all $\gamma \et \delta \in \Sqt$ and $a \et b \in \F_m$, if $\gamma \vdash_v a \et \Delta \vdash_v b$, then $v(\bar \gamma) \leq v(a) \et v(\bar \delta) \leq v(b)$, so $v\left(\ol{\gamma, \delta}\right) = v(\bar \gamma) \cdot v(\bar \delta) \leq v(a) \cdot v(b) = v(a \cdot b)$, ie $\gamma, \delta \vdash_v a \cdot b$. Hence $\vdash_v$ satisfies $\link{[R.]}$. 

The fact that $\vdash_v$ satisfies $\link{[R/]}$ and $\link{[R\setminus]}$ is an immediate consequence of $/$ and $\setminus$ in $A$ being the right and left residual of $\cdot$ and $v$ sending being a morphism for these symbols. Let us prove $\link{[R/]}$, as $\link{[R\setminus]}$ can be proven in the same way. 

For all $\gamma \in \Sqt$ and $a \et b \in \F_m$ such that $\gamma, a \vdash_v b$, $v(\bar \gamma) \cdot v(a) \leq v(b)$, ie \linebreak $v(\bar \gamma) \leq v(b) / v(a) = v(b/a)$, ie $\gamma \vdash_v b/a$. Hence $\vdash_v$ satisfies $\link{[R/]}$. 

Finally, there remains to prove that $\vdash_v$ satisfies $\link{[L\setminus]}$ and $\link{[L/]}$. Since the two proofs are analogous, let us check that $\vdash_v$ satisfies $\seq{{\Gamma[b] \vdash c}, {\gamma \vdash a}}{{\Gamma[b / a, \gamma] \vdash c}}\link{[L/]}$. 

For this, we will use the fact that, for all context $\Gamma$, there exists $r_\Gamma \colon A \rightarrow A$ such that, for all $x \et y \in A$, $\ol\Gamma[x] \leq y \Leftrightarrow x \leq r_\Gamma[y]$, which can be proven by induction over the $\{, \, , \circ, \epsilon\}$-terms. 

Let $\Gamma$ be a $\{, \, , \circ, \epsilon\}$-term, $\gamma \in \Sqt$, $a$, $b$ and $c \in \F_m$ such that $\Gamma[b] \vdash_v c \et \gamma \vdash_v a$. \linebreak $v(\ol{\Gamma[b]}) \leq v(c)$ and $v(\ol \gamma) \leq v(a)$, so $\ol\Gamma[v(b)] \leq v(c)$, so $v(b) \leq r_\Gamma[v(c)]$ and thus \linebreak $v(b) / v(a) \cdot v(\ol\gamma) \leq v(b) / v(a) \cdot v(a) \leq v(b) \leq r_\Gamma[v(c)],$ ie $\ol\Gamma[v(\gamma) \cdot v(a) / v(b)] \leq v(c)$, which is equivalent to $v(\ol{\Gamma[\gamma, a / b]}) \leq v(c)$, itself equivalent to $\Gamma[\gamma, a / b] \vdash_v c$. Hence $\vdash_v$ satisfies $\link{[L/]}$. 

We have thus proven that $\vdash_v$ satisfies all the introduction rules of \MGL. 

Let $r \in R$. We can write $r = \seq{{\Gamma[G_1(\gamma_1, \ldots, \gamma_n)] \vdash a}, \ldots, {\Gamma[G_k(\gamma_1, \ldots, \gamma_n)] \vdash a}}{{\Gamma[G_0(\gamma_1, \ldots, \gamma_n)] \vdash a}}$. For all context $\Gamma$, $G_1, \ldots, G_k$ $\{, \, , \circ, \epsilon\}$-terms of arity $n$, $\gamma_1, \ldots, \gamma_n \in \Sqt$ and $a \in \F_m$ such that \linebreak $\Gamma[G_1(\gamma_1, \ldots, \gamma_n)] \vdash_v a, \ldots, \Gamma[G_k(\gamma_1, \ldots, \gamma_n)] \vdash_v a$, $$\ol\Gamma[\ol G_1(v(\ol\gamma_1), \ldots, v(\ol\gamma_n))] \leq v(a), \ldots, \ol\Gamma[\ol G_k(v(\ol\gamma_1), \ldots, v(\ol\gamma_n))] \leq v(a),$$ so, since $A$ satisfies $R^\bullet$, $\ol\Gamma[\ol G_0(v(\ol\gamma_1), \ldots, v(\ol\gamma_n))] \leq v(a)$, ie $\Gamma[G_0(\gamma_1, \ldots, \gamma_n)] \vdash_v a$. Hence $\vdash_v$ satisfies $r$. 

We finally prove that $\vdash_v$ satisfies the cut rule. For all context $\Gamma$, $\gamma \in \Sqt$, $a \et b \in \F_m$ such that $\gamma \vdash_v a \et \Gamma[a] \vdash_v b$, $v(\ol \gamma) \leq v(a) \et \ol\Gamma[v(a)] \leq v(b)$, so, since $\ol \Gamma$ is non-decreasing, $\ol\Gamma[v(\ol\gamma)] \leq v(b)$, ie $\Gamma[\gamma] \vdash_v b$. Thus, $\vdash_v$ satisfies the cut rule. 

This concludes the proof of soundness of models of $T_\MGL \cup R^\bullet$ with respect to $\MGL \cup R \cup \{\hyperlink{Cut}{\text{(Cut)}}\}$. We can prove soundness of models of $T_\InMGL \cup R^\bullet$ with respect to $\InMGL \cup R \cup \{\hyperlink{Cut}{\text{(Cut)}}\}$ in the same way. From this soundness, one can deduce the soundness of the models of $T_\MGL \cup R^\bullet$ (resp. $T_\InMGL \cup R^\bullet$) with respect to $\MGL \cup R$ (resp. $\InMGL \cup R$). 

\ul{Completeness:} 

Let us prove completeness for $\T_\MGL \cup R^\bullet$. The proof for $\T_\InMGL \cup R^\bullet$ is the same and left to the reader. 

Let $(a_1 , \ldots , a_n , b) \in {\F_m}^{n+1}$ and $G$ be a $\{, , \circ, \epsilon\}$-term such that $\ol G(a_1, \ldots, a_n) \Rightarrow b$ is true in $\T_\MGL \cup R^\bullet$. Since $\Sqt_{R,+}$ is a model of $\T_\MGL \cup R^\bullet$, $\ol G(f(a_1), \ldots, f(a_n)) \subset f(b)$. However, according to \textbf{Lemma} \ref{f is nice}, $f(b) \subset g(b)$ and for all $a \in \F_m$, $a \in f(a)$. Hence \linebreak $G(a_1, \ldots, a_n) \in \ol G(f(a_1), \ldots, f(a_n)) \subset g(b)$, ie $G(a_1, \ldots, a_n) \vdash b$. 
\end{proof} 

\begin{cor}[Cut Admissibility theorem] \label{Cut Admissibility annexes} 

For every $(a_1, \ldots, a_n, b) \in \F_m^{n+1}$ and $\{, \, , \circ, \epsilon\}$-term $G$ such that $G(a_1, \ldots, a_n) \vdash_\cut b$, $G(a_1, \ldots, a_n) \vdash b$. 

\end{cor} 

\begin{proof} 

For every $(a_1, \ldots, a_n, b) \in \F_m^{n+1}$ and $\{, \, , \circ, \epsilon\}$-term $G$ such that $G(a_1, \ldots, a_n) \vdash_\cut b$, for all model $A$ of $\T_\MGL \cup R^\bullet$ and $v\colon \F_m \rightarrow A$, thanks to the soundness with respect to $\MGL \cup R \cup \{\hyperlink{Cut}{\text{(Cut)}}\}$, \linebreak $\ol G(v(a_1), \ldots, v(a_n)) \subset v(b)$, so, thanks to completeness of models of $\T_\MGL \cup R^\bullet$ with respect to $\MGL \cup R$, $G(a_1, \ldots, a_n) \vdash b$. 

Once again, it works the same way for the involutive case. 
\end{proof} 

As a final remark, we can notice that adding a countable rule in analytic form still gives a system that enjoys cut elimination, soundness and completeness.

\end{document}